\documentclass[a4paper,11pt]{article} 
\usepackage[english]{babel}
\usepackage[a4paper]{geometry}
\usepackage{amsmath}
\usepackage{amsthm}
\usepackage{amssymb}
\usepackage{color}

\usepackage{hyperref}         %%%% click to contents
\def\bR{\mathbb{R}}

\def\bN{\mathbb{N}}
\def\NN{\mathbb{N}}
\def\bZ{\mathbb{Z}}
\def\cC{\mathcal{C}}

\def\cA{\mathcal{A}}
\def\cM{\mathcal{M}}

\def\cV{\mathcal{V}}

\def\cF{\mathcal{F}}
\def\cG{\mathcal{G}}
\def\cL{\mathcal{L}}
\def\cJ{\mathcal{J}}
\def\cN{\mathcal{N}}
\def\cE{\mathcal{E}}
\def\cK{\mathcal{K}}
\def\cH{\mathcal{H}}
\def\eps{\varepsilon}
\def\ph{\varphi}

\def\wt{\widetilde}

\def\indic{\hbox{\raise-2pt \hbox{\indbf 1}}}

\let\==\equiv

\let\0=\noindent

\def\*{{\hfill\break\null\hfill\break}}

\def\tende#1{\,\vtop{\ialign{##\crcr\rightarrowfill\crcr
             \noalign{\kern-1pt\nointerlineskip}
             \hskip3.pt${\scriptstyle #1}$\hskip3.pt\crcr}}\,}
\def\otto{\,{\kern-1.truept\leftarrow\kern-5.truept\to\kern-1.truept}\,}

\def\tr{{\rm tr}}

\def\Re{{\rm Re}\,}

\newtheorem{theorem}{Theorem}[section]  % use thm for %Theorems to keep numbering consistent
\newtheorem{cor}[theorem]{Corollary}
\newtheorem{prop}[theorem]{Proposition}
\newtheorem{lemma}[theorem]{Lemma}

\numberwithin{equation}{section}

\def\be{\begin{equation}}
\def\ee{\end{equation}}

\newcommand{\hc}{\mbox{h.c.}}

\let\a=\alpha \let\b=\beta    \let\g=\gamma     \let\d=\delta     
             \let\l=\lambda
                  \let\p=\pi        \let\r=\rho
\let\s=\sigma \let\t=\tau         \let\ph=\varphi   
   \let\o=\omega     
        \let\L=\Lambda

\begin{document}

\title{Bose-Einstein Condensation Beyond \\the Gross-Pitaevskii Regime}

\author{Arka Adhikari$^1$, Christian Brennecke$^2$, Benjamin Schlein$^3$ \\
\\
Department of Mathematics, Harvard University, \\
One Oxford Street, Cambridge MA 02138, USA$^{1,2}$ \\
\\
Institute of Mathematics, University of Zurich, \\
Winterthurerstrasse 190, 8057 Zurich, Switzerland$^3$}

\maketitle

\begin{abstract}
We consider N bosons in a box with volume one, interacting through a two-body potential with scattering length of the order $N^{-1+\kappa}$, for $\kappa>0$. Assuming that  $\kappa\in (0;1/43)$, we show that low-energy states exhibit Bose-Einstein condensation and we provide bounds on the expectation and on higher moments of the number of excitations. 
\end{abstract} 

\section{Introduction}
\label{sec:intro}

We consider systems of $N\in\NN$ bosons trapped in the box $\Lambda = [0;1]^3$ with periodic boundary conditions (the three dimensional torus with volume one) and interacting through a repulsive potential with scattering length of the order $N^{-1+\kappa}$, for $\kappa \in (0; 1/43)$. We are  interested in the limit of large $N$. The Hamilton operator has the form 
\begin{equation}\label{eq:HN} H_N = \sum_{i=1}^N -\Delta_{x_i} + \sum_{1\leq i<j\leq N} N^{2-2\kappa} V (N^{1-\kappa}(x_i -x_j)) \end{equation}
and acts on a dense subspace of $L^2_s (\Lambda^N)$, the Hilbert space consisting of functions in $L^2 (\L^N)$ that are invariant with respect to permutations of the $N\in\NN$ particles. Here we assume the interaction potential $V \in L^3 (\bR^3)$ to have compact support and to be non-negative, ie. $V(x) \geq 0$ for almost all $x \in \bR^3$. 

\smallskip

For $\kappa = 0$, the Hamilton operator (\ref{eq:HN}) describes bosons in the so-called Gross-Pitaevskii limit. This regime is frequently used to model trapped Bose gases observed in recent experiments. Another important regime is the thermodynamic limit, where $N$ bosons interacting through a fixed potential $V$ (independent of $N$) are trapped in the box $\Lambda_L = [0;L]^3$ and where the limits $N,L \to \infty$ are taken, keeping the density $\rho = N/ L^3$ fixed. After rescaling lengths (introducing new coordinates $x' = x / L$) the Hamilton operator of the Bose gas in the thermodynamic limit is given (up to a multiplicative constant) by (\ref{eq:HN}), with $\kappa = 2/3$. Choosing $0 < \kappa < 2/3$, we are interpolating therefore between the Gross-Pitaevskii and the thermodynamic limits. 

\smallskip

The goal of this paper is to show that low-energy states of (\ref{eq:HN}) exhibit Bose-Einstein condensation in the zero-momentum mode $\ph_0 \in L^2 (\Lambda)$ defined by $\ph_0 (x) = 1$ for all $x \in \L$ and to give bounds on the number of excitations of the condensate. To achieve this goal, it is convenient to switch to an equivalent representation of the bosonic system, removing the condensate and focusing instead on its orthogonal excitations. To this end, we notice that every $\psi_N \in L^2_s (\Lambda^N)$ can be uniquely decomposed as 
\[ \psi_N = \alpha_0 \ph_0^{\otimes N} + \alpha_1 \otimes_s \ph_0^{\otimes (N-1)} + \alpha_2 \otimes_s \ph_0^{\otimes (N-2)} + \dots + \alpha_N \]
where $\otimes_s$ denotes the symmetric tensor product and $\alpha_j \in L^2_\perp (\Lambda)^{\otimes_s j}$ for all $j = 0, \dots , N$, with $L^2_\perp (\Lambda)$ the orthogonal complement in $L^2 (\Lambda)$ of $\ph_0$. This observation allows us to define a unitary map $U_N : L^2_s (\Lambda^N) \to \cF_+^{\leq N} = \bigoplus_{j=0}^N L^2_\perp (\Lambda)^{\otimes_s j}$ by setting \begin{equation}\label{eq:def-UN} 
U_N \psi_N = \{ \alpha_0, \alpha_1, \dots , \alpha_N \}.\end{equation} 
The truncated Fock space $\cF_+^{\leq N} = \bigoplus_{j=0}^N L^2_\perp (\Lambda)^{\otimes_s j}$ is used to describe orthogonal excitations of the condensate (some properties of the map $U_N$ will be discussed in Section \ref{sec:fock} below). On $\cF_+^{\leq N}$, we introduce the number of particles  operator, defining $(\cN_+ \xi)^{(n)} = n \xi^{(n)}$ for every $\xi = \{ \xi^{(0)} , \dots \xi^{(N)} \} \in \cF_+^{\leq N}$. 

\smallskip

We are now ready to state our main theorem, which provides estimates of the expectation and on higher moments of the number of orthogonal excitations of the Bose-Einstein condensate for low-energy states of (\ref{eq:HN}). 

\begin{theorem}\label{thm:main} 
Let $ V\in L^3(\mathbb{R}^3)$ be pointwise non-negative and spherically symmetric. Let $\frak{a}_0 > 0$ denote the scattering length of $V$. Let $H_N$ be defined as in \eqref{eq:HN} with $0 < \kappa < 1/43$. Then, for every $\eps > 0$, there exists a constant $C>0$ such that 
\begin{equation}\label{eq:gs-est} \big| E_N -4\pi\mathfrak{a}_0 N^{1+\kappa}\big| \leq C N^{43\kappa + \eps} .\end{equation}
for all $N \in \bN$ large enough. 

Let $\psi_N \in L^2_s (\Lambda^N)$ with $\| \psi_N \| = 1$ and \begin{equation}\label{eq:variance} \langle \psi_N, (H_N - E_N)^2 \psi_N \rangle \leq \zeta^2,\end{equation} for a $\zeta > 0$. Then, for every $\eps > 0$ there exists a constant $C > 0$ such that  \begin{equation}\label{eq:BEC0} \langle U_N \psi_N , \cN_+ \, U_N \psi_N \rangle \leq C \left[ \zeta + \zeta^2 N^{13\kappa + \eps - 1} + N^{43 \kappa + 4\eps} \right] \end{equation}
for all $N \in \bN$ large enough. If moreover $\psi_N  = \chi (H_N \leq E_N + \zeta) \psi_N$, then for all $k \in \bN$ and all $\eps > 0$ there exists $C > 0$ such that  
\begin{equation}\label{eq:BEC} \langle U_N \psi_N , \cN_+^k \, U_N \psi_N \rangle \leq C \left[ N^{20\kappa + \eps}  \zeta^2  + N^{44 \kappa + 2 \eps} \right]^{k} \end{equation} 
for all $N \in \bN$ large enough. 
\end{theorem}

The convergence $E_N/4\pi \frak{a}_0 N^{1+\kappa} \to 1$, as $N \to \infty$, has been first established, for Bose gases trapped by an external potential, in \cite{LSY} (the choice $\kappa > 0$ corresponds, in the terminology of \cite{LSY}, to the Thomas-Fermi limit). 

\smallskip

It follows from (\ref{eq:BEC0}) that the one-particle density matrix $\gamma_N = \tr_{2,\dots , N} |\psi_N \rangle \langle \psi_N |$ associated with a normalized $\psi_N \in L^2_s (\L^N)$ satisfying 
(\ref{eq:variance}) is such that  
\begin{equation}\label{eq:BEC2}
\begin{split} 
 1- \langle \ph_0 , \gamma_N \ph_0 \rangle = \; &\frac{1}{N} \left[ N - \langle \psi_N, a^* (\ph_0) a(\ph_0) \psi_N \rangle \right] \\ = \; &\frac{1}{N} \langle U_N \psi_N , \cN_+ \, U_N \psi_N 
 \rangle\\  \leq \; & 
 C \left[ \zeta N^{-1} + \zeta^2 N^{13\kappa + \eps - 2} + N^{43 \kappa + 4\eps-1} \right] \end{split} \end{equation}
as $N \to \infty$. Here we used the formula $U_N  a^* (\ph_0) a(\ph_0) U_N = N - \cN_+$; see (\ref{eq:U-rules}) below. Eq. (\ref{eq:BEC2}) implies that low-energy states of (\ref{eq:HN}) exhibit  complete Bose-Einstein condensation, {\bf if $\kappa < 1/43$}. 

\smallskip

We remark that the estimate (\ref{eq:BEC}) follows, in our analysis, from a stronger bound controlling not only the number but also the energy of the excitations of the condensate. As we will explain in Section \ref{sec:quadra}, in order to estimate the energy of excitations in low-energy states, we first need to remove (at least part of) their correlations. If we choose, as we do in (\ref{eq:BEC}), $\psi_N  \in L^2_s (\Lambda^N)$ with $\| \psi_N \| = 1$ and $\psi_N  = \chi (H_N \leq E_N + \zeta) \psi_N$, we can introduce the corresponding renormalized excitation vector $\xi_N = e^B U_N \psi_N \in \cF_+^{\leq N}$, with the antisymmetric operator $B$ defined as in \eqref{eq:genBog} (the unitary operator $e^B$ will be referred to as a generalized Bogoliubov transformation). We will show in Section \ref{sec:proof}, that for every $k \in \bN$, there exists $C > 0$ such that 
\begin{equation}\label{eq:HNk} \langle \xi_N , (\cH_N + 1) ( \cN_+ + 1)^{2k} \xi_N \rangle \leq C \left[ N^{20\kappa + \eps}  \zeta^2  + N^{44 \kappa + 2 \eps} \right]^{2k+1} \end{equation}
for all $N$ large enough. Here $\cH_N = \cK + \cV_N$, where 
\begin{equation}\label{eq:KcVN}   
\cK = \sum_{p \in \Lambda_+^*} p^2 a_p^* a_p , \quad \text{and} \quad \cV_N = 
\frac{1}{2N} \sum_{\substack{p,q \in \Lambda_+^*, r \in \Lambda^* : \\ r \not = -p, -q}} N^{\kappa} \widehat{V} (r/N^{1-\kappa})  a_{p+r}^* a_q^* a_{q+r} a_p \end{equation}
are the kinetic and potential energy operators, restricted to $\cF_+^{\leq N}$ (here $\widehat{V}$ is the Fourier transform of the potential $V$, defined as in (\ref{eq:FoT})). Eq. (\ref{eq:BEC}) follows then from (\ref{eq:HNk}), because $\cN_+$ commutes with $\cH_N$, $\cN_+ \leq \cK \leq \cH_N$ and  because conjugation with the generalized Bogoliubov transformation $e^B$ does not change the number of particles substantially; see Lemma \ref{lm:Ngrow} (for $k \in \bN$ even, we also use simple interpolation). 

\smallskip

In the Gross-Pitaevskii regime corresponding to $\kappa = 0$ the convergence $\gamma_N \to |\ph_0 \rangle \langle \ph_0|$ has been first established in \cite{LS1,LS2,LSSY}  and later, using a different approach, in \cite{NRS} \footnote{Going through the proof of \cite[Theorem 5.1]{LSSY}, one can observe that the authors actually show that $1 - \langle \ph_0 , \gamma_N \ph_0 \rangle \leq C N^{-2/17}$.}. In this case (ie. $\kappa = 0$), the bounds (\ref{eq:gs-est}), (\ref{eq:BEC0}) and (\ref{eq:BEC}) with $\varepsilon = 0$ (which are optimal in their $N$-dependence) have been shown in \cite{BBCS}. Previously, they have been established in \cite{BBCS0}, under the additional assumption of small potential. A simpler proof of the results of \cite{BBCS0}, extended also to systems of bosons trapped by an external potential, has been recently given in \cite{NNRT}. The result of \cite{BBCS} was used in \cite{BBCS2} to determine the second order corrections to the ground state energy and the low-energy excitation spectrum of the Bose gas in the Gross-Pitaevskii regime. Note that our approach in the present paper could be easily extended to the case $\kappa = 0$, leading to the same bounds obtained in \cite{BBCS}. We exclude the case $\kappa = 0$ because we would have to modify certain definitions, making the notation more complicated (for example, the sets $P_H$ in (\ref{eq:PellHPellL1}) and $P_L$ in (\ref{eq:defPHPL}) would have to be defined in terms of cutoffs independent of $N$). 
 
\smallskip

The methods of \cite{LS1,LS2,LSSY} can also be extended to show Bose-Einstein condensation for low-energy states of (\ref{eq:HN}), for some $\kappa > 0$. In fact, following the proof of \cite[Theorem 5.1]{LSSY}, it is possible to show that, for a normalized $\psi_N \in L^2_s (\Lambda^N)$ with $\| \psi_N \| =1$ and such that $\langle \psi_N, H_N \psi_N \rangle \leq E_N + \zeta$, the expectation of the number of excitations is bounded by 
\begin{equation}\label{eq:BEC-LSSY} \langle U_N \psi_N , \cN_+ U_N \psi_N \rangle \leq C \left[ N^{\frac{15 +20\kappa}{17}} + \zeta \right]   \end{equation}
which implies complete Bose-Einstein condensation for low-energy states, for all $\kappa < 1/10$. For sufficiently small $\kappa > 0$, Theorem \ref{thm:main} improves (\ref{eq:BEC-LSSY}) because it gives a better rate\footnote{For $\kappa > 0$, the rate (\ref{eq:BEC}) is not expected to be optimal. Bogoliubov theory predicts that the number of excitations of the Bose-Einstein condensate in a Bose gas with density $\rho$ is of the order $N \rho^{1/2}$; see \cite{bog}. In our regime, this corresponds to $N^{3\kappa/2}$ excitations.} (if $\kappa < 15/711$) and because, through (\ref{eq:BEC}), it also provides (under stronger conditions on $\psi_N$) bounds for higher moments of the number of excitations $\cN_+$.

\smallskip

In \cite{FS}, in a slightly different setting, the authors obtain a bound of the form (\ref{eq:BEC}) for $k=1$, for the choice $\kappa = 1/(55+1/3) $ (for normalized $\psi_N \in L^2_s (\L^N)$ that satisfy $ \langle \psi_N, H_N\psi_N\rangle\leq E_N + \zeta$). They use this 
result to show a lower bound on the ground state energy of the dilute Bose gas in the thermodynamic limit matching the prediction of Lee-Yang and Lee-Huang-Yang \cite{LY, LHY}. 

\smallskip

After completion of our work, two more papers have appeared whose results are related with Theorem \ref{thm:main}. Based on localization arguments from \cite{BFS, FS}, a bound for the expectation of $\cN_+$ in low-energy states has been shown in \cite{F}, establishing Bose-Einstein condensation for all $\kappa < 2/5$  (as pointed out there, using a refined analysis similar to that of \cite{FS}, the range of $\kappa$ can be slightly improved). On the other hand, following an approach similar to \cite{BBCS0}, but with substantial simplifications (partly due to the fact that the author works in the grand canonical, rather than the canonical, ensemble), a new proof of Bose-Einstein condensation was obtained in \cite{H}, in the Gross-Pitaevskii regime, under the assumption of small potential. There is hope that the approach of \cite{H} can be extended beyond the Gross-Pitaevskii regime, providing a simplified proof of Theorem \ref{thm:main}, potentially allowing for larger values of $\kappa$. 

\smallskip

The derivation of the bounds (\ref{eq:BEC0}), (\ref{eq:BEC}), (\ref{eq:HNk}) is crucial to resolve the low-energy spectrum of the Hamiltonian (\ref{eq:HN}). The extension of estimates on the ground state energy and on the excitation spectrum obtained in \cite{BBCS2} for the Gross-Pitaevskii limit, to regimes with $\kappa > 0$ small enough will be addressed in a separate paper \cite{BCS}, using the results of Theorem \ref{thm:main}. With our techniques, it does not seem possible to obtain such precise information on the spectrum of (\ref{eq:HN}) using only previously available bounds like (\ref{eq:BEC-LSSY}). 

\smallskip

Let us now briefly explain the strategy we use to prove Theorem \ref{thm:main}. The first part of our analysis follows closely \cite{BBCS}. We start in Section \ref{sec:fock} by introducing the excitation Hamiltonian $\cL_N = U_N H_N U_N^*$, acting on the truncated Fock space $\cF_+^{\leq N}$; the result is given in (\ref{eq:cLN}), (\ref{eq:cLNj}). The vacuum expectation $\langle \Omega, \cL_N \Omega \rangle = N^{1+\kappa} \widehat{V} (0)/2$ is still very far from the correct ground state energy of $\cL_N$ (and thus of $H_N$); the difference is of order $N^{1+\kappa}$. This is a consequence of the definition (\ref{eq:def-UN}) of the unitary map $U_N$, whose action removes products of the condensate wave function $\ph_0$, leaving however all correlations among particles in the wave functions $\alpha_j \in L^2_\perp (\L)^{\otimes_s j}$, $j=1, \dots , N$. 

\smallskip

To factor out correlations, we introduce in Section \ref{sec:quadra} a renormalized excitation Hamiltonian $\cG_N = e^{-B} \cL_N e^B$, defined through unitary conjugation of $\cL_N$ with a generalized Bogoliubov transformation $e^B$. The antisymmetric operator $B: \cF_+^{\leq N} \to \cF_+^{\leq N}$ is quadratic in the modified creation and annihilation operators $b_p, b_p^*$ defined, for every momentum $p \in \L^*_+ = 2\pi \bZ^3 \backslash \{0 \}$, in (\ref{eq:bp-de}) ($b^*_p$ creates a particle with momentum $p$ annihilating, at the same time, a particle with momentum zero; in other words, $b_p^*$ creates an excitation, moving a particle out of the condensate). The properties of $\cG_N$ are listed in Prop. \ref{prop:GN}. In particular, Prop. \ref{prop:GN} implies that, to leading order, $\langle \Omega, \cG_N \Omega \rangle \simeq 4\pi \frak{a}_0 N^{1+\kappa}$, if $\kappa$ is small enough. 

\smallskip

Unfortunately, $\cG_N$ is not coercive enough to prove directly that low-energy states exhibit condensation (in the sense that it is not clear how to estimate the difference between $\cG_N$ and its vacuum expectation from below by the number of particle operator $\cN_+$). For this reason, in Section \ref{sec:cubic}, we define yet another renormalized excitation Hamiltonian $\cJ_N = e^{-A} \cG_N e^A$, where now $A$ is the antisymmetric operator (\ref{eq:defA}), cubic in (modified) creation and annihilation operators (to be more precise, we only conjugate the main part of $\cG_N$ with $e^A$; see (\ref{eq:defJN})). Important properties of $\cJ_N$ are stated in Prop. \ref{prop:JN}. Up to negligible errors, the conjugation with $e^A$ completes the renormalization of quadratic and cubic terms; in (\ref{eq:defJNeff}), these terms have the same form they would have for particles interacting through a mean-field potential with Fourier transform $8\pi \frak{a}_0 N^\kappa {\bf 1} (|p| < N^\alpha)$, with a parameter $\alpha > 0$ that will be chosen small enough, depending on $\kappa$ (in other words, the renormalization procedure allows us to replace, in all quadratic and cubic terms, the original interaction with Fourier transform $N^{-1+\kappa} \widehat{V} (p/ N^{1-\kappa})$ decaying only for momenta $|p| > N^{1-\kappa}$, with a potential whose Fourier transform already decays on scales $N^{\alpha} \ll N^{1-\kappa}$). 

\smallskip

The main problem with $\cJ_N$ is that its quartic terms (the restriction of the initial potential energy on the orthogonal complement of the condensate wave function) are still proportional to the local interaction with Fourier transform $N^{-1+\kappa} \widehat{V} (p / N^{1-\kappa})$. 

\smallskip 

One possibility to solve this problem is to neglect the original quartic terms (they are positive) and insert instead quartic terms proportional to the renormalized mean-field potential $8\pi \frak{a}_0 N^\kappa {\bf 1} (|p| < N^\alpha)$, so that Bose-Einstein condensation follows as it does for mean-field systems (see \cite{Sei}). Since (with the notation $\check{\chi}$ for the inverse Fourier transform of the characteristic function on the ball of radius one) 
\[ \begin{split} \frac{8\pi \frak{a}_0 N^\kappa}{N} \sum_{|r| < N^\alpha} a_{p+r}^* a_q^* a_{q+r} a_p  &= 8\pi \frak{a}_0 N^{3\alpha + \kappa -1} \int \check{\chi} (N^\alpha (x-y)) \check{a}_x^* \check{a}_y^* \check{a}_y \check{a}_x \, dx dy  \\ &\leq C N^{3\alpha + \kappa - 1} \cN_+^2 \end{split} \]
and since we know from (\ref{eq:BEC-LSSY}), that $\cN_+ \lesssim N^{\frac{15 + 20 \kappa}{17}}$ in low-energy states, the insertion of the renormalized quartic terms produces an error that can be controlled by localization in the number of particles, if 
\[ 3 \alpha + \kappa - 1 + \frac{15 + 20 \kappa}{17}  = 3 \alpha + \frac{37}{17} \kappa - \frac{2}{17} < 0 \]
This strategy was used in \cite{BBCS} to prove Bose-Einstein condensation with optimal rate in the Gross-Pitaevskii regime $\kappa = 0$ (in this case, one can choose $\alpha = 0$). %For positive values $\k>0$, we need at least that $\alpha>2\kappa$ and therefore the localization in the number of particles is only useful as long as $\kappa \in (0;2/139)$. 

\smallskip

Here, we follow a different approach. We perform a last renormalization step, conjugating $\cJ_N$ through a unitary operator $e^D$, with $D$  quartic in creation and annihilation operators. This leads to a new Hamiltonian 
$\cM_N= e^{-D} \cJ_N e^D$ (in fact, it is more convenient to conjugate only the main part of $\cJ_N$, ignoring small contributions that can be controlled by other means; see (\ref{eq:defMN})), where the original interaction $N^{-1+\kappa} \widehat{V} (p/ N^{1-\kappa})$ is replaced by the mean-field potential $8\pi \frak{a}_0 N^\kappa {\bf 1} (|p| < N^\alpha)$ in all relevant terms \footnote{Observe that the renormalized potential with Fourier transform $8\pi \frak{a}_0N^{-1+\kappa} {\bf 1} (|p| < N^\alpha)$ that emerges in our rigorous analysis after a series of unitary transformations is reminiscent of the interaction that appears through an ad-hoc substitution in the pseudo-potential method of \cite{HY,LHY}.}. Condensation can then be shown as it is done for mean-field systems, with no need for localization.   
%If one could replace this interaction with its renormalized version $8\pi \frak{a}_0 N^{-1+\kappa} {\bf %1} (|p| < N^\alpha )$, condensation would follow easily, similarly as it is proven for mean-field %systems (essentially, by 
%completing the square). While in the Gross-Pitaevskii regime $\kappa =0$ it was possible to 
%introduce the renormalized quartic terms using localization in the number of excitations (and the %previously available results of \cite{NRS}), to handle the case $\kappa > 0$ we need to perform a 
This is the main novelty of our analysis, compared with \cite{BBCS}. In Section \ref{sec:quartic}, we define the final Hamiltonian $\cM_N$ and in  Prop. \ref{prop:MN} we bound it from below. The proof of Prop. \ref{prop:MN},  which is technically the main part of our paper, is deferred to Section~\ref{sec:MN}. In Section~\ref{sec:proof}, we combine the results of the previous sections to conclude the proof of Theorem~\ref{thm:main}. 

\smallskip

The results we prove with our new technique are stronger than what we would obtain using the approach of \cite{BBCS} in the sense that they allow for larger values of $\kappa$ and  better rates. More importantly, we believe that the approach we propose here is more natural and that it leaves more space for extensions.  In particular, with the final quartic renormalization step, we map the original Hamilton operator (\ref{eq:HN}), with an interaction varying on momenta of order $N^{1-\kappa}$, into a new Hamiltonian having the same form, but now with an interaction restricted to momenta smaller than $N^{\alpha}$. If $\alpha < 1-\kappa$, this leads to an effective regularization of the potential and it suggests that further improvements may be achieved by iteration; we plan to follow this strategy, which bears some similarities to the renormalization group analysis developed in \cite{BFKT}, in future work.

\smallskip

In order to control errors arising from the quartic conjugation, it is important to use observables that were not employed in \cite{BBCS}. In particular, the expectation of the number of excitations with large momenta \[ \cN_{\geq N^\gamma} = \sum_{p\in \L^*_+ : |p| \geq N^\gamma} a_p^* a_p \] and of its powers $\cN^2_{\geq N^\gamma}, \cN^3_{\geq N^\gamma}$, as well as the expectation of products of the form $\cK_L \cN_{\geq N^\gamma}$ and $\cK_L \cN_{\geq N^\gamma}^2$, involving the kinetic energy operator restricted to low momenta $\cK_L$, will play a crucial role in our analysis. It will therefore be important to establish bounds for the growth of these observables through all steps of the renormalization procedure (Lemma \ref{lm:NresgrowA}, Lemma \ref{lm:KresgrowA}, Lemma \ref{lm:NresgrowD}, Lemma \ref{lm:KresgrowD}). In Section \ref{sec:proof}, an important step in the proof of Theorem \ref{thm:main} will consist in controlling the expectation of these observables on low-energy states of the renormalized Hamiltonian $\cG_N$. 

\medskip

\emph{Acknowledgements.} We would like to thank C. Boccato and S. Cenatiempo for many helpful discussions with regards to the quartic renormalization. B. Schlein gratefully acknowledges partial support from the NCCR SwissMAP, from the Swiss National Science Foundation through the Grant ``Dynamical and energetic properties of Bose-Einstein condensates'' and from the European Research Council through the ERC-AdG CLaQS. 

\section{The Excitation Hamiltonian}
\label{sec:fock}

We denote by $\cF =  \bigoplus_{n \geq 0} L^2 (\Lambda)^{\otimes_s n}$ the bosonic Fock space over the one-particle space $L^2 (\L)$ and by $\Omega= \{ 1, 0, \dots \}$ the vacuum vector. We can define the number of particles operator $\cN$ by setting $(\cN \psi)^{(n)} = n \psi^{(n)}$ for all $\psi = \{ \psi^{(0)}, \psi^{(1)}, \dots \}$ in a dense subspace of $\cF$. For every one-particle wave function $g \in L^2 (\L)$, we define the creation operator $a^* (g)$ and its hermitian conjugate, the annihilation operator $a (g)$, through 
\[ \begin{split} 
(a^* (g) \Psi)^{(n)} (x_1, \dots , x_n) &= \frac{1}{\sqrt{n}} \sum_{j=1}^n g (x_j) \Psi^{(n-1)} (x_1, \dots , x_{j-1}, x_{j+1} , \dots , x_n) 
\\
(a (g) \Psi)^{(n)} (x_1, \dots , x_n) &= \sqrt{n+1} \int_\Lambda  \bar{g} (x) \Psi^{(n+1)} (x,x_1, \dots , x_n) \, dx \end{split} \]
Creation and annihilation operators are defined on the domain of $\cN^{1/2}$, where they satisfy the bounds
\[ \| a (f) \psi \| \leq \| f \| \| \cN^{1/2} \psi \|, \qquad  \| a^* (f) \psi \| \leq \| f \| \| (\cN_+ + 1)^{1/2} \psi \| \]
and the canonical commutation relations 
\begin{equation}\label{eq:ccr} 
[ a (g), a^* (h) ] = \langle g,h \rangle , \quad [ a(g), a(h)] = [a^* (g), a^* (h) ] = 0 
\end{equation}
for all $g,h \in L^2 (\Lambda)$ ($\langle . , . \rangle$ denotes here the inner product on $L^2 (\Lambda)$). For $p \in \Lambda^* = 2\pi \bZ^3$, we define the plane wave $\ph_p \in L^2 (\L)$ through $\ph_p (x) = e^{-i p \cdot x}$ for all $x \in \L$, and the operators $a^*_p = a (\ph_p)$ and $a_p = a (\ph_p)$ creating and, respectively, annihilating a particle with momentum $p$. It is sometimes convenient to switch to position space, introducing operator valued distributions $\check{a}_x, \check{a}_x^*$ such that  
\begin{equation*}%\label{eq:axf}
 a(f) = \int_\L \bar{f} (x) \,  \check{a}_x \, dx , \quad a^* (f) = \int_\L  f(x) \, \check{a}_x^* \, dx  \end{equation*}
In terms of creation and annihilation operators, the number of particles operator can be written as 
\[ \cN = \sum_{p \in \L^*} a_p^* a_p = \int a_x^* a_x \, dx \]

We will describe excitations of the Bose-Einstein condensate on the truncated Fock space 
\[ \cF_+^{\leq N} = \bigoplus_{j=0}^N L_\perp^2 (\Lambda)^{\otimes_s j} \]
constructed over the orthogonal complement $L^2_\perp (\Lambda)$ of the condensate wave function $\ph_0$. On $\cF_+^{\leq N}$, we denote the number of particles operator by $\cN_+$. It is given by $\cN_+ =  \sum_{p \in \Lambda^*_+} a_p^* a_p  $, where $\Lambda^*_+ = \Lambda^* \backslash \{ 0 \} = 2\pi \bZ^3 \backslash \{ 0 \}$ is the momentum space for excitations. Given $\Theta \geq 0$, we also introduce the restricted number of particles operators 
\begin{equation}\label{eq:defcN>}
		\cN_{\geq \Theta} = \sum_{p\in \Lambda_+^*: |p| \geq \Theta} a^*_p a_p, 		\end{equation}
measuring the number of excitations with momentum larger or equal to $\Theta$, and $ \cN_{< \Theta} = \cN_+ - \cN_{\geq \Theta}$. 
		
Consider the operator $U_N : L^2_s (\L^N) \to \cF_+^{\leq N}$ defined in (\ref{eq:def-UN}).  
Identifying $\psi_N \in L^2_s (\L^N)$ with the Fock space vector $\{ 0, \dots , 0,  \psi_N, 0, \dots \}$, 
we can also express $U_N$ in terms of creation and annihilation operators; we obtain  
\[ U_N  = \bigoplus_{n=0}^N  (1-|\ph_0 \rangle \langle \ph_0|)^{\otimes n} \frac{a(\ph_0)^{N-n}}{\sqrt{(N-n)!}} \] 
It is then easy to check that $U_N^* : \cF_{+}^{\leq N} \to L^2_s (\Lambda^N)$ is given by 
\[ U_N^* \, \{ \alpha^{(0)}, \dots , \alpha^{(N)} \} = \sum_{n=0}^N \frac{a^* (\ph_0)^{N-n}}{\sqrt{(N-n)!}} \, \alpha^{(n)} \]
and that $U_N^* U_N = 1$, ie. $U_N$ is unitary. 

Using $U_N$, we can define the excitation Hamiltonian $\cL_N := U_N H_N U_N^*$, acting on a dense subspace of $\cF_+^{\leq N}$. To compute $\cL_N$, we first write the Hamiltonian (\ref{eq:HN}) in momentum space, in terms of creation and annihilation operators. We find 
\begin{equation}\label{eq:Hmom} H_N = \sum_{p \in \Lambda^*} p^2 a_p^* a_p + \frac{1}{2N^{1-\kappa}} \sum_{p,q,r \in \Lambda^*} \widehat{V} (r/N^{1-\kappa}) a_{p+r}^* a_q^* a_{p} a_{q+r} 
\end{equation}
where 
\begin{equation}\label{eq:FoT} \widehat{V} (k) = \int_{\bR^3} V (x) e^{-i k \cdot x} dx \end{equation}
is the Fourier transform of $V$, defined for all $k \in \bR^3$ (in fact, (\ref{eq:HN}) is the restriction of (\ref{eq:Hmom}) to the $N\in\NN$-particle sector of the Fock space $\cF$). We can now determine the excitation Hamiltonian $\cL_N$ using the following rules, describing the action of the unitary operator $U_N$ on products of a creation and an annihilation operator (products of the form $a_p^* a_q$ can be thought of as operators mapping $L^2_s (\Lambda^N)$ to itself). For any $p,q \in \Lambda^*_+ = 2\pi \bZ^3 \backslash \{ 0 \}$, we find (see \cite{LNSS}):
\begin{equation}\label{eq:U-rules}
\begin{split} 
U_N \, a^*_0 a_0 \, U_N^* &= N- \cN_+  \\
U_N \, a^*_p a_0 \, U_N^* &= a^*_p \sqrt{N-\cN_+ } \\
U_N \, a^*_0 a_p \, U_N^* &= \sqrt{N-\cN_+ } \, a_p \\
U_N \, a^*_p a_q \, U_N^* &= a^*_p a_q 
\end{split} \end{equation}
We conclude that 
\begin{equation}\label{eq:cLN} \cL_N =  \cL^{(0)}_{N} + \cL^{(2)}_{N} + \cL^{(3)}_{N} + \cL^{(4)}_{N} \end{equation}
with
\begin{equation}\label{eq:cLNj} \begin{split} 
\cL_{N}^{(0)} =\;& \frac{N-1}{2N}N^{\kappa}\widehat{V} (0) (N-\cN_+ ) + \frac{N^{\kappa}\widehat{V} (0)}{2N} \cN_+  (N-\cN_+ ) \\
\cL^{(2)}_{N} =\; &\sum_{p \in \Lambda^*_+} p^2 a_p^* a_p + \sum_{p \in \Lambda_+^*} N^{\kappa}\widehat{V} (p/N^{1-\kappa}) \left[ b_p^* b_p - \frac{1}{N} a_p^* a_p \right] \\ &+ \frac{1}{2} \sum_{p \in \Lambda^*_+} N^{\kappa}\widehat{V} (p/N^{1-\kappa}) \left[ b_p^* b_{-p}^* + b_p b_{-p} \right] \\
\cL^{(3)}_{N} =\; &\frac{1}{\sqrt{N}} \sum_{p,q \in \Lambda_+^* : p+q \not = 0} N^{\kappa}\widehat{V} (p/N^{1-\kappa}) \left[ b^*_{p+q} a^*_{-p} a_q  + a_q^* a_{-p} b_{p+q} \right] \\
\cL^{(4)}_{N} =\; & \frac{1}{2N} \sum_{\substack{p,q \in \Lambda_+^*, r \in \Lambda^*: \\ r \not = -p,-q}} N^{\kappa}\widehat{V} (r/N^{1-\kappa}) a^*_{p+r} a^*_q a_p a_{q+r} 
\end{split} \end{equation}
where we introduced generalized creation and annihilation operators  
\begin{equation}\label{eq:bp-de} 
b^*_p = a^*_p \, \sqrt{\frac{N-\cN_+}{N}} , \qquad \text{and } \quad  b_p = \sqrt{\frac{N-\cN_+}{N}} \, a_p 
\end{equation}
for all $p \in \Lambda^*_+$. Observe that, by (\ref{eq:U-rules}), 
\[ U_N^* b_p^* U_N = a^*_p  \frac{a_0}{\sqrt{N}}, \qquad U_N^* b_p U_N = \frac{a_0^*}{\sqrt{N}} a_p \]
In other words, $b_p^*$ creates a particle with momentum $p \in \Lambda^*_+$ but, at the same time, it annihilates a particle from the condensate; it creates an excitation, preserving the total number of particles in the system. On states exhibiting complete Bose-Einstein condensation in the zero-momentum mode $\ph_0$, we have $a_0 , a_0^* \simeq \sqrt{N}$ and we can therefore expect that $b_p^* \simeq a_p^*$ and that $b_p \simeq a_p$. Modified creation and annihilation operators satisfy the commutation relations 
\begin{equation}\label{eq:comm-bp} \begin{split} [ b_p, b_q^* ] &= \left( 1 - \frac{\cN_+}{N} \right) \delta_{p,q} - \frac{1}{N} a_q^* a_p 
\\ [ b_p, b_q ] &= [b_p^* , b_q^*] = 0 
\end{split} \end{equation}
Furthermore, we find 
\begin{equation}\label{eq:comm2} [ b_p, a_q^* a_r ] = \delta_{pq} b_r, \qquad  [b_p^*, a_q^* a_r] = - \delta_{pr} b_q^* \end{equation}
for all $p,q,r \in \Lambda_+^*$; this implies in particular that $[b_p , \cN_+] = b_p$, $[b_p^*, \cN_+] = - b_p^*$. It is also useful to notice that the operators $b^*_p, b_p$, like the standard creation and annihilation operators $a_p^*, a_p$, can be bounded by the square root of the number of particles operators; we find
\begin{equation*}
\begin{split} 
\| b_p \xi \| &\leq \Big\| \cN_+^{1/2} \Big( \frac{N+1-\cN_+}{N} \Big)^{1/2} \xi \Big\| \leq \| \cN_+^{1/2} \xi \|  \\ 
\| b^*_p \xi \| &\leq \Big\| (\cN_+ +1)^{1/2} \Big( \frac{N-\cN_+ }{N} \Big)^{1/2} \xi \Big\| \leq  \| (\cN_+ + 1)^{1/2} \xi \| 
\end{split} \end{equation*}
for all $\xi \in \cF^{\leq N}_+$. Since $\cN_+  \leq N$ on $\cF_+^{\leq N}$, the operators $b_p^* , b_p$ are bounded, with $\| b_p \|, \| b^*_p \| \leq (N+1)^{1/2}$. 

We can also define modified operator valued distributions 
\[ \check{b}_x = \sqrt{\frac{N-\cN_+}{N}} \, \check{a}_x, \qquad \text{and } \quad  \check{b}^*_x = \check{a}^*_x \, \sqrt{\frac{N-\cN_+}{N}} \]
in position space, for $x \in \Lambda$. The commutation relations (\ref{eq:comm-bp}) take the form 
\begin{equation*}
\begin{split}  [ \check{b}_x, \check{b}_y^* ] &= \left( 1 - \frac{\cN_+}{N} \right) \delta (x-y) - \frac{1}{N} \check{a}_y^* \check{a}_x \\ 
[ \check{b}_x, \check{b}_y ] &= [ \check{b}_x^* , \check{b}_y^*] = 0 
\end{split} \end{equation*}
Moreover, (\ref{eq:comm2}) translates to 
\begin{equation*}
\begin{split}
[\check{b}_x, \check{a}_y^* \check{a}_z] &=\delta (x-y)\check{b}_z, \qquad 
[\check{b}_x^*, \check{a}_y^* \check{a}_z] = -\delta (x-z) \check{b}_y^*
\end{split} \end{equation*}
which also implies that $[ \check{b}_x, \cN_+ ] = \check{b}_x$, $[ \check{b}_x^* , \cN_+ ] = - \check{b}_x^*$. 

\section{Renormalized Excitation Hamiltonian} 
\label{sec:quadra}

Conjugation with $U_N$ extracts, from the original quartic interaction in (\ref{eq:Hmom}), some constant and some quadratic contributions, collected in $\cL^{(0)}_N$ and $\cL^{(2)}_N$ in (\ref{eq:cLNj}). For bosons described by the Hamiltonian (\ref{eq:HN}), this is not enough; there are still large contributions to the energy that are hidden in $\cL^{(3)}_N$ and $\cL^{(4)}_N$. 

To extract the missing energy, we have to take into account correlations. To this end, we consider the ground state solution $f_\ell$ of the Neumann problem 
\begin{equation}\label{eq:scatl} \left[ -\Delta + \frac{1}{2} V \right] f_{\ell} = \lambda_{\ell} f_\ell \end{equation}
on the  ball $|x| \leq N^{1-\kappa}\ell$ (we omit the $N\in\NN$-dependence in the notation for $f_\ell$ and for $\lambda_\ell$; notice that $\lambda_\ell$ scales as $N^{3\kappa-3}$), with the normalization $f_\ell (x) = 1$ if $|x| = N^{1-\kappa} \ell$. By scaling, we observe that $f_\ell (N^{1-\kappa}.)$ satisfies the equation 
\[ \left[ -\Delta + \frac12 N^{2-2\kappa} V (N^{1-\kappa}x) \right] f_\ell (N^{1-\kappa}x) = N^{2-2\kappa} \lambda_\ell f_\ell (N^{1-\kappa}x) \]
on the ball $|x| \leq \ell$. From now on, we fix some $0 < \ell < 1/2$, so that the ball of radius $\ell$ is contained in the box $\Lambda= [-1/2 ; 1/2]^3$. 
We then extend $f_\ell (N^{1-\kappa}.)$ to $\Lambda$, by setting $f_{N} (x) = f_\ell (N^{1-\kappa}x)$, if $|x| \leq \ell$ and $f_{N} (x) = 1$ for $x \in \Lambda$, with $|x| > \ell$. As a consequence,  
\begin{equation}\label{eq:scatlN}
 \left[ -\Delta + \frac12N^{2-2\kappa} V (N^{1-\kappa}.) \right] f_{N} = N^{2-2\kappa} \lambda_\ell f_{N} \chi_\ell, 
\end{equation}
where $\chi_\ell$ denotes the characteristic function of the ball of radius $\ell$. The Fourier coefficients of the function $f_{N}$ are given by 
\begin{equation}\label{eq:fellN} \widehat{f}_{N} (p) = \int_\Lambda f_\ell (N^{1-\kappa}x) e^{-i p \cdot x} dx \end{equation}
for all $p \in \L^*$. Next, we define $w_\ell (x) = 1- f_\ell (x)$ for $|x| \leq N^{1-\kappa} \ell$ and $w_\ell (x) = 0$ for all $|x| > N^{1-\kappa} \ell$. Its rescaled version $w_{N} : \Lambda \to \bR$ is defined through $w_{N} (x) = w_{\ell} (N^{1-\kappa}x)$ if $|x| \leq \ell$ and $w_{N} (x) = 0$ if $x \in \L$ with $|x| > \ell$.  The Fourier coefficients of $w_{N}$ are given by  
\[  \widehat{w}_{N} (p) = \int_{\Lambda} w_\ell (N^{1-\kappa}x) e^{-i p \cdot x} dx = \frac{1}{N^{3-3\kappa}} \widehat{w}_\ell (p/N^{1-\kappa}), \]
where \[ \widehat{w}_\ell (k) = \int_{\bR^3} w_\ell (x) e^{-ik \cdot x} dx \] denotes the  Fourier transform of the (compactly supported) function $w_\ell$. We find $\widehat{f}_{N} (p) = \delta_{p,0} - N^{3\kappa-3} \widehat{w}_\ell (p/N^{1-\kappa})$. {F}rom (\ref{eq:scatlN}), we obtain  
\begin{equation}\label{eq:wellp}
\begin{split}  
- p^2 \widehat{w}_\ell (p/N^{1-\kappa}) +  \frac{N^{2-2\kappa}}{2} \sum_{q \in \L^*} \widehat{V} ((p-q)/N^{1-\kappa}) \widehat{f}_{N} (q) = N^{5-5\kappa} \lambda_\ell \sum_{q \in \L^*} \widehat{\chi}_\ell (p-q) \widehat{f}_{N} (q).
\end{split} 
\end{equation}
The next lemma summarizes important properties of the functions $w_\ell $ and $ f_\ell$. Its proof can be found in \cite[Appendix A]{BBCS} (replacing $N\in\NN$ by $N^{1-\kappa}$ and noting that still $ N^{1-\kappa}\ell \gg 1$ for $N\in\NN$ sufficiently large and fixed $ \ell \in (0;1/2) $).
%\cite[Lemma A.1]{ESY0} and in \cite[Lemma 4.1]{BS}. 
% Notice that this lemma is the reason why we require that $V \in L^3 
%(\bR^3)$; for the rest of the analysis $V \in L^2 (\bR^3)$ would be %enough. 
\begin{lemma} \label{sceqlemma}
Let $V \in L^3 (\bR^3)$ be non-negative, compactly supported and spherically symmetric. Fix $\ell > 0$ and let $f_\ell$ denote the solution of \eqref{eq:scatl}. For $N\in\NN$ large enough the following properties hold true.
\begin{enumerate}
\item [i)] We have 
\begin{equation}\label{eq:lambdaell} 
  \lambda_\ell = \frac{3\frak{a}_0 }{ N^{3-3\kappa}\ell^3} \left(1 +\mathcal{O} \big(\frak{a}_0  / \ell N^{1-\kappa}\big) \right).
\end{equation}
% \[ \lambda_\ell = \frac{3\frak{a}_0 }{(\ell N)^3} \left(1 + \mathcal{O} (\frak{a}_0  / \ell N) \right) \]
\item[ii)] We have $0\leq f_\ell, w_\ell\leq1$. Moreover there exists a constant $C > 0$ such that 
%\begin{equation}\label{eq:Vfa0} \left|  \int  V(x) f_\ell (x) dx - 8\pi \frak{a}_0  \right| \leq \frac{C \frak{a}_0^2}{\ell N} \, . 
%\end{equation}    
\begin{equation} \label{eq:Vfa0} 
\left|  \int  V(x) f_\ell (x) dx - 8\pi \frak{a}_0   \right| \leq \frac{C \frak{a}_0^2}{\ell N^{1-\kappa}}.
\end{equation}
\item[iii)] There exists a constant $C>0 $ such that 
	\begin{equation}\label{3.0.scbounds1} 
	w_\ell(x)\leq \frac{C}{|x|+1} \quad\text{ and }\quad |\nabla w_\ell(x)|\leq \frac{C }{x^2+1}. 
	\end{equation}
for all $x \in \bR^3$ and all $N\in\NN$ large enough. 
\item[iv)] There exists a constant $C > 0$ such that 
\[ |\widehat{w}_{N} (p)| \leq \frac{C}{N^{1-\kappa} p^2} \]
for all $p \in \bR^3$ and all $N\in\NN$ large enough (such that $N^{1-\kappa} \geq \ell^{-1}$).  
\end{enumerate}        
\end{lemma}

We define $\eta: \L^* \to \bR$ through
\begin{equation}\label{eq:defeta}
\eta_p = -N \widehat{w}_{N} (p) = - \frac {N^\kappa} {N^{2-2\kappa}} \widehat{\o}_\ell(p/N^{1-\kappa}).
\end{equation}
In position space, this means that for $ x\in \Lambda$, we have 
		\begin{equation}\label{eq:defetax} \check{\eta}(x) = - N w_\ell( N^{1-\kappa} x),\end{equation}
so that we have in particular the $L^\infty$-bound
		\begin{equation}\label{eq:etainfbnd} \|\check{\eta}\|_\infty \leq CN.
		\end{equation}
Lemma \ref{sceqlemma} also implies
\be \label{eq:modetap}
|\eta_p| \leq \frac{CN^{\kappa}}{|p|^2}
\ee
for all $p \in \L_+^*=2\pi \bZ^3 \backslash \{0\}$, and  for some constant $C>0$ independent of $N\in\NN$ (for $N\in\NN$ large enough). {F}rom (\ref{eq:wellp}), we find the relation  
\begin{equation}\label{eq:eta-scat0}
\begin{split} 
p^2 \eta_p + \frac{1}{2} N^{\kappa}(\widehat{V} (./N^{1-\kappa}) *\widehat{f}_{N}) (p) = N^{3-2\kappa} \l_\ell (\widehat{\chi}_\ell * \widehat{f}_{N}) (p)
\end{split} \end{equation}
or equivalently, expressing the r.h.s. through the coefficients $\eta_p$, 
\begin{equation}\label{eq:eta-scat}
\begin{split} 
p^2 \eta_p + \frac{1}{2} N^{\kappa}\widehat{V} (p/N^{1-\kappa}) & + \frac{1}{2N} \sum_{q \in \Lambda^*} N^{\kappa}\widehat{V} ((p-q)/N^{1-\kappa}) \eta_q \\ &\hspace{1cm} = N^{3-2\kappa} \lambda_\ell \widehat{\chi}_\ell (p) + N^{2-2\kappa} \lambda_\ell \sum_{q \in \Lambda^*} \widehat{\chi}_\ell (p-q) \eta_q.
\end{split} \end{equation}
%Moreover, with (\ref{3.0.scbounds1}), we find 
%\begin{equation}\label{eq:L2eta} \| \eta \|^2 = \| \check{\eta} \|^2 = \int_{|x| \leq \ell}  N^{2} |w_\ell (N^{1-\kappa} x)|^2 dx \leq CN^{2\kappa} \int_{|x| \leq \ell} \frac{1}{|x|^2} dx \leq CN^{2\kappa} \ell \end{equation}
%In particular, we can make $\| \eta \|$ arbitrarily small, choosing $\ell$ small enough (but preserving the condition $ N^{1-\kappa}\ell \gg 1$). 
In our analysis, it is useful to restrict $\eta$ to high momenta. To this end, let $\alpha > 0$ and \begin{equation}\label{eq:PellHPellL1} P_{H}= \{p\in \Lambda_+^*: |p|\geq N^{\alpha}\}.		\end{equation}
We define $ \eta_H\in \ell^2(\Lambda_+^*)$ by
\be \label{eq:defetaH}
\eta_H (p)=\eta_p\, \chi(p \in P_H) = \eta_p \chi (|p| \geq N^{\alpha}) \,.
\ee
Eq. \eqref{eq:modetap} implies that 
\begin{equation}\label{eq:etaHL2}
\| \eta_H \| \leq C N^{\kappa -\a/2}
\end{equation}
and we assume from now on that $\a > 2\kappa$ such that in particular
		\begin{equation}\label{eq:L2eta} \lim_{N\to \infty}\| \eta_H \| = 0. \end{equation}
Notice, on the other hand, that the $H^1$-norm of $\eta$ and $\eta_{H}$ diverge, as $N \to \infty$. From (\ref{eq:defetax}) and Lemma \ref{sceqlemma}, part iii), we find 
\be \label{eq:H1eta}
\sum_{p \in P_H} p^2 |\eta_p|^2  \leq \sum_{p \in \L_+^*} p^2 |\eta_p|^2  = \int |\nabla \check{\eta} (x)|^2 dx \leq C N^{1+\kappa}
\ee
for all $N \in \bN$ large enough. We will mostly use the coefficients $\eta_p$ with $p\neq 0$. Sometimes, however, it will be useful to have an estimate on $\eta_0$ (because Eq. \eqref{eq:eta-scat} involves $\eta_0$). From Lemma \ref{sceqlemma}, part iii), we obtain
\begin{equation}\label{eq:wteta0}  |\eta_0| \leq N^{3\kappa-2} \int_{\bR^3} w_\ell (x) dx \leq C N^{\kappa}\ell^2 \end{equation}

It will also be useful to have bounds for the function $\check{\eta}_H : \L \to \bR$, having Fourier coefficients $\eta_H (p)$ as defined in (\ref{eq:defetaH}). Writing $\eta_H (p) = \eta_p - \eta_p \chi (|p| \leq N^{\alpha})$, we obtain 
\[ \check{\eta}_H (x) = \check{\eta} (x) - \sum_{\substack{p \in \L^* :\\  |p| \leq N^{\alpha}}} \eta_p e^{i p \cdot x} = -N w_\ell (N^{1-\kappa}x) - \sum_{\substack{p \in \L^* :\\  |p| \leq N^{\alpha}}} \eta_p e^{i p \cdot x} \]
so that 
\begin{equation}\label{eq:etax}
|\check{\eta}_H (x)| \leq C N + CN^{\kappa}\sum_{\substack{p \in \L^* :\\  |p| \leq N^{\alpha}}} |p|^{-2} \leq C (N + N^{\alpha+\kappa}) \leq C (N + N^{\alpha+\kappa})
\end{equation}
for all $x \in \L$, if $N \in \bN$ is large enough. 

With the coefficients (\ref{eq:defetaH}), we define the antisymmetric operator 
\begin{equation}\label{eq:genBog} B = \frac{1}{2} \sum_{p \in P_H} \left( \eta_p b_p^* b_{-p}^* - \bar{\eta}_p b_{-p} b_p \right) \end{equation}
and the generalized Bogoliubov transformation $e^B : \cF_+^{\leq N} \to \cF_+^{\leq N}$. A first important observation is that conjugation with this unitary operator does not change the number of particles by too much. The proof of the following Lemma can be found in \cite[Lemma 3.1]{BS} (a similar result has been previously established in \cite{Sei}).
\begin{lemma}\label{lm:Ngrow}
Assume $B$ is defined as in (\ref{eq:genBog}), with the coefficients $\eta_p$ as in (\ref{eq:defeta}), satisfying (\ref{eq:L2eta}). For every $n \in \bN$, there exists a constant 
$C > 0$ such that 
\begin{equation}\label{eq:bd-Beta} e^{-B} (\cN_+ +1)^{n} e^{B} \leq C (\cN_+ +1)^{n}  
\end{equation}
as an operator inequality on $\cF_+^{\leq N}$ (the constant depends only on $\| \eta_H \|$ and on $n \in \bN$). 
\end{lemma}
 
With the generalized Bogoliubov transformation $e^B$, we can now define the renormalized excitation Hamiltonian $\cG_{N} : \cF^{\leq N}_+ \to \cF^{\leq N}_+$ by setting 
\begin{equation}\label{eq:GN} 
\cG_{N} = e^{-B} \cL_N e^{B} = e^{-B} U_N H_N U_N^* e^{B}.
\end{equation}

In the next propositions, we collect important properties of $\cG_{N}$. Recall the notation $\cH_N = \cK + \cV_N$, introduced in (\ref{eq:KcVN}). 
\begin{prop} \label{prop:GN} Let $V\in L^3(\bR^3)$ be compactly supported, pointwise non-negative and spherically symmetric. Let $\cG_N$ be defined as in (\ref{eq:GN}). Assume that the exponent $\alpha$ introduced in (\ref{eq:PellHPellL1}) is such that 
\begin{equation}\label{eq:alp} 
\alpha > 6 \kappa, \qquad 2\alpha + 3 \kappa < 1 
\end{equation} 
Then 
\be \label{eq:GN-prel}
\cG_{N}  = 4 \pi \frak{a}_0 N^{1+\kappa} + \cH_N + \theta_{\cG_{N}}
\ee
and there exists $C>0$ such that, for all $\delta > 0$ and all $N \in \bN$ large enough, we have
\begin{equation}\label{eq:Gbd0}
\pm \theta_{\cG_{N}} \leq \delta \cH_N + C \delta^{-1}N^{\alpha+2\kappa} \cN_+ + CN^{\alpha+2\kappa}
\end{equation}
and the improved lower bound 
\begin{equation}\label{eq:theta-err}
\theta_{\cG_{N}} \geq  - \delta \cH_N - C \delta^{-1} N^{\kappa} \cN_+ - C N^{\alpha+2\kappa}.  
\end{equation}  
Furthermore, for $ \beta>0$, denote by $ \cG_N^{\text{eff}}$ the excitation Hamiltonian
\begin{equation}\begin{split}
		\label{eq:GNeff} \cG^{\text{eff}}_{N}  =&\;4\pi \frak{a}_0 N^{\kappa} (N-\cN_+) + \big[\widehat V(0)-4\pi \frak{a}_0\big]N^{\kappa}\cN_+\frac{(N-\cN_+)}{N}\\
		&   + N^{\kappa}\widehat V(0) \sum_{p\in P_H^c}   a^*_pa_p (1-\cN_+/N) + 4\pi \frak{a}_0 N^{\kappa}\sum_{p\in P_H^c}  \big[ b^*_p b^*_{-p} + b_p b_{-p} \big] \\
		& +  \frac{1}{\sqrt N}\sum_{\substack{ p,q\in\Lambda_+^*: |q|\leq N^{\beta},\\ p+q\neq 0 }} N^{\kappa}\widehat V(p/N^{1-\kappa})\big[ b^*_{p+q}a^*_{-p}a_q+ \emph{h.c.}\big]+\cH_N
		\end{split}\end{equation} 
Then there exists $C > 0$ such that $\cE_{\cG_{N}} = \cG_{N} - \cG^{\text{eff}}_{N}$ is bounded by 
\begin{equation}\label{eq:GeffE} \pm \cE_{\cG_{N}} \leq  C (N^{3\kappa-\alpha/2} + N^{\alpha+3\kappa/2-1/2} + N^{\kappa/2-\beta}) \cH_N + CN^{\alpha+2\kappa}\end{equation}
for all $N\in\NN$ sufficiently large. 

Furthermore, there exists a constant $C > 0$ such that 
\begin{equation} \begin{split}\label{eq:errComm}
\pm i[\cN_{\geq cN^{\gamma}}, \cG_{N}], \;\pm i [\cN_{< cN^{\gamma}}, \cG_{N}]&\leq C ( N^{\kappa +\alpha/2 -\gamma} +N^{\kappa +\gamma/2 } )  (\cH_N+1) \end{split}\end{equation}
for all $\alpha\geq \gamma>0$, $c>0$ fixed (independent of $N\in \mathbb{N}$) and $N\in \bN$ large enough.

Finally, for every $k\in\NN$, there exists a constant $C>0$ such that
		\begin{equation}\label{eq:adjGN}
		\pm \operatorname{ad}^{(k)}_{i\cN_+}(\cG_N) = \pm \big[ i\cN_+, \dots \big[i\cN_+, \cG_N\big]\dots\big]\leq C N^{\kappa +\alpha/6 } (\cH_N+1).
		\end{equation}
\end{prop}

The proof of Prop. \ref{prop:GN} is very similar to the proof of \cite[Prop. 4.2]{BBCS} and \cite[Prop. 3.2]{BBCS2}, with the appropriate modifications dictated by the different scaling of the interaction. The main novelty in Prop. \ref{prop:GN} is the bound (\ref{eq:errComm}) involving commutators of the restricted number of particles operator $\cN_{\geq c N^\gamma}$. This can be obtained similarly to the bounds for $\cE_{\cG_N}$ and for $i [ \cN_+ , \cG_N ]$, because we have a full expansion of the operator 
$\cG_N$ in a sum of terms whose commutators with $\cN_+$ and with $\cN_{\geq cN^\gamma}$ retains essentially the same form. We give a complete proof of Prop. \ref{prop:GN} in  Appendix \ref{sec:GN}. 

%%%%%%%%%%%%%%%%%%%%%%%%%%%%%%%%%%%%%%%%%%%%%%%%%%%
%%%%%%%%%%%%%%%%%%%%%%%%%%%%%%%%%%%%%%%%%%%%%%%%%%%
%%%%%%%%%%%%%%%%%%%%%%%%%%%%%%%%%%%%%%%%%%%%%%%%%%%
%%%%%%%%%%%%%%%%%%%%%%%%%%%%%%%%%%%%%%%%%%%%%%%%%%%
%%%%%%%%%%%%%%%%%%%%%%%%%%%%%%%%%%%%%%%%%%%%%%%%%%%

\section{Cubic Renormalization} \label{sec:cubic}

From Eq. (\ref{eq:GNeff}), we observe that the cubic terms in $\cG^\text{eff}_N$ still depend on the original interaction, which decays slowly in momentum (in contrast to the quadratic terms in the second line of (\ref{eq:GNeff}), where the sum is now restricted to $P_H^c = \{ p \in \L^*_+ : |p| < N^\alpha \}$). 

To renormalize the cubic terms in (\ref{eq:GNeff}), we are going to conjugate $\cG^\text{eff}_N$ with a unitary operator $e^A$, where the antisymmetric operator $A:\cF_+^{\leq N}\to \cF_+^{\leq N}$ is defined by
		\begin{equation}\label{eq:defA}
		\begin{split}
 		A&= A_1-A_1^*,\hspace{0.6cm} \text{with } \hspace{0.4cm}A_1=\frac1{\sqrt N} \sum_{r\in P_{H}, p \in P_{L}} \eta_r b^*_{r+p}a^*_{-r}a_p .
		\end{split}
		\end{equation} 
The high-momentum set $P_H = \{ p \in \Lambda^*_+ : |p| \geq N^\alpha \}$ is as in (\ref{eq:PellHPellL1}). The low-momentum set $P_L$ is defined by  
\begin{equation}\label{eq:defPHPL} P_L = \{ p \in \L_+^* : |p| \leq N^{\b} \} \end{equation}
with exponent $\beta>0$, that will be chosen as in \eqref{eq:GNeff}. 

Using the unitary operator $e^A$, we define $\cJ_N:\cF_+^{\leq N}\to \cF_+^{\leq N}$ by
\begin{equation}\label{eq:defJN} \cJ_N = e^{-A} \cG_N^{\text{eff}}e^A. 
\end{equation}
Observe here that we only conjugate the main part $\cG_N^\text{eff}$ of the renormalized excitation Hamiltonian $\cG_N$; this makes the analysis a bit simpler (the difference $\cG_N - \cG_N^\text{eff}$ is small and can be estimated before applying the cubic conjugation).
 
The next proposition summarizes important properties of $ \cJ_N$; it can be shown similarly to \cite[Prop. 5.2]{BBCS}, of course with the appropriate changes of the scaling of the interaction. For completeness, we provide a proof in Appendix \ref{sec:JN}. 
\begin{prop} \label{prop:JN} Suppose the exponents $\alpha$ and $\beta$ are such that 
\begin{equation}\label{eq:alphabeta} 
i) \; \alpha > 3 \beta + 2\kappa, \qquad ii) \; 3\alpha /2 + 2 \kappa < 1, \qquad iii) \; \alpha < 5 \beta, \qquad iv) \; \beta > 3\kappa /2, \qquad  v) \; \beta < 1/2 
\end{equation}
Let $ \cJ_N$ be defined as in \eqref{eq:defJN}, let  
		\begin{equation}\begin{split}
		\label{eq:defJNeff} \cJ_{N}^{\text{eff}} =&\;4\pi \frak{a}_0 N^{1+\kappa}  - 4\pi \frak{a}_0 N^{\kappa}\cN_+^2/N +  8\pi \mathfrak{a}_0 N^{\kappa} \sum_{p\in P_H^c} \Big[  b^*_pb_p + \frac12  b^*_p b^*_{-p} + \frac12 b_p b_{-p} \Big] \\
		& +  \frac{8\pi\mathfrak{a}_0N^{\kappa}}{\sqrt N}\sum_{\substack{ p \in P_H^c,q\in P_L: \\ p+q\neq 0 }} \big[ b^*_{p+q}a^*_{-p}a_q+ \emph{h.c.}\big]+\cH_N,
		\end{split}\end{equation} 
and set $\mu = \max (3\alpha/2 +2\kappa-1, 3\kappa/2-\beta)$ ($\mu < 0$ follows from (\ref{eq:alphabeta})). Then, there exists a constant $C > 0$ such that the self-adjoint operator $\cE_{\cJ_{N}}  = \cJ_N - \cJ_N^{\text{eff} }$ satisfies the operator inequality
\begin{equation}\label{eq:propJNerrorbnd}\pm e^A\cE_{\cJ_N}e^{-A} \leq  C(N^{-\beta/2}+ N^{\mu} )\cK + CN^{\mu}\cV_N  + CN^{\mu-\kappa}\cN_+ +CN^{\alpha+2\kappa}(1+ N^{\alpha+\beta/2-1})\end{equation}
in $ \cF_+^{\leq N}$ for all $N\in\NN$ sufficiently large. 
\end{prop}

The bounds for $\cJ_N$ given in Prop. \ref{prop:JN} are still not enough to show Theorem \ref{thm:main}. As we will discuss in the next section, the main problem is the quartic interaction term, contained in $\cH_N$, which still depends on the singular interaction potential (in all other terms on the r.h.s. of (\ref{eq:defJNeff}), the singular potential has been replaced by the regular mean-field type potential, with Fourier transform $8\pi \frak{a}_0 N^\kappa {\bf 1}_{P_H^c} (p)$, supported on momenta $|p| < N^\alpha$). To renormalize the quartic interaction, we will have to conjugate $\cJ^\text{eff}_N$ with yet another unitary operator, this time quartic in creation and annihilation operators. This last conjugation (which will be performed in the next section), will produce error 
terms. These errors will controlled in terms of the observables $\cN_+$, $\cK$ and $\cV_N$ (as in (\ref{eq:propJNerrorbnd})) but also, as we stressed at the end of Section \ref{sec:intro}, in terms of observables having the form $\cN_{\geq N^\gamma}$ (the number of excitations having momentum larger or equal to $N^\gamma$), $\cN^2_{\geq N^\gamma}$, $\cN_{\geq N^\gamma}^3$, $\cK_{\leq N^\gamma}$ (the kinetic energy of excitations with momentum below $N^\gamma$), $\cK_L \cN_{\geq N^\gamma}$. For this reason, we need to control the action of $e^A$ on all these observables. 

First of all, we bound the action of the cubic phase on the restricted number of particles operators $\cN_{\geq \theta} = \sum_{p \in \L^*_+ : |p| \geq \theta} a_p^* a_p$. We will make use of the  pull-through formula $a_p \cN_{\geq \theta} = (\cN_{\geq \theta} + {\bf 1}_{[\theta, \infty)} (p)) a_p$, which in particular implies that  
\begin{equation}\label{eq:Ntheta-bds}
\begin{split}
\| (\cN_{\geq \theta} +1)^{1/2} a_p \xi \| & \leq C \| a_p (\cN_{\geq \theta} +1)^{1/2}\xi \|, \\
\| (\cN_{\geq \theta} +1)^{-1/2} a_p \xi \| & \leq C \| a_p (\cN_{\geq \theta} +1)^{-1/2}\xi \| \, .
\end{split} 
\end{equation} 
\begin{lemma}\label{lm:NresgrowA} Assume the exponents $\alpha, \beta$ satisfy (\ref{eq:alphabeta}) (in fact, here it is enough to assume that $\alpha > 2\kappa$). Let $ k\in\NN_0$, $m=0,1,2$, $0 < \gamma\leq\alpha$, $c \geq 0$ (and $c < 1$ if $\gamma=\alpha$). Then, there exists a constant $C>0$ such that the operator inequalities
		\begin{equation}\label{eq:NresgrowA}
		\begin{split}
		e^{-sA} (\cN_+ + 1)^{k }(\cN_{\geq c N^{\gamma}}+1)^m e^{sA} \leq &\; C(\cN_+ + 1)^{k } (\cN_{\geq  cN^{\gamma}}+1)^m 	
		\end{split}
		\end{equation}
for all $s\in [-1;1] $ and all $N\in\NN$. 
\end{lemma}		

\begin{proof} The case $m=0$ follows from $m=1$. We start therefore with the case $m=1$. For $\xi \in \cF_+^{\leq N}$, we define the function $ \varphi_\xi:\mathbb{R}\to \mathbb{R}$ by 	
		\[\varphi_\xi (s)  = \langle \xi, e^{-sA} (\cN_+ + 1)^{k }(\cN_{\geq c N^{\gamma}}+1) e^{sA}  \xi \rangle\]
which has derivative
		\begin{equation} \label{eq:NresgrowlemA1}
		\begin{split}
		\partial_s \varphi_\xi (s)& =2 \Re\langle e^{sA} \xi, (\cN_+ + 1)^{k}\big[\cN_{\geq c N^{\gamma}}, A_1 \big]  e^{sA}  \xi \rangle\\
		&\hspace{0.4cm} + 2\Re\langle e^{sA} \xi, \big[(\cN_+ + 1)^{k}, A_1 \big](\cN_{\geq c N^{\gamma}}+1)  e^{sA}  \xi \rangle, 
		\end{split}
		\end{equation}
where $A_1$ as in \eqref{eq:defA}. By the assumptions on $\gamma$ and $c$, we have $N^\alpha\geq  N^{\alpha}-N^{\beta} \geq cN^{\gamma}$ for $N\in \NN$ large enough. This implies in particular that 
		\[ [\cN_{\geq cN^{\gamma} }, b^*_{p+r}] =  b^*_{p+r}, \hspace{0.2cm} [\cN_{\geq cN^{\gamma} }, a^*_{-r}] =  a^*_{-r}, \hspace{0.2cm} [\cN_{\geq cN^\gamma}, a_p] = \chi(|p|\geq cN^\gamma) a_p\]
for $ r\in P_H$ and $p\in P_L$, by \eqref{eq:ccr} and \eqref{eq:comm2}. We then obtain
		\begin{equation}\label{eq:NresgrowlemA2}
		\begin{split}
		\big[\cN_{\geq c N^{\gamma}}, A_1 \big] & = \frac2{\sqrt N} \sum_{\substack{ r\in P_{H}, p \in P_{L} }} \eta_r b^*_{r+p}a^*_{-r}a_p - \frac1{\sqrt N} \sum_{\substack{ r\in P_{H}, p \in P_{L},\\ |p| \geq cN^{\gamma} }} \eta_r b^*_{r+p}a^*_{-r}a_p
		\end{split}
		\end{equation}
as well as
		\begin{equation}\label{eq:NresgrowlemA3}
		\begin{split}
		\big[(\cN_++1)^k, A_1 \big] & = \frac k{\sqrt N} \sum_{r\in P_{H}, p \in P_{L}}  \eta_r b^*_{r+p}a^*_{-r}a_p  (\cN_++\Theta(\cN_+)+1)^{k-1},
		\end{split}
		\end{equation}
for some function $ \Theta: \NN\to (0;1)$ by the mean value theorem. Using the pull-through formula $ \cN_+ a^*_p = a^*_p (\cN_++1)$ and Cauchy-Schwarz, we estimate
		\[\begin{split}
		&\bigg| \frac1{\sqrt N} \sum_{\substack{ r\in P_{H}, p \in P_{L} }} \eta_r  \langle e^{sA}\xi, (\cN_+ + 1)^{k} b^*_{r+p}a^*_{-r}a_p   e^{sA}\xi\rangle\bigg| \\
		& \leq  \frac1{\sqrt N} \bigg(  \sum_{\substack{ r\in P_{H}, p \in P_{L} }}  \| (\cN_{\geq cN^{\gamma} }+1)^{-1/2} a_{r+p}a_{-r}  (\cN_+ + 1)^{k/2}e^{sA}\xi \|^2  \bigg)^{1/2}\\
		&\hspace{1cm} \times  \bigg(  \sum_{\substack{ r\in P_{H}, p \in P_{L} }} \eta_r^2\|(\cN_{\geq cN^{\gamma} }+1)^{1/2} a_p   (\cN_+ + 1)^{k/2} e^{sA}\xi \|^2\bigg)^{1/2}
		\end{split}\]
With the operator inequality $\cN_{\geq cN^{\gamma} } \geq  \cN_{\geq N^{\alpha}} $ and with (\ref{eq:Ntheta-bds}), we find that
		\begin{equation}\label{eq:NresgrowlemA4}
		\begin{split}
		&\bigg| \frac1{\sqrt N} \sum_{\substack{ r\in P_{H}, p \in P_{L} }} \eta_r  \langle e^{sA}\xi, (\cN_+ + 1)^{k} b^*_{r+p}a^*_{-r}a_p   e^{sA}\xi\rangle\bigg| \\ 
		& \leq  \frac C{\sqrt N} \bigg(  \sum_{\substack{ r\in P_{H}, p \in P_{L}:|p+r|\geq cN^{\gamma} }}  \| a_{p+r} (\cN_{\geq cN^{\gamma} }+1)^{-1/2} a_{-r}  (\cN_+ + 1)^{k/2}e^{sA}\xi \|^2  \bigg)^{1/2}\\
		&\hspace{1cm} \times \|\eta_H\| \bigg(  \sum_{\substack{ p \in P_{L} } } \|a_p(\cN_{\geq cN^{\gamma} }+1)^{1/2}    (\cN_+ + 1)^{k/2} e^{sA}\xi \|^2\bigg)^{1/2} \\
		&\leq \frac {CN^{\kappa-\alpha/2}}{\sqrt N} \|(\cN_{\geq N^{\alpha} }+1)^{1/2}   (\cN_+ + 1)^{k/2} e^{sA}\xi \| \|(\cN_{\geq cN^{\gamma} }+1)^{1/2}   (\cN_+ + 1)^{(k+1)/2} e^{sA}\xi \|\\
		&\leq CN^{\kappa-\alpha/2}\|(\cN_{\geq cN^{\gamma} }+1)^{1/2}   (\cN_+ + 1)^{k/2} e^{sA}\xi \|^2  = CN^{\kappa-\alpha/2} \varphi_\xi(s).
		\end{split}\end{equation}
The same arguments show that
		\begin{equation}\label{eq:NresgrowlemA5}\begin{split}
		&\bigg| \frac1{\sqrt N} \sum_{\substack{ r\in P_{H}, p \in P_{L},\\ |p|\geq cN^{\gamma} }} \eta_r  \langle e^{sA}\xi, (\cN_+ + 1)^{k} b^*_{r+p}a^*_{-r}a_p   e^{sA}\xi\rangle\bigg| \\ 
		& \leq  \frac C{\sqrt N} \bigg(  \sum_{\substack{ r\in P_{H}, p \in P_{L}:|p+r|\geq cN^{\gamma} }}  \| a_{p+r} (\cN_{\geq cN^{\gamma} }+1)^{-1/2} a_{-r}  (\cN_+ + 1)^{k/2}e^{sA}\xi \|^2  \bigg)^{1/2}\\
		&\hspace{1cm} \times \|\eta_H\| \bigg(  \sum_{\substack{ p \in P_{L} } } \|a_p(\cN_{\geq cN^{\gamma} }+1)^{1/2}    (\cN_+ + 1)^{k/2} e^{sA}\xi \|^2\bigg)^{1/2} \\
		&\leq CN^{\kappa-\alpha/2} \varphi_\xi(s).
		\end{split}\end{equation}
Finally, we have that 
		\begin{equation}\label{eq:NresgrowlemA7}\begin{split}
		&\bigg | \frac k{\sqrt N} \sum_{r\in P_{H}, p \in P_{L}}  \eta_r \langle e^{sA}\xi,b^*_{r+p}a^*_{-r}a_p  (\cN_++\Theta(\cN_+)+1)^{k-1}(\cN_{\geq cN^{\gamma}} +1)e^{sA}\xi\rangle\bigg| \\ 
		& \leq  \frac C{\sqrt N} \bigg(  \sum_{\substack{ r\in P_{H}, p \in P_{L}:|p+r|\geq cN^{\gamma} }}  \|  a_{r+p}a_{-r}  (\cN_+ + 1)^{(k-1)/2}e^{sA}\xi \|^2  \bigg)^{1/2}\\
		&\hspace{1cm} \times  \bigg(  \sum_{\substack{ r\in P_{H}, p \in P_{L}}} \eta_r^2\| a_p   (\cN_+ + 1)^{(k-1)/2}(\cN_{\geq cN^{\gamma}} +1) e^{sA}\xi \|^2\bigg)^{1/2}\\
		&\leq CN^{\kappa-\alpha/2} \|(\cN_{\geq cN^{\gamma} }+1)^{1/2}   (\cN_+ + 1)^{k/2} e^{sA}\xi \|^2  = CN^{\kappa-\alpha/2} \varphi_\xi(s).
		\end{split}\end{equation}
Recalling \eqref{eq:NresgrowlemA1}, \eqref{eq:NresgrowlemA2} and that $ \alpha\geq 2\kappa$, the bounds \eqref{eq:NresgrowlemA4} to \eqref{eq:NresgrowlemA7} show that
		\[\partial_s\varphi_\xi(s) \leq CN^{\kappa-\alpha/2} \varphi_\xi(s) \leq C \varphi_\xi(s). \]
Since the bounds are independent of $\xi \in \cF_+^{\leq N}$ and the same bounds hold true replacing $ A$ by $-A$ in the definition of $ \varphi_\xi$, the first inequality in \eqref{eq:NresgrowA} follows by Gronwall's Lemma. 

To prove \eqref{eq:NresgrowA} with $m=2$, we proceed similarly. Given $\xi\in\cF_+^{\leq N}$, we define the function $ \psi_\xi:\mathbb{R}\to\mathbb{R}$ by 
		\[\psi_\xi(s) = \langle \xi, e^{-sA} (\cN_+ + 1)^{k }(\cN_{\geq c N^{\gamma}}+1)^2 e^{sA}  \xi \rangle. \]
Its derivative is equal to
		\begin{equation}\label{eq:NresgrowlemA8}\begin{split}
		\partial_s \psi_\xi(s)&=  2\Re \langle e^{sA} \xi, (\cN_+ + 1)^{k}\big[(\cN_{\geq c N^{\gamma}}+1)^2, A_1\big]  e^{sA}  \xi \rangle\\
		&\hspace{0.4cm} + 2\Re \langle e^{sA} \xi, \big[(\cN_+ + 1)^{k}, A_1 \big](\cN_{\geq c N^{\gamma}}+1)^2  e^{sA}  \xi \rangle \\
		&=  2\Re\langle e^{sA} \xi, (\cN_+ + 1)^{k}\big[\cN_{\geq c N^{\gamma}}, \big[\cN_{\geq c N^{\gamma}}, A_1\big]\big]  e^{sA}  \xi \rangle\\
		& \hspace{0.4cm} +4\Re \langle e^{sA} \xi, (\cN_+ + 1)^{k}\big[\cN_{\geq c N^{\gamma}}, A_1\big](\cN_{\geq c N^{\gamma}}+1)  e^{sA}  \xi \rangle \\ 
		&\hspace{0.4cm} + 2\Re\langle e^{sA} \xi, \big[(\cN_+ + 1)^{k}, A_1 \big](\cN_{\geq c N^{\gamma}}+1)^2  e^{sA}  \xi \rangle.
		\end{split}\end{equation}
Comparing the contribution containing the double commutator in the last line on the r.h.s. of the last equation with \eqref{eq:NresgrowlemA2} and using once again that $N^\alpha\geq  N^{\alpha}-N^{\beta} \geq cN^{\gamma}$ for $N\in \NN$ large enough, we observe that 
		\begin{equation}\label{eq:NresgrowlemA9}
		\begin{split}
		\big[\cN_{\geq c N^{\gamma}}, \big[\cN_{\geq c N^{\gamma}}, A_1 \big] \big]& = \frac4{\sqrt N} \sum_{\substack{ r\in P_{H}, p \in P_{L}}} \eta_r b^*_{r+p}a^*_{-r}a_p  - \frac3{\sqrt N} \sum_{\substack{ r\in P_{H}, p \in P_{L},\\ |p| \geq cN^{\gamma} }} \eta_r b^*_{r+p}a^*_{-r}a_p.
		\end{split}
		\end{equation}
Hence, the bounds \eqref{eq:NresgrowlemA4} and \eqref{eq:NresgrowlemA5} prove that
		\[\big| \langle e^{sA} \xi, (\cN_+ + 1)^{k}\big[\cN_{\geq c N^{\gamma}}, \big[\cN_{\geq c N^{\gamma}}, A_1\big]\big]  e^{sA}  \xi \rangle\big|\leq C\varphi_\xi(s) \leq C \psi_\xi(s). \]
To bound the second contribution on the r.h.s. in \eqref{eq:NresgrowlemA8}, we recall \eqref{eq:NresgrowlemA2} and we estimate
		\[\begin{split}
		&\bigg| \frac1{\sqrt N} \sum_{\substack{ r\in P_{H}, p \in P_{L} }} \eta_r  \langle e^{sA}\xi, (\cN_+ + 1)^{k} b^*_{r+p}a^*_{-r}a_p (\cN_{\geq cN^\gamma} +1)   e^{sA}\xi\rangle\bigg| \\ 
		&\hspace{0.4cm}+\bigg| \frac1{\sqrt N} \sum_{\substack{ r\in P_{H}, p \in P_{L},\\ |p|\geq cN^{\gamma} }} \eta_r  \langle e^{sA}\xi, (\cN_+ + 1)^{k} b^*_{r+p}a^*_{-r}a_p (\cN_{\geq cN^\gamma} +1)   e^{sA}\xi\rangle\bigg| \\ 
		& \leq  \frac C{\sqrt N} \bigg(  \sum_{\substack{ r\in P_{H}, p \in P_{L}:|p+r|\geq cN^{\gamma} }}  \| a_{p+r}  a_{-r}  (\cN_+ + 1)^{k/2}e^{sA}\xi \|^2  	\bigg)^{1/2}\\
		&\hspace{1cm} \times \|\eta_H\| \bigg(  \sum_{\substack{ p \in P_{L} } } \|a_p    (\cN_+ + 1)^{k/2} (\cN_{\geq cN^\gamma} +1) e^{sA}\xi \|^2\bigg)^{1/2} \\
		 &\leq CN^{\kappa-\alpha/2}\|(\cN_{\geq cN^{\gamma} }+1)  (\cN_+ + 1)^{k/2} e^{sA}\xi \|^2 = CN^{\kappa-\alpha/2} \psi_\xi(s)
		\end{split}\]		
Finally, the last contribution in \eqref{eq:NresgrowlemA8} can be bounded as in \eqref{eq:NresgrowlemA7}, using \eqref{eq:NresgrowlemA3}. We have
		\[\begin{split}
		&\bigg | \frac k{\sqrt N} \sum_{r\in P_{H}, p \in P_{L}}  \eta_r \langle e^{sA}\xi,b^*_{r+p}a^*_{-r}a_p  (\cN_++\Theta(\cN_+)+1)^{k-1}(\cN_{\geq cN^{\gamma}} +1)^2e^{sA}\xi\rangle\bigg| \\ 
		& \leq  \frac C{\sqrt N} \bigg(  \sum_{\substack{ r\in P_{H}, p \in P_{L}:|p+r|\geq cN^{\gamma} }}  \|  a_{r+p}a_{-r}  (\cN_+ + 1)^{k/2}e^{sA}\xi \|^2  \bigg)^{1/2}\\
		&\hspace{1cm} \times  \bigg(  \sum_{\substack{ r\in P_{H}, p \in P_{L}}} \eta_r^2\| a_p   (\cN_+ + 1)^{(k-2)/2}(\cN_{\geq cN^{\gamma}} +1)^2 e^{sA}\xi \|^2\bigg)^{1/2}\\
		&\leq CN^{\kappa-\alpha/2} \|(\cN_{\geq cN^{\gamma} }+1)   (\cN_+ + 1)^{k/2} e^{sA}\xi \|^2  = CN^{\kappa-\alpha/2} \psi_\xi(s),
		\end{split}\]
where, in the last step, we used that $ \cN_{\geq cN^{\gamma}}\leq \cN_+ $. In conclusion, we have proved that
		\[\partial_s \psi_\xi (s)\leq CN^{\kappa-\alpha/2} \psi_\xi(s)\leq C \psi_\xi(s).\]
Since the bounds are independent of $\xi\in \cF_+^{\leq N}$ and the same bounds hold true replacing $-A$ by $A$ in the definition $\psi_\xi$, Gronwall's lemma implies the last inequality in \eqref{eq:NresgrowA}. 
\end{proof}

We denote the kinetic energy restricted to low momenta by 
\begin{equation}\label{eq:defKres}
		\begin{split}
		\cK_{\leq cN^{\gamma}} = \sum_{p\in \Lambda_+^*: |p|\leq cN^{\gamma}} p^2a^*_pa_p. 
		\end{split}\end{equation}	
We will need the following estimates for the growth of the restricted kinetic energy. 
\begin{lemma}\label{lm:KresgrowA} Assume the exponents $\alpha,\beta$ satisfy (\ref{eq:alphabeta}) (here we only need $\alpha \geq 2\kappa$ and $\alpha > \beta$). Let $0 < \gamma_1, \gamma_2 \leq \alpha$, and $c_1, c_2 \geq 0$ (and also $c_j < 1$, if $\gamma_j = \alpha$, for $j=1,2$). Then, there exists a constant $C>0$ such that the operator inequalities		
		\begin{equation}\label{eq:KresgrowA}
		\begin{split}		
		e^{-sA} \cK_{\leq c_1 N^{\gamma_1}} e^{sA} \leq&\; \cK_{\leq c_1 N^{\gamma_1}}  + N^{2\beta+2\kappa-\alpha-1} (\cN_{\geq \frac12 N^{\alpha}} +1)^2,\\
		e^{-sA} \cK_{\leq c_1 N^{\gamma_1}} (\cN_{\geq c_2 N^{\gamma_2}}+1)e^{sA} \leq&\; \cK_{\leq c_1 N^{\gamma_1}}(\cN_{\geq c_2 N^{\gamma_2}} +1)  \\
		& + N^{2\beta+2\kappa -\alpha-1} (\cN_{\geq c_2 N^{\gamma_2}} +1)^2(\cN_{\geq \frac12 N^{\alpha}} +1)\\
		\end{split}
		\end{equation}
for all $s\in [-1;1] $ and all $N\in\NN$ sufficiently large. 
\end{lemma}	
\begin{proof} Like the previous Lemma \ref{lm:NresgrowA}, this is an application of Gronwall's lemma. Let us start to prove the first inequality in \eqref{eq:KresgrowA}. Fix $\xi \in \cF_+^{\leq N}$ and define $\varphi_\xi:\mathbb{R}\to \mathbb{R} $ by $\varphi_\xi (s) = \langle \xi, e^{-sA}  \cK_{\leq c_1 N^{\gamma_1}} e^{sA}\xi \rangle $ such that
		\[\begin{split} 
		\partial_s \varphi_\xi(s) &= 2\Re\langle \xi, e^{-sA} [  \cK_{\leq c_1 N^{\gamma_1}} , A_1 ]   e^{sA}\xi \rangle.\\
		%&  \hspace{0.4cm} +2\Re\langle \xi, e^{-sD} \cK_L \big[ (\cN_+ + 1)^{k }, D_1\big]   e^{sD}\xi \rangle.
		\end{split}\]
We notice first that 
		\[ \big[ \cK_{\leq c_1 N^{\gamma_1}}, b^*_{p+r} \big] = \big[ \cK_{\leq c_1 N^{\gamma_1}}, a^*_{-r} \big]=0 \]
if $ r\in P_H $ and $p\in P_L$, because $ |r|, |p+r| \geq N^{\alpha}- N^{\beta} > c_1 N^{\gamma_1}$ for all $N\in\NN$. %Hence, with the identity \eqref{eq:NresgrowlemA3}, we find that
		%\[\begin{split}
		%&  \langle \xi, e^{-sD}\cK_L \big[ (\cN_+ + 1)^{k }, D_1\big]   e^{sD}\xi \rangle  \\
		%& =  \frac k{\sqrt N} \sum_{r\in P_{H}, p \in P_{L}}  \eta_r  \langle e^{sD} \xi, \cK_L b^*_{r+p}a^*_{-r}a_p  (\cN_++\Theta(\cN_+)+1)^{k-1} e^{sD}\xi \rangle \\
		%&=  \frac k{\sqrt N} \sum_{r\in P_{H}, p \in P_{L}}  \eta_r  \langle e^{sD} \xi,  b^*_{r+p}a^*_{-r} \cK_L a_p  (\cN_++\Theta(\cN_+)+1)^{k-1} e^{sD}\xi \rangle \\
		%\end{split}\]
%and we can bound this contribution by
		%\begin{equation}\label{eq:KresgrowD2}
		%\begin{split}
		%& \big|  \langle \xi, e^{-sD}\cK_L \big[ (\cN_+ + 1)^{k }, D_1\big]   e^{sD}\xi \rangle \big| \\
		%&= \bigg| \frac k{\sqrt N} \sum_{r\in P_{H}, p}  \eta_r  \langle e^{sD} \xi,  b^*_{r+p}a^*_{-r} \cK_L a_p  (\cN_++\Theta(\cN_+)+1)^{k-1} e^{sD}\xi \rangle \bigg| \\
		%&= \bigg| \frac k{\sqrt N} \sum_{r\in P_{H}, p,q \in P_{L}}  q^2 \eta_r  \langle e^{sD} \xi,  b^*_{r+p}a^*_{-r} a^*_q a_q a_p  (\cN_++\Theta(\cN_+)+1)^{k-1} e^{sD}\xi \rangle \bigg| \\
		%&\leq \frac{C}{\sqrt N } \bigg( \sum_{r\in P_{H}, p,q \in P_{L}}q^2 \| a_{p+r}  a_{-r} a_q  (\cN_++1)^{(k-1)/2}  e^{sD}\xi \|^2\bigg)^{1/2}\\
		%&\hspace{1cm} \times \bigg( \sum_{r\in P_{H}, p,q \in P_{L}}q^2 \eta_r^2 \| a_{p} a_q   (\cN_++1)^{(k-1)/2}  e^{sD}\xi \|^2\bigg)^{1/2}\\
		%&\leq CN^{\kappa-\alpha/2} \| \cK_L^{1/2} (\cN_++1)^{k/2} e^{sD}\xi\|^2 = C N^{\kappa-\alpha/2}\varphi_\xi(s).
		%\end{split}
		%\end{equation}
%Now, let's switch to the first contribution on the r.h.s. of \eqref{eq:KresgrowD1}. 
Using the commutation relations \eqref{eq:ccr}, we then compute
		\begin{equation}\label{eq:KresgrowA3}
		\begin{split}
		[ \cK_{\leq c_1 N^{\gamma_1}}  , A_1 ] & =  - \frac1{\sqrt N} \sum_{\substack{ r\in P_{H}, p \in P_{L}: |p|\leq c_1 N^{\gamma_1} }} p^2 \eta_r b^*_{r+p}a^*_{-r}a_p.
		\end{split}
		\end{equation}
With \eqref{eq:KresgrowA3} and $|p|\leq N^{\beta}$ for $p\in P_L$, we then find that
		\begin{equation}\label{eq:KresgrowA4}
		\begin{split}
		&\big| \langle \xi, e^{-sA} [ \cK_{\leq c_1 N^{\gamma_1}}, A_1  ]  e^{sA}\xi \rangle\big| \\
		&\leq \frac {CN^{\beta}}{\sqrt N} \sum_{\substack{ r\in P_{H}, p \in P_{L}: |p|\leq c_1 N^{\gamma_1}  }} |p| |\eta_r| \|  a_{r+p}a_{-r} e^{sA}\xi \| \|  a_pe^{sA}\xi \| \\
		&\leq \frac {CN^{\beta+\kappa-\alpha/2}}{\sqrt N} \| (\cN_{\geq \frac12 N^\alpha }+1) e^{sA}\xi \| \| \cK_{\leq c_1 N^{\gamma_1}}^{1/2}e^{sA}\xi \|.
		\end{split}
		\end{equation}
Finally, using Lemma \ref{lm:NresgrowA} (with $ c=\frac12$, $ \gamma=\alpha$ and $N\in\NN$ sufficiently large), we conclude
		\begin{equation*}
		\begin{split}
		 \partial_s\varphi_\xi(s)  &\leq CN^{\beta+\kappa-\alpha/2-1/2} \| (\cN_{\geq \frac12 N^\alpha }+1) e^{sA}\xi \| \| \cK_{\leq c_1 N^{\gamma_1}}^{1/2}e^{sA}\xi \| \\
		&\leq CN^{2\beta+2\kappa-\alpha-1} \langle \xi, (\cN_{\geq \frac12 N^\alpha }+1)^2 \xi \rangle + C \varphi_\xi(s).
		\end{split}
		\end{equation*}
This proves the first inequality in \eqref{eq:KresgrowA}, by Gronwall's lemma. 

Next, let us prove the second inequality in \eqref{eq:KresgrowA}. We define $ \psi_\xi:\mathbb{R}\to \mathbb{R}$ by 
		\[ \psi_\xi(s) = \langle \xi, e^{-sA}  \cK_{\leq c_1 N^{\gamma_1}} (\cN_{\geq c_2 N^{\gamma_2}}+1) e^{sA}\xi \rangle, \]
and we compute
		\[\begin{split}
		\partial_s \psi_\xi(s) &= 2\Re \langle \xi, e^{-sA}  \big[ \cK_{\leq c_1 N^{\gamma_1}}, A_1\big] (\cN_{\geq c_2 N^{\gamma_2}}+1) e^{sA}\xi \rangle\\
		& \hspace{0.4cm} + 2\Re \langle \xi, e^{-sA}  \cK_{\leq c_1 N^{\gamma_1}} \big[ \cN_{\geq c_2 N^{\gamma_2}}, A_1 \big] e^{sA}\xi \rangle.
		\end{split}\]
First, we proceed as in \eqref{eq:KresgrowA4} and obtain with (\ref{eq:Ntheta-bds}) that
		\begin{equation}\label{eq:KresgrowA5}
		\begin{split}
		&\big| \langle \xi, e^{-sA} [ \cK_{\leq c_1 N^{\gamma_1}}, A_1 ](\cN_{\geq c_2 N^{\gamma_2}}+1)  e^{sA}\xi \rangle\big| \\
		&\leq \frac {CN^{\beta}}{\sqrt N} \sum_{\substack{ r\in P_{H}, p \in P_{L}:\\ |p|\leq c_1 N^{\gamma_1}  }} |p| |\eta_r| \| a_{r+p}(\cN_{\geq c_2 N^{\gamma_2}}+1)^{1/2} a_{-r} e^{sA}\xi \| \|a_p (\cN_{\geq c_2 N^{\gamma_2}}+1)^{1/2}e^{sA}\xi \| \\
		&\leq  \frac {CN^{\beta+\kappa-\alpha/2}}{\sqrt N} \| (\cN_{\geq c_2 N^{\gamma_2} }+1)(\cN_{\geq \frac12 N^\alpha }+1)^{1/2}e^{sD}\xi \| \| \cK_{\leq c_1 N^{\gamma_1}}^{1/2}(\cN_{\geq c_2 N^{\gamma_2}}+1)^{1/2}e^{sA}\xi \|.
		\end{split}
		\end{equation}
Eq. \eqref{eq:KresgrowA5} and Lemma \ref{lm:NresgrowA} then imply
		\begin{equation}\label{eq:KresgrowA6}\begin{split}
		&\big| \langle \xi, e^{-sA} [ \cK_{\leq c_1 N^{\gamma_1}}, A_1 ](\cN_{\geq c_2 N^{\gamma_2}}+1)  e^{sA}\xi \rangle\big| \\
		&\hspace{1.5cm}\leq CN^{2\beta+2\kappa-\alpha-1}  \langle \xi, (\cN_{\geq c_2 N^{\gamma_2} }+1)^2(\cN_{\geq \frac12 N^\alpha }+1)\xi \rangle +  C\psi_\xi(s).
		\end{split} \end{equation}
Next, we recall the identity in \eqref{eq:NresgrowlemA2} and that 
		\[ \big[ \cK_{\leq c_1 N^{\gamma_1}}, b^*_{p+r} \big] = \big[ \cK_{\leq c_1 N^{\gamma_1}}, a^*_{-r} \big]=0 \]
whenever $ r\in P_H, p \in P_L$ and $N\in\NN$, by assumption on $ c_1$ and $\gamma_1$. We then estimate
		\begin{equation}\label{eq:KresgrowA7}\begin{split}
		&\big| \langle \xi, e^{-sA}  \cK_{\leq c_1 N^{\gamma_1}}\big[ \cN_{\geq c_2 N^{\gamma_2}}, A_1\big]  e^{sA}\xi \rangle\big| \\
		&\leq \frac {C}{\sqrt N} \sum_{\substack{ r\in P_{H}, p \in P_{L},\\ v\in \Lambda_+^*:  |v|\leq c_1N^{\gamma_1} }} |v|^2 |\eta_r| \| a_{r+p}(\cN_{\geq c_2 N^{\gamma_2} }+1)^{-1/2} a_{-r} a_v e^{sD}\xi \|\\
		&\hspace{5cm}\times \|a_p (\cN_{\geq c_2 N^{\gamma_2} }+1)^{1/2}a_v e^{sD}\xi \| \\
		&\leq C N^{\kappa - \alpha/2} \langle e^{sA}\xi,  \cK_{\leq c_1 N^{\gamma_1}} (\cN_{\geq c_2 N^{\gamma_2}}+1) e^{sA}\xi \rangle \leq  C \psi_\xi(s).
		\end{split}\end{equation} 
Hence, putting \eqref{eq:KresgrowA6} and \eqref{eq:KresgrowA7} together, we have proved that
		\[ \partial_s \psi_\xi(s) \leq  CN^{2\beta+2\kappa-\alpha-1}  \langle \xi, (\cN_{\geq c_2 N^{\gamma_2} }+1)^2(\cN_{\geq \frac12 N^\alpha }+1)\xi \rangle+C \psi_\xi(s).\]
This implies the second bound in \eqref{eq:KresgrowA}, by Gronwall's lemma.
\end{proof}	

Next, we seek a bound for the growth of the potential energy operator. To this end, we first compute the commutator of $\cV_N$ with the antisymmetric operator $A$. We introduce here the shorthand notation for the low-momentum part of the kinetic energy 
\begin{equation}\label{eq:defKL} \cK_L= \sum_{p\in\Lambda_+^*: |p|\leq N^{\beta} } p^2 a^*_pa_p = \sum_{p\in P_L } p^2 a^*_pa_p.
\end{equation}

\begin{prop}\label{prop:comm1VN} Assume the exponents $\alpha,\beta$ satisfy (\ref{eq:alphabeta}). There exists a constant $C>0$ such that 
		\begin{equation} \label{eq:propcomm1VN}
		\begin{split}
		[\cV_N, A] & = \frac{1}{\sqrt N} \sum_{\substack{ u\in \Lambda_+^*, p\in P_L:\\ p+u\neq 0 }} N^{\kappa} (\widehat{V}(./N^{1-\kappa})\ast  \eta/N)(u) \big[ b^*_{p+u} a^*_{-u}a_p +\emph{h.c.}\big] + \cE_{[\cV_N,A]}
		\end{split}
		\end{equation}
where the self-adjoint operator $\cE_{[\cV_N,A]} $ satisfies 
		\begin{equation}\label{eq:propcommVNAerror}\begin{split}
		\pm \cE_{[\cV_N,A]}&\;\leq \delta \cV_N   + \delta^{-1}CN^{\kappa-2\beta -1} \cK_L (\cN_{\geq \frac12 N^{\alpha}}+1)+ \delta^{-1}C N^{2\alpha  +3\kappa-2} \cN_+\\
		&\;\hspace{0.4cm} +  \delta^{-1}CN^{\kappa -1} (\cN_{\geq \frac12 N^{\alpha}}+1)^2 
		\end{split}\end{equation}	
for all $\delta>0$ and for all $N\in\NN$ sufficiently large.
\end{prop}

\begin{proof} From \eqref{eq:defA} we have 
		\[[\cV_N, A] = [\cV_N, A_1] +\text{h.c.}\]
Following \cite[Prop. 8.1]{BBCS}, we find
		\begin{equation}\label{eq:commVNA11}\begin{split}
		[\cV_N, A_1]+\text{h.c.} =&\; \frac{1}{\sqrt N} \sum^*_{u\in \Lambda_+^*, v\in P_L} N^{\kappa} (\widehat{V}(./N^{1-\kappa})\ast  \eta/N)(u) b^*_{u+v} a^*_{-u}a_v \\
		& + \Theta_{1}+\Theta_{2}+\Theta_{3}+\Theta_{4} +\text{h.c.},
		\end{split}\end{equation}
where
		\begin{equation}\label{eq:commVNA12}\begin{split}
		\Theta_{1} &= - \frac{1}{N^{3/2}}\sum^*_{\substack{u \in\Lambda^*, v \in P_L, \\ r \in P_H^c\cup\{0\}}} N^{\kappa} \widehat{V}((u-r)/N^{1-\kappa})\eta_r b^*_{u+v} a_{-u}^* a_v,\\
		 \Theta_{2} &= \frac{1}{N^{3/2}}\sum^*_{\substack{u\in \Lambda^*,p\in \Lambda_+^*,\\ r\in P_{H} , v\in P_{L} }} N^{\kappa} \widehat{V} (u/N^{1-\kappa}) \eta_r b_{p+u}^* a_{v+r-u}^*a_{-r}^*a_{p}a_{v} ,\\
		 \Theta_{3} &= \frac{1}{N^{3/2}}\sum^*_{\substack{u\in \Lambda^*,p\in \Lambda_+^*,\\ r\in P_{H}, v\in P_{L} }}N^{\kappa}  \widehat{V} (u/N^{1-\kappa}) \eta_r b_{v+r}^* a_{p+u}^*a_{-r-u}^*a_{p}a_{v} ,\\
		 \Theta_{4} &= -\frac{1}{N^{3/2}}\sum^*_{\substack{u\in \Lambda^*,p\in \Lambda_+^*,\\ r\in P_{H} , v\in P_{L} }}N^{\kappa}  \widehat{V} (u/N^{1-\kappa}) \eta_r  b_{v+r}^* a_{-r}^*a_{p+u}^*a_{p}a_{v+u} .
		\end{split}\end{equation}
Here and in the following the notation $ \sum^*$ indicates that we only sum over those momenta for which the arguments of the creation and annihilation operators are non-zero. The first term on the r.h.s. of \eqref{eq:commVNA11} appears explicitly in \eqref{eq:propcomm1VN}, so let us estimate next the size of the operators $ \Theta_{1}$ to $\Theta_{4}$, defined in \eqref{eq:commVNA12}. The bounds can be obtained similarly as in the proof of \cite[Prop. 8.1]{BBCS}. 

Consider first $\Theta_{1}$. For $\xi\in \cF_+^{\leq N}$, we switch to position space and find 
\begin{equation}\label{eq:Theta11} \begin{split}
	|\langle \xi, \Theta_{1}\xi \rangle |\leq &\; \frac{1}{N^{1/2}} \sum_{ r\in P_H^c}  |\eta_r|\bigg( \int_{\Lambda^2}dxdy\; N^{2-2\kappa}V(N^{1-\kappa}(x-y))  \|\check{b}_x \check{a}_y \xi \|^2 \bigg)^{1/2}\\
	&\;\hspace{1.5cm} \times \bigg( \int_{\Lambda^2}dxdy\; N^{2-2\kappa}V(N^{1-\kappa}(x-y))   \Big\| \sum_{v\in P_L} e^{ivx}a_v \xi \Big\|^2 \bigg)^{1/2}\\
	\leq &\; CN^{\alpha  +3\kappa/2-1} \| \cV_N^{1/2}\xi \| \bigg( \int_{\Lambda}dx\; e^{i(v-v')x} \sum_{v, v'\in P_L} \langle\xi, a^*_{v'}a_{v} \xi\rangle \bigg)^{1/2}\\
	\leq&\;  CN^{\alpha  +3\kappa/2-1}\| \cV_N^{1/2}\xi \|\| \cN_{\leq N^{\beta}}^{1/2}\xi\|.
	\end{split}\end{equation}
The term $\Theta_{2}$ on the r.h.s. of (\ref{eq:commVNA12}) can be controlled by 
\[\begin{split}
	|\langle \xi, \Theta_{2}\xi \rangle |=&\; \bigg| \frac{1}{ N^{1/2} } \int_{\Lambda^2}dxdy \;N^{2-2\kappa}V(N^{1-\kappa}(x-y))\!\!\!\!\!\sum_{r\in P_H,  v\in P_{L} }\!\!\! e^{ivy}e^{iry} \eta_r \langle \xi, \check{b}_{x}^*\check{a}_y^*a^*_{-r} \check{a}_x a_v\xi \rangle \bigg|  \\
	\leq &\; \frac{ \|{\eta}_{H}\|  }{N^{1/2}} \bigg [ \int_{\Lambda^2}dxdy\; N^{2-2\kappa}V(N^{1-\kappa}(x-y)) \sum_{  v\in P_{L} }|v|^{-2}  \| \check{b}_x \check{a}_y  \xi \|^2 \bigg)^{1/2}\\
	&\;\hspace{0cm} \times \bigg( \int_{\Lambda^2}dxdy\; N^{2-2\kappa}V(N^{1-\kappa}(x-y)) \sum_{  v\in P_{L} } |v|^{2} \|(\cN_{\geq \frac12 N^{\alpha}}+1)^{1/2}\check{a}_x a_v \xi \|^2 \bigg)^{1/2}\\	
	\leq&\;  CN^{ \beta/2+3\kappa/2 -\alpha/2-1/2}\| \cV_N^{1/2}\xi \|\| \cK_L^{1/2}(\cN_{\geq \frac12 N^{\alpha}}+1)^{1/2}\xi\|.
	\end{split}\]
In the last step we used (\ref{eq:Ntheta-bds}) to estimate 
		\begin{equation}\label{eq:Nresax}\begin{split} \int_\Lambda dx\; \|(\cN_{\geq \frac12 N^{\alpha}}+1)^{1/2}\check{a}_x \xi\|^2 &= \sum_{p\in\Lambda_+^*}  \|(\cN_{\geq \frac12 N^{\alpha}}+1)^{1/2} a_p \xi\|^2\\
		&\leq C \sum_{p\in\Lambda_+^*}  \| a_p(\cN_{\geq \frac12 N^{\alpha}}+1)^{1/2}  \xi\|^2 \\ &= C\|\cN_+^{1/2}(\cN_{\geq \frac12 N^{\alpha}}+1)^{1/2} \xi\|^2 
		\end{split} \end{equation}
for any $\xi \in \cF_+^{\leq N}$. The contributions $\Theta_{3}$ and $\Theta_{4}$ can be bounded similarly. We find
	\[\begin{split}
	|\langle \xi, \Theta_{3}\xi \rangle |=&\; \bigg| \frac{1}{ N^{1/2} } \int_{\Lambda^2}dxdy \;N^{2-2\kappa}V(N^{1-\kappa}(x-y))\!\!\!\sum_{r \in P_H, v\in P_L } \!\!\! e^{-iry} \eta_r\langle \xi, b_{v+r}^*\check{a}_x^*\check{a}^*_y \check{a}_x a_v \xi \rangle\bigg|  \\
	\leq &\; \frac{C\|\eta_H\|}{N^{1/2}} \bigg( \int_{\Lambda^2}dxdy\; N^{2-2\kappa}V(N^{1-\kappa}(x-y))   \sum_{ v\in P_{L} }|v|^{-2} \|   \check{a}_x \check{a}_y  \xi \|^2 \bigg)^{1/2}\\
	&\;\hspace{0cm} \times \bigg( \int_{\Lambda^2}dxdy\; N^{2-2\kappa}V(N^{1-\kappa}(x-y))\sum_{ v\in P_{L} }  |v|^{2} \|(\cN_{\geq \frac12 N^{\alpha}}+1)^{1/2}\check{a}_x a_v \xi \|^2 \bigg)^{1/2}\\	
	\leq&\; CN^{\beta/2 + 3\kappa/2 -\alpha/2-1/2}  \| \cV_N^{1/2}\xi \|\| \cK_L^{1/2}(\cN_{\geq \frac12 N^{\alpha}}+1)^{1/2}\xi\|
	\end{split}\]
as well as
	\[ \begin{split}
	|\langle \xi, \Theta_{4}\xi \rangle |=&\; \bigg| \frac{1}{ N^{1/2} } \int_{\Lambda^2}dxdy \;N^{2-2\kappa}V(N^{1-\kappa}(x-y))\!\!\!\!\sum_{r\in P_H, v\in P_L} \!\!\!\!\eta_r e^{-ivy} \langle \xi, b_{v+r}^*a^*_{-r}\check{a}^*_x \check{a}_x \check{a}_y \xi \rangle\bigg|  \\
	\leq &\; \frac{C\|\eta_H\|}{N^{1/2}} \bigg [ \int_{\Lambda^2} dx dy \; N^{2-2\kappa}V(N^{1-\kappa}(x-y))\sum_{v\in P_{L}}\| \check{a}_x \check{a}_y  \xi \|^2 \bigg)^{1/2}\\
	&\;\hspace{0.4cm} \times \bigg [ \int_{\Lambda^2}dxdy \; N^{2-2\kappa}V(N^{1-\kappa}(x-y)) \sum_{r\in P_H,  v\in P_{L}} 
	\| \check{a}_x a_{v+r}a_{-r} \xi \|^2 \bigg)^{1/2}\\	
	\leq&\;  CN^{3\beta/2 + 3\kappa/2 -\alpha/2-1/2} \| \cV_N^{1/2}\xi \|\| (\cN_{\geq \frac12 N^{\alpha}}+1)\xi\|.
	\end{split} \]
Summarizing (using $\alpha > 3\beta + 2 \kappa$) we proved that 
		\begin{equation} \label{eq:commVNA13}
		\begin{split}
		\pm \sum_{i=1}^4(\Theta_{i} +\text{h.c.}) &\;\leq \delta \cV_N +\delta^{-1}C N^{2\alpha  +3\kappa-2} \cN_+ + \delta^{-1}CN^{\kappa-2\beta -1} \cK_L (\cN_{\geq \frac12 N^{\alpha}}+1)\\
		&\;\hspace{0.4cm} +  \delta^{-1}CN^{\kappa -1} (\cN_{\geq \frac12 N^{\alpha}}+1)^2 
		\end{split}
		\end{equation}
for any $\delta>0$. Setting $\cE_{[\cV_N,A]} = \sum_{i=1}^4(\Theta_{i} +\text{h.c.}) $, this proves the claim.
\end{proof}

From Proposition \ref{prop:comm1VN} we immediately get a bound for the action of $e^A$ on $\cV_N$. 
\begin{cor}\label{cor:VNgrowA} Assume the exponents $\alpha,\beta$ satisfy (\ref{eq:alphabeta}).  Then there exists a constant $C>0$ such that 
		\begin{equation} \label{eq:corVNgrowA}
		\begin{split}
		e^{-sA} \cV_N e^{sA} &\leq C \cV_N +C(N^\kappa +N^{2\alpha+3\kappa-2}) (\cN_++1) \\
		& \hspace{0.4cm}+CN^{\kappa-2\beta-1} \cK_L (\cN_{\geq \frac12 N^{\alpha} }+1) + C N^{\kappa-3\beta-2} (\cN_{\geq \frac12 N^{\alpha} }+1)^3.
		\end{split}\end{equation}	
for all $s\in [-1;1]$ and $N\in\NN$ large enough.
\end{cor}
\begin{proof} We apply Gronwall's lemma. Given $\xi\in\cF_+^{\leq N}$, we define $\varphi_\xi(s) = \langle\xi, e^{-sA }\cV_N e^{sA}\xi\rangle $
and compute its derivative s.t.
		\[\partial_s \varphi_\xi(s) =\langle\xi, e^{-sA }[\cV_N,A] e^{sA}\xi\rangle.  \]
Hence, we can apply \eqref{eq:propcomm1VN} and estimate 
		\[\begin{split} 
		&\bigg|\frac{1}{\sqrt N} \sum_{\substack{ u\in \Lambda_+^*, v\in P_L:\\ v+u\neq 0 }} N^{\kappa} \langle e^{sA }\xi, (\widehat{V}(./N^{1-\kappa})\ast  \eta/N)(u)b^*_{v+u} a^*_{-u}a_v e^{sA}\xi\rangle\bigg| \\
		& \leq \frac{N^{\kappa/2}\|\check{\eta}\|_\infty}{N} \bigg( \int_{\Lambda^2}dxdy\; N^{2-2\kappa}V(N^{1-\kappa}(x-y)) \|\check{a}_x \check{a}_y e^{sA}\xi\|^2 \bigg)^{1/2}\\
		&\hspace{3cm} \times \bigg( \int_{\Lambda^2}dxdy\; N^{3-3\kappa}V(N^{1-\kappa}(x-y)) \Big\|\sum_{v\in P_L} e^{ivx}a_v e^{sA}\xi\|^2 \bigg)^{1/2}\\		
		&\leq CN^{\kappa/2}\|\cV_N^{1/2} e^{sA}\xi\| \| \cN_{\leq N^\beta} e^{sA}\xi \| \leq   CN^{\kappa} \langle\xi, e^{-sA}\cN_+ e^{sA}\xi\rangle + C \varphi_\xi(s).
		\end{split}\]
Here, we used \eqref{eq:etainfbnd}, which shows that $\|\check{\eta}\|_\infty\leq CN $. Using Lemma \ref{lm:NresgrowA}, this simplifies to
		\begin{equation}\label{eq:corVN1}\begin{split} 
		&\bigg|\frac{1}{\sqrt N} \sum_{\substack{ u\in \Lambda_+^*, v\in P_L:\\ v+u\neq 0 }} N^{\kappa} \langle e^{sD }\xi, (\widehat{V}(./N^{1-\kappa})\ast  \eta/N)(u)b^*_{u+v} a^*_{-u}a_v e^{sD}\xi\rangle\bigg| \\
		&\leq C \varphi_\xi(s)  +  CN^\kappa\langle\xi,  (\cN_++1) \xi\rangle.
		\end{split}\end{equation}
Together with \eqref{eq:propcomm1VN}, the bound \eqref{eq:propcommVNAerror} (choosing $\delta=1$) and an application of Lemma \ref{lm:NresgrowA} as well as of Lemma \ref{lm:KresgrowA}, the claim follows from Gronwall's lemma.
\end{proof}

%%%%%%%%%%%%%%%%%%%%%%%%%%%%%%%%%%%%%%%%%%%%%%%%%%%
%%%%%%%%%%%%%%%%%%%%%%%%%%%%%%%%%%%%%%%%%%%%%%%%%%%
%%%%%%%%%%%%%%%%%%%%%%%%%%%%%%%%%%%%%%%%%%%%%%%%%%%
%%%%%%%%%%%%%%%%%%%%%%%%%%%%%%%%%%%%%%%%%%%%%%%%%%%
%%%%%%%%%%%%%%%%%%%%%%%%%%%%%%%%%%%%%%%%%%%%%%%%%%%

\section{Quartic Renormalization}
\label{sec:quartic} 

To explain why the bounds for $\cJ_N$ obtained in Prop. \ref{prop:JN} are not enough to show Theorem \ref{thm:main}, we introduce, for $ r\in \Lambda_+^*$, the operators 		
\begin{equation}\label{eq:defcrdr}
		c^*_r = \frac{1}{\sqrt N} \sum_{\substack{v\in\Lambda_+^*:v\neq -r,\\v\in P_L, v+r\in P_L^c }} a^*_{v+r}a_v, \hspace{1cm} e^*_r = \frac{1}{2 \sqrt N} \sum_{\substack{v\in\Lambda_+^*:v\neq -r,\\v\in P_L, v+r\in P_L }}a^*_{v+r}a_v.
		\end{equation}
We denote the adjoints of $ c^*_r$ and $e^*_r$ by $c_r$ and $e_r$, respectively. Notice in particular that $ e^*_r = e_{-r}$ for all $r\in\Lambda_+^*$. A straightforward computation shows that 
		\begin{equation}\label{eq:cubicbcd}\begin{split}
		 &\frac{8\pi\mathfrak{a}_0N^{\kappa}}{\sqrt N}\sum_{\substack{ p \in P_H^c,q\in P_L: \\ p+q\neq 0 }} \big[ b^*_{p+q}a^*_{-p}a_q+ \text{h.c.}\big]\\
		 & \hspace{1cm}= 8\pi\mathfrak{a}_0N^{\kappa} \sum_{p\in P_H^c} \Big[ b^*_{-p}e_{-p} + e^*_{-p}b_{-p} + b^*_{-p}e^*_{p} + e_{p}b_{-p} + b^*_{-p}c^*_{p} + c_p b_{-p} \Big].
		 \end{split} \end{equation}
Together with \eqref{eq:defJNeff}, this suggests to bound the Hamiltonian $ \cJ_N$ from below by completing the square in the operators $g_r^* := b^*_r + c^*_r + e^*_r$ and $g_r := b_r + c_r + e_r$, for $ r\in P_H^c\subset \Lambda_+^*$. A better look at (\ref{eq:defJNeff}) reveals, however, that several terms that are needed to complete the square are still hidden in the energy $\cH_N$. Since these terms are not small, we need to extract them from $\cH_N$ by conjugation with a unitary operator $e^{D}$, with 		
\begin{equation}\label{eq:defD}
\begin{split}
 D&= D_1 -D_1^*, \hspace{0.6cm} \text{ where }\hspace{0.4cm} D_1 = \frac1{2 N} \sum_{r\in P_{H}, p, q \in P_{L}} \eta_r a^*_{p+r}a^*_{q-r}a_pa_q. 
\end{split}
\end{equation} 
Since $[D,\cN_+] = 0$, we have the identity
\begin{equation}\label{eq:NgrowD} e^{-sD} (\cN_++1)^k e^{sD}  = (\cN_++1)^k
\end{equation}	
for all $k \in \NN$. 

Using $e^D$, we define the final excitation Hamiltonian 
\begin{equation}\label{eq:defMN} \cM_N = e^{-D} \cJ_N^{\text{eff}} e^D.
\end{equation}
The next proposition provides an important lower bound for $\cM_N$. Its proof is given in Section \ref{sec:MN}. 
\begin{prop} \label{prop:MN} Suppose the exponents $\alpha$ (in the definition of the set $P_H$ in (\ref{eq:PellHPellL1})) and $\beta$ (in the definition of the set $P_L$ in (\ref{eq:defPHPL})) are such that 
\begin{equation}\label{eq:condab} i) \hspace{0.2cm} \alpha> 3\beta+2\kappa ,\hspace{0.5cm} ii)\hspace{0.2cm} 1> \alpha+\beta +2\kappa ,\hspace{0.5cm} iii)\hspace{0.2cm} 5\beta >\alpha , \hspace{0.5cm} iv)\hspace{0.2cm} \beta > 3\kappa,  \hspace{0.5cm} v)\hspace{0.2cm} 1/2>\beta, \end{equation}
Set $ \gamma= \min(\alpha, 1-\alpha-\kappa)$ ($\gamma > 0$ from (\ref{eq:condab})) and let $m_0\in \mathbb{R}$ be s.t. $ m_0\beta=\alpha$. Let $V\in L^3(\bR^3)$ be compactly supported, pointwise non-negative and spherically symmetric. Then, $ \cM_N$, as defined as in \eqref{eq:defMN}, is bounded from below by
		\begin{equation}
		\begin{split}\label{eq:MN} 
		\cM_{N} \geq &\;4\pi \frak{a}_0 N^{1+\kappa}+ \frac14 \cK  + \cE_{\cM_N}
		\end{split}
		\end{equation} 
for a self-adjoint operator $ \cE_{\cM_N}$ satisfying 
		\begin{equation}\label{eq:propMNerrorbnd} 
		\begin{split}
		e^{A}e^D \cE_{\cM_{N}} &e^{-D} e^{-A}\\
		\geq \; &- C N^{-\beta} \cK - C N^{-\beta-\kappa}\cV_N - C N^{\beta + 2\kappa -1} \cK \cN_{\geq N^\beta} - C N^{\alpha+\beta+ 2\kappa-1}  \cK \cN_{\geq N^{\lfloor m_0 \rfloor \beta}}  \\
		&- C \sum_{j=3}^{2\lfloor m_0\rfloor -1}N^{ j\beta/2 + \beta/2+2\kappa-1}\cK \cN_{\geq \frac12 N^{j\beta/2}} - CN^{3\alpha+\kappa} 
		\end{split}  
		\end{equation}
for all $N\in\NN$ sufficiently large. 
\end{prop}

%%%%%%%%%%%%%%%%%%%%%%%%%%%%%%%%%%%%%%%%%%%%%%%%%%%
%%%%%%%%%%%%%%%%%%%%%%%%%%%%%%%%%%%%%%%%%%%%%%%%%%%
%%%%%%%%%%%%%%%%%%%%%%%%%%%%%%%%%%%%%%%%%%%%%%%%%%%
%%%%%%%%%%%%%%%%%%%%%%%%%%%%%%%%%%%%%%%%%%%%%%%%%%%
%%%%%%%%%%%%%%%%%%%%%%%%%%%%%%%%%%%%%%%%%%%%%%%%%%%

\section{Proof of Theorem \ref{thm:main}}
\label{sec:proof}

For $\varepsilon>0 $ sufficiently small, we define 
\begin{equation}\label{eq:defalphabeta} \alpha = 14\kappa +4\varepsilon , \hspace{1cm} \beta = 4\kappa + \varepsilon \, . 
\end{equation}
The choice $\kappa < 1/43$ guarantees, if $\eps > 0$ is small enough, that all conditions in (\ref{eq:condab}) (and thus also in (\ref{eq:alp}) and (\ref{eq:alphabeta})) are satisfied.

From (\ref{eq:GN-prel}) and (\ref{eq:Gbd0}), we obtain the upper bound 
\begin{equation} \label{eq:up-pr} E_N \leq 4 \pi \frak{a}_0 N^{1+\kappa} + C N^{16\kappa + 4 \eps} \end{equation}
for the ground state energy of $H_N$. From (\ref{eq:GN-prel}) and (\ref{eq:theta-err}), on the other hand, we obtain 
\[ \cH_N \leq 2 (\cG_N \textendash 4 \pi \frak{a}_0 N^{1+\kappa}) + C N^\kappa \cN_+ + C N^{16\kappa + 4 \eps} \]
With (\ref{eq:up-pr}) and setting $\cG'_N = \cG_N - E_N$, we deduce that 
\begin{equation}\label{eq:prop34} \cH_N \leq 2 \cG'_N + C N^\kappa \cN_+ + C N^{16\kappa + 4 \eps} \end{equation}
 
Next, we prove (\ref{eq:BEC0}). From (\ref{eq:GeffE}) and (\ref{eq:prop34}) we arrive  at 
\[\begin{split}
\cG_N &=  \cG_N^{\text{eff}}+ \cE_{\cG_N} \geq \cG_N^\text{eff} - C N^{-(7\kappa+2\eps)/2} \cG'_N
- C N^{-(5\kappa+2\eps)/2} \cN_+ - C N^{16\kappa + 4 \eps} 
\end{split} \]
Writing $\cG_\text{eff} = e^A \cJ_N e^{-A}$ and recalling that $\kappa < 1/43$ (and that $\eps > 0$ is small enough), Prop. \ref{prop:JN} and  (\ref{eq:prop34}) imply that 
\[\begin{split}
\cG_N \geq \; & e^A \cJ_N^{\text{eff}}  e^{-A} + e^A \cE_{\cJ_N}  e^{-A} - C N^{- (7\kappa + 2 \eps) /2} \cG'_N  - CN^{-(5\kappa+2\eps)/2}  \cN_+ - C N^{16\kappa + 4 \eps} \\ \geq \; &e^A \cJ_N^{\text{eff}}  e^{-A} - C N^{-(5 \kappa + 2\eps)/2}  \cG'_N -  C N^{- (3\kappa + 2\eps)/2}  \cN_+  - C N^{16\kappa+4\eps} \end{split} \]
Inserting $\cJ_\text{eff} = e^D \cM_N e^{-D}$ and applying Prop. \ref{prop:MN}, we obtain  
\begin{equation}\label{eq:low-k0} \begin{split}
\cG_N \geq \; & 4 \pi \frak{a}_0 N^{1+\kappa} + \frac{1}{4} e^A e^D \cK e^{-D} e^{-A} + e^A e^D \cE_{\cM_N} e^{-D} e^{-A}  \\ & - C N^{-(5 \kappa + 2\eps)/2}  \cG'_N -  C N^{- (3\kappa + 2\eps)/2}  \cN_+  - C N^{16\kappa+4\eps}  \end{split} \end{equation} 
With $\cK \geq (2\pi)^2 \cN_+$ and Lemma \ref{lm:NresgrowA} (with $m=0$ and  $k=1$) we have 
\begin{equation}\label{eq:low-K}\begin{split} e^A e^D \cK e^{-D} e^{-A} &\geq (2\pi)^2  e^A e^D \cN_+ e^{-D} e^{-A}  = (2\pi)^2  e^A \cN_+ e^{-A}  \geq c \cN_+  \end{split} \end{equation} 
for a constant $c > 0$ small enough (but independent of $N$). If $N$ is large enough, we conclude (using also the upper bound (\ref{eq:up-pr})), that 
\begin{equation}\label{eq:cN+bd} \cN_+ \leq C \cG'_N - C e^A e^D \cE_{\cM_N} e^{-D} e^{-A} + C N^{16 \kappa + 4 \eps} \end{equation}
To bound the error term $e^A e^D \cE_{\cM_N} e^{-D} e^{-A}$, we need (according to (\ref{eq:propMNerrorbnd})) to control observables of the form $N^{-1} \cK \cN_{\geq c N^\gamma}$. To this end, we observe, first of all, that, by Cauchy-Schwarz and by (\ref{eq:prop34}), 
\begin{equation}\label{eq:commbnd1}
		\begin{split}
		N^{-1}\cK\cN_{\geq cN^{\gamma}}\leq&\;  \delta^{-1} N^{\kappa-2\gamma }\cK +   \delta N^{2\gamma-\kappa-2}\cK\cN_{\geq cN^{\gamma}}^2\\
		\leq &\; \delta^{-1} N^{\kappa-2\gamma} \cK + 2 \delta N^{2\gamma-\kappa-2}\cN_{\geq cN^{\gamma}} \cG_N' \cN_{\geq cN^{\gamma}} + C \delta  N^{-1}  \cK\cN_{\geq cN^{\gamma}}.
		\end{split}
		\end{equation}
Choosing $\delta>0$ sufficiently small, we thus have
		\begin{equation}\label{eq:commbnd2}
		\begin{split}
		&N^{-1}\cK \cN_{\geq cN^{\gamma}} \leq  C N^{\kappa-2\gamma} \cK + C N^{2\gamma-\kappa-2} \cN_{\geq cN^{\gamma}}\cG_N' \cN_{\geq cN^{\gamma}}.
		\end{split}
		\end{equation}
We write 
\begin{equation}\label{eq:commbnd3}
\cN_{\geq cN^{\gamma}} \cG_N' \cN_{\geq cN^{\gamma}} = \cN_{\geq cN^{\gamma}}^2 \cG_N'  +  \cN_{\geq cN^{\gamma}} [\cG_N', \cN_{\geq cN^{\gamma}}] .
		\end{equation}
Using \eqref{eq:prop34} (similarly as we did in (\ref{eq:commbnd1})) and $\cN_{\geq cN^\gamma} \leq N$, $\cN_{\geq cN^{\gamma}}\leq CN^{-2\gamma}\cK $, we can bound the expectation of the first term on the 
r.h.s. of the last equation, for an arbitrary $\xi \in \cF_+^{\leq N}$, by  
		\begin{equation}\label{eq:commbnd4}\begin{split}
		&|\langle\xi,  \cN_{\geq cN^{\gamma}}^2\cG_N'  \xi \rangle | \\
		&\;\;\; \leq \langle\xi,  \cN_{\geq cN^{\gamma}}^3 \xi \rangle^{1/2}\langle\xi, \cG_N' \cN_{\geq cN^{\gamma}}\cG_N' \xi \rangle^{1/2}\\
		&\;\;\; \leq CN^{1/2-\gamma} \langle\xi,  \cK\cN_{\geq cN^{\gamma}}^2 \xi \rangle^{1/2} \langle \xi, \cG^{'2}_N \xi \rangle^{1/2} \\
		&\;\;\;\leq  CN^{1/2-\gamma} \langle \xi, \cG^{'2}_N \xi \rangle^{1/2}\langle\xi,  \cN_{\geq cN^{\gamma}} \cG_N' \cN_{\geq cN^{\gamma}} \xi \rangle^{1/2} \\ &\hspace{2cm}+  CN^{1+\kappa/2-2\gamma}\langle \xi, \cG^{'2}_N \xi \rangle^{1/2} \langle\xi,  \cK\cN_{\geq cN^{\gamma}} \xi \rangle^{1/2}\\
		&\;\;\;\leq  \delta \langle\xi,  \cN_{\geq cN^{\gamma}} \cG_N' \cN_{\geq cN^{\gamma}} \xi \rangle + C \delta^{-1} N^{1-2\gamma} \langle \xi, \cG^{'2}_N \xi \rangle \\ &\hspace{2cm}+  C \delta N^{1+\kappa-2\gamma}  \langle\xi,  \cK\cN_{\geq cN^{\gamma}} \xi \rangle^{1/2}.
		\end{split}\end{equation}
On the other hand, to estimate the commutator term in equation \eqref{eq:commbnd3}, we notice that $ \cA := (\cH_N+1)^{-1/2} i [\cG_N',\cN_{\geq cN^{\gamma}}  ] (\cH_N+1)^{-1/2} $ is a bounded, self-adjoint operator with $ \| \cA \| \leq CN^{\kappa +\alpha/2-\gamma} + CN^{\kappa +\gamma/2}$, by \eqref{eq:errComm}. Setting $ \mu = \max(\alpha, 3\gamma)$, this implies, with (\ref{eq:prop34}), 
		\begin{equation}\label{eq:commbnd5}
		\begin{split}
		&| \langle\xi,  \cN_{\geq cN^{\gamma}}[\cG_N', \cN_{\geq cN^{\gamma}}]\xi \rangle | \\
		& \leq \delta \langle\xi,  \cN_{\geq cN^{\gamma}}(\cH_N +1)\cN_{\geq cN^{\gamma}}\xi \rangle  + C \delta^{-1} N^{2\kappa-2\gamma+\mu}\langle\xi,  (\cH_N +1) \xi \rangle\\
		& \leq 2\delta \langle\xi,  \cN_{\geq cN^{\gamma}}\cG_N' \cN_{\geq cN^{\gamma}}\xi \rangle + C \delta N^{1+\kappa-2\gamma} \langle\xi,  \cK \cN_{\geq cN^{\gamma}}\xi \rangle \\
		&\hspace{0.4cm}+  C \delta^{-1} N^{3\kappa-2\gamma+\mu} \langle \xi, \cN_+ \xi \rangle + C \delta^{-1} N^{3\kappa+\alpha -2\gamma+\mu} \| \xi \|^2 \, 
		\end{split}
		\end{equation}
for all $\xi \in \cF_+^{\leq N}$. Plugging \eqref{eq:commbnd4} and \eqref{eq:commbnd5} into \eqref{eq:commbnd3}, we find that, for sufficiently small $\delta > 0$,  
		\begin{equation}\label{eq:commbnd6}\begin{split}
		 \cN_{\geq cN^{\gamma}} \cG_N' \cN_{\geq cN^{\gamma}}  \leq \; &C \delta  N^{1+\kappa-2\gamma}  \cK\cN_{\geq cN^{\gamma}}  + C \delta^{-1} N^{1- 2\gamma} \cG^{'2}_N \\
		&+ C \delta^{-1} N^{3\kappa-2\gamma+\mu} \cN_+  + C \delta^{-1} N^{3\kappa -2\gamma+\mu +\alpha} 
		\end{split}\end{equation}
Inserting into \eqref{eq:commbnd2} and choosing $\delta > 0$ small enough, we obtain 
		\begin{equation}\label{eq:first-low} \begin{split} 
		&N^{-1} \cK \cN_{\geq cN^{\gamma}} \leq CN^{\kappa-2\gamma} \cK + C N^{-\kappa-1} \cG^{'2}_N + C N^{2\kappa +\mu-2} \cN_+ + C N^{2\kappa + \mu +\alpha-2} \end{split} \end{equation}
Applying (\ref{eq:first-low}) to the r.h.s. of  (\ref{eq:propMNerrorbnd}) we find, using also (\ref{eq:prop34}), (\ref{eq:defalphabeta}), and the choice $\kappa < 1/43$,  
\begin{equation}\label{eq:bnd-cEMN} \begin{split} e^A e^D \cE_{\cM_N} e^{-D} e^{-A}  &\geq - C N^{-\eps}  \cN_+  - 
C N^{-(\kappa+\eps)} \cG'_N - C N^{13\kappa + 3 \eps -1} \cG^{'2}_N - C N^{43 \kappa + 12 \eps} \end{split} \end{equation} 
Inserting the last equation into (\ref{eq:cN+bd}) and using (\ref{eq:up-pr}), we conclude that, for $N$ large enough,  
\[ \cN_+ \leq C \cG'_N + C N^{13\kappa + 3\eps-1} \cG^{'2}_N + C N^{43 \kappa + 12\eps} \]
For $\psi_N \in L^2_s (\Lambda^N)$ with $\| \psi_N \| = 1$ and $\langle \psi_N, (H_N - E_N)^2 \psi_N \rangle \leq \zeta^2$, the corresponding excitation vector $\xi_N = e^{B} U_N \psi_N$ is such that $\langle \xi_N, \cG^{'2}_N \xi_N \rangle \leq \zeta^2$ and thus 
\[ \langle \xi_N, \cN_+ \xi_N \rangle \leq C \left[ \zeta + \zeta^2 N^{13\kappa+3\eps-1} +  N^{43 \kappa + 12\eps} \right] \]
which proves (\ref{eq:BEC0}), using Lemma \ref{lm:Ngrow}. From (\ref{eq:prop34}), we obtain also 
\begin{equation}\label{eq:k0-case}  \langle \xi_N, \cH_N \xi_N \rangle \leq C \left[ \zeta N^{\kappa} + \zeta^2 N^{14\kappa + 3\eps-1}  +  N^{44 \kappa + 12\eps} \right],  \end{equation} 
 an estimate that will be needed to arrive at (\ref{eq:BEC}).
 
Evaluating (\ref{eq:bnd-cEMN}) on a normalized ground state $\xi_N$ of $\cG_N$ and inserting the result in (\ref{eq:low-k0}) we also deduce that 
\[ E_N \geq 4 \pi \frak{a}_0 N^{1+\kappa} - C N^{43 \kappa + 12\eps} \]
Together with the upper bound (\ref{eq:up-pr}), this concludes the proof of (\ref{eq:gs-est}). 

We still have to show (\ref{eq:BEC}) for $k > 0$. To this end, we will prove the stronger bound (\ref{eq:HNk}); Eq. (\ref{eq:BEC}) follows then immediately from $\cN_+ \leq \cH_N$ and by 
Lemma \ref{lm:Ngrow}. We denote by $Q_\zeta$ the spectral subspace of $\cG_N$ associated with energies below $E_N + \zeta$. We use induction to show that, for all $k \in \bN$, there exists a constant $C > 0$ (depending on $k$) such that  
\begin{equation}\label{eq:ind-assu} \sup_{\xi \in  Q_{\zeta}} \frac{\langle \xi , (\cH_N +1)  (\cN_++1)^{2k} \xi \rangle}{ \| \xi \|^2} \leq C \left[ N^{44\kappa+12\eps} + \zeta^2 N^{20\kappa + 5\eps}  \right]^{2k+1} \end{equation}
for all $k \in \bN$. This proves (\ref{eq:HNk}) and thus, with the bound $ \cN_+\leq \cH_N$ and with Lemma \ref{lm:Ngrow}, also (\ref{eq:BEC}). The case $k=0$ follows from (\ref{eq:k0-case}). From now on, we assume (\ref{eq:ind-assu}) to hold true and we prove the same bound, with $k$ replaced by $(k+1)$ (and with a new constant $C$). To this end, we start by observing that, combining (\ref{eq:prop34}) and \eqref{eq:cN+bd}, 
\[ \cH_N +1 \leq  C N^\kappa \cG'_N - C N^\kappa e^A e^D \cE_{\cM_N} e^{-D} e^{-A} + C N^{17 \kappa + 4 \eps} \] 
Hence
\begin{equation}\label{eq:HNk+1}  \begin{split} (\cN_+ +1)^{2(k+1)} (\cH_N +1) &= (\cN_+ +1)^{k+1} (\cH_N +1) (\cN_+ +1)^{k+1} \\ & \leq C N^\kappa  (\cN_+ +1)^{k+1} \cG'_N (\cN_+ +1)^{k+1} \\ &\hspace{.3cm} - C N^\kappa (\cN_+ +1)^{k+1} e^A e^D \cE_{\cM_N} e^{-D} e^{-A} (\cN_+ +1)^{k+1}  \\ &\hspace{.3cm} + C N^{17\kappa+4\eps} (\cN_+ +1)^{2(k+1)}  \end{split} \end{equation}
We estimate the first term on the r.h.s. by 
\[ \begin{split} 
N^\kappa &(\cN_+ +1)^{k+1} \cG'_N (\cN_+ +1)^{k+1} \\ &\leq N^\kappa (\cN_+ +1)^{2(k+1)} \cG'_N + N^\kappa (\cN_+ +1)^{k+1} [ \cG'_N, (\cN_+ +1)^{k+1} ] \\ &=  N^\kappa (\cN_+ +1)^{2(k+1)} \cG'_N + N^\kappa  \sum_{j=1}^{k+1} {k+1 \choose j} (\cN_+ +1)^{k+1} \text{ad}_{\cN_+}^{(j)} (\cG_N) (\cN_+ +1)^{k+1-j} \end{split} \]
By Cauchy-Schwarz, we find 
\[ \begin{split} N^\kappa &(\cN_+ +1)^{k+1} \cG'_N (\cN_+ +1)^{k+1}  \\ &\leq N^\kappa  (\cN_+ +1)^{2(k+1)} + N^\kappa  \cG'_N (\cN_+ +1)^{2(k+1)} \cG'_N  \\ &\hspace{.4cm} + N^\kappa  \sum_{j=1}^{k+1} {k+1 \choose j} (\cN_+ +1)^{k+1} \text{ad}_{\cN_+}^{(j)} (\cG_N) (\cN_+ +1)^{k+1-j} \end{split} \]
With $(\cN_+ +1)^{2(k+1)} \leq (\cN_+ +1)^{2k+1} (\cH_N+1)$ and with the estimate 
\begin{equation}\label{eq:est-Aj} \| (\cH_N + 1)^{-1/2} \text{ad}_{\cN_+}^{(j)} (\cG_N) (\cH_N +1)^{-1/2} \| \leq C N^{7\kappa/3 + 2 \eps/3} \end{equation} 
from (\ref{eq:adjGN}) we obtain, using again Cauchy-Schwarz,  
\[ \begin{split} 
N^\kappa \langle \xi, &(\cN_+ +1)^{k+1} \cG'_N (\cN_+ +1)^{k+1}  \xi \rangle  \\ \leq \; &C \left[ N^\kappa \zeta^2 + N^{7\kappa/3 + 2 \eps/3} \right] \| \xi \|^2 \\ & \times \left[ \sup_{\xi \in Q_{\zeta}} \frac{\langle \xi, (\cN_+ +1)^{2(k+1)} (\cH_N+1) \xi \rangle}{\| \xi \|^2} \right]^{1/2}  \left[ \sup_{\xi \in Q_{\zeta}} \frac{\langle \xi, (\cN_+ +1)^{2k} (\cH_N+1) \xi \rangle}{\| \xi \|^2} \right]^{1/2} \end{split} \]
for every $\xi \in Q_{\zeta}$. Hence, for any $\delta > 0$, we have 
\begin{equation}\label{eq:NkG'Nk-bd} \begin{split} 
N^\kappa &\frac{\langle \xi, (\cN_+ +1)^{k+1} \cG'_N (\cN_+ +1)^{k+1}  \xi \rangle}{\| \xi \|^2} \\ &\leq \delta  \sup_{\xi \in Q_{\zeta}} \frac{\langle \xi, (\cN_+ +1)^{2(k+1)} (\cH_N+1) \xi \rangle}{\| \xi \|^2}  \\ &\hspace{2cm}  + C \delta^{-1} \left[ N^\kappa \zeta^2 + N^{7\kappa/3 + 2 \eps/3} \right]^2  \sup_{\xi \in Q_{\zeta}} \frac{\langle \xi, (\cN_+ +1)^{2k} (\cH_N +1) \xi \rangle}{\| \xi \|^2} \end{split} \end{equation}
To bound the contribution proportional to $e^A e^D \cE_{\cM_N} e^{-D} e^{-A}$ on the r.h.s. of (\ref{eq:HNk+1}), we have to control, according to (\ref{eq:commbnd2}), terms of the form
\[ \begin{split} (\cN_+ &+1)^{k+1} \cN_{\geq cN^\gamma} \cG'_N \cN_{\geq cN^\gamma} (\cN_+ +1)^{k+1} \\ & = 
((\cN_++1)^{k+1} \cN_{\geq cN^\gamma})^2 \cG'_N + (\cN_+ +1)^{k+1} \cN_{\geq cN^\gamma} \left[ \cG'_N, (\cN_+ +1)^{k+1} \cN_{\geq cN^\gamma} \right] \\ & =: \text{A} + \text{B} \end{split} \]
For an arbitrary $\xi\in Q_{\zeta}$, we can bound the expectation of $\text{A}$ by Cauchy-Schwarz as 
\begin{equation}\label{eq:A-bd} \begin{split}  \frac{\langle \xi, \text{A} \xi \rangle}{\| \xi \|^2} &\leq \frac{\langle \xi, ((\cN_+ +1)^{k+1} \cN_{\geq cN^\gamma})^2 \xi \rangle}{\| \xi \|^2}   + \frac{\langle \cG'_N \xi, ((\cN_+ +1)^{k+1} \cN_{\geq cN^\gamma})^2 \cG'_N \xi \rangle}{\| \xi \|^2}  \\ &\leq N^2 (1+ \zeta^2) \, \sup_{\xi \in Q_{\zeta}} \frac{\langle \xi, (\cN_+ +1)^{2k} \cN_{\geq cN^\gamma}^2 \xi \rangle}{\| \xi \|^2}  \\ & \leq N^{2-2\gamma} (1+ \zeta^2) \, \sup_{\xi \in Q_{\zeta}} \frac{\langle \xi, (\cN_+ +1)^{2k+1} \cK  \xi \rangle}{\| \xi \|^2} \\  & \leq N^{2-2\gamma} (1+ \zeta^2) \, \left[ \sup_{\xi \in Q_{\zeta}} \frac{\langle \xi, (\cN_+ +1)^{2k} \cK  \xi \rangle}{\| \xi \|^2} \right]^{1/2} \\&\hspace{5cm} \times \left[ \sup_{\xi \in Q_{\zeta}} \frac{\langle \xi, (\cN_+ +1)^{2(k+1)} \cK  \xi \rangle}{\| \xi \|^2} \right]^{1/2}  \end{split} \end{equation}
As for the term $\text{B}$, we can write 
\[ \begin{split}  \text{B} = \; &(\cN_+ +1)^{k+1} \cN_{\geq c N^\gamma}^2  \left[ \cG'_N, (\cN_+ +1)^{k+1} \right] \\ &+ (\cN_+ +1)^{k+1} \cN_{\geq c N^\gamma}  \left[ \cG'_N, \cN_{\geq cN^\gamma}  \right] (\cN_+ +1)^{k+1}  \\
= \;  & \sum_{j=1}^{k+1} {k+1 \choose j} (\cN_+ +1)^{k+1} \cN_{\geq c N^\gamma}^2  \text{ad}^{(j)}_{\cN_+} (\cG'_N) (\cN_+ +1)^{k+1-j} \\ &+ 
(\cN_+ +1)^{k+1} \cN_{\geq c N^\gamma}  \left[ \cG'_N, \cN_{\geq cN^\gamma}  \right] (\cN_+ +1)^{k+1} 
\end{split} \]
From (\ref{eq:est-Aj}) and using (\ref{eq:errComm}) to estimate 
\[ \| (\cH_N + 1)^{-1/2} [ \cN_{\geq c N^\gamma} , \cG_N ] (\cH_N + 1)^{-1/2} \| \leq C N^{8\kappa+2\eps-\gamma}+C N^{\kappa+\gamma/2},  \]
we obtain for every $\xi \in Q_{\zeta}$ that 
\[\begin{split}  |\langle \xi, \text{B} \xi \rangle | \leq \; &CN^{7\kappa/3+2\eps/3} \| (\cH_N+1)^{1/2} \cN_{\geq c N^\gamma}^2 (\cN_+ +1)^{k+1} \xi \| \| (\cH_N + 1)^{1/2} (\cN_+ +1)^{k} \xi \| \\ 
&+ C N^{8\kappa+2\eps-\gamma}\| (\cH_N+1)^{1/2} \cN_{\geq c N^\gamma} (\cN_+ +1)^{k+1} \xi \| \| (\cH_N + 1)^{1/2} (\cN_+ +1)^{k+1} \xi \|
\\ &+C N^{\kappa+\gamma/2}\| (\cH_N+1)^{1/2} \cN_{\geq c N^\gamma} (\cN_+ +1)^{k+1} \xi \| \| (\cH_N + 1)^{1/2} (\cN_+ +1)^{k+1} \xi \|.
\end{split} \]
Applying the bounds $ \cN_+\leq N$, $ \cN_{\geq c N^\gamma}\leq C N^{-2\gamma}\cK$ and \eqref{eq:prop34} yields on the one hand
		\[\begin{split}
		&\| (\cH_N+1)^{1/2} \cN_{\geq c N^\gamma} (\cN_+ +1)^{k+1} \xi \| \| (\cH_N + 1)^{1/2} (\cN_+ +1)^{k+1} \xi \|\\ 
		&\hspace{1cm}\leq C \| \cG_N'\cN_{\geq  cN^\gamma}  (\cN_++1)^{k+1} \xi\|\| (\cH_N + 1)^{1/2} (\cN_+ +1)^{k+1} \xi \|\\
		&\hspace{1.5cm} + CN^{1+\kappa/2-\gamma}\| (\cH_N + 1)^{1/2} (\cN_+ +1)^{k+1} \xi \|^2\\
		& \hspace{1cm}\leq\delta \langle \xi,  (\cN_++1)^{k+1}\cN_{\geq  cN^\gamma} \cG_N'\cN_{\geq  cN^\gamma}  (\cN_++1)^{k+1} \xi\rangle \\
		&\hspace{1.5cm} +C(\delta^{-1}+N^{1+\kappa/2-\gamma})\| (\cH_N + 1)^{1/2} (\cN_+ +1)^{k+1} \xi \|^2
		\end{split}\]
for any $\delta> 0$. Since $ 8\kappa+2\eps-\gamma \leq  1+\kappa/2-\gamma$ and $ \kappa+\gamma/2 \leq  1+\kappa/2-\gamma$ for all $\gamma\leq \alpha$ if $ \kappa<1/43$, this implies with the choice $ \delta =\frac14 (N^{8\kappa+2\eps-\gamma} + N^{ \kappa+\gamma/2} )^{-1} $ that
		\begin{equation}
		\begin{split}\label{eq:cGNNgNg}
		 |\langle \xi, \text{B} \xi \rangle | \leq \;& CN^{7\kappa/3+2\eps/3} \| (\cH_N+1)^{1/2} \cN_{\geq c N^\gamma}^2 (\cN_+ +1)^{k+1} \xi \| \| (\cH_N + 1)^{1/2} (\cN_+ +1)^{k} \xi \|\\
		 & +C(N^{1+17\kappa/2+2\eps-\gamma}+N^{1+3\kappa/2-\gamma/2})\| (\cH_N + 1)^{1/2} (\cN_+ +1)^{k+1} \xi \|^2\\
		 & + \frac 14 \langle \xi,  (\cN_++1)^{k+1}\cN_{\geq  cN^\gamma} \cG_N'\cN_{\geq  cN^\gamma}  (\cN_++1)^{k+1} \xi\rangle.
		\end{split}
		\end{equation}
%Since $[\cH_N , \cN_+] = 0$, we immediately find \[ \| (\cH_N + 1)^{1/2} (\cN_+ +1)^{k+1} \xi \| \leq N \| (\cH_N + 1)^{1/2} (\cN_+ +1)^{k} \xi \|.\] 
On the other hand, we can estimate 
\begin{equation}\label{eq:cHNNg} 
\begin{split}& \| (\cH_N + 1)^{1/2} \cN_{\geq cN^\gamma}^2 (\cN_+ +1)^{k+1} \xi \| \\ &\hspace{1cm}\leq  N \|  (\cK + 1)^{1/2} \cN_{\geq cN^\gamma}(\cN_+ +1)^{k+1} \xi \| + \| \cV_N^{1/2} \cN_{\geq cN^\gamma}^2 (\cN_+ +1)^{k+1} \xi \|. \end{split} \end{equation} 
Expressing $\cV_N$ in position space, we find, with $\phi= \cN_{\geq cN^\gamma}(\cN_+ +1)^{k+1} \xi$, 
\begin{equation}\label{eq:cVNNg} \| \cV_N^{1/2} \cN_{\geq cN^\gamma} \phi \|^2 = \int dx dy \, N^{2-2\kappa} V (N^{1-\kappa} (x-y)) \| \check a_x \check a_y \cN_{\geq cN^\gamma} \phi \|^2 \end{equation} 
We have 
\[ \check a_x \cN_{\geq cN^\gamma} = (\cN_{\geq cN^\gamma} + 1) \check a_x - a (\check{\chi}_x) \]
where \[ \check{\chi}_x (y) = \check{\chi} (y-x) = \sum_{p \in \L^*_+ : |p| \leq c N^\gamma} e^{i p \cdot (x-y)} \,  \] 
is such that $\| \check{\chi}_x \| = \| \chi \| \leq  C N^{3\gamma/2}$. Hence, we find 
\[ \begin{split} 
 \| \check a_x \check a_y &\cN_{\geq cN^\gamma} \phi \| \leq N \| \check a_x \check a_y \phi \| + N^{1/2} \| \check{\chi}_x \| \| \check a_y \phi \| + N^{1/2} \| \check{\chi}_y \| \| \check a_x \phi \|.\end{split} \]
Inserting in (\ref{eq:cVNNg}), we find 
\[ \| \cV_N^{1/2} \cN_{\geq cN^\gamma} \phi \|^2 \leq C N^2 \| \cV_N^{1/2}  \phi \|^2 + C N^{3\gamma +\kappa} \| \cN_+^{1/2} \phi \|^2 . \]
 From (\ref{eq:cHNNg}), we conclude that
 \[ \| (\cH_N + 1)^{1/2} \cN_{\geq cN^\gamma}^2 \cN_+^{k+1} \xi \| \leq N   \| (\cH_N + 1)^{1/2}\cN_{\geq cN^\gamma} (\cN_++1)^{k+1} \xi \| \]
 for all $\gamma \leq \alpha = 14 \kappa + 4\eps$, if $\kappa < 1/43$. Using now similar arguments as before \eqref{eq:cGNNgNg}, we conclude that, together with \eqref{eq:cGNNgNg}, we have
 		\[\begin{split}
		|\langle \xi, \text{B} \xi \rangle | \leq&\; \frac12  \langle \xi,  (\cN_++1)^{k+1}\cN_{\geq  cN^\gamma} \cG_N'\cN_{\geq  cN^\gamma}  (\cN_++1)^{k+1} \xi\rangle \\
		&\; +C N^{2+10\kappa/3+2\eps/3-\gamma} \| (\cH_N+1)^{1/2} (\cN_+ +1)^{k+1} \xi \| \| (\cH_N + 1)^{1/2} (\cN_+ +1)^{k} \xi \| \\
		&\; +C N^{2+14\kappa/3+4\eps/3}  \| (\cH_N + 1)^{1/2} (\cN_+ +1)^{k} \xi \|^2 \\
		&\; + C(N^{1+17\kappa/2+2\eps-2\gamma}+N^{1+3\kappa/2-\gamma/2})\| (\cH_N + 1)^{1/2} (\cN_+ +1)^{k+1} \xi \|^2\\
		\end{split}\]
Combining this with (\ref{eq:A-bd}), we arrive at 
\[ \begin{split} &\frac{\langle \xi, (\cN_+ +1)^{k+1} \cN_{\geq c N^\gamma} \cG'_N \cN_{\geq c N^\gamma} (\cN_+ +1)^{k+1} \xi \rangle}{\| \xi \|^2} \\ 
& \leq \left[ N^{2-2\gamma} \zeta^2 + N^{2+10\kappa/3+2\eps/3-\gamma} \right]  \left[ \sup_{\xi \in Q_{\zeta}}  \frac{\langle \xi, (\cN_+ +1)^{2k} (\cH_N +1) \xi \rangle}{\| \xi \|^2} \right]^{1/2} \\ &\hspace{5cm} \times \left[  \sup_{\xi \in Q_{\zeta}}  \frac{\langle \xi, (\cN_+ +1)^{2(k+1)} (\cH_N +1) \xi \rangle}{\| \xi \|^2} \right]^{1/2} \\
& \hspace{0.5cm} + CN^{2+14\kappa/3+4\eps/3}\left[ \sup_{\xi \in Q_{\zeta}}  \frac{\langle \xi, (\cN_+ +1)^{2k} (\cH_N +1) \xi \rangle}{\| \xi \|^2} \right]\\
& \hspace{0.5cm} + C(N^{1+17\kappa/2+2\eps-2\gamma}+N^{1+3\kappa/2-\gamma/2})\left[ \sup_{\xi \in Q_{\zeta}}  \frac{\langle \xi, (\cN_+ +1)^{2(k+1)} (\cH_N +1) \xi \rangle}{\| \xi \|^2} \right]
\end{split} \]
for all $\xi \in Q_z$. With (\ref{eq:commbnd2}), we obtain
\[ \begin{split} 
&\frac{N^{-1} \langle \xi,  (\cN_+ +1)^{k+1}  \cK \cN_{\geq c N^\gamma} (\cN_+ +1)^{k+1} \xi \rangle}{\| \xi \|^2} \\ & \hspace{.5cm}  \leq \; C N^{\kappa-2\gamma} \frac{\langle \xi,  (\cN_+ +1)^{k+1}  \cK  (\cN_+ +1)^{k+1} \xi \rangle}{\| \xi \|^2} \\ & \hspace{.9cm}+ C \left[ N^{-\kappa} \zeta^2 + N^{\gamma+7\kappa/3+2\eps/3} \right]  \left[ \sup_{\xi \in Q_{\zeta}}  \frac{\langle \xi, (\cN_+ +1)^{2k} (\cH_N +1) \xi \rangle}{\| \xi \|^2} \right]^{1/2} \\ &\hspace{5cm} \times \left[  \sup_{\xi \in Q_{\zeta}}  \frac{\langle \xi, (\cN_+ +1)^{2(k+1)} (\cH_N +1) \xi \rangle}{\| \xi \|^2} \right]^{1/2}\\
& \hspace{0.9cm} + CN^{2\gamma+11\kappa/3+4\eps/3}\left[ \sup_{\xi \in Q_{\zeta}}  \frac{\langle \xi, (\cN_+ +1)^{2k} (\cH_N +1) \xi \rangle}{\| \xi \|^2} \right]\\
& \hspace{0.9cm}+ C(N^{15\kappa/2+2\eps-1}+N^{\kappa/2+3\gamma/2-1})\left[ \sup_{\xi \in Q_{\zeta}}  \frac{\langle \xi, (\cN_+ +1)^{2(k+1)} (\cH_N +1) \xi \rangle}{\| \xi \|^2} \right].
 \end{split} \]
Applying this bound to (\ref{eq:propMNerrorbnd}) and recalling that $ \kappa<1/43$, we conclude that 
\[ \begin{split} 
&\frac{N^\kappa \langle \xi,  (\cN_+ +1)^{k+1} e^A e^D \cE_{\cM_N} e^{-D} e^{-A} (\cN_+ +1)^{k+1} \xi \rangle}{\| \xi \|^2} \\ & \hspace{.3cm}  \geq \; - C N^{-\eps} \left[ \sup_{\xi \in Q_{\zeta}}\frac{\langle \xi , (\cH_N+1) (\cN_+ +1)^{2(k+1)} \xi \rangle}{\| \xi \|^2}\right] \\ &\hspace{.7cm} - C \left[ N^{20\kappa + 5\eps} \zeta^2 + N^{44\kappa + 12 \eps} \right] \left[ \sup_{\xi \in Q_{\zeta}}  \frac{\langle \xi, (\cN_+ +1)^{2k} (\cH_N+1) \xi \rangle}{\| \xi \|^2} \right]^{1/2} \\ &\hspace{5cm} \times  \left[  \sup_{\xi \in Q_{\zeta}}  \frac{\langle \xi, (\cN_+ +1)^{2(k+1)} (\cH_N +1) \xi \rangle}{\| \xi \|^2} \right]^{1/2}.\end{split} \]
Therefore, for any $\delta > 0$, we find (if $N$ is large enough)  
\[ \begin{split} 
&\frac{N^\kappa \langle \xi,  (\cN_+ +1)^{k+1} e^A e^D \cE_{\cM_N} e^{-D} e^{-A} (\cN_+ +1)^{k+1} \xi \rangle}{\| \xi \|^2} \\ & \hspace{.5cm}  \geq -\delta \sup_{\xi \in Q_{\zeta}} \frac{\langle \xi, (\cH_N +1) (\cN_+ +1)^{2(k+1)} \xi \rangle}{\| \xi \|^2} \\ &\hspace{1cm} - C \delta^{-1} \left[ N^{20\kappa+5\eps} \zeta^2 + N^{44 \kappa + 12 \eps} \right]^2 \sup_{\xi \in Q_{\zeta}} \frac{\langle \xi, (\cH_N +1) (\cN_+ +1)^{2k} \xi \rangle}{\| \xi \|^2}.
\end{split} \]

From the last bound, \eqref{eq:NkG'Nk-bd} and (\ref{eq:HNk+1}), we obtain 
\[  \begin{split} &\frac{\langle \xi ,  (\cN_+ +1)^{2(k+1)} (\cH_N +1) \xi \rangle}{ \| \xi \|^2}  \\ &\leq \; \delta \sup_{\xi \in Q_{\zeta}} \frac{\langle \xi ,  (\cN_+ +1)^{2(k+1)} (\cH_N +1) \xi \rangle}{ \| \xi \|^2} \\ &\hspace{2cm} + C \delta^{-1} \left[ N^{20\kappa+5\eps}  \zeta^2  + N^{44 \kappa + 12 \eps} \right]^2   \sup_{\xi \in Q_{\zeta}} \frac{\langle \xi.,  (\cN_+ +1)^{2k} (\cH_N +1) \xi \rangle}{ \| \xi \|^2} \end{split} \]
for any $\xi \in Q_{\zeta}$. Taking the supremum over all $\xi \in Q_{\zeta}$, and choosing $\delta > 0$ small enough, we arrive at 
\[ \begin{split}  \sup_{\xi \in Q_{\zeta}} & \frac{\langle \xi ,  (\cN_+ +1)^{2(k+1)} (\cH_N +1) \xi \rangle}{ \| \xi \|^2}  \\ &\leq C  \left[ N^{20\kappa + 5 \eps}  \zeta^2  + N^{44 \kappa + 12 \eps} \right]^2  \sup_{\xi \in Q_{\zeta}} \frac{\langle \xi,  (\cN_+ +1)^{2k} (\cH_N +1) \xi \rangle}{ \| \xi \|^2} \\ & \leq C \left[ N^{20\kappa + 5 \eps}  \zeta^2  + N^{44 \kappa + 12 \eps} \right]^{2k+1} \end{split} \]
by the induction assumption. 

\qed

%%%%%%%%%%%%%%%%%%%%%%%%%%%%%%%%%%%%%%%%%%%%%%%%%%%
%%%%%%%%%%%%%%%%%%%%%%%%%%%%%%%%%%%%%%%%%%%%%%%%%%%
%%%%%%%%%%%%%%%%%%%%%%%%%%%%%%%%%%%%%%%%%%%%%%%%%%%
%%%%%%%%%%%%%%%%%%%%%%%%%%%%%%%%%%%%%%%%%%%%%%%%%%%
%%%%%%%%%%%%%%%%%%%%%%%%%%%%%%%%%%%%%%%%%%%%%%%%%%%

\section{Analysis of $ \cM_N $}\label{sec:MN}

This section is devoted to the proof of Proposition \ref{prop:MN}. In Subsection \ref{sec:MNNKres} we establish bounds on the  growth of the number of excitations and of their energy with respect to the action of $e^D$, with the quartic operator $D = D_1 - D_1^*$ with 
\begin{equation}\label{eq:def-D12} D_1 = \frac{1}{2N} \sum_{r \in P_H, p,q \in P_L} \eta_r a_{p+r}^* a_{q-r}^* a_p a_q \end{equation}
as defined in (\ref{eq:defD}). In Subsection \ref{sub:MNparts} we compute the different parts of the excitation Hamiltonian $\cM_N$, introduced in (\ref{eq:defMN}). Finally, in Subsection \ref{sub:MNproof}, we conclude the proof of Prop. \ref{prop:MN}. 

\subsection{Growth of Number and Energy of Excitations}
\label{sec:MNNKres}

The first lemma of this section controls the growth of the number of excitations with high momentum. 
\begin{lemma}\label{lm:NresgrowD} Assume the exponents $\alpha, \beta$ satisfy (\ref{eq:condab}). Let $k\in\NN_0$, $m=1,2,3$, $0< \gamma\leq \alpha$ and $c>0$ ($c  < 1$ if $\gamma = \alpha$). Then, there exists a constant $C>0$ such that 
		\begin{equation}\label{eq:NresgrowD}
		\begin{split}
		e^{-sD} (\cN_+ + 1)^{k }(\cN_{\geq c N^{\gamma}}+1)^m e^{sD} \leq &\; C(\cN_+ + 1)^{k } (\cN_{\geq  cN^{\gamma}}+1)^m ,		
		\end{split}
		\end{equation}
for all $s\in [-1;1] $ and all $N\in\NN$ large enough. 
\end{lemma}		

\begin{proof} Since $[\cN_+ , \cN_{\geq cN^\gamma} ] = 0$ and $[ \cN_+ , D] =0$, it is enough to prove the lemma for $k=0$. We consider first $m=1$. For $\xi \in \cF_+^{\leq N}$, we define the function $ \varphi_\xi:\mathbb{R}\to \mathbb{R}$ by 	
		\[\varphi_\xi (s)  = \langle \xi, e^{-sD} (\cN_{\geq c N^{\gamma}}+1) e^{sD}  \xi \rangle\]
so that differentiating yields
		\begin{equation} \label{eq:NresgrowlemD1}
		\begin{split}
		\partial_s \varphi_\xi (s)& =2 \Re\langle e^{sD} \xi, \big[\cN_{\geq c N^{\gamma}}, D_1 \big]  e^{sD}  \xi \rangle
		\end{split}
		\end{equation}
with $D_1$ as in \eqref{eq:def-D12}. By assumption,  $N^\alpha\geq  N^{\alpha}-N^{\beta} \geq cN^{\gamma}$ for sufficiently large $N\in \NN$. This implies that 
		\[   [\cN_{\geq cN^{\gamma} }, a^*_{p+r}] =  a^*_{p+r}, \hspace{0.2cm}  [\cN_{\geq cN^{\gamma} }, a^*_{q-r}] =  a^*_{q-r}\]
for $ r\in P_H$ and $p, q\in P_L$, by \eqref{eq:ccr} and \eqref{eq:comm2}. We then compute
		\begin{equation}\label{eq:NresgrowlemD2}
		\begin{split}
		\big[\cN_{\geq c N^{\gamma}}, D_1 \big] & = \frac1{ N} \sum_{\substack{ r\in P_{H}, p, q \in P_{L}}} \eta_r a^*_{p+r}a^*_{q-r}a_pa_q  -  \frac1{ N} \sum_{\substack{ r\in P_{H}, p, q \in P_{L}, \\ |p|\geq cN^{\gamma}  }} \eta_r a^*_{p+r}a^*_{q-r}a_pa_q.
		\end{split}
		\end{equation}
and apply Cauchy-Schwarz to obtain
		\begin{equation}\label{eq:NresgrowlemD6}\begin{split}
		| \partial_s \varphi_\xi (s)|		& \leq  \frac C{N} \bigg(  \sum_{\substack{ r\in P_{H}, p, q \in P_{L}, \\ |p+r|\geq cN^{\gamma}, |q-r|\geq cN^{\gamma} }}  \| a_{p+r} (\cN_{\geq cN^{\gamma} }+1)^{-1/2} a_{q-r}  e^{sD}\xi \|^2  \bigg)^{1/2}\\
		&\hspace{1cm} \times \|\eta_H\| \bigg(  \sum_{\substack{ p,q  \in P_{L} } } \|a_p(\cN_{\geq cN^{\gamma} }+1)^{1/2} a_q  e^{sD}\xi \|^2\bigg)^{1/2} \\
		 &\leq CN^{\kappa+3\beta/2-\alpha/2} \varphi_\xi(s)\leq C\varphi_\xi(s).
		\end{split}\end{equation}
Since the bound is independent of $\xi \in \cF_+^{\leq N}$ and it also holds true if we replace $ D$ by $-D$ in the definition of $ \varphi_\xi$, this proves \eqref{eq:NresgrowD}, for $m=1$. 

For $m=3$, we define  
		\[\psi_\xi(s) = \langle \xi, e^{-sD} (\cN_{\geq c N^{\gamma}}+1)^3 e^{sD}  \xi \rangle \]
with derivative 
\[ \partial_s \psi_\xi(s) =  2\Re\langle e^{sD} \xi, [ (\cN_{\geq c N^{\gamma}} + 1)^3, D_1 ]  e^{sD} \xi \rangle \]		
We have 
\begin{equation}\label{eq:3-comm} \begin{split} [ (\cN_{\geq c N^{\gamma}}+1)^3 , D_1 ] = \; &3 ( \cN_{\geq c N^{\gamma}}+1) [ \cN_{\geq c N^{\gamma}} , D_1 ] (\cN_{\geq c N^{\gamma}}+1) \\ &+ [ \cN_{\geq c N^{\gamma}} , [ \cN_{\geq c N^{\gamma}} , [ \cN_{\geq c N^{\gamma}} , D_1 ] ] ]. \end{split} \end{equation} 
The contribution of the first term on the r.h.s. of (\ref{eq:3-comm}) can be controlled as in \eqref{eq:NresgrowlemD6} (replacing $e^{sD} \xi$ with $(\cN_{\geq  c N^{\gamma}}+1) e^{sD} \xi$). With \eqref{eq:NresgrowlemD2} and using again that $N^\alpha\geq  N^{\alpha}-N^{\beta} \geq cN^{\gamma}$, we obtain that  
\[  \begin{split} [ \cN_{\geq c N^{\gamma}} , [ \cN_{\geq c N^{\gamma}} , & [ \cN_{\geq c N^{\gamma}} , D_1 ] ] ] \\ = \; & \frac4{ N} \!\!\sum_{\substack{ r\in P_{H}, p, q \in P_{L} }} \eta_r a^*_{p+r}a^*_{q-r}a_pa_q  -  \frac7{ N} \!\!\sum_{\substack{ r\in P_{H}, p, q \in P_{L}, \\ |p|\geq cN^{\gamma}  }} \eta_r a^*_{p+r}a^*_{q-r}a_pa_q \\ &+  \frac{3}{ N} \!\!\sum_{\substack{ r\in P_{H}, p, q \in P_{L}, \\ |p|,|q| \geq cN^{\gamma}  }} \eta_r a^*_{p+r}a^*_{q-r} a_p a_q .
		\end{split} \]
All these contributions can be controlled like those in \eqref{eq:NresgrowlemD2}. We conclude that 
\[ | \partial_s \psi_\xi (s)| \leq C \psi_\xi (s) \]
This proves \eqref{eq:NresgrowD} with $m=3$. The case $m=2$ follows by operator monotonicity of the function $x \mapsto x^{2/3}$. 
\end{proof}

Next, we prove bounds for the growth of the low-momentum part of the kinetic energy, defined as in  (\ref{eq:defKres}).
\begin{lemma}\label{lm:KresgrowD} Assume the exponents $\alpha, \beta$ satisfy (\ref{eq:condab}). Let $0< \gamma_1, \gamma_2 \leq \alpha$, $c_1, c_2 \geq 0$ (and $c_j \leq 1$ if $\gamma_j = \alpha$, for $j=1,2$). Then, there exists a constant $C>0$ such that 
\begin{equation}\label{eq:KresgrowD}
\begin{split}		
e^{-sD} \cK_{\leq c_1 N^{\gamma_1}} e^{sD} \leq&\; \cK_{\leq c_1 N^{\gamma_1}}  + N^{2\beta -1} (\cN_{\geq \frac12 N^{\alpha}} +1)^2,\\
		e^{-sD} \cK_{\leq c_1 N^{\gamma_1}} (\cN_{\geq c_2 N^{\gamma_2}}+1)e^{sD} \leq&\; \cK_{\leq c_1 N^{\gamma_1}}(\cN_{\geq c_2 N^{\gamma_2}} +1)  \\
		& + N^{2\beta -1} (\cN_{\geq c_2 N^{\gamma_2}} +1)^2(\cN_{\geq \frac12 N^{\alpha}} +1)\\
		\end{split}
		\end{equation}
for all $s\in [-1;1] $ and all $N\in\NN$ sufficiently large. 
\end{lemma}	
\begin{proof} Fix $\xi \in \cF_+^{\leq N}$ and define $\varphi_\xi:\mathbb{R}\to \mathbb{R} $ by $\varphi_\xi (s) = \langle \xi, e^{-sD}  \cK_{\leq c_1 N^{\gamma_1}} e^{sD}\xi \rangle $ such that
		\[\begin{split} 
		\partial_s \varphi_\xi(s) &= 2\Re\langle \xi, e^{-sD} [  \cK_{\leq c_1 N^{\gamma_1}} , D_1  ]   e^{sD}\xi \rangle.\\
		%&  \hspace{0.4cm} +2\Re\langle \xi, e^{-sD} \cK_L \big[ (\cN_+ + 1)^{k }, D_1\big]   e^{sD}\xi \rangle.
		\end{split}\]
We notice that
		\[ \big[ \cK_{\leq c_1 N^{\gamma_1}}, a^*_{p+r} \big]=\big[ \cK_{\leq c_1 N^{\gamma_1}}, a^*_{q-r} \big]=0 \]
if $ r\in P_H $ and $p,q\in P_L$, because $ |r|, |p+r|, |q-r|\geq N^{\alpha}- N^{\beta} > c_1 N^{\gamma_1}$ for $N\in\NN$ large enough. %Hence, with the identity \eqref{eq:NresgrowlemA3}, we find that
		%\[\begin{split}
		%&  \langle \xi, e^{-sD}\cK_L \big[ (\cN_+ + 1)^{k }, D_1\big]   e^{sD}\xi \rangle  \\
		%& =  \frac k{\sqrt N} \sum_{r\in P_{H}, p \in P_{L}}  \eta_r  \langle e^{sD} \xi, \cK_L b^*_{r+p}a^*_{-r}a_p  (\cN_++\Theta(\cN_+)+1)^{k-1} e^{sD}\xi \rangle \\
		%&=  \frac k{\sqrt N} \sum_{r\in P_{H}, p \in P_{L}}  \eta_r  \langle e^{sD} \xi,  b^*_{r+p}a^*_{-r} \cK_L a_p  (\cN_++\Theta(\cN_+)+1)^{k-1} e^{sD}\xi \rangle \\
		%\end{split}\]
%and we can bound this contribution by
		%\begin{equation}\label{eq:KresgrowD2}
		%\begin{split}
		%& \big|  \langle \xi, e^{-sD}\cK_L \big[ (\cN_+ + 1)^{k }, D_1\big]   e^{sD}\xi \rangle \big| \\
		%&= \bigg| \frac k{\sqrt N} \sum_{r\in P_{H}, p}  \eta_r  \langle e^{sD} \xi,  b^*_{r+p}a^*_{-r} \cK_L a_p  (\cN_++\Theta(\cN_+)+1)^{k-1} e^{sD}\xi \rangle \bigg| \\
		%&= \bigg| \frac k{\sqrt N} \sum_{r\in P_{H}, p,q \in P_{L}}  q^2 \eta_r  \langle e^{sD} \xi,  b^*_{r+p}a^*_{-r} a^*_q a_q a_p  (\cN_++\Theta(\cN_+)+1)^{k-1} e^{sD}\xi \rangle \bigg| \\
		%&\leq \frac{C}{\sqrt N } \bigg( \sum_{r\in P_{H}, p,q \in P_{L}}q^2 \| a_{p+r}  a_{-r} a_q  (\cN_++1)^{(k-1)/2}  e^{sD}\xi \|^2\bigg)^{1/2}\\
		%&\hspace{1cm} \times \bigg( \sum_{r\in P_{H}, p,q \in P_{L}}q^2 \eta_r^2 \| a_{p} a_q   (\cN_++1)^{(k-1)/2}  e^{sD}\xi \|^2\bigg)^{1/2}\\
		%&\leq CN^{\kappa-\alpha/2} \| \cK_L^{1/2} (\cN_++1)^{k/2} e^{sD}\xi\|^2 = C N^{\kappa-\alpha/2}\varphi_\xi(s).
		%\end{split}
		%\end{equation}
%Now, let's switch to the first contribution on the r.h.s. of \eqref{eq:KresgrowD1}. 
Using \eqref{eq:ccr}, we then compute
		\begin{equation}\label{eq:KresgrowD3}
		\begin{split}
		[ \cK_{\leq c_1 N^{\gamma_1}} , D_1 ]  & =  -  \frac1{ N} \sum_{\substack{ r\in P_{H}, p, q \in P_{L}: |p|\leq c_1 N^{\gamma_1}   }} p^2 \eta_r a^*_{p+r}a^*_{q-r}a_pa_q.
		\end{split}
		\end{equation}
and, using that $|p|\leq N^{\beta}$ for $p\in P_L$, we obtain with Cauchy-Schwarz
		\begin{equation}\label{eq:KresgrowD4}
		\begin{split}
		&\big| \langle \xi, e^{-sD} [ \cK_{\leq c_1 N^{\gamma_1}}, D_1 ]  e^{sD}\xi \rangle\big| \\
		&\leq  \frac {CN^{\beta}}{ N} \sum_{\substack{ r\in P_{H}, p, q \in P_{L}: |p|\leq c_1 N^{\gamma_1}   }} |p| |\eta_r| \|  a_{r+p}a_{q-r} e^{sD}\xi \| \|  a_p a_qe^{sD}\xi \| \\
		&\leq CN^{5\beta/2+\kappa-\alpha/2-1/2} \| (\cN_{\geq \frac12 N^\alpha }+1)e^{sD}\xi \| \| \cK_{\leq c_1 N^{\gamma_1}}^{1/2}e^{sD}\xi \|.
		\end{split}
		\end{equation}
With Lemma \ref{lm:NresgrowD} choosing $ c=\frac12$ and $ \gamma=\alpha$, this implies for $N\in\NN$ large enough that
		\begin{equation*}
		\begin{split}
		 \partial_s\varphi_\xi(s)  &\leq CN^{5\beta/2+\kappa-\alpha/2-1/2} \| (\cN_{\geq \frac12 N^\alpha }+1) e^{sD}\xi \| \| \cK_{\leq c_1 N^{\gamma_1}}^{1/2}e^{sD}\xi \| \\
		&\leq CN^{2\beta-1} \langle \xi, (\cN_{\geq \frac12 N^\alpha }+1)^2 \xi \rangle + C \varphi_\xi(s).
		\end{split}
		\end{equation*}
This proves the first inequality in \eqref{eq:KresgrowD}, by Gronwall's lemma and $ \alpha>   3\beta+2\kappa\geq 0$. 

Next, let us prove the second inequality in \eqref{eq:KresgrowD}. We define $ \psi_\xi:\mathbb{R}\to \mathbb{R}$ by 
		\[ \psi_\xi(s) = \langle \xi, e^{-sD}  \cK_{\leq c_1 N^{\gamma_1}} (\cN_{\geq c_2 N^{\gamma_2}}+1) e^{sD}\xi \rangle, \]
and we compute
		\[\begin{split}
		\partial_s \psi_\xi(s) &= 2\Re \langle \xi, e^{-sD}  \big[ \cK_{\leq c_1 N^{\gamma_1}}, D_1\big] (\cN_{\geq c_2 N^{\gamma_2}}+1) e^{sD}\xi \rangle\\
		& \hspace{0.4cm} + 2\Re \langle \xi, e^{-sD}  \cK_{\leq c_1 N^{\gamma_1}} \big[ \cN_{\geq c_2 N^{\gamma_2}}, D_1 \big] e^{sD}\xi \rangle.
		\end{split}\]
First, we proceed as in \eqref{eq:KresgrowD4} and obtain with (\ref{eq:Ntheta-bds}) that 
		\begin{equation*}
		\begin{split}
		&\big| \langle \xi, e^{-sD} [ \cK_{\leq c_1 N^{\gamma_1}}, D_1 ](\cN_{\geq c_2 N^{\gamma_2}}+1)  e^{sD}\xi \rangle\big| \\
		&\leq \frac {CN^{\beta}}{ N} \sum_{\substack{ r\in P_{H}, p, q \in P_{L}: \\|p|\leq c_1 N^{\gamma_1}   }}|p| |\eta_r| \| a_{r+p} a_{q-r} (\cN_{\geq c_2 N^{\gamma_2}}+1)^{1/2} e^{sD}\xi \| \\
		&\hspace{4.5cm}\times\|a_q a_p (\cN_{\geq c_2 N^{\gamma_2}}+1)^{1/2} e^{sD}\xi \| \\
		&\leq  CN^{5\beta/2+\kappa-\alpha/2-1/2}\| (\cN_{\geq c_2 N^{\gamma_2} }+1)(\cN_{\geq \frac12 N^\alpha }+1)^{1/2}e^{sD}\xi \| \\
		&\hspace{4.5cm}\times\| \cK_{\leq c_1 N^{\gamma_1}}^{1/2}(\cN_{\geq c_2 N^{\gamma_2}}+1)^{1/2}e^{sD}\xi \|.
		\end{split}
		\end{equation*}
Here, we used in the last step that $ [a_{q-r}, \cN_{\geq c_2N^{\gamma_2}}] = a_{q-r}$ for $r\in P_H$, $q\in P_L$ and that $ \cN_{c_2 N^{\gamma_2}} \geq \cN_{ N^{\alpha}-N^{\beta}}$ for $N\in\NN$ large enough. The last bound and Lemma \ref{lm:NresgrowD} imply that
		\begin{equation}\label{eq:KresgrowD6}\begin{split}
		&\big| \langle \xi, e^{-sD} [ \cK_{\leq c_1 N^{\gamma_1}}, D_1 ](\cN_{\geq c_2 N^{\gamma_2}}+1)  e^{sD}\xi \rangle\big| \\
		&\hspace{2cm}\leq CN^{2\beta-1}  \langle \xi, (\cN_{\geq c_2 N^{\gamma_2} }+1)^2(\cN_{\geq \frac12 N^\alpha }+1)\xi \rangle +  C\psi_\xi(s).
		\end{split} \end{equation}
Next, we recall the identity \eqref{eq:NresgrowlemD2} and that 
		\[ \big[ \cK_{\leq c_1 N^{\gamma_1}}, a^*_{p+r} \big]=\big[ \cK_{\leq c_1 N^{\gamma_1}}, a^*_{q-r} \big]=0 \]
whenever $ r\in P_H, p,q \in P_L$ and $N\in\NN$ is sufficiently large. We then obtain
		\begin{equation}\label{eq:KresgrowD7}\begin{split}
		&\big| \langle \xi, e^{-sD}  \cK_{\leq c_1 N^{\gamma_1}}\big[ \cN_{\geq c_2 N^{\gamma_2}}, D_1 \big]  e^{sD}\xi \rangle\big| \\
		&\leq  \frac {C}{N} \sum_{\substack{ r\in P_{H}, p,q \in P_{L},\\ v\in \Lambda_+^*:  |v|\leq c_1N^{\gamma_1} }} |v|^2 |\eta_r| \| a_{r+p}(\cN_{\geq c_2 N^{\gamma_2} }+1)^{-1/2} a_{q-r} a_v e^{sD}\xi \|\\
		&\hspace{5cm}\times \|a_pa_q (\cN_{\geq c_2 N^{\gamma_2} }+1)^{1/2}a_v e^{sD}\xi \| \\
		&\leq C N^{3\beta/2 +\kappa - \alpha/2} \langle e^{sD}\xi,  \cK_{\leq c_1 N^{\gamma_1}} (\cN_{\geq c_2 N^{\gamma_2}}+1) e^{sD}\xi \rangle \leq  C \psi_\xi(s).
		\end{split}\end{equation} 
Hence, putting \eqref{eq:KresgrowD6} and \eqref{eq:KresgrowD7} together, we have proved that
		\[ \partial_s \psi_\xi(s) \leq  CN^{2\beta-1}  \langle \xi, (\cN_{\geq c_2 N^{\gamma_2} }+1)^2(\cN_{\geq \frac12 N^\alpha }+1)\xi \rangle+C \psi_\xi(s),\]
which implies the second bound in \eqref{eq:KresgrowD}, by Gronwall's lemma.
\end{proof}	

It will also be important to control the potential energy operator, restricted to low momenta. We define %{\bf[is this needed, or can we use $\cV_N \leq N^{\beta+\kappa} \cK_L$?]} 
		\begin{equation}\label{eq:defVNL}
		\cV_{N,L} = \frac1{2N}\sum_{\substack{ u\in \Lambda^*, p,q\in \Lambda_+^*:\\ p+u, q+u, p,q  \in P_L}} N^{\kappa}\widehat{V}(u/N^{1-\kappa})a^*_{p+u}a^*_{q}a_pa_{q+u}.
		\end{equation}
Notice that $ \cV_{N,L} = \cV_{N,L}^*$ by symmetry of the momentum restrictions. To calculate $e^D \cV_{N,L} e^{-D}$, we will use the next lemma, which will also be useful in the next subsections. 
\begin{lemma}\label{lm:VFresgrowD} Assume the exponents $\alpha, \beta$ satisfy (\ref{eq:condab}).  Let $ F = (F_p)_{p\in\Lambda_+^*}\in \ell^\infty(\Lambda_+^*)$ and define 		
\begin{equation}\label{eq:defZ}  \emph{Z} = \frac1{2N}\sum_{\substack{ u\in \Lambda^*, p,q\in \Lambda_+^*:\\ p+u, q+u, p,q  \in P_L}} F_u a^*_{p+u}a^*_{q}a_pa_{q+u}  \end{equation} 
Then, there exists a constant $C>0$ such that  	
		\begin{equation}\label{eq:VFresgrowD}
		\begin{split}		
		\pm \big( e^{-sD} \emph{Z} e^{sD}-Z\big) &\leq  C\|F\|_\infty N^{\beta -1}  \cK_L(\cN_{\geq \frac12 N^{\alpha} }+1) +  C\|F\|_\infty N^{3\beta -2}  (\cN_{\geq \frac12 N^{\alpha} }+1)^3
		\end{split}
		\end{equation}
for all $s\in [-1;1] $, and for all $N\in\NN$ sufficiently large. 
\end{lemma}
\begin{proof} Given $ \xi\in \cF_+^{\leq N}$, we define $ \varphi_\xi:\mathbb{R}\to \mathbb{R}$ by
		\[ \varphi_\xi(s) = \langle \xi, e^{-sD} \text{Z} e^{sD}\xi\rangle, \]
which has derivative
		\[\begin{split} \partial_s \varphi_\xi(s) &= 2\Re \langle \xi, e^{-sD} [ \text{Z}, D_1] e^{sD}\xi\rangle.\\
		\end{split}\]
By assumption, we have $ \alpha> 3\beta + 2\kappa$ so that $ |r|, |v+r|, |w-r| \geq N^\alpha-N^\beta > N^\beta $ if $ r\in P_H$ and $ v,w\in P_L$, for sufficiently large $N\in\NN$. This implies in particular that
		\[[a_{p}a_{q+u}, a^*_{v+r}a^*_{w-r}] = 0\]
whenever $q+u, p\in P_L $ and $ r\in P_H$, $ v,w\in P_L$. As a consequence, we find
		\begin{equation}\label{eq:commVNLD1}
		\begin{split}
		[\text{Z}, D_1] &= - \frac1{2N^{2}}\sum_{\substack{ u\in \Lambda^*, r\in P_H, v,w \in P_L: \\ w-u, v+u  \in P_L}} F_u \eta_r a^*_{v+r}a^*_{w-r}  a_{w-u} a_{v+u}\\
		&\hspace{0.4cm} - \frac1{N^{2}}\sum_{\substack{ u\in \Lambda^*, r\in P_H, v, w, p\in P_L:\\ p+u, v+u  \in P_L}}F_u \eta_r a^*_{v+r}a^*_{w-r} a^*_{p+u} a_w a_{v+u} a_p.
		\end{split}
		\end{equation}
With (\ref{eq:Ntheta-bds}) and $ N^\alpha-N^\beta> \frac 12 N^\alpha$ for $N\in \NN$ large enough, we can bound 
		\[\begin{split}
		&\bigg| \frac1{N^{2}}\sum_{\substack{ u\in \Lambda^*, r\in P_H, v,w \in P_L: \\ w-u, v+u  \in P_L}}F_u \eta_r \langle e^{sD}\xi,a^*_{v+r}a^*_{w-r}  a_{w-u} a_{v+u}e^{sD}\xi\rangle\bigg| \\
		&\leq \frac{C\|F\|_\infty }{N^{2}}\bigg(\sum_{\substack{ u\in \Lambda^*, r\in P_H, v,w \in P_L: \\ w-u, v+u  \in P_L}} |v+u|^{-2}\| a_{v+r} (\cN_{\geq \frac12 N^{\alpha} }+1)^{-1/2}a_{w-r} e^{sD}\xi  \|^2 \bigg)^{1/2}\\
		&\hspace{1cm}\times \bigg(\sum_{\substack{ u\in \Lambda^*, r\in P_H, v,w \in P_L: \\ w-u, v+u  \in P_L}}\eta_r^2 |v+u|^{2} \| a_{w-u} (\cN_{\geq \frac12 N^{\alpha} }+1)^{1/2}a_{v+u}e^{sD}\xi  \|^2 \bigg)^{1/2}\\
		&\leq C\|F\|_\infty N^{7\beta/2 + \kappa-\alpha/2-3/2} \| (\cN_{\geq \frac12 N^{\alpha} }+1)^{1/2}  e^{sD}\xi\| \| \cK_L^{1/2}(\cN_{\geq \frac12 N^{\alpha} }+1)^{1/2}  e^{sD}\xi\|.
		\end{split}\]
and 
		\[\begin{split}
		&\bigg| \frac1{N^{2}}\sum_{\substack{ u\in \Lambda^*, r\in P_H, v, w, p\in P_L:\\ p+u, v+u  \in P_L}} F_u \eta_r \langle e^{sD}\xi, a^*_{v+r}a^*_{w-r} a^*_{p+u} a_w a_{v+u} a_pe^{sD}\xi\rangle\bigg| \\
		&\leq \frac{C\|F\|_\infty}{N^{2}}\bigg(\sum_{\substack{ u\in \Lambda^*, r\in P_H, v, w, p\in P_L:\\ p+u, v+u  \in P_L}} |p+u|^2|p|^{-2} \| a_{v+r} (\cN_{\geq \frac12 N^{\alpha} }+1)^{-1/2}a_{w-r} a_{p+u} e^{sD}\xi  \|^2 \bigg)^{1/2}\\
		&\hspace{1cm}\times \bigg(\sum_{\substack{ u\in \Lambda^*, r\in P_H, v, w, p\in P_L:\\ p+u, v+u  \in P_L}}\eta_r^2 |p|^{2}|p+u|^{-2} \| a_w (\cN_{\geq \frac12 N^{\alpha} }+1)^{1/2}a_{v+u}a_pe^{sD}\xi  \|^2 \bigg)^{1/2}\\
		%& \leq CN^{3\beta + \kappa-\alpha/2-1} \langle \xi, e^{-sD} (\cN_{\geq \frac12 N^{\alpha} }+1)(\cN_++1)e^{sD}\xi\rangle\\
		&\leq C\|F\|_\infty N^{5\beta/2 + \kappa-\alpha/2-1} \langle \xi, e^{-sD} \cK_L(\cN_{\geq \frac12 N^{\alpha} }+1)e^{sD}\xi\rangle.
		\end{split}\]
Lemma \ref{lm:NresgrowD}, Lemma \ref{lm:KresgrowD} and the assumption $ \alpha> 3\beta + 2\kappa\geq 0$ implies 
		\[\begin{split}
		\pm \partial_s \varphi_s(\xi) & \leq C\|F\|_\infty N^{\beta -1} \langle \xi,  \cK_L(\cN_{\geq \frac12 N^{\alpha} }+1)\xi\rangle +  C\|F\|_\infty N^{3\beta -2}\langle \xi, (\cN_{\geq \frac12 N^{\alpha} }+1)^3\xi\rangle.
		\end{split}\]
Hence, integrating the last equation from zero to $ s\in [-1; 1]$ proves the lemma.
\end{proof}
With $\sup_{p\in\Lambda^*}| N^{\kappa}\widehat{V}(p/N^{1-\kappa})|\leq CN^\kappa$, we obtain immediately the following result.
\begin{cor} \label{cor:VNresgrowD}
Assume the exponents $\alpha, \beta$ satisfy (\ref{eq:condab}). Then there exists a constant $C>0$ such that 		
\begin{equation*}
\begin{split}		
\pm \big( e^{-sD} \cV_{N,L} e^{sD} -\cV_{N,L} \big)&\leq    CN^{\beta + \kappa-1} \cK_L(\cN_{\geq \frac12 N^{\alpha} }+1)  +  CN^{3\beta + \kappa-2}  (\cN_{\geq \frac12 N^{\alpha} }+1)^3
\end{split}
\end{equation*}
for all $s\in [-1;1] $, and for all $N\in\NN$ sufficiently large. 
\end{cor}

We also need rough bounds for the conjugation of the full potential energy operator $\cV_N$. To this end, we will make use of the following estimate for the commutator of $\cV_N$ with $D = D_1 - D_1^*$, with $D_1$ defined in \eqref{eq:def-D12}.	
\begin{prop}\label{prop:commVND} 
Assume the exponents $\alpha, \beta$ satisfy (\ref{eq:condab}). Then   
		\begin{equation} \label{eq:propcommVND}
		\begin{split}
		[\cV_N, D] & =  \frac{1}{2N} \sum_{\substack{ u\in \Lambda_+^*, p,q\in P_L:\\ p+u, q-u\neq 0 }} N^{\kappa} (\widehat{V}(./N^{1-\kappa})\ast  \eta/N)(u) \big( a^*_{p+u} a^*_{q-u}a_pa_q +\emph{h.c.}\big) \\
		&\hspace{0.4cm}+ \cE_{[\cV_N,D]}
		\end{split}
		\end{equation}
and there exists a constant $C>0$ such that 
		\begin{equation}\label{eq:propcommVNDerror}\begin{split}
		\pm \cE_{[\cV_N,D]} &\;\hspace{0cm}\leq \delta  \cV_N + CN^{\alpha+\kappa-1}\cV_N+ CN^{\alpha+\kappa-1}\cV_{N,L} \\
		&\;\hspace{0.4cm} +  \delta^{-1}C N^{\beta +\kappa-1} \cK_L (\cN_{\geq \frac12 N^{\alpha} }+1)  + \delta^{-1}C N^{3\beta +\kappa-1} (\cN_{\geq \frac12 N^{\alpha} }+1)^2
		\end{split}\end{equation}	
for all $\delta >0$ and for all $N\in\NN$ sufficiently large.
\end{prop}

\begin{proof} We have 
		\[[\cV_N, D] = [\cV_N, D_1] +\text{h.c.}\]
To compute the commutator $ [ \cV_N, D_1 ]$, we compute first of all that 
		\[\begin{split}
		&[a^*_{p+u}a^*_q a_p a_{q+u}, a^*_{v+r} a^*_{w-r} a_v a_w] \\
	 	&\hspace{2.5cm}= a^*_{p+u}a^*_q a_{q+u}  a^*_{w-r} a_v a_w\delta_{p, v+r} + a^*_{p+u}a^*_q a_p a^*_{w-r} a_v a_w \delta_{q+u, v+r}\\
		&\hspace{3cm} + a^*_{p+u}a^*_q a^*_{v+r}a_{q+u}   a_v a_w\delta_{p, w-r} + a^*_{p+u}a^*_q a^*_{v+r}  a_p  a_v a_w \delta_{q+u, w-r}\\
		&\hspace{3cm} - a^*_{v+r} a^*_{w-r}a^*_q a_w  a_{p} a_{q+u}\delta_{p+u, v} -  a^*_{v+r} a^*_{w-r}a^*_{p+u} a_w  a_{p} a_{q+u} \delta_{q, v}\\
		&\hspace{3cm} -  a^*_{v+r} a^*_{w-r} a_v  a^*_q  a_{p} a_{q+u}\delta_{p+u, w} -  a^*_{v+r} a^*_{w-r} a_v  a^*_{p+u}   a_{p} a_{q+u} \delta_{q, w}.
		\end{split}\]
Putting the terms in the first and last line on the r.h.s. into normal order, we obtain 
		\begin{equation}\label{eq:commVND1}\begin{split}
		[\cV_N, D_1]+\text{h.c.} =&\; \frac{1}{2N} \sum^*_{u\in \Lambda^*, v, w \in P_L} N^{\kappa} (\widehat{V}(./N^{1-\kappa})\ast  \eta/N)(u) a^*_{v+u} a^*_{w-u}a_va_w \\
		& +\Phi_{1}+\Phi_{2}+\Phi_{3}+\Phi_{4}  +\text{h.c.},
		\end{split}\end{equation}
where
		\begin{equation}\label{eq:commVND2}\begin{split}
		\Phi_{1} &= - \frac{1}{2N^{2}}\sum^*_{\substack{u \in\Lambda^*, v,w \in P_L, \\ r \in P_H^c\cup\{0\}}} N^{\kappa} \widehat{V}((u-r)/N^{1-\kappa})\eta_r a^*_{v+u} a_{w-u}^* a_va_w,\\
		 \Phi_{2} &= -\frac{1}{2N^{2}}\sum^*_{\substack{u\in \Lambda^*, r\in P_{H},\\ v,w\in P_{L} }}N^{\kappa}  \widehat{V} (u/N^{1-\kappa}) \eta_r a_{v+r}^* a_{w-r}^* a_{w-u}a_{v+u}  ,\\
		 \Phi_{3} &=  \frac{1}{N^{2}}\sum^*_{\substack{u\in \Lambda^*,q\in \Lambda_+^*,\\ r\in P_{H}, v,w\in P_{L} }}N^{\kappa}  \widehat{V} (u/N^{1-\kappa}) \eta_r a_{w-r+u}^*  a_{v+r}^*a_{q}^*a_{q+u}a_{v}a_{w} ,\\
		\Phi_{4} &= - \frac{1}{N^{2}}\sum^*_{\substack{u\in \Lambda^*,q\in \Lambda_+^*,\\ r\in P_{H}, v,w\in P_{L} }}N^{\kappa}  \widehat{V} (u/N^{1-\kappa}) \eta_r a_{v+r}^* a_{w-r}^* a_{q}^*a_{w}a_{v-u}a_{q+u}.
		\end{split}\end{equation}
The first term on the r.h.s. in \eqref{eq:commVND1} appears explicitly in \eqref{eq:propcommVND}. Hence, let us estimate the size of the operators $ \Phi_{1}$ to $\Phi_{4}$, defined in \eqref{eq:commVND2}. 

Starting with $\Phi_1$, we switch to position space and find
		\begin{equation}\label{eq:commVND3} \begin{split}
	|\langle \xi, \Phi_1\xi \rangle |\leq &\; \frac{1}{N} \sum_{ r\in P_H^c\cup\{0\}}  |\eta_r|\bigg( \int_{\Lambda^2}dxdy\; N^{2-2\kappa}V(N^{1-\kappa}(x-y))\|\check{b}_x \check{a}_y \xi \|^2 \bigg)^{1/2}\\
	&\;\hspace{0cm} \times \bigg( \int_{\Lambda^2}dxdy\; N^{2-2\kappa}V(N^{1-\kappa}(x-y))  \Big\| \sum_{w,v\in P_L} e^{ivx+iwy}a_v a_w\xi \Big\|^2 \bigg)^{1/2}\\
	%\leq &\; CN^{\alpha +\kappa-1} \| \cV_N^{1/2}\xi \| \\
	%&\hspace{0cm}\times\bigg(\sum_{\substack{v, w\in P_L\\ p,q \in P_L}} \int_{\Lambda^2}dxdy\; N^{2-2\kappa}V(N^{1-\kappa}(x-y))e^{-i(v-p)x-i(w-q)y}  \langle\xi, a^*_{v}a^*_wa_pa_{q} \xi\rangle \bigg)^{1/2}\\
	\leq&\;  CN^{\alpha +\kappa-1} \| \cV_N^{1/2}\xi \|\| \cV_{N,L}^{1/2}\xi\|.
	\end{split}\end{equation}
The term $\Phi_2$ on the r.h.s. of (\ref{eq:commVND2}) can be controlled by 
	\[\begin{split}
	|\langle \xi, \Phi_{2}\xi \rangle |=&\; \bigg| \frac{1}{ N } \int_{\Lambda^2}dxdy \;N^{2-2\kappa}V(N^{1-\kappa}(x-y))\sum^*_{\substack{ r\in P_H, \\ v,w \in P_{L} } } e^{-iwx}e^{-ivy} \eta_r \langle \xi, a^*_{v+r}a^*_{w-r}\check{a}_x\check{a}_y\xi \rangle \bigg|  \\
	\leq &\; \frac{ CN^{3
\beta}  \|\eta_H\|}{N} \bigg ( \int_{\Lambda^2}dxdy\; N^{2-2\kappa}V(N^{1-\kappa}(x-y))  \| \check{a}_x \check{a}_y  \xi \|^2 \bigg)^{1/2}\\
	&\;\hspace{0.4cm} \times \bigg( \int_{\Lambda^2}dxdy\; N^{2-2\kappa}V(N^{1-\kappa}(x-y)) \sum_{r\in P_H,  v,w\in P_{L} }  \|a_{v+r}a_{w-r} \xi \|^2 \bigg)^{1/2}\\	
	\leq&\;  CN^{9\beta/2 +3\kappa/2 -\alpha/2 -3/2}\| \cV_N^{1/2}\xi \|\| (\cN_{\geq \frac12 N^{\alpha}}+1)\xi\|.
	\end{split}\]
Finally, the contributions $\Phi_3$ and $\Phi_{4}$ can be bounded as follows. We obtain
	\[\begin{split}
	|\langle \xi, \Phi_{3}\xi \rangle |\leq &\; \frac{1}{ N } \int_{\Lambda^2}dxdy \;N^{2-2\kappa}V(N^{1-\kappa}(x-y))\sum_{\substack{ r\in P_H, \\ v,w \in P_{L} }} |\eta_r| |\langle \xi, a_{v+r}^*\check{a}_x^*\check{a}^*_y \check{a}_y a_v a_w \xi \rangle |  \\
	\leq &\; \frac{CN^{3\beta/2}\|\eta_H\|}{N} \bigg( \int_{\Lambda^2}dxdy\; N^{2-2\kappa}V(N^{1-\kappa}(x-y))   \sum_{ v\in P_{L} }|v|^{-2} \|   \check{a}_x \check{a}_y  \xi \|^2 \bigg)^{1/2}\\
	&\;\hspace{0cm} \times \bigg( N^{\kappa-1}\int_{\Lambda}dx\; \sum_{ v,w\in P_{L} }  |v|^{2} \|(\cN_{\geq \frac12 N^{\alpha}}+1)^{1/2}\check{a}_x a_wa_v \xi \|^2 \bigg)^{1/2}\\	
	\leq&\; CN^{2\beta + 3\kappa/2 -\alpha/2-1/2}  \| \cV_N^{1/2}\xi \|\| \cK_L^{1/2}(\cN_{\geq \frac12 N^{\alpha}}+1)^{1/2}\xi\|
	\end{split}\]
as well as
	\[ \begin{split}
	|\langle \xi, \Phi_{4}\xi \rangle |\leq &\;  \frac{1}{ N } \int_{\Lambda^2}dxdy \;N^{2-2\kappa}V(N^{1-\kappa}(x-y))\!\!\!\!\sum_{r\in P_H, v,w\in P_L}\!\!\!\! |\eta_r|  | \langle \xi, a_{v+r}^*a^*_{w-r}\check{a}^*_y a_w \check{a}_x \check{a}_y \xi \rangle |  \\
	\leq &\; \frac{CN^{3\beta/2}\|\eta_H\|}{N} \bigg [ \int_{\Lambda^2} dx dy \; N^{2-2\kappa}V(N^{1-\kappa}(x-y))\| \check{a}_x \check{a}_y  \xi \|^2 \bigg)^{1/2}\\
	&\;\hspace{0cm} \times \bigg ( N^{\kappa-1}\int_{\Lambda}dy \;  \sum_{\substack{ r\in P_H, \\ v,w \in P_{L} }} 
	\| \check{a}_y a_{v+r}a_{w-r} (\cN_++1)^{1/2} \xi \|^2 \bigg)^{1/2}\\	
	\leq&\;  CN^{ 3\beta +3\kappa/2 -\alpha/2-1/2} \| \cV_N^{1/2}\xi \|\| (\cN_{\geq \frac12 N^{\alpha}}+1)\xi\|.
	\end{split} \]
In conclusion, the previous bounds imply with the assumption (\ref{eq:condab}) (in particular, since $ \alpha> 3\beta +2\kappa$ and $3\beta -2 < 0$) that
		\begin{equation} \label{eq:commVNA23}
		\begin{split}
		&\pm ( \Phi_{1} + \Phi_{2}+\Phi_{3}+\Phi_{4}+\text{h.c.} )\\
		&\;\hspace{0cm}\leq \delta  \cV_N + CN^{\alpha+\kappa-1}\cV_N+ CN^{\alpha+\kappa-1}\cV_{N,L}  +  \delta^{-1}C N^{\beta +\kappa-1} \cK_L (\cN_{\geq \frac12 N^{\alpha} }+1) \\
		&\;\hspace{0.4cm} + \delta^{-1}C N^{3\beta +\kappa-1} (\cN_{\geq \frac12 N^{\alpha} }+1)^2 
		\end{split}
		\end{equation}
holds true in $ \cF_+^{\leq N}$ for any $\delta>0$. This concludes the proof.	
\end{proof}
With Prop. \ref{prop:commVND}, we obtain a bound for the growth of $\cV_N$. 
\begin{cor}\label{cor:VNgrowD} 
Assume the exponents $\alpha, \beta$ satisfy (\ref{eq:condab}). Then there exists a constant $C>0$ such that the operator inequality
		\begin{equation*}
		\begin{split}
		e^{-sD} \cV_N e^{sD} &\leq C \cV_N + C \cV_{N,L} + CN^{\beta + \kappa-1}  \cK_L(\cN_{\geq \frac12 N^{\alpha} }+1) +  CN^{3\beta + \kappa}(\cN_{\geq \frac12 N^{\alpha} }+1).
		\end{split}\end{equation*}	
for all $ s\in [-1;1]$ and for all $N\in\NN$ sufficiently large.
\end{cor}
\begin{proof} We apply Gronwall's lemma. Given a normalized vector $\xi\in\cF_+^{\leq N}$, we define $\varphi_\xi(s) = \langle\xi, e^{-sD }\cV_N e^{sD}\xi\rangle $
and compute its derivative s.t.
		\[\partial_s \varphi_\xi(s) =\langle\xi, e^{-sD }[\cV_N,D] e^{sD}\xi\rangle.  \]
Hence, we can apply \eqref{eq:propcommVND} and estimate 
		\begin{equation}\label{eq:corVNgrowD1}\begin{split} 
		& \bigg| \frac{1}{2N} \sum_{\substack{ u\in \Lambda_+^*, v,w\in P_L:\\ v+u, w-u\neq 0 }} N^{\kappa} (\widehat{V}(./N^{1-\kappa})\ast  \eta/N)(u) \langle e^{sD }\xi, a^*_{v+u} a^*_{w-u}a_va_w e^{sD}\xi\rangle \bigg|\\
		& \leq \frac{\|\check{\eta}\|_\infty}{N} \bigg( \int_{\Lambda^2}dxdy\; N^{2-2\kappa}V(N^{1-\kappa}(x-y)) \|\check{a}_x \check{a}_y e^{sD}\xi\|^2 \bigg)^{1/2}\\
		&\hspace{2cm} \times \bigg( \int_{\Lambda^2}dxdy\; N^{2-2\kappa}V(N^{1-\kappa}(x-y)) \Big\|\sum_{v,w\in P_L}e^{ivx+iwy} a_va_w e^{sD}\xi\Big\|^2 \bigg)^{1/2}\\		
		&\leq C \| \cV_N^{1/2} e^{sD}\xi\| \|\cV_{N,L}^{1/2}e^{sD}\xi\| \leq C \varphi_\xi(s) + C\langle\xi,  e^{-sD}\cV_{N,L}e^{sD}\xi\rangle.
		\end{split}\end{equation}
Here, we used \eqref{eq:etainfbnd}, which shows that $\|\check{\eta}\|_\infty\leq CN $. Using Corollary \ref{cor:VNresgrowD} (recalling that $ \alpha> 3\beta +2\kappa$ and $ 2\beta\leq 1$) and $ \cN_{\geq \frac12 N^{\alpha} }\leq N$ in $ \cF_+^{\leq N}$, this simplifies to
		\[\begin{split} 
		&\bigg| \frac{1}{2N} \sum_{\substack{ u\in \Lambda_+^*, v,w\in P_L:\\ v+u, w-u\neq 0 }} N^{\kappa} (\widehat{V}(./N^{1-\kappa})\ast  \eta/N)(u) \langle e^{sD }\xi, a^*_{v+u} a^*_{w-u}a_va_w e^{sD}\xi\rangle \bigg|\\
		&\leq C \varphi_\xi(s) + C \langle\xi,  \cV_{N,L}\xi\rangle + CN^{\beta + \kappa-1} \langle \xi,  \cK_L(\cN_{\geq \frac12 N^{\alpha} }+1)\xi\rangle +  CN^{3\beta + \kappa} \langle \xi, (\cN_{\geq \frac12 N^{\alpha} }+1)\xi\rangle.
		\end{split}\]
Together with \eqref{eq:propcommVND}, the bound \eqref{eq:propcommVNDerror} (choosing $\delta=1$) and an application of Lemma~\ref{lm:NresgrowD} and of Lemma \ref{lm:KresgrowD}, the claim follows now from Gronwall's lemma.
\end{proof}

Finally, we need control for the growth of the full kinetic energy operator $\cK$. To this end, we need to estimate its commutator with $D$. 
\begin{prop}\label{prop:commKD}  Assume the exponents $\alpha, \beta$ satisfy (\ref{eq:condab}).
Let $m_0 \in\mathbb{R}$ be such that $m_0\beta  = \alpha$ (from (\ref{eq:condab}) it follows that $3 < m_0 < 5$). Then 
		\begin{equation} \label{eq:propcommKD}
		\begin{split}
		[\cK, D] & = -  \frac{1}{2N} \sum_{\substack{ u\in \Lambda^*, p,q\in P_L:\\ p+u, q-u\neq 0 }} N^{\kappa} (\widehat{V}(./N^{1-\kappa})\ast  \widehat{f}_{N})(u) \big( a^*_{p+u} a^*_{q-u}a_pa_q +\emph{h.c.}\big) \\
		&\hspace{0.4cm}+ \cE_{[\cK,D]},
		\end{split}
		\end{equation}
where the self-adjoint operator $\cE_{[\cK,D]} $ satisfies  
		\begin{equation}\label{eq:propcommKDerror}\begin{split}
		\pm \cE_{[\cK,D]} \leq&\;  CN^{5\beta/4+\kappa}\cK_{\leq 2 N^{3\beta/2}} +\delta  \cK + C \delta^{-1}\sum_{j=3}^{2\lfloor m_0\rfloor -1}N^{ j\beta/2 + 3\beta/2+2\kappa-1}\cK_L(\cN_{\geq \frac12 N^{j\beta/2}}+1) \\
		&\; +C \delta^{-1}N^{ \alpha +\beta +2\kappa-1}\cK_L(\cN_{\geq \frac12 N^{\lfloor m_0 \rfloor \beta}} +1)  +C \\ 
		\end{split}\end{equation}	
for all $\delta>0$ and for all $N\in\NN$ sufficiently large.
\end{prop}
\begin{proof} Using that $[\cK, D]  =  [\cK, D_1] +\text{h.c.}$, a straight forward computation shows that	
		\begin{equation}\label{eq:commKD1}
		\begin{split}
		[\cK, D_1] + \text{h.c.}&= -\frac{1}{2N} \sum_{\substack{ r\in \Lambda^*, v,w\in P_L:\\ v+r,w-r\neq 0 }} N^{\kappa} (\widehat{V}(./N^{1-\kappa})\ast  \widehat{f}_{N})(r)  a^*_{v+r} a^*_{w-r}a_va_w\\
		&\hspace{0.4cm} + \Sigma_{1} +\Sigma_{2} + \Sigma_{3}  +\text{h.c.}, 
		\end{split}
		\end{equation}
where
		\begin{equation}\label{eq:commKD2}
		\begin{split}
		\Sigma_{1} =&\; \frac{1}{2N} \sum_{\substack{ r\in P_H^c\cup\{0\}, v,w\in P_L:\\ v+r, w-r\neq 0 }} N^{\kappa} (\widehat{V}(./N^{1-\kappa})\ast  \widehat{f}_{N})(r)  a^*_{v+r} a^*_{w-r}a_va_w  , \\ 
		\Sigma_{2} =&\; \frac{1}{2N} \sum_{\substack{ r\in P_H, v,w \in P_L:\\ v+r, w-r\neq 0 }} N^{3-2\kappa}\lambda_\ell (\widehat{\chi}_\ell\ast \widehat{f}_N)(r)  a^*_{v+r} a^*_{w-r}a_va_w  , \\ 
		\Sigma_{3} =&\;  \frac{2}{N } \sum_{\substack{ r\in P_H, v,w \in P_L:\\ v+r, w-r\neq 0 } } r\cdot v\,\eta_r a^*_{v+r}a^*_{w-r}a_va_w.
		\end{split}
		\end{equation}	
Let us estimate the size of the operators $ \Sigma_{1}, \Sigma_{2}$ and $\Sigma_{3}$. Using $ \big|(\widehat{V}(./N^{1-\kappa})\ast  \widehat{f}_{N})(r)\big|\leq C$, we control the operator $ \Sigma_1$ by
		\begin{equation}\begin{split}\label{eq:commKD3}
		&|\langle\xi, \Sigma_1\xi\rangle| =\bigg| \frac{1}{2N} \sum_{\substack{ r\in P_H^c\cup\{0\}, v,w \in P_L:\\ v+r,w-r\neq 0 }} N^{\kappa} (\widehat{V}(./N^{1-\kappa})\ast  \widehat{f}_{N})(r)  \langle\xi,  b^*_{v+r} a^*_{w-r}a_va_w \xi\rangle\bigg|\\
		&\leq \frac{CN^\kappa}{ N}  \sum_{\substack{r\in\Lambda^*, v,w\in P_L:  |r|\leq N^{3\beta/2}, \\ v+r,w-r\neq 0 }} \|a_{w-r}  a_{v+r}\xi\| \|a_va_w \xi \|  \\
		&\hspace{0.4cm}+ \frac{CN^\kappa}{N} \sum_{j=3}^{2\lfloor m_0\rfloor-1} \!\!\!\!\!\!\!\!\sum_{\substack{r\in P_H^c\cup\{0\}, v,w\in P_L: \\ N^{j\beta/2}\leq |r|\leq N^{(j+1)\beta/2}, \\ v+r,w-r\neq 0 }} \hspace{-0.95cm}\|a_{w-r} (\cN_{\geq \frac12 N^{j\beta/2}}+1)^{-1/2} a_{v+r}\xi\| \|a_v(\cN_{\geq \frac12 N^{j\beta/2}}+1)^{1/2}a_w\xi \| \\
		&\hspace{0.4cm}+ \frac{CN^\kappa}{ N} \sum_{\substack{r\in P_H^c\cup\{0\}, v,w\in P_L: \\ N^{\lfloor m_0 \rfloor \beta}\leq |r|\leq N^{\alpha}, \\ v+r,w-r\neq 0 }} \hspace{-0.75cm}\|a_{w-r} (\cN_{\geq \frac12 N^{\lfloor m_0 \rfloor\beta}}+1)^{-1/2} a_{v+r}\xi\| \|a_v(\cN_{\geq \frac12 N^{\lfloor m_0 \rfloor\beta}}+1)^{1/2}a_w\xi \|.
		%&\leq  CN^{\kappa+3\beta/2} \langle\xi, \cN_+\xi\rangle   +\sum_{j=3}^{m_0}CN^{\kappa+3j\beta/2} \langle\xi, (\cN_{\geq \frac12 N^{(j-1)\beta}}+1)\xi\rangle
		\end{split}\end{equation}	
By Cauchy-Schwarz, the first term on the r.h.s. of \eqref{eq:commKD3} can be controlled by 
		\[\frac{CN^\kappa}{ N}  \sum_{\substack{r\in\Lambda^*, v, w\in P_L:  |r|\leq N^{3\beta/2}, \\ v+r,w-r\neq 0 }} \|a_{w-r}  a_{v+r}\xi\| \|a_va_w \xi \| \leq CN^{5\beta/4+\kappa} \langle\xi, \cK_{\leq 2 N^{3\beta/2} } \xi\rangle. \]
The second contribution on the r.h.s. of \eqref{eq:commKD3} can be bounded by
		%\[\begin{split}
		%&\frac{CN^\kappa}{\sqrt N}\sum_{\substack{r\in P_H^c, v\in P_L: \\ N^{\beta}\leq |r|\leq N^{3\beta/2}, \\ v+r\neq 0 }} \hspace{-0.75cm}\|a_{-r} (\cN_{\geq N^{\beta}}+1)^{-1/2} a_{v+r}\xi\| \|a_v(\cN_{\geq N^{\beta}}+1)^{1/2}\xi\| \\
		%&\; \leq \delta \langle \xi, \cK \xi\rangle + \delta^{-1}CN^{\beta +2\kappa}\langle\xi, (\cN_{\geq N^{\beta}}+1)\xi\rangle.
		%\end{split}\]
		\begin{equation}\label{eq:commKD31}\begin{split}
		&\frac{CN^\kappa}{N} \sum_{j=3}^{2\lfloor m_0\rfloor-1} \!\!\!\!\!\!\!\!\sum_{\substack{r\in P_H^c\cup\{0\}, v,w\in P_L: \\ N^{j\beta/2}\leq |r|\leq N^{(j+1)\beta/2}, \\ v+r,w-r\neq 0 }} \hspace{-0.95cm}\|a_{w-r} (\cN_{\geq \frac12 N^{j\beta/2}}+1)^{-1/2} a_{v+r}\xi\| \|a_w(\cN_{\geq \frac12 N^{j\beta/2}}+1)^{1/2}a_v\xi \|  \\
		%&= \frac{CN^\kappa}{\sqrt N} \sum_{j=3}^{2\lfloor m_0\rfloor-1} \!\!\!\!\sum_{\substack{r\in P_H^c, v\in P_L: \\ N^{j\beta/2}\leq |r|\leq N^{(j+1)\beta/2}, \\ v+r\neq 0 }} \Big(|v+r|\|a_{-r} (\cN_{\geq \frac12 N^{j\beta/2}}+1)^{-1/2} a_{v+r}\xi\| \Big)\\
		%&\hspace{6cm}\times\Big(|v+r|^{-1} \|a_v(\cN_{\geq \frac12 N^{j\beta/2}}+1)^{1/2}\xi \|\Big) \\
		&\;\leq  C \sum_{j=3}^{2\lfloor m_0\rfloor-1}N^{ j\beta/4 + 3\beta/4+\kappa-1/2} \| \cK^{1/2}\xi\|  \| \cK_L^{1/2} (\cN_{\geq \frac12 N^{j\beta/2}}+1)^{1/2}\xi\|.
		\end{split}\end{equation}
Similarly, we find that 
		\begin{equation}\label{eq:commKD32}\begin{split}
		&\frac{CN^\kappa}{ N} \sum_{\substack{r\in P_H^c\cup\{0\}, v,w\in P_L: \\ N^{\lfloor m_0 \rfloor \beta}\leq |r|\leq N^{\alpha}, \\ v+r,w-r\neq 0 }} \hspace{-0.75cm}\|a_{w-r} (\cN_{\geq \frac12 N^{\lfloor m_0 \rfloor\beta}}+1)^{-1/2} a_{v+r}\xi\| \|a_w(\cN_{\geq \frac12 N^{\lfloor m_0 \rfloor\beta}}+1)^{1/2}a_v\xi \| \\
		&\;\leq  C N^{ \alpha/2 +\beta/2+\kappa-1/2} \| \cK^{1/2}\xi\|  \| \cK_L^{1/2}(\cN_{\geq \frac12 N^{\lfloor m_0 \rfloor \beta}}+1)^{1/2}\xi\|.
		\end{split}\end{equation}
In summary, the previous three bounds imply that  
		\begin{equation}\label{eq:commKD4}
		\begin{split}
		\pm \Sigma_1\leq &\;   CN^{5\beta/4+\kappa}\cK_{\leq 2 N^{3\beta/2}} +\delta  \cK  + C \delta^{-1}N^{ \alpha +\beta +2\kappa-1}\cK_L(\cN_{\geq \frac12 N^{\lfloor m_0 \rfloor \beta}} +1)  \\
		&\;+ C \delta^{-1}\sum_{j=3}^{2\lfloor m_0\rfloor -1}N^{ j\beta/2 + 3\beta/2+2\kappa-1}\cK_L(\cN_{\geq \frac12 N^{j\beta/2}}+1) 
		\end{split}
		\end{equation}
for some constant $C>0$ and all $\delta>0$.

Next, let us switch to $ \Sigma_2$ and $\Sigma_3$, defined in \eqref{eq:commKD2}. Using Lemma \ref{sceqlemma} $i)$ and the bound \ref{eq:chiastfNbnd}, Cauchy-Schwarz and $ \alpha>3\beta +2\kappa$ impy
		\begin{equation}\label{eq:commKD5}
		\begin{split}
		|\langle\xi, \Sigma_2\xi\rangle| \leq &\; \frac{CN^{\kappa} }{N} \sum_{\substack{ r\in P_H, v,w\in P_L:\\ v+r,w-r\neq 0 }} |r|^{-2} \|  a_{v+r} (\cN_{\geq \frac12 N^{\alpha} }+1)^{-1/2}a_{w-r} \xi\| \|a_v (\cN_{\geq \frac12 N^{\alpha}}+1)^{1/2}a_w\xi\|\\
		\leq &\;CN^{-\beta-1/2}\| (\cN_{\geq \frac12 N^{\alpha}}+1)^{1/2}\xi\| \| \cK_L^{1/2}(\cN_{\geq \frac12 N^{\alpha}}+1)^{1/2}\xi\| .
		\end{split}
		\end{equation}
Similarly, we obtain
		\begin{equation}\label{eq:commKD6}
		\begin{split}
		|\langle\xi, \Sigma_3\xi\rangle| \leq \;&  \frac{C}{N } \sum_{r\in P_H, v,w\in P_L } |r| | v|  |\eta_r| \|  a_{v+r} (\cN_{\geq \frac12 N^{\alpha} }+1)^{-1/2}a_{w-r} \xi\| \|a_v a_w(\cN_{\geq \frac12 N^{\alpha}}+1)^{1/2}\xi\|\\
		\leq &\; CN^{-1/2} \| \cK^{1/2}\xi\|  \| \cK_L^{1/2}(\cN_{\geq \frac12 N^{\alpha}}+1)^{1/2} \xi\|,
		\end{split}
		\end{equation}
where we used that $ |r|/|v+r|\leq 2$ for $r\in P_H$, $v\in P_L$ and $N\in\NN$ large enough. Combining \eqref{eq:commKD4}, \eqref{eq:commKD5} and \eqref{eq:commKD6} and defining $ \cE_{[\cK,D]} = \sum_{i=1}^3(\Sigma_i+\text{h.c.})$ proves the claim.
\end{proof}

\begin{cor}\label{cor:KgrowD} Assume the exponents $\alpha, \beta$ satisfy (\ref{eq:condab}). Let $ m_0\in\mathbb{R}$ be such that $m_0 \beta =\alpha$ ($3 < m_0 < 5$ from (\ref{eq:condab})). Then, there exists a constant $C>0$ such that 
		\begin{equation} \label{eq:corKgrowD}
		\begin{split}
		e^{-sD} \cK e^{sD} &\leq  C \cK +C\cV_N+ C\cV_{N,L} + CN^{5\beta/4+\kappa}\cK_{\leq N^{3\beta/2}} \\
		&\hspace{0.4cm}+ C \sum_{j=3}^{2\lfloor m_0\rfloor -1}N^{ j\beta/2 + 3\beta/2+2\kappa-1}\Big[\cK_L + N^{2\beta}(\cN_{\geq \frac12 N^\alpha} +1)\Big](\cN_{\geq \frac12 N^{j\beta/2}}+1)\\
		&\hspace{0.4cm}+C N^{ \alpha +\beta +2\kappa-1}\Big[\cK_L + N^{2\beta}(\cN_{\geq \frac12 N^\alpha} +1)\Big](\cN_{\geq \frac12 N^{\lfloor m_0 \rfloor \beta}} +1) + CN^{13\beta /4 +\kappa} 
		\end{split}\end{equation}	
for all $ s\in [-1;1]$ and for all $N\in\NN$ sufficiently large.
\end{cor}
\begin{proof} Given $ \xi\in \cF_+^{\leq N}$, we define $ \varphi_\xi(s)  = \langle\xi, e^{-sD} \cK e^{sD}\xi\rangle$. Differentiation yields
		\[ \partial_s \varphi_\xi(s) = \langle\xi, e^{-sD} [\cK, D]e^{sD}\xi\rangle, \]
s.t., to bound the derivative of $ \varphi_\xi$, we can apply Proposition \ref{prop:commKD}. Arguing exactly as in \eqref{eq:corVNgrowD1}, we obtain with $ \sup_{x\in\Lambda} |f_N(x)|\leq 1$ the operator inequality 
		\[\begin{split}
		&\pm\frac{1}{2N} \sum_{\substack{ u\in \Lambda_+^*, v,w\in P_L:\\ v+u, w-u\neq 0 }} N^{\kappa} (\widehat{V}(./N^{1-\kappa})\ast  \widehat{f}_N)(u) a^*_{v+u} a^*_{w-u}a_va_w  \leq C  \cV_N + C\cV_{N,L} .
		\end{split}\]

Now, the claim follows from the bound \eqref{eq:propcommKDerror} (choosing $\delta=1$), the previous bound and an application of Corollary \ref{cor:VNgrowD}, Corollary \ref{cor:VNresgrowD}, Lemma \ref{lm:NresgrowD}, Lemma \ref{lm:KresgrowD} and the operator bound $ \cN_{\geq \frac12 N^\alpha}\leq 4 N^{-2\alpha}\cK$, by Gronwall's Lemma.
\end{proof}

%%%%%%%%%%%%%%%%%%%%%%%%%%%%%%%%%%%%%%%%%%%%%%%%%%%%%%%%%%%%%%%%%%%
%%%%%%%%%%%%%%%%%%%%%%%%%%%%%%%%%%%%%%%%%%%%%%%%%%%%%%%%%%%%%%%%%%%
%%%%%%%%%%%%%%%%%%%%%%%%%%%%%%%%%%%%%%%%%%%%%%%%%%%%%%%%%%%%%%%%%%%

\subsection{Action of Quartic Renormalization on Excitation Hamiltonian} 
\label{sub:MNparts}

We compute now the main contributions to $\cM_N= e^{-D} \cJ_{N}^\text{eff} e^{D}$. From (\ref{eq:defJNeff}) and recalling that $[\cN_+ , D] = 0$, we can decompose  
\begin{equation}\label{eq:decoMN} \cM_{N}= 4\pi \frak{a}_0 N^{1+\kappa}  -4\pi \frak{a}_0 N^{\kappa-1}\cN_+^2/N +\cM_{N}^{(2)}+\cM_{N}^{(3)}+\cM_{N}^{(4)} \end{equation}
where the operators $ \cM_{N}^{(i)}, i= 2,3,4,$ are defined by
\begin{equation}\begin{split}\label{eq:MN0to4} 
\cM_{N}^{(2)} =&\; 8\pi\mathfrak{a}_0N^\kappa\sum_{p\in P_H^c}   e^{-D}b^*_pb_p  e^D + 4\pi \frak{a}_0 N^{\kappa}\sum_{p\in P^c_H} e^{-D} \big[ b^*_p b^*_{-p} + b_p b_{-p} \big] e^D \\
\cM_{N}^{(3)}  =&\;   \frac{8\pi\mathfrak{a}_0N^\kappa}{\sqrt N}\sum_{\substack{ p\in P_H^c , q\in P_L:\\ p+q\neq 0}} e^{-D}\big[ b^*_{p+q}a^*_{-p}a_q+ \text{h.c.}\big]e^D ,\\
\cM_{N}^{(4)} =&\; e^{-D}\cH_N e^D = e^{-D}\cK e^D +e^{-D}\cV_N e^D.		
\end{split}\end{equation}

\subsubsection{Analysis of $ \cM_N^{(2)}$}\label{sec:MN2}

In this section, we determine the main contributions to $ \cM_N^{(2)}$, defined in \eqref{eq:MN0to4} by
		\begin{equation}\label{eq:defMN2}
		\cM_N^{(2)} = 8\pi\mathfrak{a}_0N^\kappa\sum_{p\in P_H^c}   e^{-D}b^*_pb_p  e^D + 4\pi \frak{a}_0 N^{\kappa}\sum_{p\in P^c_H} e^{-D} \big[ b^*_p b^*_{-p} + b_p b_{-p} \big] e^D
		\end{equation}
The main result of this section is the following proposition.
\begin{prop}\label{prop:MN2} Assume the exponents $\alpha, \beta$ satisfy (\ref{eq:condab}). Then 	\begin{equation}\label{eq:propMN2id} \cM_N^{(2)}=  8\pi\mathfrak{a}_0N^\kappa\sum_{p\in P_H^c} \Big[  b^*_pb_p  +\frac12  b^*_p b^*_{-p} +\frac12  b_p b_{-p} \Big] + \cE_{\cM_N}^{(2)} 
		\end{equation}
and there exists a constant $C>0$ such that   
		\begin{equation}\label{eq:propMN2errorbnd}
		\begin{split}
		\pm e^{A}e^{D} \cE_{\cM_N}^{(2)}  e^{-D} e^{-A}&\leq C N^{-\beta-2\kappa} \cK  + CN^\kappa
		\end{split}
		\end{equation}
for all $N\in\NN$ sufficiently large.
\end{prop}
\begin{proof} We start with the identity
		\begin{equation}\label{eq:propMN21}
		\begin{split}
		&\cM_N^{(2)} -  8\pi\mathfrak{a}_0N^\kappa\sum_{p\in P_H^c} \Big[  b^*_pb_p  +\frac12  b^*_p b^*_{-p} +\frac12  b_p b_{-p} \Big] \\
		& = 8\pi\mathfrak{a}_0N^\kappa \int_0^1dt \;\sum_{p\in P_H^c}  e^{-tD}\big[   b^*_pb_p  +\frac12  b^*_p b^*_{-p} +\frac12  b_p b_{-p} , D_1 \big]e^{tD} +\text{h.c.}
		\end{split}
		\end{equation}
and a straight-forward computation shows that  
		\[\begin{split}  
		& \big[  b^*_pb_p  +\frac12  b^*_p b^*_{-p} +\frac12  b_p b_{-p}, a_{v+r}^* a_{w-r}^* a_w a_v  \big] \\
		&= b^*_{v+r}a^*_{w-r}a_v b_w \big( \delta_{p,v+r} + \delta_{p,w-r}- \delta_{p,v} - \delta_{p,w}\big) \\
		& \hspace{0.4cm} -  \frac12 b^*_{v+r}b^*_{w-r} \big(\delta_{p,w}\delta_{-p,v} + \delta_{-p,w}\delta_{p,v}) +  \frac12b_{v}b_{w} (\delta_{p,w-r}\delta_{-p,v+r} + \delta_{-p,w-r}\delta_{p,v+r}\big) \\
		&\hspace{0.4cm} -  \frac12b^*_{v+r}b^*_{w-r}\big( a^*_{-p}a_w\delta_{p,v} + a^*_{p}a_w\delta_{-p,v} + a^*_{-p}a_v\delta_{p,w} + a^*_{p}a_v\delta_{-p,w}  \big)\\
		&\hspace{0.4cm} + \frac12\big( a^*_{w-r}a_{-p}\delta_{p,v+r} + a^*_{v+r}a_{-p}\delta_{p,w-r} + a^*_{w-r}a_{p}\delta_{-p,v+r} + a^*_{v+r}a_{p}\delta_{-p,w-r} \big) b_v b_w.
		\end{split}\]
As a consequence, we find that 
		\begin{equation}\label{eq:propMN22}
		\begin{split}
		&\cM_N^{(2)} -  8\pi\mathfrak{a}_0N^\kappa\sum_{p\in P_H^c} \Big[  b^*_pb_p  +\frac12  b^*_p b^*_{-p} +\frac12  b_p b_{-p} \Big]  =  \int_0^1\!\!dt\;  e^{-tD}\sum_{j =1}^5\big( \text{V}_j + \text{h.c.}\big )e^{tD},
		\end{split}
		\end{equation}
where 
		\begin{equation}\label{eq:propMN23}
		\begin{split}
		\text{V}_1 &= - \frac{8\pi\mathfrak{a}_0N^\kappa}{2N}\sum_{\substack{ r\in P_H, v \in P_L }} \eta_r b^*_{v+r}b^*_{-v-r}, \\
		\text{V}_2 &= \frac{8\pi\mathfrak{a}_0N^\kappa}{2N}\sum_{\substack{ r\in P_H, v\in P_L: \\ v+r \in P_H^c , v+r\neq 0  }}\eta_r  b_v b_{-v} , \\
		\text{V}_3 &=\frac{8\pi\mathfrak{a}_0N^\kappa}{2N}\sum_{\substack{ r\in P_H, v,w\in P_L: \\v+r,w-r\neq 0  }}\eta_r \big(-2 + \chi_{\{ r+v\in P_H^c\} }+ \chi_{\{ w-r\in P_H^c\} }\big) b^*_{v+r}a^*_{w-r}a_v b_w , \\
		\text{V}_4 &=- \frac{8\pi\mathfrak{a}_0N^\kappa}{N}\sum_{\substack{ r\in P_H, v,w\in P_L: \\v+r,w-r\neq 0  }}\eta_r b^*_{v+r}b^*_{w-r} a^*_{-v}a_w , \\
		\text{V}_5 &= \frac{8\pi\mathfrak{a}_0N^\kappa}{N}\sum_{\substack{ r\in P_H, v,w\in P_L: \\r-w\in P_H^c, v+r,w-r\neq 0  }}\eta_r a^*_{v+r}a_{r-w}  b_v b_w . \\
		\end{split}
		\end{equation}
Here $ \chi_{\{ p \in S\}}$ denotes as usual the characteristic function for the set $S\subset \Lambda_+^*$, evaluated at $ p\in \Lambda_+^*$. Let us briefly explain how to bound the different contributions $ \text{V}_1$ to $\text{V}_5$, defined in \eqref{eq:propMN23}. Using Cauchy-Schwarz, the first two contributions are bounded by 
		\[\begin{split}
		 \pm (\text{V}_1 + \text{V}_2) & \leq C N^{2\kappa+3\beta -\alpha/2 -1} (\cN_{\geq \frac12 N^\alpha}+1) + CN^{2\kappa+3\beta/2-1}(\cK_L+1)
		 \end{split}\]
where, for $ \text{V}_2$, we used that $v+r\in P_H^c$ implies that $ |r|\leq N^{\alpha}+N^{\beta}$ and furthermore that $ \sum_{N^\alpha\leq |r|\leq N^\alpha+N^\beta}|\eta_r|\leq N^{\kappa +\beta}$.
The contributions $ \text{V}_3$ to $\text{V}_5$, on the other hand, can be controlled by
		\[\begin{split}
		&|\langle\xi, (\text{V}_3 + \text{V}_4+\text{V}_5)\xi\rangle|\\
		&\leq \frac{CN^{\kappa}}N\sum_{\substack{ r\in P_H, v,w\in P_L: \\v+r,w-r\neq 0  }} |\eta_r| \| a_{v+r}(\cN_{\geq \frac12 N^\alpha}+1)^{-1/2}a_{w-r} \xi \|   \|a_v (\cN_{\geq \frac12 N^\alpha}+1)^{1/2} a_w\xi\|\\
		&\hspace{0.4cm} + \frac{CN^{\kappa}}N\sum_{\substack{ r\in P_H, v,w\in P_L: \\v+r,w-r\neq 0  }}  |\eta_r| \| a_{v+r}(\cN_{\geq \frac12 N^\alpha}+1)^{-1/2}a_{w-r}a_w \xi \|   \|a_v (\cN_{\geq \frac12 N^\alpha}+1)^{1/2} \xi\|\\
		&\hspace{0.4cm} + \frac{CN^{\kappa}}N\sum_{\substack{ r\in P_H, v,w\in P_L: \\v+r,w-r\neq 0  }}|\eta_r| \| a_{v+r}\xi \|   \|a_va_w a_{w-r}\xi\|\\
		&\leq CN^{2\kappa+3\beta/2-\alpha/2} \langle\xi, (\cN_{\geq \frac12 N^\alpha}+1)\xi\rangle\leq C N^\kappa \langle\xi, (\cN_{\geq \frac12 N^\alpha}+1)\xi\rangle
		\end{split}\]
for any $ \xi \in \cF_+^{\leq N}$. In conclusion (since $2\kappa + 3\beta-\alpha/2-1 < \kappa$ from (\ref{eq:condab})), we have proved that 
		\[\pm \sum_{j =1}^5\big( \text{V}_j + \text{h.c.}\big )\leq  CN^{2\kappa+3\beta/2-1} \cK_L + CN^\kappa (\cN_{\geq \frac12 N^\alpha} + 1) .\]
Now, applying this bound together with \eqref{eq:propMN22}, Lemma \ref{lm:NresgrowA}, Lemma \ref{lm:KresgrowA}, Lemma \ref{lm:NresgrowD}, Lemma \ref{lm:KresgrowD} and the operator inequality $ \cN_{\geq \frac12 N^\alpha}\leq 4 N^{-2\alpha}\cK$ proves the claim.
\end{proof}

%%%%%%%%%%%%%%%%%%%%%%%%%%%%%%%%%%%%%%%%%%%%%%%%%%%%%%%%%%%%%%%%%%%
%%%%%%%%%%%%%%%%%%%%%%%%%%%%%%%%%%%%%%%%%%%%%%%%%%%%%%%%%%%%%%%%%%%
%%%%%%%%%%%%%%%%%%%%%%%%%%%%%%%%%%%%%%%%%%%%%%%%%%%%%%%%%%%%%%%%%%%

\subsubsection{Analysis of $  \cM_N^{(3)}$} \label{sec:MN3}

In this section, we determine the main contributions to $  \cM_N^{(3)}$, defined in \eqref{eq:MN0to4} by
		\begin{equation}\label{eq:defMN3} \cM_N^{(3)}=   \frac{8\pi \mathfrak{a}_0 N^\kappa}{\sqrt N}\sum_{\substack{ p\in P_H^c, q\in P_L: \\p+q\neq 0}} e^{-D}\big(b^*_{p+q}a^*_{-p}a_q+ \text{h.c.}\big)e^D . 
		\end{equation}
\begin{prop}\label{prop:MN3} Assume the exponents $\alpha, \beta$ satisfy (\ref{eq:condab}). Then we have that
		\begin{equation}\label{eq:propMN3id} \cM_N^{(3)}=   \frac{8\pi \mathfrak{a}_0 N^\kappa}{\sqrt N}\sum_{\substack{ p\in P_H^c, q\in P_L: \\p+q\neq 0}} \big(b^*_{p+q}a^*_{-p}a_q+ \emph{h.c.}\big) + \cE_{\cM_N}^{(3)} 
		\end{equation}
and there exists a constant $C>0$ such that 
		\begin{equation}
		\begin{split}
		\pm &e^{A}e^{D} \cE_{\cM_N}^{(3)}  e^{-D} e^{-A} \\
		& \hspace{0.5cm}\leq CN^{-\beta}\cK +  CN^{\alpha + \beta/2+2\kappa-1 }\cK (\cN_{\geq \frac12 N^\alpha}+1) + CN^{\alpha + \beta/2+2\kappa}  
		\end{split}
		\end{equation} 
for all $N\in\NN$ sufficiently large.
\end{prop}
\begin{proof} Let us define the operator $ \text{Y}:\cF_+^{\leq N}\to \cF_+^{\leq N} $ by 
		\begin{equation}\label{eq:defopY} \text{Y} = \frac{8\pi \mathfrak{a}_0 N^\kappa}{\sqrt N}\sum_{\substack{ p\in P_H^c, q\in P_L: \\p+q\neq 0}} \big(b^*_{p+q}a^*_{-p}a_q+ \text{h.c.}\big),
		\end{equation}
so that $ \cM_N^{(3)} = e^{-D}\text{Y}e^D$. We recall the definition \eqref{eq:def-D12} and observe that
		\begin{equation}\label{eq:propMN31}e^{-D}\text{Y}e^D- \text{Y} = \int_0^1 ds\; e^{-sD} [\text{Y}, D_1] e^{sD} + \text{h.c.}.\end{equation}
This implies that it is enough to control the commutator $  [\text{Y}, D_1]$ after conjugation with $ e^{tD}$, for any $t\in [-1;1] $. Note that, if $ p\in P_H^c, q\in P_L, r\in P_H$ and $v,w \in P_L$, we have $ |v+r|\geq N^\alpha-N^\beta>\frac12 N^\alpha>N^\beta$ s.t. $[a^*_{-p}a_q, a^*_{v+r}a^*_{w-r}] = 0 $, for $N\in\NN$ large enough. Then, a lengthy, but straight-forward calculation shows that 
		\[\begin{split}
		 [b^*_{p+q}a^*_{-p}a_q, a^*_{v+r}a^*_{w-r}a_va_w] &= -b^*_{v+r}a^*_{w-r}a_q (\delta_{-p, w}\delta_{p+q,v} + \delta_{-p,v}\delta_{p+q,w})\\
		 &\hspace{0.4cm} -b^*_{p+q}a^*_{v+r}a^*_{w-r}a_q (a_w \delta_{-p,v}+ a_v\delta_{-p,w} ) \\
		 &\hspace{0.4cm}  -b^*_{-p}a^*_{v+r}a^*_{w-r}a_q (a_w\delta_{p+q,v}+a_v\delta_{p+q,w}) 
		\end{split}\]
and 
		\[\begin{split}
		[a^*_qa_{-p}b_{p+q}, a^*_{v+r}a^*_{w-r}a_va_w] &= a^*_q a_v b_w \delta_{-p,w-r}\delta_{p+q,v+r} + a^*_q a_v b_w \delta_{-p,v+r}\delta_{p+q,w-r}\\
		&\hspace{0.4cm}+ a^*_qa^*_{w-r}a_va_wb_{p+q}\delta_{-p,v+r}+ a^*_qa^*_{v+r}a_va_wb_{p+q}\delta_{-p,w-r}\\
		&\hspace{0.4cm}- a^*_{v+r}a^*_{w-r}a_wa_{-p}b_{p+q}\delta_{q,v}- a^*_{v+r}a^*_{w-r}a_va_{-p}b_{p+q}\delta_{q,w}\\
		&\hspace{0.4cm}+ a^*_qa^*_{w-r}a_{-p}a_vb_{w}\delta_{p+q,v+r}+ a^*_qa^*_{v+r}a_{-p}a_vb_{w}\delta_{p+q,w-r}.
		\end{split} \]
As a consequence, we conclude that
		\begin{equation}\label{eq:propMN32}
		\begin{split}
		[\text{Y}, D_1]+\text{h.c.}& =\sum_{i=1}^6 (\Psi_{i}+\text{h.c.}) ,
		\end{split}
		\end{equation}
where 
		\begin{equation}\label{eq:defPsi1to6}
		\begin{split}
		\Psi_{1} &= - \frac{8\pi\mathfrak{a}_0 N^\kappa}{N^{3/2}} \sum_{\substack{r\in P_H,  v,w\in P_L : \\ v+w \in P_L}}^* \eta_r b^*_{v+r}a^*_{w-r}a_{v+w} , \\
		\Psi_{2}&=  \frac{8\pi\mathfrak{a}_0 N^\kappa}{N^{3/2}} \sum_{\substack{ r\in P_H,  v,w\in P_L:\\ v+r, r - w \in P_H^c, v+w \in P_L}}^* \eta_r a^*_{v+w} a_v b_w,\\
		\Psi_{3} &=  -\frac{16\pi\mathfrak{a}_0 N^\kappa}{N^{3/2}} \sum_{r\in P_H, q, v,w\in P_L}^*\eta_r b^*_{q-v}a^*_{v+r}a^*_{w-r}a_qa_w, \\
		\Psi_{4} &= \frac{8\pi\mathfrak{a}_0 N^\kappa}{N^{3/2}} \sum_{\substack{ r\in P_H,  q,v,w\in P_L:\\ v+r\in P_H^c}}^* \eta_r a^*_qa^*_{w-r}a_va_wb_{q-v-r}, \\
		\Psi_{5}&=\frac{8\pi\mathfrak{a}_0 N^\kappa}{N^{3/2}} \sum_{\substack{ r\in P_H,  q,v,w\in P_L:\\ v+r-q \in P_H^c}}^* \eta_r a^*_qa^*_{w-r}a_va_wb_{q-v-r},  \\	
		\Psi_{6} &= - \frac{8\pi\mathfrak{a}_0 N^\kappa}{N^{3/2}} \sum_{p\in P_H^c, r\in P_H,  v,w\in P_L}^* \eta_r  a^*_{v+r}a^*_{w-r}a_wa_{-p}b_{p+v}. \\			
		\end{split}
		\end{equation}
Let us explain how to control the operators $ \Psi_1$ to $\Psi_6$, defined in \eqref{eq:defPsi1to6}. We start with $\Psi_1$. Given $ \xi\in \cF_+^{\leq N}$, we find that 
		\[\begin{split}
		|\langle\xi, \Psi_1\xi\rangle | &= \bigg| \frac{8\pi\mathfrak{a}_0 N^\kappa}{N^{3/2}} \sum_{r\in P_H,  v,w\in P_L}^* \eta_r \langle\xi, b^*_{v+r}a^*_{w-r}a_{v+w}\xi\rangle\bigg| \\
		&\leq \frac{C N^\kappa}{N^{3/2}}  \sum_{r\in P_H,  v,w\in P_L}^*|\eta_r| \| (\cN_{\geq \frac12 N^\alpha}+1)^{-1/2}a_{v+r}a_{w-r} \xi \| \| (\cN_{\geq \frac12 N^\alpha}+1)^{1/2}a_{v+w} \xi \|\\
		&\leq C N^{3\beta + 2\kappa-\alpha/2-1} \langle\xi, (\cN_{\geq \frac12 N^\alpha}+1)\xi\rangle \leq C N^{3\beta/2 + \kappa-1} \langle\xi, (\cN_{\geq \frac12 N^\alpha}+1)\xi\rangle.
		\end{split}\]
The contribution $ \Psi_2$ can be bounded by
		\[\begin{split}
		|\langle\xi, \Psi_2\xi\rangle|  &= \bigg| \frac{8\pi\mathfrak{a}_0 N^\kappa}{N^{3/2}} \sum_{\substack{ r\in P_H,  v,w\in P_L:\\ v+r\in P_H^c}}^* \eta_r \langle\xi a^*_{v+w} a_v b_w\xi\rangle\bigg| \\
		&\leq C N^{\beta/2 + \kappa-1} \langle\xi, \cK_{\leq 2N^\beta}\xi\rangle \sum_{\substack{N^\alpha \leq |r|\leq N^{\alpha}+N^\beta } } |\eta_r|\leq C N^{3\beta/2 + 2\kappa-1} \langle\xi, \cK_{\leq 2N^\beta}\xi\rangle. 
		\end{split}\]
Notice here, that we used that $ |r|\leq N^\alpha+N^\beta$ if $ r+v\in P_H^c$ and $v\in P_L$. Next, we apply as usual Cauchy-Schwarz to estimate the terms $ \Psi_3$ to $\Psi_5  $ by 
		\[\begin{split} 
		&| \langle\xi, \Psi_3\xi\rangle +  \langle\xi, \Psi_4\xi\rangle + \langle\xi, \Psi_5\xi\rangle | \\
		&\hspace{2cm}\leq  C N^{3\beta  + 2\kappa-\alpha/2} \langle\xi, (\cN_{\geq \frac12 N^\alpha}+1)\xi\rangle \leq C N^{3\beta/2  + \kappa} \langle\xi, (\cN_{\geq \frac12 N^\alpha}+1)\xi\rangle
		 \end{split}\]
for all $ \alpha>3\beta +2\kappa$. Finally, the term $ \Psi_6$ can be controlled by 
		\[\begin{split}
		| \langle\xi, \Psi_6\xi\rangle | &= \bigg|  \frac{8\pi\mathfrak{a}_0 N^\kappa}{N^{3/2}} \sum_{p\in P_H^c, r\in P_H,  v,w\in P_L}^* \eta_r  \langle\xi, a^*_{v+r}a^*_{w-r}a_wa_{-p}b_{p+v} \xi\rangle  \bigg| \\
		&\leq CN^{\kappa-3/2} \sum_{p\in P_H^c, r\in P_H,  v,w\in P_L}^*|w|^{-1} \| (\cN_{\geq N^\alpha/2} + 1)^{-1/2}a_{v+r}a_{w-r} \xi \|   \\
		&\hspace{5cm}\times|w||\eta_r|\| a_wa_{-p}b_{p+v} (\cN_{\geq N^\alpha/2} + 1)^{1/2} \xi \| \\
		&\leq  CN^{\alpha + \beta/2+2\kappa-1 } \langle\xi, \cK_L (\cN_{\geq \frac12 N^\alpha}+1)\xi\rangle + CN^{\alpha + \beta/2+2\kappa}\langle\xi, (\cN_{\geq \frac12 N^\alpha}+1)\xi\rangle.
		\end{split}\]
In conclusion, the previous estimates show that
		\[\begin{split}
		 \pm \bigg[ \sum_{i=1}^6 (\Psi_{i}+\text{h.c.}) \bigg] & \leq C N^{3\beta/2 + 2\kappa-1}\cK_{\leq 2N^\beta} + CN^{\alpha + \beta/2+2\kappa-1 }\cK_L (\cN_{\geq \frac12 N^\alpha}+1)\\
		 &\hspace{0.4cm} + CN^{\alpha + \beta/2+2\kappa}(\cN_{\geq \frac12 N^\alpha}+1),
		 \end{split}\]
so that, together with \eqref{eq:propMN31} and \eqref{eq:propMN32}, an application of the Lemmas  \ref{lm:NresgrowA},  \ref{lm:KresgrowA}, \ref{lm:NresgrowD}, \ref{lm:KresgrowD} and the operator bound $ \cN_{\geq \frac12 N^\alpha}\leq 4N^{-2\alpha}\cK $ proves the claim. 	
\end{proof}

%%%%%%%%%%%%%%%%%%%%%%%%%%%%%%%%%%%%%%%%%%%%%%%%%%%%%%%%%%%%%%%%%%%
%%%%%%%%%%%%%%%%%%%%%%%%%%%%%%%%%%%%%%%%%%%%%%%%%%%%%%%%%%%%%%%%%%%
%%%%%%%%%%%%%%%%%%%%%%%%%%%%%%%%%%%%%%%%%%%%%%%%%%%%%%%%%%%%%%%%%%%

\subsubsection{Analysis of $ \cM_N^{(4)}$} \label{sec:MN4}

In this section, we determine the main contributions to $ \cM_N^{(4)} = e^{-D}\cH_N e^D$, defined in \eqref{eq:MN0to4}. To this end, we start with the observation that
		\begin{equation}\label{eq:MN41}
		\cM_N^{(4)}  = \cH_N + \int_0^1ds\; e^{-sD} \Big( [\cK, D_1] + [\cV_N, D_1]\Big) e^{sD} +\text{h.c.},
		\end{equation}
with $D_1$ defined in \eqref{eq:def-D12}. By Proposition \ref{prop:commVND} and Proposition \ref{prop:commKD}, this implies that  
		\begin{equation}\label{eq:MN42}
		\begin{split}
		\cM_N^{(4)}  &= \cH_N - \frac{N^\kappa}{2N}\sum_{\substack{ r\in \Lambda^*, v,w\in P_L: \\ v+r, w-r\neq 0}}\int_0^1 ds\; \widehat{V}(r/N^{1-\kappa}) e^{-sD}\big(a^*_{v+r}a^*_{w-r}a_va_w+\text{h.c.}\big)e^{sD}\\
		&\hspace{0.4cm} + \int_0^1ds\; e^{-sD} \Big( \cE_{ [\cK, D] }+ \cE_{[\cV_N, D]}\Big) e^{sD},
		\end{split}
		\end{equation}
where we used that $ \widehat{V}(\cdot/N^{1-\kappa})\ast (\widehat{f}_N-\eta/N)(r) = \widehat{V}(\cdot/N^{1-\kappa})(r)$ for all $r\in\Lambda_+^*$. Moreover, the operators $ \cE_{[\cV_N,D]}$ and $ \cE_{[\cK, D]}$ are explicitly given by
		\begin{equation}\label{eq:MN43}
		\begin{split}
		\cE_{[\cV_N,D]} = \sum_{i=1}^4 \big(\Phi_i+\text{h.c.} \big) , \hspace{0.5cm} \cE_{[\cK,D]} =  \sum_{j=1}^3 \big(\Sigma_j+\text{h.c.}\big)
		\end{split}
		\end{equation}
where we recall the definitions \eqref{eq:commVND2} and \eqref{eq:commKD2}. Let us analyse the different contributions in \eqref{eq:MN42}, separately. We start with the second term on the r.h.s. of \eqref{eq:MN42}. 
%%%%%%%%%%%%%%%%%%%%%%%%%%%%%%%%%%%%%%%%%%%%%%%%%%%%%%%%
\begin{prop}\label{prop:secondcommVND} Assume the exponents $\alpha, \beta$ satisfy (\ref{eq:condab}). Then we have 
		\begin{equation}\label{eq:propsecondcommVND}
		\begin{split}
		&\frac{1}{2N}\sum_{\substack{ u\in \Lambda^*, p, q\in P_L: \\ p+u, q-u\neq 0}}  N^\kappa\widehat{V}(r/N^{1-\kappa}) e^{-sD}\big(a^*_{p+u}a^*_{q-u}a_p a_q+\emph{h.c.}\big)e^{sD} \\
		&= \frac{1}{2N}\sum_{\substack{ u\in \Lambda^*, p, q\in P_L: \\ p+u, q-u\neq 0}}  N^\kappa\widehat{V}(r/N^{1-\kappa}) \big(a^*_{p+u}a^*_{q-u}a_p a_q+\emph{h.c.}\big)\\
		&\hspace{0.4cm} + \frac{s}{N} \sum^*_{\substack{ u\in \Lambda^*, v, w \in P_L:\\v+u, w-u\in P_L}} N^{\kappa} (\widehat{V}(./N^{1-\kappa})\ast  \eta/N)(u) a^*_{v+u} a^*_{w-u}a_va_w+ \cE_1(s) + \cE_2(s)
		\end{split}
		\end{equation}
and there exists a constant $C>0$ s.t. $ \cE_1(s) $ and $ \cE_2(s)$ satisfy
		\begin{equation}\label{eq:propsecondcommVNDerrors}\begin{split}\pm \cE_1(s) &\leq C(N^{\alpha+\beta+2\kappa-1} +N^{-3\beta-3\kappa}) \cK + CN^{2\beta+\kappa},\\
		\pm \cE_2(s) &\leq  CN^{\beta+\kappa-1}\cK_L(\cN_{\geq \frac12N^\alpha}+1) + C(N^{-\beta-\kappa} + CN^{3\beta/2 + \kappa/2-1}) \int_0^s dt\;  e^{-tD}   \cV_N  e^{tD}\\
		&\hspace{0.4cm} + CN^{2\beta+2\kappa-1}\int_0^s dt\;  e^{-tD}  \cK_{\leq 2N^\beta}(\cN_{\geq \frac12 N^\alpha}+1)e^{tD},
		\end{split}\end{equation}
for all $ \delta>0$, $ s\in [-1;1]$ and for all $N\in\NN$ sufficiently large.
\end{prop}
\begin{proof} For definiteness, let's denote by $ \text{W}:\cF_+^{\leq N}\to\cF_+^{\leq N} $ the operator 
		\begin{equation} \text{W} = \frac{1}{2N}\sum_{\substack{ u\in \Lambda^*, p, q\in P_L: \\ p+u, q-u\neq 0}} N^\kappa\widehat{V}(u/N^{1-\kappa}) \big(a^*_{p+u}a^*_{q-u}a_p a_q+\text{h.c.}\big)
		\end{equation}
and consider the identity 
		\begin{equation}
		\begin{split}\label{eq:secondcommVND0}
		& e^{-sD}\text{W} e^{sD} - \text{W} \\ &= \int_0^s dt\; e^{-tD} [ \text{W}, D_1] e^{tD} +\text{h.c.}\\
		& = \frac{1}{2N}\int_0^s dt\;\sum_{\substack{ u\in \Lambda^*, p, q\in P_L: \\ p+u, q-u\neq 0}}  N^\kappa \widehat{V}(r/N^{1-\kappa}) e^{-tD} \Big[\big(a^*_{p+u}a^*_{q-u}a_p a_q+\text{h.c.}\big), D_1\Big]e^{tD} +\text{h.c.}
		\end{split}
		\end{equation}
Now, observe that 
		\[ [a_p, a^*_{v+r} ] = [a_q, a^*_{v+r} ] = [a_p, a^*_{w-r} ] =[a_q, a^*_{w-r} ] =0  \]
for all $  p,q \in P_L$ and $ r\in P_H$, $v,w\in P_L$ and $N\in\NN$ sufficiently large. Then, proceeding as in the proof of Proposition \ref{prop:commVND}, we obtain
		\begin{equation}\label{eq:secondcommVND01}\begin{split}
		&[a^*_{p+u}a^*_{q-u} a_p a_{q}, a^*_{v+r} a^*_{w-r} a_v a_w] \\
	 	&\hspace{2.5cm}=  - a^*_{v+r} a^*_{w-r}a^*_{q-u} a_w  a_{p} a_{q}\delta_{p+u, v} -  a^*_{v+r} a^*_{w-r}a^*_{p+u} a_w  a_{p} a_{q} \delta_{q-u, v}\\
		&\hspace{3cm} -  a^*_{v+r} a^*_{w-r} a_v  a^*_{q-u}  a_{p} a_{q}\delta_{p+u, w} -  a^*_{v+r} a^*_{w-r} a_v  a^*_{p+u}   a_{p} a_{q} \delta_{q-u, w}.
		\end{split}\end{equation}
and 
		\begin{equation}\label{eq:secondcommVND02}
		\begin{split}
		&[a^*_{p}a^*_q a_{p-u} a_{q+u},  a^*_{v+r} a^*_{w-r} a_v a_w] \\
	 	&\hspace{2.5cm}= a^*_{p}a^*_q a_{q+u}  a^*_{w-r} a_v a_w\delta_{p-u, v+r} + a^*_{p}a^*_q a_{p-u} a^*_{w-r} a_v a_w \delta_{q+u, v+r}\\
		&\hspace{3cm} + a^*_{p}a^*_q a^*_{v+r}a_{q+u}   a_v a_w\delta_{p-u, w-r} + a^*_{p}a^*_q a^*_{v+r}  a_{p-u}  a_v a_w \delta_{q+u, w-r}\\
		&\hspace{3cm} - a^*_{v+r} a^*_{w-r}a^*_q a_w  a_{p-u} a_{q+u}\delta_{p, v} -  a^*_{v+r} a^*_{w-r}a^*_{p} a_w  a_{p-u} a_{q+u} \delta_{q, v}\\
		&\hspace{3cm} -  a^*_{v+r} a^*_{w-r} a_v  a^*_q  a_{p-u} a_{q+u}\delta_{p, w} -  a^*_{v+r} a^*_{w-r} a_v  a^*_{p}   a_{p-u} a_{q+u} \delta_{q, w}.
		\end{split}\end{equation}
Combining the last two identities and putting non-normally ordered contributions into normal order, we find that 
		\begin{equation}\label{eq:secondcommVND1}\begin{split}
		[\text{W}, D_1]+\text{h.c.} =&\; \frac{1}{N} \sum^*_{\substack{ u\in \Lambda^*, v, w \in P_L:\\v+u, w-u\in P_L}} N^{\kappa} (\widehat{V}(./N^{1-\kappa})\ast  \eta/N)(u) a^*_{v+u} a^*_{w-u}a_va_w \\
		& +\sum_{j=1}^6 \big( \zeta_{j}+\text{h.c.}\big),
		\end{split}\end{equation}
where
		\begin{equation}\label{eq:secondcommVND2}\begin{split}
		\zeta_{1} &= - \frac{1}{2N^{2}}\sum^*_{\substack{u \in\Lambda^*, v,w \in P_L: \\ v+u, w-u\in P_L,\\ r \in P_H^c\cup\{0\}}} N^{\kappa} \widehat{V}((u-r)/N^{1-\kappa})\eta_r a^*_{v+u} a_{w-u}^* a_va_w,\\
		\zeta_{2} & = -\frac{1}{2N^{2}}\sum^*_{\substack{u\in \Lambda^*, r\in P_{H},\\ v,w\in P_{L}:\\ w-u, v+u\in P_L }}N^\kappa\widehat{V}(u/N^{1-\kappa})  \eta_r a^*_{v+r} a^*_{w-r}   a_{w-u} a_{v+u} ,\\
		 \zeta_{3} &= -\frac{1}{2N^{2}}\sum^*_{\substack{u\in \Lambda^*, r\in P_{H},\\ v,w\in P_{L} }}N^\kappa\widehat{V}(u/N^{1-\kappa})  \eta_r  a^*_{v+r} a^*_{w-r}  a_{w-u} a_{v+u} ,\\
		 \zeta_{4} &= -\frac{1}{N^{2}}\sum^*_{\substack{u\in \Lambda^*, r\in P_{H},\\ v,w,q\in P_L:\\ v-u\in P_{L} }}N^\kappa\widehat{V}(u/N^{1-\kappa})  \eta_r a^*_{v+r} a^*_{w-r}a^*_{q-u} a_w  a_{v-u} a_{q},\\
		 \zeta_{5} &= \frac{1}{N^{2}}\sum^*_{\substack{u\in \Lambda^*, r\in P_{H},\\ v,w,q \in P_L:\\ v+r+u\in P_{L} }}N^\kappa\widehat{V}(u/N^{1-\kappa})  \eta_r a^*_{v+r+u}a^*_q a^*_{w-r}a_{q+u}   a_v a_w ,\\
		 \zeta_{6} &= -\frac{1}{N^{2}}\sum^*_{\substack{u\in \Lambda^*, r\in P_{H},\\ v,w,q\in P_{L} }}N^\kappa\widehat{V}(u/N^{1-\kappa})   \eta_r a^*_{v+r} a^*_{w-r}a^*_q a_w  a_{v-u} a_{q+u}.
		 \end{split}\end{equation}
		 
Let us briefly explain how to control the operators $ \zeta_{1}$ to $\zeta_{6}$, defined in \eqref{eq:secondcommVND2}. 

Noting that $ v+u\in P_L $ implies $ |u|\leq 2N^\beta$ whenever $v\in P_L$, the first two contributions $ \zeta_1$ and $\zeta_2$ in \eqref{eq:secondcommVND2} can be controlled by
		\begin{equation}\label{eq:secondcommVND333}\begin{split}
		&|\langle\xi, \zeta_1\xi\rangle| +|\langle\xi, \zeta_2\xi\rangle|\\
		&\leq \frac{CN^\kappa}{2N^{2}}\sum^*_{\substack{u \in\Lambda^*, v,w \in P_L: \\ v+u, w-u\in P_L,\\ r \in P_H^c\cup\{0\}}} | \eta_r| \frac{|w-u|}{|v|}\|  a_{v+u} a_{w-u} \xi\| \frac{|v|}{|w-u|}\|  a_va_w \xi \| \\
		&\hspace{0.4cm} +\frac{C N^\kappa}{2N^{2}}\sum^*_{\substack{u\in \Lambda^*, r\in P_{H},\\ v,w\in P_{L}:\\ w-u, v+u\in P_L }}  |\eta_r| \| a_{v+r} (\cN_{\geq \frac12N^\alpha}+1)^{-1/2}a_{w-r}  \xi\| \| a_{w-u} (\cN_{\geq \frac12N^\alpha}+1)^{1/2}a_{v+u} \xi\|\\
		%&\leq   N^{7\beta/2+2\kappa-\alpha/2-3/2}\|(\cN_{\geq \frac12N^\alpha}+1)^{1/2} \xi \| \| \cK_{\leq 2N^\beta}^{1/2}(\cN_{\geq \frac12N^\alpha}+1)^{1/2}\xi \|\\
		%&\hspace{0.4cm} + CN^{\alpha+\beta+2\kappa-1} \langle \xi,  \cK_{\leq 2N^\beta}\xi\rangle \\
		&\leq CN^{\alpha+\beta+2\kappa-1} \langle \xi,  \cK_{\leq 2N^\beta}\xi\rangle + N^{7\beta/2+2\kappa-\alpha/2-1}\langle\xi, (\cN_{\geq \frac12N^\alpha}+1)\xi\rangle\\
		&\hspace{0.4cm}  + N^{7\beta/2+2\kappa-\alpha/2-2}\langle\xi, \cK_{\leq 2N^\beta} (\cN_{\geq \frac12N^\alpha}+1)\xi\rangle\\ 
		&\leq CN^{\alpha+\beta+2\kappa-1} \langle \xi,  \cK_{\leq 2N^\beta}\xi\rangle + CN^{2\beta+\kappa-1}\langle\xi, (\cN_{\geq \frac12N^\alpha}+1)\xi\rangle.
		\end{split}\end{equation}
By switching to position space, the term $\zeta_3$ can be bounded by
		\[\begin{split}
		& |\langle\xi, \zeta_3\xi\rangle | \leq CN^{3\beta/2+ \kappa-\alpha/2-1} \bigg(\int_{\Lambda^2}dxdy\; N^{2-2\kappa}V(N^{1-\kappa}(x-y)) \| \check{a}_x\check{a}_y\xi\|^2\bigg)^{1/2}\\
		 &\hspace{0.5cm} \times \bigg(\int_{\Lambda^2}dxdy\; N^{2-2\kappa}V(N^{1-\kappa}(x-y)) \sum_{r\in P_H, w\in P_L}\Big\| \sum_{v\in P_L} e^{ivx} a_{v+r}a_{w-r}\xi\Big\|^2\bigg)^{1/2}\\
		 &\leq CN^{3\beta/2+ \kappa-\alpha/2-1} \|\cV_N^{1/2}\xi\|  \bigg(N^{\kappa-1}\int_{\Lambda}dx\; \sum_{r\in P_H, w\in P_L}\Big\| \sum_{v\in P_L} e^{ivx} a_{v+r}a_{w-r}\xi\Big\|^2\bigg)^{1/2}\\
		 %&\leq CN^{\kappa/2-3/2}\|\cV_N^{1/2}\xi\|\|(\cN_{\geq \frac12 N^\alpha}+1)\xi\|\\
		% &\leq CN^{3\beta/2+ \kappa/2-1}\langle\xi, \cV_N\xi\rangle + CN^{3\beta/2+ \kappa/2-2}\langle\xi, (\cN_{\geq \frac12 N^\alpha}+1)^2\xi\rangle\\
		 &\leq CN^{3\beta/2 + \kappa/2-1}\langle\xi, \cV_N\xi\rangle + CN^{3\beta/2 +\kappa/2} .
		 \end{split}\]
We proceed similarly as above for the terms $ \zeta_4$ and $\zeta_5$ which yields 
		\begin{equation}\label{eq:secondcommVND334}\begin{split}
		&|\langle\xi, \zeta_4\xi\rangle| +|\langle\xi, \zeta_5\xi\rangle|\\
		&\leq \frac{CN^\kappa}{N^{2}}\sum^*_{\substack{u\in \Lambda^*, r\in P_{H},\\ v,w,q\in P_L: v-u\in P_{L} }} |q|^{-1} |q-u| \| a_{v+r} (\cN_{\geq \frac12 N^\alpha}+1)^{-1/2} a_{w-r}a_{q-u}\xi\|\\
		&\hspace{5cm}\times|\eta_r| |q| |q-u|^{-1} \| a_w   (\cN_{\geq \frac12 N^\alpha}+1)^{1/2}a_{v-u} a_{q}\xi\|\\
		&\hspace{0.4cm} +  \frac{CN^\kappa}{N^{2}}\sum^*_{\substack{u\in \Lambda^*, r\in P_{H},\\ v,w,q\in P_L:\\ v+r+u\in P_{L} }} \Big( |q| |v|^{-1} \| a_{v+r+u}a_q a_{w-r}\xi\| \Big)\Big( |\eta_r|  |q|^{-1} |v|\| a_{q+u}   a_v a_w\xi\|\Big)\\
		&\leq CN^{5\beta/2+2\kappa-\alpha/2 -1} \langle\xi, \cK_{\leq 3N^\beta}(\cN_{\geq \frac12 N^\alpha}+1)\xi\rangle \\
		&\leq C N^{\beta+\kappa -1} \langle\xi, \cK_{\leq 3N^\beta}(\cN_{\geq \frac12 N^\alpha}+1)\xi\rangle,
		\end{split}\end{equation}
where, for $ \zeta_5$, we used that $ v+r+u\in P_L$ implies that $ |u|\geq \frac34 N^\alpha$, and thus $|q+u|\geq \frac12 N^\alpha $, whenever $ v,q\in P_L $, $r\in P_H$ and $N\in \NN$ sufficiently large (otherwise $ |v+r+u|\geq \frac14N^{\alpha}-N^\beta > N^\beta$ for large enough $N\in\NN$). Finally, $ \zeta_6$ can be controlled by
		\[\begin{split}
		&|\langle\xi, \zeta_6\xi\rangle|  \\
		& =  \bigg| \frac{1}{N }\sum^*_{\substack{r\in P_{H},\\ v,w,q\in P_{L} }}\int_{\Lambda^2} N^{2-2\kappa}V(N^{1-\kappa}(x-y)) e^{-ivx-iqy}  \eta_r \langle\xi, a^*_{v+r} a^*_{w-r}a^*_q a_w  \check{a}_{x} \check{a}_y\xi\rangle\bigg| \\
		&\leq CN^{\beta/2 + \kappa-\alpha/2-1/2} \| \cV_N^{1/2}\xi\| \bigg( N^{\kappa-1}\int_{\Lambda}dx\; \sum^*_{\substack{r\in P_{H},\\ w,q\in P_{L} }} |q| \Big\|\sum_{v\in P_L}e^{-ivx}a_{v+r} a_{w-r}a_q \xi\Big\|^2\bigg)^{1/2}\\
		&\leq CN^{\beta/2 +\kappa/2-1/2} \| \cV_N^{1/2}\xi\|\| \cK_L^{1/2}(\cN_{\geq \frac12 N^\alpha}+1)^{1/2}\xi\|
		\end{split}\]
In summary, the previous estimates show that 
		\begin{equation}\label{eq:secondcommVND3}\begin{split}
		\pm  \sum_{j=1}^6 \big( \zeta_{j}+\text{h.c.}\big)& \leq \delta \cV_N + CN^{3\beta/2 + \kappa/2-1}\cV_N + CN^{\alpha+\beta+2\kappa-1}  \cK_{\leq 2N^\beta} + CN^{2\beta+\kappa}  \\
		&\hspace{0.4cm}+ C(1+\delta^{-1})N^{\beta+\kappa -1}\cK_{\leq 3N^\beta}(\cN_{\geq \frac12 N^\alpha}+1) 
		\end{split}\end{equation}
for all $ \delta>0$. On the other hand, by Lemma \ref{lm:VFresgrowD}, we also know that 
		\begin{equation}\label{eq:secondcommVND4}
		\begin{split}
		&\pm \bigg[ \frac{1}{N} \sum^*_{\substack{ u\in \Lambda^*, v, w \in P_L:\\v+u, w-u\in P_L}} N^{\kappa} (\widehat{V}(./N^{1-\kappa})\ast  \eta/N)(u)\int_0^s dt\;  e^{-tD}a^*_{v+u} a^*_{w-u}a_va_w e^{tD} \\
		& \hspace{3cm}- \frac{s}{N} \sum^*_{\substack{ u\in \Lambda^*, v, w \in P_L:\\v+u, w-u\in P_L}} N^{\kappa} (\widehat{V}(./N^{1-\kappa})\ast  \eta/N)(u) a^*_{v+u} a^*_{w-u}a_va_w \bigg]\\
		&\leq  C N^{-3\beta-3\kappa}\cK  +  CN^{3\beta+\kappa-2} + CN^{\beta+\kappa-1}\cK_L(\cN_{\geq \frac12N^\alpha}+1)\, .
		\end{split}
		\end{equation}

Now, going back to \eqref{eq:secondcommVND0}, the bounds \eqref{eq:secondcommVND3} and \eqref{eq:secondcommVND4} imply that
		\begin{equation}\label{eq:secondcommVND5}\begin{split}
		e^{-sD} \text{W} e^{sD} & =  \text{W} + \frac{s}{N} \sum^*_{\substack{ u\in \Lambda^*, v, w \in P_L:\\v+u, w-u\in P_L}} N^{\kappa} (\widehat{V}(./N^{1-\kappa})\ast  \eta/N)(u) a^*_{v+u} a^*_{w-u}a_va_w\\
		&\hspace{0.4cm} + \cE_1(s)+\cE_2(s,\delta),
		\end{split}\end{equation}
where the self-adjoint operators $ \cE_1(s)$ and $\cE_2(s)$ are bounded by
		\[\begin{split}\pm \cE_1(s) &\leq C(N^{\alpha+\beta+2\kappa-1} +N^{-3\beta-3\kappa}) \cK + CN^{2\beta+\kappa},
		\end{split}\]	
as well as
		\[\begin{split}\pm \cE_2(s,\delta) &\leq  CN^{\beta+\kappa-1}\cK_L(\cN_{\geq \frac12N^\alpha}+1) + C(\delta + CN^{3\beta/2 + \kappa/2-1}) \int_0^s dt\;  e^{-tD}   \cV_N  e^{tD}\\
		&\hspace{0.4cm} + C(1+\delta^{-1})N^{\beta+\kappa -1}\int_0^s dt\;  e^{-tD}  \cK_{\leq 2N^\beta}(\cN_{\geq \frac12 N^\alpha}+1)e^{tD},
		\end{split}\]
for all $ \delta>0$ and uniformly in $ s\in [-1;1]$. Defining $  \cE_2(s) =  \cE_2(s, N^{-\beta-\kappa})$, this concludes the proof. 		
\end{proof}
%%%%%%%%%%%%%%%%%%%%%%%%%%%%%%%%%%%%%%%%%%%%%%%%%%%%%%%%

Equipped with Proposition \ref{prop:secondcommVND}, we go back to \eqref{eq:MN42} and conclude that 
\begin{equation}\label{eq:MN44}
\begin{split}
\cM_N^{(4)}  &\geq  \cH_N - \frac{1}{2N}\sum_{\substack{ r\in \Lambda^*, v,w\in P_L: \\ v+r, w-r\neq 0}}  \widehat{V}(r/N^{1-\kappa})\big(a^*_{v+r}a^*_{w-r}a_va_w+\text{h.c.}\big)\\
		&\hspace{0.4cm} -\frac1{2N}\sum^*_{\substack{ u\in \Lambda^*, v, w \in P_L:\\v+u, w-u\in P_L}} N^{\kappa} (\widehat{V}(./N^{1-\kappa})\ast  \eta/N)(u) a^*_{v+u} a^*_{w-u}a_va_w \\
		&\hspace{0.4cm} - \frac18\cK - CN^{2\beta+\kappa} + \int_0^{1} ds \, \cE_2(s) + \int_0^1ds\; e^{-sD} \Big(  \cE_{[\cV_N, D]} + \cE_{ [\cK, D] }\Big) e^{sD},
		\end{split}
		\end{equation}
for all $ \alpha\geq 3\beta +2\kappa\geq 0$ with $ \alpha+\beta+2\kappa-1<0$, $0\leq \kappa<\beta$ and $N\in\NN$ large enough.

Next, let us analyse the error terms related to $ \cE_2(s)  $ and $ \cE_{[\cV_N, D]}$ further. The bounds \eqref{eq:propsecondcommVNDerrors} and \eqref{eq:commVNA23} (with $\delta = c N^{-\beta-\kappa}$ for a sufficiently small $c>0$; this choice guarantees that we can extract the term $\cV_{N,L}$ in (\ref{eq:MN45}), with an error that can be absorbed in $\cK$) imply,  together with Lemma \ref{lm:NresgrowD}, Lemma \ref{lm:KresgrowD}, Corollary \ref{cor:VNresgrowD} and Corollary \ref{cor:VNgrowD} and with the assumption (\ref{eq:condab}) on the exponents $\alpha, \beta$, that  
		\[\begin{split}
		&\int_0^{1} ds\; \Big(e^D  \cE_2(s)e^{-D}  + e^{(1-s)D}  \cE_{[\cV_N, D]} e^{-(1-s)D} \Big) \\
		&\hspace{1cm}\geq - CN^{2\beta+2\kappa-1} \cK_L(\cN_{\geq \frac12N^\alpha}+1) - \wt{C} N^{-\beta-\kappa} (\cV_N + \cV_{N,L}) - CN^{2\beta}(\cN_{\geq\frac12 N^\alpha}+1) \\
		&\hspace{1.4cm} - CN^{4\beta+2\kappa-1}(\cN_{\geq\frac12 N^\alpha}+1)^2 
		\end{split}\]
for all $N\in\NN$ large enough and for an arbitrarily small constant $\wt{C} > 0$. With Corollary \ref{cor:VNresgrowD} and \eqref{eq:MN44}, we conclude that 
		\begin{equation}\label{eq:MN45}
		\begin{split}
		\cM_N^{(4)}  &\geq  \cH_N - \frac{1}{2N}\sum_{\substack{ r\in \Lambda^*, v,w\in P_L: \\ v+r, w-r\neq 0}}  \widehat{V}(r/N^{1-\kappa})\big(a^*_{v+r}a^*_{w-r}a_va_w+\text{h.c.}\big)\\
		&\hspace{0.4cm} -\frac1{2N}\sum^*_{\substack{ u\in \Lambda^*, v, w \in P_L:\\v+u, w-u\in P_L}} N^{\kappa} (\widehat{V}(./N^{1-\kappa})\ast  \eta/N)(u) a^*_{v+u} a^*_{w-u}a_va_w \\
		&\hspace{0.4cm} - \frac14\cK - CN^{2\beta+\kappa}- \wt{C} N^{-\beta - \kappa}\cV_{N,L}  + \int_0^1ds\; e^{-sD}  \cE_{ [\cK, D] }  e^{sD} + \cE_{\cM_N}^{(41)},
		\end{split}
		\end{equation}
where the error $ \cE_{\cM_N}^{(41)}$ is such that 
		\begin{equation*}%\label{eq:EMN41}
		\begin{split}
		e^{D} \cE_{\cM_N}^{(41)} e^{-D}  &\geq - CN^{2\beta+2\kappa-1}\cK_L(\cN_{\geq \frac12N^\alpha}+1) - CN^{-\beta-\kappa} \cV_N \\
		&\hspace{0.4cm} - CN^{2\beta} \cN_{\geq\frac12 N^\alpha} - CN^{4\beta+2\kappa-1} \cN_{\geq\frac12 N^\alpha}^2   
		\end{split}
		\end{equation*}
Applying Lemma \ref{lm:NresgrowA}, \ref{lm:KresgrowA} and Corollary \ref{cor:VNgrowA}, we deduce with the operator inequality $ \cN_{\geq \frac12 N^\alpha}\leq 4N^{-2\alpha}\cK$ that 		\begin{equation}\label{eq:EMN41bnd}
		\begin{split}
		e^A e^{D} \cE_{\cM_N}^{(41)} e^{-D} e^{-A} &\geq - CN^{-\beta} \cK - CN^{-\beta-\kappa} \cV_N - C N^{2\beta+2\kappa-1} \\
		&\hspace{0.4cm} - CN^{2\beta+2\kappa-1}\cK \cN_{\geq \frac12N^\alpha}  % - CN^{\alpha+\beta+2\kappa-1}(\cN_{\geq\frac12 N^\alpha}+1)^2   
		\end{split}
		\end{equation}
for all $N\in\NN$ large enough. 

Now, we switch to the contribution containing the operator $  \cE_{ [\cK, D] }$ on the r.h.s. of the lower bound \eqref{eq:MN45}. We recall once again that 
		\[ \int_0^1ds\; e^{-sD} \cE_{[\cK,D]} e^{sD} = \int_0^1ds\;   \sum_{j=1}^3 e^{-sD}\big(\Sigma_j+\text{h.c.}\big)e^{sD},\]
where the operators $ \Sigma_1, \Sigma_2$ and $\Sigma_3$ were defined in \eqref{eq:commKD2}. It turns out that $ \Sigma_2$ and $\Sigma_3$ are negligible errors while $ \Sigma_1$ still contains an important contribution of leading order. We start with the analysis of the contribution related to $ \Sigma_1$.
%%%%%%%%%%%%%%%%%%%%%%%%%%%%%%%%%%%%%%%%%%%%%%%%%%%%%%%%%%%%%%%%%
\begin{prop}\label{prop:secondcommKD} Assume the exponents $\alpha, \beta$ satisfy (\ref{eq:condab}). Then we have that 
		\begin{equation}\label{eq:propsecondcommKD}
		\begin{split}
		&\frac{1}{2N}\sum_{\substack{ u\in P_H^c\cup\{0\}, p, q\in P_L: \\ p+u, q-u\neq 0}}  N^\kappa\big(\widehat{V}(/N^{1-\kappa}) \ast \widehat{f}_N\big)(u) e^{-sD}\big(a^*_{p+u}a^*_{q-u}a_p a_q+\emph{h.c.}\big)e^{sD} \\
		&= \frac{1}{2N}\sum_{\substack{ u\in P_H^c\cup \{0\}, p, q\in P_L: \\ p+u, q-u\neq 0}} N^\kappa \big(\widehat{V}(/N^{1-\kappa})\ast \widehat{f}_N\big)(u) \big(a^*_{p+u}a^*_{q-u}a_p a_q+\emph{h.c.}\big)  + \cE_3(s) 
		\end{split}
		\end{equation}
and there exists a constant $C>0$ such that 		
\begin{equation}\label{eq:propsecondcommKDerror}\begin{split}\pm &e^{A}e^D\cE_3(s) e^{-D}e^{-A} \\
		&\leq  C N^{\alpha + \beta+ 2 \kappa -1}  \cK + C N^{\alpha + \beta+ 2 \kappa -1}  \cK \cN_{\geq \frac12 N^\alpha} +CN^{4\beta +2\kappa} + C N^{\alpha + 3\beta+ 2 \kappa -1}
		\end{split}\end{equation}
for all $ s\in [-1;1]$ and for all $N\in\NN$ sufficiently large.
\end{prop}
\begin{proof} We proceed as in Proposition \ref{prop:secondcommVND} and recall $ \Sigma_1:\cF_+^{\leq N}\to \cF_+^{\leq N}$ to be
		\[\begin{split}
		&\Sigma_{1} =  \frac{1}{2N}\sum_{\substack{ u\in P_H^c\cup\{0\}, p, q\in P_L: \\ p+u, q-u\neq 0}} N^\kappa \big(\widehat{V}(/N^{1-\kappa}) \ast \widehat{f}_N\big)(u) \big(a^*_{p+u}a^*_{q-u}a_p a_q+\text{h.c.}\big).
		\end{split}\]
We then have 
		\begin{equation}\label{eq:secondcommKD1}\begin{split}
		&e^{-sD}\Sigma_{1} e^{sD} - \Sigma_{1} = \int_0^sdt \; e^{-tD} [ \Sigma_{1}, D_1] e^{tD} + \text{h.c.}\\ 
		%&=  \frac{1}{2N}\sum_{\substack{ u\in P_H^c, p, q\in P_L: \\ p+u, q-u\neq 0}} \int_0^s dt\; N^\kappa \big(\widehat{V}(/N^{1-\kappa}) \ast \widehat{f}_N\big)(u) e^{-tD} \Big[ \big(a^*_{p+u}a^*_{q-u}a_p a_q+\text{h.c.}\big), D_1 \Big] e^{tD} +\text{h.c.} 
		\end{split}\end{equation}
Similarly as in \eqref{eq:secondcommVND1} and \eqref{eq:secondcommVND2}, we find that 
		\begin{equation} \label{eq:secondcommKD2}[\Sigma_{1}, D_1] +\text{h.c.} = \sum_{i=1}^8 (\Gamma_i +\text{h.c.}),\end{equation}
where
		\begin{equation*}
		\begin{split}
		\Gamma_{1} &= \frac{1}{N^2} \sum^*_{\substack{ u\in P_H^c\cup\{0\}, r\in P_H, v, w \in P_L:\\v+u+r, w-u-r\in P_L}} N^{\kappa}  \big(\widehat{V}( . /N^{1-\kappa}) \ast \widehat{f}_N\big)(u)  \eta_r a^*_{v+u+r} a^*_{w-u-r}a_va_w,\\
		\Gamma_{2} & = -\frac{1}{2N^{2}}\sum^*_{\substack{u\in P_H^c\cup\{0\}, r\in P_{H},\\ v,w\in P_{L}:\\ w-u, v+u\in P_L }}N^\kappa \big(\widehat{V}(. /N^{1-\kappa}) \ast \widehat{f}_N\big)(u) \eta_r a^*_{v+r} a^*_{w-r}   a_{w-u} a_{v+u} ,\\
		 \Gamma_{3} &= -\frac{1}{2N^{2}}\sum^*_{\substack{u\in P_H^c\cup\{0\}, r\in P_{H},\\ v,w\in P_{L} }}N^\kappa \big(\widehat{V}(. /N^{1-\kappa}) \ast \widehat{f}_N\big)(u)\eta_r  a^*_{v+r} a^*_{w-r}  a_{w-u} a_{v+u} ,\\
		 \Gamma_{4} &= -\frac{1}{N^{2}}\sum^*_{\substack{u\in P_H^c\cup\{0\}, r\in P_{H},\\ v,w,q\in P_L:\\ v-u\in P_{L} }}N^\kappa \big(\widehat{V}(. /N^{1-\kappa}) \ast \widehat{f}_N\big)(u) \eta_r a^*_{v+r} a^*_{w-r}a^*_{q-u} a_w  a_{v-u} a_{q},\\
		 \Gamma_{5} &= \frac{1}{N^{2}}\sum^*_{\substack{u\in P_H^c\cup\{0\}, r\in P_{H},\\ v,w,q \in P_L:\\ v+r+u\in P_{L} }}N^\kappa \big(\widehat{V}(. /N^{1-\kappa}) \ast \widehat{f}_N\big)(u) \eta_r a^*_{v+r+u}a^*_q a^*_{w-r}a_{q+u}   a_v a_w ,\\
		 \Gamma_{6} &= -\frac{1}{N^{2}}\sum^*_{\substack{u\in P_H^c\cup\{0\}, r\in P_{H},\\ v,w,q\in P_{L} }}N^\kappa \big(\widehat{V}(. /N^{1-\kappa}) \ast \widehat{f}_N\big)(u) \eta_r a^*_{v+r} a^*_{w-r}a^*_q a_w  a_{v-u} a_{q+u}.
		 \end{split}
		 \end{equation*}
The operators $ \Gamma_1$ to $\Gamma_6$ can be bounded similarly as in the proof of Proposition \ref{prop:secondcommVND}. Let us start with $ \Gamma_1$. Applying as usual Cauchy-Schwarz implies that
		\[\begin{split} 
		&|\langle\xi, \Gamma_1\xi\rangle| \leq  \frac{CN^\kappa}{N^2} \sum^*_{\substack{ u\in P_H^c\cup\{0\}, r\in P_H, v, w \in P_L:\\ v+u+r, w-u-r\in P_L}} \bigg( |v|^{-1}   \| a_{v+u+r} a_{w-u-r}\xi\|\bigg) \bigg(  |\eta_r||v| \| a_va_w\xi\|\bigg) \\
		&\leq  CN^{\alpha/2+ 5\beta/2+ 2\kappa-1/2} \| \xi \| \|\cK_L^{1/2}\xi\|\leq  CN^{\alpha+ \beta+ 2\kappa-1}\langle\xi, \cK_L\xi\rangle + CN^{4\beta +2\kappa} \| \xi \|^2 
		\end{split}\]
where we used that $v+u+r \in P_L $ implies $ |u|\geq N^{\alpha}-3N^\beta$ and $ |r|\leq N^{\alpha}+3N^\beta$ whenever $ u\in P_H^c, r\in P_H$ and $ v\in P_L$ (otherwise $ |u+r+v| \geq |r| - |u|-N^\beta\geq 2N^\beta > N^\beta$ if either $ |u|\leq N^{\alpha}-3N^\beta$ or $ |r|\geq N^{\alpha}+3N^\beta$, in contradiction to $ u+r+v\in P_L$) for $N\in\NN$ sufficiently large. Notice in addition that $ \sum_{ N^\alpha-3N^\beta \leq |u|\leq N^\alpha} \leq C N^{2\alpha+\beta}$.

The term $ \Gamma_2$ can be estimated exactly as the term $ \zeta_2$ in \eqref{eq:secondcommVND333}, that is
		\[ |\langle\xi, \Gamma_2\xi\rangle| \leq CN^{\alpha+\beta+2\kappa-1} \langle \xi,  \cK_{\leq 2N^\beta}\xi\rangle + CN^{2\beta+\kappa-1}\langle\xi, (\cN_{\geq \frac12N^\alpha}+1)\xi\rangle.\]
The contribution $\Gamma_3$ can be controlled by 
		\[\begin{split}
		|\langle\xi, \Gamma_3\xi\rangle| &\leq  \frac{CN^\kappa}{2N^{2}}\!\!\!\!\sum^*_{\substack{u\in P_H^c\cup\{0\}, r\in P_{H},\\ v,w\in P_{L} }}\!\!\!\!|\eta_r| \|   a_{v+r}(\cN_{\geq \frac12N^\alpha}+1)^{-1/2} a_{w-r}\xi\| \|   a_{w-u} (\cN_{\geq \frac12N^\alpha}+1)^{1/2}a_{v+u} \xi\| \\
		&\leq CN^{\alpha + 3\beta + 2\kappa-1}   \langle\xi, (\cN_{\geq \frac12N^\alpha}+1)\xi\rangle.
		\end{split}\]
The terms $ \Gamma_4$ and $\Gamma_5$ can be bounded exactly as in \eqref{eq:secondcommVND334}. We find
		\[|\langle\xi, \Gamma_4\xi\rangle| +|\langle\xi, \Gamma_5\xi\rangle|  \leq C N^{\beta+\kappa -1} \langle\xi, \cK_{\leq 2N^\beta}(\cN_{\geq \frac12 N^\alpha}+1)\xi\rangle,\]
Finally, the last contribution $ \Gamma_6$ is bounded by
		\[\begin{split}
		|\langle\xi, \Gamma_6\xi\rangle| &\leq \frac{CN^\kappa}{N^{2}}\sum^*_{\substack{u\in P_H^c\cup\{0\}, r\in P_{H},\\ v,w,q\in P_{L} }} \Big( |q||w|^{-1} \|  a_{v+r} (\cN_{\geq \frac12N^\alpha}+1)^{-1/2}a_{w-r}a _q  \xi\| \Big)\\
		&\hspace{4cm} \times \Big( |\eta_r||w| |q|^{-1}\| a_{v-u} (\cN_{\geq \frac12N^\alpha}+1)^{1/2}  a_wa_{q+u} \xi \|\Big)\\
		&\leq C N^{\alpha + \beta+ 2 \kappa -1}  \langle\xi, \cK_L(\cN_{\geq \frac12 N^\alpha}+1)\xi\rangle.
		\end{split}\]
In conclusion, the above estimates imply that
		\[\begin{split} 
		\pm  \sum_{i=1}^6 \big( \Gamma_i+\text{h.c.}\big)&\leq  C N^{\alpha + \beta+ 2 \kappa -1}  \cK_{\leq 2N^\beta}+ C N^{\alpha + \beta+ 2 \kappa -1}  \cK_{\leq 2N^\beta}(\cN_{\geq \frac12 N^\alpha}+1)\\
		& \hspace{0.4cm} + CN^{\alpha + 3\beta + 2\kappa-1}   (\cN_{\geq \frac12N^\alpha}+1)+CN^{4\beta +2\kappa}
		\end{split}\]
for all $ \alpha>3\beta +2\kappa\geq 0$ and for all $N\in\NN$ sufficiently large. Combining this estimate with the identites \eqref{eq:secondcommKD1} and  \eqref{eq:secondcommKD2}, and applying Lemma \ref{lm:NresgrowA}, \ref{lm:KresgrowA}, Lemma \ref{lm:NresgrowD} as well as Lemma \ref{lm:KresgrowD} together with the operator inequality $ \cN_{\geq \frac12 N^\alpha}\leq 4 N^{-2\alpha }\cK$ proves the proposition.
\end{proof}
%%%%%%%%%%%%%%%%%%%%%%%%%%%%%%%%%%%%%%%%%%%%%%%%%%%%%%%%%%%%%%%%%
Applying Proposition \ref{prop:secondcommKD} to the lower bound \eqref{eq:MN45} and defining $ \cE_{\cM_N}^{(42)} = \int_0^1 ds\; \cE_3(s)$ with $ \cE(s)$ from Proposition \ref{prop:secondcommKD}, we conclude that 
		\begin{equation}\label{eq:MN46}
		\begin{split}
		\cM_N^{(4)}  &\geq  \cH_N - \frac{1}{2N}\sum_{\substack{ r\in \Lambda^*, v,w\in P_L: \\ v+r, w-r\neq 0}}  \widehat{V}(r/N^{1-\kappa})\big(a^*_{v+r}a^*_{w-r}a_va_w+\text{h.c.}\big)\\
		&\hspace{0.4cm} -\frac1{2N}\sum^*_{\substack{ u\in \Lambda^*, v, w \in P_L:\\v+u, w-u\in P_L}} N^{\kappa} (\widehat{V}(./N^{1-\kappa})\ast  \eta/N)(u) a^*_{v+u} a^*_{w-u}a_va_w \\
		&\hspace{0.4cm} +\frac{1}{2N}\sum_{\substack{ u\in P_H^c\cup \{0\}, p, q\in P_L: \\ p+u, q-u\neq 0}} N^\kappa \big(\widehat{V}(/N^{1-\kappa})\ast \widehat{f}_N\big)(u) \big(a^*_{p+u}a^*_{q-u}a_p a_q+\emph{h.c.}\big)\\
		&\hspace{0.4cm} - \frac14\cK - C N^{-\beta-\kappa}\cV_{N,L}   + \cE_{\cM_N}^{(41)}+ \cE_{\cM_N}^{(42)} + \int_0^1ds\; e^{-sD}  \big( \Sigma_2+\Sigma_3+\text{h.c.}\big) e^{sD},
		\end{split}
		\end{equation}
where $ \cE_{\cM_N}^{(41)}$ satisfies the lower bound \eqref{eq:EMN41bnd}, $\cE_{\cM_N}^{(42)}$ satisfies the bound \eqref{eq:propsecondcommKDerror} and where the operators $\Sigma_2 $ and $\Sigma_3$ were defined in \eqref{eq:commKD2}. 

Let us finally estimate the size of the error in the last line of \eqref{eq:MN46}, involving the two operators $\Sigma_2 $ and $\Sigma_3$. Using the estimate \eqref{eq:commKD5} together with Lemma \ref{lm:NresgrowA}, \ref{lm:KresgrowA}, Lemma \ref{lm:NresgrowD} and Lemma \ref{lm:KresgrowD}, we find for $ \cE_{\cM_N}^{(43)} = \int_0^1ds\; e^{-sD}  \big( \Sigma_2+\text{h.c.}\big) e^{sD}$ 
		\begin{equation}\label{eq:EMN43bnd}
		\begin{split}
		e^{A}e^{D}\cE_{\cM_N}^{(43)}e^{-D}e^{-A} \geq  - CN^{-\beta -1} \cK \cN_{\geq \frac12 N^\alpha}-CN^{-5\beta-4\kappa}\cK  -CN^\beta.
		\end{split}
		\end{equation}
Finally, consider the operator $ \cE_{\cM_N}^{(44)} = \int_0^1ds\; e^{-sD}  \big( \Sigma_3+\text{h.c.}\big) e^{sD}$, with $\Sigma_3$ defined in \eqref{eq:commKD2}. Let $ m_0\in\mathbb{R}$ be such that $ m_0\beta =\alpha$ (in particular, $ \lfloor m_0 \rfloor \geq 3$). Here, we use the bound \eqref{eq:commKD6} to find first of all that 
		\[\begin{split}
		 \cE_{\cM_N}^{(44)} \geq - \int_0^1 ds\; \| \cK^{1/2}e^{sD} \xi\| \Big( N^{-1/2}\big \|\cK_L^{1/2}(\cN_{\geq \frac12 N^\alpha}+1)^{1/2}\xi \big\| +  N^{\beta-1}\big \| (\cN_{\geq \frac12 N^\alpha}+1)^{3/2}\xi \big\|\Big)
		\end{split}\]
for any $ \xi \in \cF_+^{\leq N}$ with $ \|\xi\|=1$. Notice that we applied once again Lemma \ref{lm:NresgrowD} and Lemma \ref{lm:KresgrowD} in the second factor. With Corollary \ref{cor:KgrowD}, the first factor is bounded by
		\[\begin{split}
		 &\cE_{\cM_N}^{(44)} \\
		 &\geq - C \bigg( \|\cK^{1/2}\xi\|  + \| \cV_N^{1/2}\xi\|  + \| \cV_{N,L}^{1/2}\xi\| + N^{5\beta/8+\kappa/2}\| \cK_{\leq N^{3\beta/2}}^{1/2} \xi\| \\
		 &\hspace{1.2cm}+N^{-1/2}\big \|\cK_L^{1/2}(\cN_{\geq \frac12 N^\alpha}+1)^{1/2}\xi \big\| + N^{3\beta/2+\kappa/2} \\
		 &\hspace{1.2cm}  +  \sum_{j=3}^{2\lfloor m_0\rfloor -1}N^{ j\beta/4 + 3\beta/4+\kappa-1/2}\Big[ \big \|\cK_L^{1/2}(\cN_{\geq \frac12 N^{j\beta/2}}+1)^{1/2}\xi\big \| \\
		 &\hspace{6.1cm} + N^\beta\big \|(\cN_{\geq \frac12 N^{\alpha}}+1)^{1/2}(\cN_{\geq \frac12 N^{j\beta/2}}+1)^{1/2}\xi\big \|\Big] \\
		  &\hspace{1.2cm}  +  N^{ \alpha/2 +\beta/2 +\kappa-1/2}\Big[ \big \|\cK_L^{1/2}(\cN_{\geq \frac12 N^{\lfloor m_0 \rfloor \beta}} +1)^{1/2}\xi\big \| \\
		  &\hspace{4.5cm} + N^\beta \big \|(\cN_{\geq \frac12 N^{\alpha}}+1)^{1/2}(\cN_{\geq \frac12 N^{\lfloor m_0 \rfloor \beta}} +1)^{1/2}\xi\big \| \Big] \bigg) \\
		 &\hspace{1cm}\times\bigg( N^{-1/2}\big \|\cK_L^{1/2}(\cN_{\geq \frac12 N^\alpha}+1)^{1/2}\xi \big\| +  N^{\beta-1}\big \| (\cN_{\geq \frac12 N^\alpha}+1)^{3/2}\xi \big\|\bigg)
		\end{split}\]
for all exponents $\alpha, \beta$ satisfying (\ref{eq:condab}) and $N\in\NN$ sufficiently large. It follows that 		
\begin{equation}\label{eq:MN47}\begin{split}
\cE_{\cM_N}^{(44)} \geq \cE_{\cM_N}^{(441)} +\cE_{\cM_N}^{(442)} +\cE_{\cM_N}^{(443)}, 
\end{split}\end{equation}
where
\begin{equation}\label{eq:EMN441442bnd}
\begin{split}
 \cE_{\cM_N}^{(441)} = - \frac18 \cK - \wt{C} N^{-\alpha} \cV_{N,L}-CN^{3\beta+\kappa}, \hspace{1cm} \cE_{\cM_N}^{(442)} = N^{-\alpha}\cV_N
\end{split}
\end{equation}
with an arbitrarily small constant $\wt{C} >0$ and where, after an additional application of Lemmas \ref{lm:NresgrowA}, \ref{lm:KresgrowA}, \ref{lm:NresgrowD} and \ref{lm:KresgrowD} together with the operator bound $\cN_{\geq \Theta}\leq \Theta^{-2}\cK $, the error $ \cE_{\cM_N}^{(443)}$ is such that 
		\begin{equation}\label{eq:EMN443bnd}
		\begin{split}
		&e^A e^{D} \cE_{\cM_N}^{(443)}e^{-D}e^{-A}\\
		&\geq -C N^{\alpha+\beta+2\kappa-1}\cK-C N^{\alpha-1}\cK \cN_{\geq \frac12N^\alpha} -  C N^{\alpha+3\beta+2\kappa-1}\\
		&\hspace{0.5cm}  - C \sum_{j=3}^{2\lfloor m_0\rfloor -1}N^{ j\beta/2 + \beta/2+2\kappa-1}\cK \cN_{\geq \frac12 N^{j\beta/2}}-C N^{ \alpha +\beta +2\kappa-1}\cK\cN_{\geq \frac12 N^{\lfloor m_0 \rfloor \beta}}
		\end{split}
		\end{equation}
for all exponents $\alpha, \beta$ satisfying (\ref{eq:condab}) and $N\in\NN$ sufficiently large.

Choosing $\widetilde C>0$ sufficiently large (but independently of $N\in\NN$) and arguing as right before \eqref{eq:MN45}, we deduce that 		
\begin{equation}\label{eq:EMN442bnd}
		\begin{split}
		&e^A \Big( \widetilde C N^{-\alpha} e^D\cV_{N,L } e^{-D}+ e^D \cE_{\cM_N}^{(442)} e^{-D}\Big)e^{-A} \\
		&\hspace{3cm}\geq- C N^{-\alpha}\cV_N - CN^{-3\beta-\kappa}\cN_+- CN^{-2\beta-\kappa-1}\cK\cN_{\geq\frac12N^\alpha}
		\end{split}
		\end{equation}
for all $\alpha, \beta$ satisfying (\ref{eq:condab}) and $N\in\NN$ sufficiently large. This follows through another application of Corollary \ref{cor:VNgrowA}, Corollary \ref{cor:VNresgrowD} and Corollary \ref{cor:VNgrowD}, together with Lemma \ref{lm:NresgrowA}, Lemma \ref{lm:KresgrowA}, Lemma \ref{lm:NresgrowD} and Lemma \ref{lm:KresgrowD}. We summarize these bounds in the following corollary.

\begin{cor}\label{cor:MN4} Let $ m_0\in\mathbb{R}$ be such that $ m_0\beta =\alpha$ and let $ \cM_N^{(4)}$ be defined as in \eqref{eq:MN0to4}. For every $\wt{C} > 0$, there exists a constant $C>0$ such that  
		\begin{equation}\label{eq:corMN4}
		\begin{split}
		\cM_N^{(4)}  &\geq \frac12\cK +\cV_N - \frac{1}{2N}\sum_{\substack{ r\in \Lambda^*, v,w\in P_L: \\ v+r, w-r\neq 0}}  \widehat{V}(r/N^{1-\kappa})\big(a^*_{v+r}a^*_{w-r}a_va_w+\text{h.c.}\big)\\
		&\hspace{0.4cm} -\frac1{2N}\sum^*_{\substack{ u\in \Lambda^*, v, w \in P_L:\\v+u, w-u\in P_L}} N^{\kappa} (\widehat{V}(./N^{1-\kappa})\ast  \eta/N)(u) a^*_{v+u} a^*_{w-u}a_va_w \\
		&\hspace{0.4cm} +\frac{1}{2N}\sum_{\substack{ u\in P_H^c\cup \{0\}, p, q\in P_L: \\ p+u, q-u\neq 0}} N^\kappa \big(\widehat{V}(/N^{1-\kappa})\ast \widehat{f}_N\big)(u) \big(a^*_{p+u}a^*_{q-u}a_p a_q+\emph{h.c.}\big)\\
		&\hspace{0.4cm} -\wt{C} N^{-\beta-\kappa} \cV_{N,L}   + \cE_{\cM_N}^{(4)}
		\end{split}
		\end{equation}
where  
		\begin{equation}\label{cor:error}
		\begin{split}
		&e^Ae^D  \cE_{\cM_N}^{(4)} e^{-D}e^{-A}\\
		& \geq - C N^{-\beta}\cK - C N^{-\beta-\kappa} \cV_N  -C N^{ \alpha +\beta +2\kappa-1}\cK\cN_{\geq \frac12 N^{\lfloor m_0 \rfloor \beta}} \\
		&\hspace{0.5cm}  - C \sum_{j=3}^{2\lfloor m_0\rfloor -1}N^{ j\beta/2 + \beta/2+2\kappa-1}\cK \cN_{\geq \frac12 N^{j\beta/2}}- CN^{2\beta+\kappa} 
		\end{split}
		\end{equation}
for all  exponents $\alpha,\beta$ satisfying (\ref{eq:condab}) and for all $N\in\NN$ sufficiently large. 
\end{cor}
\begin{proof} The proof follows from defining $ \cE_{\cM_N}^{(4)} = \sum_{j=1}^3\cE_{\cM_N}^{(4j)} + \sum_{j=1}^3\cE_{\cM_N}^{(44j)}$ and combining \eqref{eq:EMN41bnd}, \eqref{eq:MN46}, \eqref{eq:propsecondcommKDerror}, \eqref{eq:EMN43bnd}, \eqref{eq:MN47}, \eqref{eq:EMN441442bnd}, \eqref{eq:EMN442bnd}, \eqref{eq:EMN443bnd} and the operator bound $ \cN_+\leq (2\pi)^{-2}\cK$ in $ \cF_+^{\leq N}$.
\end{proof}

%&\hspace{0.4cm}+ C \sum_{j=3}^{2\lfloor m_0\rfloor -1}N^{ j\beta/2 + 3\beta/2+2\kappa-1}\Big[\cK_L + N^{2\beta}(\cN_{\geq \frac12 N^\alpha} +1)\Big](\cN_{\geq \frac12 N^{j\beta/2}}+1)\\
%		&\hspace{0.4cm}+C N^{ \alpha +\beta +2\kappa-1}\Big[\cK_L + N^{2\beta}(\cN_{\geq \frac12 N^\alpha} +1)\Big](\cN_{\geq \frac12 N^{\lfloor m_0 \rfloor \beta}} +1) + CN^{3\beta+\kappa} 

\subsection{Proof of Proposition \ref{prop:MN}}\label{sub:MNproof}

Recall from (\ref{eq:decoMN}) the decomposition  
		\[ \cM_N = 4\pi \mathfrak{a}_0N^{1+\kappa}-4\pi\mathfrak{a}_0N^{\kappa-1}\cN_+^{\,2}/N + \cM_N^{(2)}+ \cM_N^{(3)} + \cM_N^{(4)} \]
Collecting the results of Proposition \ref{prop:MN2}, Proposition \ref{prop:MN3} and Corollary \ref{cor:MN4}, we deduce that 
		\begin{equation}\label{eq:proofpropMN1}
		\begin{split}
		\cM_N&\geq 4\pi \mathfrak{a}_0N^{1+\kappa}-4\pi\mathfrak{a}_0N^{\kappa-1}\cN_+^{\,2} + 8\pi\mathfrak{a}_0N^\kappa\sum_{p\in P_H^c}\Big[ b^*_pb_p +\frac12 b^*_pb^*_{-p} + \frac12 b_pb_{-p}\Big]\\
		&\hspace{0.4cm} + \frac{8\pi\mathfrak{a}_0N^\kappa}{\sqrt N}\sum_{\substack{ p\in P_H^c,q\in P_L:\\p+q\neq 0 }} \big[ b^*_{-p}a^*_{p+q}a_q +\text{h.c.}\big]+\frac12\cK\\
		&\hspace{0.4cm}  +\cV_N - \frac{1}{2N}\sum_{\substack{ r\in \Lambda^*, v,w\in P_L: \\ v+r, w-r\neq 0}}  \widehat{V}(r/N^{1-\kappa})\big(a^*_{v+r}a^*_{w-r}a_va_w+\text{h.c.}\big)\\
		&\hspace{0.4cm} -\frac1{2N}\sum^*_{\substack{ r\in \Lambda^*, v, w \in P_L:\\v+r, w-r\in P_L}} N^{\kappa} (\widehat{V}(./N^{1-\kappa})\ast  \eta/N)(r) a^*_{v+r} a^*_{w-r}a_va_w \\
		&\hspace{0.4cm} +\frac{1}{2N}\sum_{\substack{ r\in P_H^c\cup \{0\}, v, w\in P_L: \\ v+r,w-r\neq 0}} N^\kappa \big(\widehat{V}(./N^{1-\kappa})\ast \widehat{f}_N\big)(r) \big(a^*_{v+r}a^*_{w-r}a_v a_w+\text{h.c.}\big)\\
		&\hspace{0.4cm}  - \wt{C} N^{-\beta-\kappa}\cV_{N,L}   +\cE'_{\cM_N},
		\end{split}
		\end{equation}
where $\cE'_{\cM_N}$ satisfies the lower bound 
		\begin{equation}\label{eq:proofpropMN2}
		\begin{split}
		e^Ae^D \cE'_{\cM_N} e^{-D}e^{-A}& \geq - C N^{-\beta} \cK - C N^{-\beta-\kappa}\cV_N  -C N^{ \alpha +\beta +2\kappa-1}\cK\cN_{\geq \frac12 N^{\lfloor m_0 \rfloor \beta}} \\
		&\hspace{0.5cm}  - C \sum_{j=3}^{2\lfloor m_0\rfloor -1}N^{ j\beta/2 + \beta/2+2\kappa-1}\cK \cN_{\geq \frac12 N^{j\beta/2}}- CN^{\alpha + \beta/2 +2\kappa} 
		\end{split}
		\end{equation}
for all $N\in\NN$ sufficiently large. 

We combine next the terms on the third, fourth and fifth lines in \eqref{eq:proofpropMN1}. We first notice that
		\begin{equation}\label{eq:proofpropMN4}
		\begin{split}
		&\frac{1}{2N}\sum_{\substack{ r\in \Lambda^*, v,w\in P_L: \\ v+r, w-r\neq 0}}  \widehat{V}(r/N^{1-\kappa})\big(a^*_{v+r}a^*_{w-r}a_va_w+a^*_{v}a^*_{w}a_{w-r}a_{v+r}\big)\\
		&= \frac{1}{2N}\sum_{\substack{ r\in \Lambda^*, v,w\in \Lambda_+^* :\\ v,w\in P_L, \\ v+r, w-r\neq 0}}\!\!\!  \widehat{V}(r/N^{1-\kappa})a^*_{v+r}a^*_{w-r}a_va_w + \frac{1}{2N}\!\!\sum_{\substack{ r\in \Lambda^*, v,w\in \Lambda_+^* : \\v+r,w-r\in P_L}} \!\!\!\!\! \widehat{V}(r/N^{1-\kappa})a^*_{v+r}a^*_{w-r}a_va_w \\
		& = \frac{1}{2N}\sum^*_{\substack{ r\in \Lambda^*, v,w\in \Lambda_+^* :\\ (v,w)\in P_L^2\text{ or }( v+r, w-r)\in P_L^2}}  \widehat{V}(r/N^{1-\kappa})a^*_{v+r}a^*_{w-r}a_va_w\\
		&\hspace{4cm} +  \frac{1}{2N}\sum^*_{\substack{ r\in \Lambda^*, v,w\in \Lambda_+^* :\\ (v,w, v+r, w-r)\in P_L^4}}  \widehat{V}(r/N^{1-\kappa})a^*_{v+r}a^*_{w-r}a_va_w
		\end{split}
		\end{equation}
Arguing in the same way for the contribution on the fifth line in \eqref{eq:proofpropMN1}, using that $ (\widehat{f}_N-\eta/N)(p)=\delta_{p,0}$ for all $p\in\Lambda_+^*$, and using that $ v\in P_L$ and $ v+r\in P_L $ implies in particular that $ r\in P_H^c$, we therefore obtain that
		\begin{equation}\label{eq:proofpropMN5}
		\begin{split}
		&\cV_N - \frac{1}{2N}\sum^*_{\substack{ r\in \Lambda^*, v,w\in \Lambda_+^*: \\ (v,w)\in P_L^2\text{ or }(v+r, w-r)\in P_L^2}}  \widehat{V}(r/N^{1-\kappa}) a^*_{v+r}a^*_{w-r}a_va_w\\
		&- \frac{1}{2N}\sum_{\substack{ r\in \Lambda^*, v,w\in \Lambda_+^*: \\ (v,w,v+r, w-r)\in P_L^4}}  \widehat{V}(r/N^{1-\kappa}) a^*_{v+r}a^*_{w-r}a_va_w \\
		& -\frac1{2N}\sum^*_{\substack{ r\in \Lambda^*, v, w \in P_L:\\v+r, w-r\in P_L}} N^{\kappa} (\widehat{V}(./N^{1-\kappa})\ast  \eta/N)(r) a^*_{v+r} a^*_{w-r}a_va_w \\
		& +\frac{1}{2N}\sum_{\substack{ r\in P_H^c\cup \{0\}, v, w\in P_L: \\ v+r, w-r\neq 0}} N^\kappa \big(\widehat{V}(/N^{1-\kappa})\ast \widehat{f}_N\big)(r) \big(a^*_{v+r}a^*_{w-r}a_v a_w+\text{h.c.}\big)\\
		& = \cV_N - \frac{1}{2N}\sum^*_{\substack{ r\in \Lambda^*, v,w\in \Lambda_+^*: \\ (v,w)\in P_L^2\text{ or }(v+r, w-r)\in P_L^2}}  \widehat{V}(r/N^{1-\kappa}) a^*_{v+r}a^*_{w-r}a_va_w \\
		&\hspace{0.4cm} +\frac{1}{2N}\sum_{\substack{ r\in P_H^c\cup \{0\}, v, w\in P_L: \\ (v,w)\in P_L^2\text{ or }(v+r, w-r)\in P_L^2}}^* N^\kappa \big(\widehat{V}(/N^{1-\kappa})\ast \widehat{f}_N\big)(r)  a^*_{v+r}a^*_{w-r}a_v a_w.
		\end{split}
		\end{equation}
Now, notice furthermore that
		\[\begin{split}
		&\cV_N - \frac{1}{2N}\sum^*_{\substack{ r\in \Lambda^*, v,w\in \Lambda_+^*: \\ (v,w)\in P_L^2\text{ or }(v+r, w-r)\in P_L^2}}  \widehat{V}(r/N^{1-\kappa}) a^*_{v+r}a^*_{w-r}a_va_w\\
		&\hspace{2cm}=\frac{1}{2N}\sum^*_{\substack{ r\in \Lambda^*, v,w\in \Lambda_+^*: \\ (v,w)\in (P_L^2)^c \text{ and  }\\(v+r, w-r)\in (P_L^2)^c}}  \widehat{V}(r/N^{1-\kappa}) a^*_{v+r}a^*_{w-r}a_va_w,
		\end{split}\]
such that, after switching to position space, the pointwise positivity $V\geq 0$ implies 
		\begin{equation}\label{eq:proofpropMN6}\begin{split}
		&\cV_N - \frac{1}{2N}\sum^*_{\substack{ r\in \Lambda^*, v,w\in \Lambda_+^*: \\ (v,w)\in P_L^2\text{ or }(v+r, w-r)\in P_L^2}}  \widehat{V}(r/N^{1-\kappa}) a^*_{v+r}a^*_{w-r}a_va_w\\
		&= \int_{\Lambda^2}dxdy\; N^{2-2\kappa}V (N^{1-\kappa}(x-y))\\
		&\hspace{0.5cm}\times\Big[ a^*\big( (\check{\chi}_{P_L^c})_x\big)a^*\big( (\check{\chi}_{P_L^c})_y\big) + a^*\big( (\check{\chi}_{P_L})_x\big) a^*\big( (\check{\chi}_{P_L^c})_y\big) + a^*\big( (\check{\chi}_{P_L^c})_x\big)a^*\big( (\check{\chi}_{P_L})_y\big)\Big] \\
		& \hspace{0.5cm}\times\Big[ a\big( (\check{\chi}_{P_L^c})_x\big)a\big( (\check{\chi}_{P_L^c})_y\big) + a\big( (\check{\chi}_{P_L})_x\big) a\big( (\check{\chi}_{P_L^c})_y\big) + a\big( (\check{\chi}_{P_L^c})_x\big)a\big( (\check{\chi}_{P_L})_y\big)\Big]\\
		& \geq 0.
		\end{split}\end{equation}
Here, we used that $ \Lambda_+^* = P_L \cup P_L^c  $ and we denote by $ \check{\chi}_{S}$ the distribution which has Fourier transform $ \chi_S$, the characteristic function of the set $S\subset \Lambda_+^*$.

Combining \eqref{eq:proofpropMN1}, \eqref{eq:proofpropMN4}, \eqref{eq:proofpropMN5} and \eqref{eq:proofpropMN6}, it follows that  
		\begin{equation}\label{eq:proofpropMN7}
		\begin{split}
		\cM_N&\geq 4\pi \mathfrak{a}_0N^{1+\kappa}-4\pi\mathfrak{a}_0N^{\kappa-1}\cN_+^{\,2} + 8\pi\mathfrak{a}_0N^\kappa\sum_{p\in P_H^c}\Big[ b^*_pb_p +\frac12 b^*_pb^*_{-p} + \frac12 b_pb_{-p}\Big]\\
		&\hspace{0.4cm} + \frac{8\pi\mathfrak{a}_0N^\kappa}{\sqrt N}\sum_{\substack{ p\in P_H^c,q\in P_L:\\p+q\neq 0 }} \big[ b^*_{-p}a^*_{p+q}a_q +\text{h.c.}\big]+\frac12\cK\\
		&\hspace{0.4cm}  +\frac{1}{2N}\sum_{\substack{ r\in P_H^c\cup \{0\}, v, w\in P_L: \\ (v,w)\in P_L^2\text{ or }(v+r, w-r)\in P_L^2}}^* N^\kappa \big(\widehat{V}(./N^{1-\kappa})\ast \widehat{f}_N\big)(r)  a^*_{v+r}a^*_{w-r}a_v a_w\\
		&\hspace{0.4cm} - \wt{C} N^{-\beta-\kappa}\cV_{N,L}   +\cE'_{\cM_N}
		\end{split}
		\end{equation}
Using Lemma \ref{sceqlemma}, part ii), we have $\big(\widehat{V}(./N^{1-\kappa})\ast \widehat{f}_N\big)(0) = 8\pi \mathfrak{a}_0 + \mathcal{O}(N^{\kappa-1}) $. This implies 		\begin{equation}\label{eq:proofpropMN8}
		\begin{split}
		\cM_N&\geq 4\pi \mathfrak{a}_0N^{1+\kappa} + 8\pi\mathfrak{a}_0N^\kappa\sum_{p\in P_H^c}\Big[ b^*_pb_p +\frac12 b^*_pb^*_{-p} + \frac12 b_pb_{-p}\Big]\\
		&\hspace{0.4cm} + \frac{8\pi\mathfrak{a}_0N^\kappa}{\sqrt N}\sum_{\substack{ p\in P_H^c,q\in P_L:\\p+q\neq 0 }} \big[ b^*_{-p}a^*_{p+q}a_q +\text{h.c.}\big]+\frac12\cK\\
		&\hspace{0.4cm}  +\frac{1}{2N}\sum_{\substack{ r\in P_H^c , v, w\in P_L: \\ (v,w)\in P_L^2\text{ or }(v+r, w-r)\in P_L^2}}^* N^\kappa \big(\widehat{V}(. /N^{1-\kappa})\ast \widehat{f}_N\big)(r)  a^*_{v+r}a^*_{w-r}a_v a_w\\ &\hspace{0.4cm} - \wt{C} N^{-\beta -\kappa} \cV_{N,L}   + \cE''_{\cM_N},
		\end{split}
		\end{equation}
where, by (\ref{eq:proofpropMN2}) and Lemmas \ref{lm:NresgrowA} and \ref{lm:NresgrowD}, 
\begin{equation}\label{eq:proofpropMN9}
\begin{split} 
e^{A}e^D \cE''_{\cM_N} e^{-D}e^{-A}& \geq - C N^{-\beta} \cK - C N^{-\beta-\kappa}\cV_N  -C N^{ \alpha +\beta +2\kappa-1}\cK\cN_{\geq \frac12 N^{\lfloor m_0 \rfloor \beta}} \\
		&\hspace{0.5cm}  - C \sum_{j=3}^{2\lfloor m_0\rfloor -1}N^{ j\beta/2 + \beta/2+2\kappa-1}\cK \cN_{\geq \frac12 N^{j\beta/2}}- CN^{\alpha + \beta/2 +2\kappa} 
\end{split} \end{equation}

Similarly, for $ r\in P_H^c$, we know that
		\[\big| \big(\widehat{V}(./N^{1-\kappa})\ast \widehat{f}_N\big)(r) - 8\pi \mathfrak{a}_0   \big| \leq CN^{\alpha+\kappa-1}. \]
Therefore, proceeding exactly as between \eqref{eq:commKD3} and \eqref{eq:commKD4}, with $\big(\widehat{V}(./N^{1-\kappa})\ast \widehat{f}_N\big)(r)$ replaced by $\big| \big(\widehat{V}(/N^{1-\kappa})\ast \widehat{f}_N\big)(r) - 8\pi \mathfrak{a}_0   \big| $, we deduce that 
		\begin{equation}\label{eq:proofpropMN10}
		\begin{split}
		\cM_N&\geq 4\pi \mathfrak{a}_0N^{1+\kappa} +\frac12\cK+ 8\pi\mathfrak{a}_0N^\kappa\sum_{p\in P_H^c}\Big[ b^*_pb_p +\frac12 b^*_pb^*_{-p} + \frac12 b_pb_{-p}\Big]\\
		&\hspace{0.4cm} + \frac{8\pi\mathfrak{a}_0N^\kappa}{\sqrt N}\sum_{\substack{ p\in P_H^c,q\in P_L:\\p+q\neq 0 }} \big[ b^*_{-p}a^*_{p+q}a_q +\text{h.c.}\big]\\
		&\hspace{0.4cm}  +\frac{4\pi\mathfrak{a}_0N^\kappa}{N}\sum_{\substack{ r\in P_H^c , v, w\in P_L: \\ (v,w)\in P_L^2\\\text{ or }(v+r, w-r)\in P_L^2}}^*  a^*_{v+r}a^*_{w-r}a_v a_w  - \wt{C} N^{-\beta-\kappa} \cV_{N,L}   +\cE'''_{\cM_N},
		\end{split}
		\end{equation}
with $\cE'''_{\cM_N}$ satisfying the same bound (\ref{eq:proofpropMN9}) as $\cE''_{\cM_N}$. Here we used Lemmas \ref{lm:NresgrowA}, \ref{lm:KresgrowA}, \ref{lm:NresgrowD} and \ref{lm:KresgrowD}, as well as the assumption (\ref{eq:condab}). 

Finally, recalling the definition \eqref{eq:defcrdr} and the identity \eqref{eq:cubicbcd}, we find 
		\begin{equation}\label{eq:proofpropMN11}
		\begin{split}
		\cM_N&\geq 4\pi \mathfrak{a}_0N^{1+\kappa} +\frac12\cK+ 8\pi\mathfrak{a}_0N^\kappa\sum_{p\in P_H^c}\Big[ b^*_pb_p +\frac12 b^*_pb^*_{-p} + \frac12 b_pb_{-p}\Big]\\
		&\hspace{0.4cm} + 8\pi\mathfrak{a}_0N^\kappa \sum_{p\in P_H^c } \Big[ b^*_{-p}e_{-p} + e^*_{-p}b_{-p} + b^*_{-p}e^*_{p} + e_{p}b_{-p} + b^*_{-p}c^*_{p} + c_p b_{-p} \Big]\\
		&\hspace{0.4cm}  +\frac{4\pi\mathfrak{a}_0N^\kappa}{N}\sum_{\substack{ r\in P_H^c , v, w\in P_L: \\ (v,w)\in P_L^2\\\text{ or }(v+r, w-r)\in P_L^2}}^*  a^*_{v+r}a^*_{w-r}a_v a_w - \wt{C} N^{-\beta-\kappa} \cV_{N,L} + \cE'''_{\cM_N}\, .		
		\end{split}
		\end{equation}
To express also the first term in the third line of \eqref{eq:proofpropMN11} in terms of the modified creation and annihilation fields defined in \eqref{eq:defcrdr}, we first observe that
		\[\begin{split}
		&\frac{4\pi\mathfrak{a}_0N^\kappa}{N}\sum_{\substack{ r\in P_H^c , v, w\in P_L: \\ (v,w)\in P_L^2\\\text{ or }(v+r, w-r)\in P_L^2}}^*  a^*_{v+r}a^*_{w-r}a_v a_w\\
		& = \frac{4\pi\mathfrak{a}_0N^\kappa}{N}\sum_{r\in P_H^c}\sum_{\substack{ v, w\in P_L: \\ (v,w)\in P_L^2\\\text{ or }(v+r, w-r)\in P_L^2}}^*  a^*_{v+r}a_v a^*_{w-r}a_{w} - \frac{4\pi\mathfrak{a}_0N^\kappa}{N}\sum_{\substack{ r\in P_H^c , v\in P_L: \\ (v,v+r)\in P_L^2}}^*  a^*_{v+r} a_{v+r}\\
		&\geq \frac{4\pi\mathfrak{a}_0N^\kappa}{N}\sum_{r\in P_H^c}\sum_{\substack{ v, w\in P_L: \\ (v,w)\in P_L^2\\\text{ or }(v+r, w-r)\in P_L^2}}^*  a^*_{v+r}a_v a^*_{w-r}a_{w}- CN^{3\beta+\kappa-1}\cN_+-C.
		\end{split}\]
Then, for a fixed $ r\in P_H^c$, we have that $$ \big\{ (v,w)\in \Lambda_+^*\times \Lambda_+^*:(v,w)\in P_L^2\text{ or }(v+r,w-r)\in P_L^2 \big \}=\bigcup_{j=1}^7S_j, $$ where
		\[\begin{split}
		S_1= \big\{ (v,w)\in \Lambda_+^*\times \Lambda_+^*: v\in P_L, w\in P_L, v+r\in P_L, w-r\in P_L \big \},\\
		S_2=\big\{ (v,w)\in \Lambda_+^*\times \Lambda_+^*: v\in P_L, w\in P_L, v+r\in P_L, w-r\in P_L^c \big \},\\
		S_3= \big\{ (v,w)\in \Lambda_+^*\times \Lambda_+^*: v\in P_L, w\in P_L, v+r\in P_L^c, w-r\in P_L \big \},\\
		S_4= \big\{ (v,w)\in \Lambda_+^*\times \Lambda_+^*: v\in P_L, w\in P_L, v+r\in P_L^c, w-r\in P_L^c \big \},\\
		S_5= \big\{ (v,w)\in \Lambda_+^*\times \Lambda_+^*: v\in P_L^c, w\in P_L, v+r\in P_L, w-r\in P_L \big \},\\
		S_6= \big\{ (v,w)\in \Lambda_+^*\times \Lambda_+^*: v\in P_L, w\in P_L^c, v+r\in P_L, w-r\in P_L \big \},\\
		S_7= \big\{ (v,w)\in \Lambda_+^*\times \Lambda_+^*: v\in P_L^c, w\in P_L^c, v+r\in P_L, w-r\in P_L \big \}.
		\end{split}\]
In particular, the union $  \bigcup_{j=1}^7S_j$ is a disjoint union. As a consequence, we find that 
		\[\begin{split}
		&\frac{4\pi\mathfrak{a}_0N^\kappa}{N}\sum_{r\in P_H^c}\sum_{\substack{ v, w\in P_L: \\ (v,w)\in P_L^2\\\text{ or }(v+r, w-r)\in P_L^2}}^*  a^*_{v+r}a_v a^*_{w-r}a_{w}\\
		&\hspace{1cm} = 8\pi \mathfrak{a}_0 N^\kappa \sum_{r\in P_H^c} \Big[ e^*_r c^*_{-r} + c_{-r}e_r+ \frac12d^*_{r}e^*_{-r}+ \frac12e_{-r}e_r +\frac12 c^*_{r}c^*_{-r} + \frac12 c_{-r}c_r  \Big] 	\\
		&\hspace{1.5cm} + 	8\pi \mathfrak{a}_0 N^\kappa \sum_{r\in P_H^c} \Big[ e^*_r e_r +c^*_r e_r + e^*_r c_r\Big]. 
		\end{split}\]
Inserting in \eqref{eq:proofpropMN10}, we obtain 
\begin{equation}\label{eq:proofpropMN11b}
		\begin{split}
		\cM_N&\geq 4\pi \frak{a}_0 N^{1+\kappa}+ \frac12 \cK  + 8\pi \mathfrak{a}_0N^\kappa \sum_{r\in P_H^c} \big( b^*_r + c^*_r +e^*_r  \big)\big( b_r + c_r +e_r  \big)\\
		& +4 \pi \mathfrak{a}_0N^\kappa \sum_{r\in P_H^c}\Big[ \big( b^*_r + c^*_r +e^*_r  \big)\big( b^*_{-r} + c^*_{-r} +e^*_{-r}  \big) + \text{h.c.}\Big] \\
		& - 8\pi \frak{a}_0 N^\kappa \sum_{r \in P_H^c} \left[ c_r^* c_r + b_r^* c_r + c_r^* b_r \right] 
		- \wt{C} N^{-\beta- \kappa} \cV_{N,L} + \cE'''_{\cM_N}
		\end{split}
		\end{equation}
		with 
\[\begin{split} 
e^{A}e^D \cE'''_{\cM_N} e^{-D}e^{-A}& \geq - C N^{-\beta} \cK - C N^{-\beta-\kappa}\cV_N  -C N^{ \alpha +\beta +2\kappa-1}\cK\cN_{\geq \frac12 N^{\lfloor m_0 \rfloor \beta}} \\
		&\hspace{0.5cm}  - C \sum_{j=3}^{2\lfloor m_0\rfloor -1}N^{ j\beta/2 + \beta/2+2\kappa-1}\cK \cN_{\geq \frac12 N^{j\beta/2}}- CN^{\alpha + \beta/2 +2\kappa} 
\end{split} \]
Let us now estimate the remaining terms on the last line of (\ref{eq:proofpropMN11b}). For $\xi \in \cF_+^{\leq N}$, we have 
		\begin{equation}\label{eq:KNbeta} \begin{split}
		\bigg| \,8\pi \mathfrak{a}_0N^\kappa \sum_{r\in P_H^c} \langle\xi , c^*_r c_r \xi\rangle\bigg| &\leq \frac{CN^\kappa}N \!\!\!\!\sum^*_{\substack{ r\in P_H^c, v,w\in P_L:\\ v\in P_L, r+v\in P_L^c,\\ w\in P_L, w+r\in P_L^c }} \Big ( |w||v|^{-1} \| a_{r+v}a _{w} \xi\| \Big) \Big(|v||w|^{-1}\|a_va_{w+r} \xi\| \Big) \\
		&  \leq CN^{\beta+ \kappa-1} \langle \xi, \cK_L(\cN_{\geq N^\beta}+1) \xi\rangle,
		\end{split}\end{equation}		
and 
		\begin{equation}\label{eq:KNbeta2}\begin{split}
		\bigg| \,8\pi \mathfrak{a}_0N^\kappa \sum_{r\in P_H^c} \langle\xi , (b^*_r c_r + c^*_r b_r)\xi\rangle\bigg|&\leq \frac14 \sum_{r\in P_H^c} \langle\xi , b^*_r b_r \xi\rangle +CN^{2\kappa} \sum_{r\in P_H^c} \langle\xi , c^*_r c_r \xi\rangle \\
		%&\leq \frac{CN^\kappa}N \!\!\!\!\sum^*_{\substack{ r\in P_H^c, v,w\in P_L:\\ v\in P_L, r+v\in P_L^c,\\ w\in P_L, w+r\in P_L^c }} \Big ( |w||v|^{-1} \| a_{r+v}a _{w} \xi\| \Big) \Big(|v||w|^{-1}\|a_va_{w+r} \xi\| \Big) \\
		&  \leq \frac14 \cK + C N^{\beta+ 2\kappa-1} \langle \xi, \cK_L(\cN_{\geq N^\beta}+1) \xi\rangle,
		\end{split}\end{equation}
Similarly, we can bound 
\[ \begin{split} N^{-\beta- \kappa} \langle \xi,  \cV_{N,L} \xi \rangle \leq \; &C N^{-\beta-1}  \sum_{\substack{u\in \Lambda^*, p,q \in \Lambda_+^*: \\ p+u, q+u, p,q  \in P_L}} \|  a_{p+u} a_{q} \xi \| \| a_p a_{q+u} \xi \| \\ \leq \; & C N^{-\beta-1} \sum_{\substack{u\in \Lambda^*, p,q \in \Lambda_+^*: \\ p+u, q+u, p,q  \in P_L}}  \frac{|q|^2}{|p|^2}  \|  a_{p+u} a_{q} \xi \|^2 \\ \leq \; &C N^{-1} \| \cK^{1/2} \cN_+^{1/2} \xi \|^2 \leq C \| \cK^{1/2} \xi \|^2 
\end{split} \]   
Thus, choosing the constant $\wt{C} > 0$ small enough and applying Lemma \ref{lm:KresgrowD}, Lemma \ref{lm:KresgrowA} and Lemma \ref{lm:NresgrowA} to the r.h.s. of (\ref{eq:KNbeta}) and to the second term on the r.h.s. of (\ref{eq:KNbeta2}), we conclude that
		\begin{equation}\label{eq:proofpropMN12}
		\begin{split}
		\cM_N&\geq 4\pi \frak{a}_0 N^{1+\kappa}+ \frac14 \cK  + 8\pi \mathfrak{a}_0N^\kappa \sum_{r\in P_H^c} \big( b^*_r + c^*_r +e^*_r  \big)\big( b_r + c_r +e_r  \big)\\
		& +4 \pi \mathfrak{a}_0N^\kappa \sum_{r\in P_H^c}\Big[ \big( b^*_r + c^*_r +e^*_r  \big)\big( b^*_{-r} + c^*_{-r} +e^*_{-r}  \big) + \text{h.c.}\Big]  + \cE''''_{\cM_N}
		\end{split}
		\end{equation}
where $ \cE''''_{\cM_N}$ is such that 
		\begin{equation}\label{eq:proofpropMN13}
		\begin{split} 
		e^Ae^D \cE''''_{\cM_N} e^{-A} e^{-D} \geq 
		&- C N^{-\beta} \cK - C N^{-\beta-\kappa}\cV_N  \\ &- C N^{\beta + 2\kappa - 1} \cK \cN_{\geq N^\beta} - C N^{ \alpha +\beta +2\kappa-1}\cK\cN_{\geq \frac12 N^{\lfloor m_0 \rfloor \beta}} \\ & - C \sum_{j=3}^{2\lfloor m_0\rfloor -1}N^{ j\beta/2 + \beta/2+2\kappa-1}\cK \cN_{\geq \frac12 N^{j\beta/2}} - CN^{\alpha + \beta/2 +2\kappa} 
\end{split} \end{equation}

We introduce the operators 
\[ g^*_r = b^*_r + c^*_r + e^*_r , \qquad g_r = b_r+c_r+e_r.\]
With the algebraic identity
		\[\begin{split}
		 \sum_{r\in P_H^c }\Big[ g^*_rg_r + \frac12 g^*_r g^*_{-r} +\frac12 g_{-r} g_{r}\Big] &=  \frac12 \sum_{r\in P_H^c }\big( g^*_r +g_{-r}\big) \big( g_r +g^*_{-r}\big) - \frac12 \sum_{r\in P_H^c }  [g_{r} , g^*_{r}],
		\end{split}\]
we conclude that 
\[\begin{split}
		\cM_{N} &\geq 4\pi \frak{a}_0 N^{1+\kappa}+ \frac14 \cK  - 4\pi \mathfrak{a}_0N^\kappa\sum_{r\in P_H^c} [g_{r} , g^*_{r}] + \cE''''_{\cM_N}
		\end{split}\]
Since 
		\[\begin{split}
		 [b_{r} , c^*_{r}]  = [b_r, e^*_r] = [c_r, b^*_r] = [e_r, b^*_r] = [c_r, e^*_r] = [e_r,c^*_r] = 0,
		\end{split}\]
we obtain that
\[\begin{split}
		[g_r, g^*_r] &= \frac{N-\cN_+}{N} - \frac1Na^*_ra_r +  \frac1N \sum_{\substack{v\in\Lambda_+^*: v\in P_L, \\ v+r\in P_L^c } }a^*_va_v-  \frac1N \sum_{\substack{v\in\Lambda_+^*: v\in P_L, \\ v+r\in P_L^c } }a^*_{v+r}a_{v+r}\\
		&\hspace{0.4cm} + \frac1{4N} \sum_{\substack{v\in\Lambda_+^*: v\in P_L, \\ v+r\in P_L } }a^*_va_v - \frac1{4N} \sum_{\substack{v\in\Lambda_+^*: v\in P_L, \\ v+r\in P_L } }a^*_{v+r}a_{v+r}.
		\end{split} \]
A straightforward computation then shows that 
		\[-4\pi \mathfrak{a}_0N^\kappa\sum_{p\in P_H^c} [g_{r} , g^*_{r}]\geq  - C N^{3\alpha+\kappa} (1-\cN_+/N) - C N^{3\alpha+\kappa} \cN_+/N\geq - C N^{3\alpha+\kappa}.  \]	
Thus
\[ \cM_{N} \geq 4\pi \frak{a}_0 N^{1+\kappa}+ \frac14 \cK  + \cE_{\cM_N} \]
where $\cE_{\cM_N}$ satisfies 
\[ \begin{split} 
		e^Ae^D \cE_{\cM_N} e^{-A} e^{-D} \geq 
		&- C N^{-\beta} \cK - C N^{-\beta-\kappa}\cV_N  \\ &- C N^{\beta+2\kappa-1} \cK \cN_{\geq N^\beta} - C N^{ \alpha +\beta +2\kappa-1} \cK \cN_{\geq \frac12 N^{\lfloor m_0 \rfloor \beta}} \\
		&- C \sum_{j=3}^{2\lfloor m_0\rfloor -1}N^{ j\beta/2 + \beta/2+2\kappa-1}\cK \cN_{\geq \frac12 N^{j\beta/2}}- CN^{3\alpha + \kappa} 
\end{split}\] 
This concludes the proof of Proposition \ref{prop:MN}.\qed

\appendix

\section{Analysis of $ \cG_N$ }\label{sec:GN}

The goal of this section is to prove Prop. \ref{prop:GN}. To reach this goal, we need precise information about the action of the generalized Bogoliubov transformation $e^B$, with the antisymmmetric operator $B$ defined as in (\ref{eq:genBog}), beyond the bound 
(\ref{lm:Ngrow}) for the growth of the number of excitations. 

To describe the action of $e^B$  on the generalized creation and annihilation operators $b^*_p, b_p$ introduced in (\ref{eq:bp-de}), we expand, for any $p \in \Lambda^*_+$,  
\[\begin{split} e^{-B(\eta)} \, b_p \, e^{B(\eta)} &= b_p + \int_0^1 ds \, \frac{d}{ds}  e^{-sB(\eta)} b_p e^{sB(\eta)} \\ &= b_p - \int_0^1 ds \, e^{-sB(\eta)} [B(\eta), b_p] e^{s B(\eta)} \\ &= b_p - [B(\eta),b_p] + \int_0^1 ds_1 \int_0^{s_1} ds_2 \, e^{-s_2 B(\eta)} [B(\eta), [B(\eta),b_p]] e^{s_2 B(\eta)} \end{split} \]
Iterating $m$ times, we find 
\begin{equation}\label{eq:BCH} \begin{split} 
e^{-B(\eta)} b_p e^{B(\eta)} = &\sum_{n=1}^{m-1} (-1)^n \frac{\text{ad}^{(n)}_{B(\eta)} (b_p)}{n!} \\ &+ \int_0^{1} ds_1 \int_0^{s_1} ds_2 \dots \int_0^{s_{m-1}} ds_m \, e^{-s_m B(\eta)} \text{ad}^{(m)}_{B(\eta)} (b_p) e^{s_m B(\eta)} \end{split} \end{equation}
where we recursively defined \[ \text{ad}_{B(\eta)}^{(0)} (A) = A \quad \text{and } \quad \text{ad}^{(n)}_{B(\eta)} (A) = [B(\eta), \text{ad}^{(n-1)}_{B(\eta)} (A) ]  \]
We are going to expand the nested commutators $\text{ad}_{B(\eta)}^{(n)} (b_p)$ and   
$\text{ad}_{B(\eta)}^{(n)} (b^*_p)$. To this end, we need to introduce some additional notation. 
We follow here \cite{BS,BBCS0, BBCS1,BBCS2, BBCS}. For $f_1, \dots , f_n \in \ell_2 (\Lambda^*_+)$, $\sharp = (\sharp_1, \dots , \sharp_n), \flat = (\flat_0, \dots , \flat_{n-1}) \in \{ \cdot, * \}^n$, we set 
\begin{equation}\label{eq:Pi2}
\begin{split}  
\Pi^{(2)}_{\sharp, \flat} &(f_1, \dots , f_n) \\ &= \sum_{p_1, \dots , p_n \in \Lambda^*}  b^{\flat_0}_{\alpha_0 p_1} a_{\beta_1 p_1}^{\sharp_1} a_{\alpha_1 p_2}^{\flat_1} a_{\beta_2 p_2}^{\sharp_2} a_{\alpha_2 p_3}^{\flat_2} \dots  a_{\beta_{n-1} p_{n-1}}^{\sharp_{n-1}} a_{\alpha_{n-1} p_n}^{\flat_{n-1}} b^{\sharp_n}_{\beta_n p_n} \, \prod_{\ell=1}^n f_\ell (p_\ell)  \end{split} \end{equation}
where, for $\ell=0,1, \dots , n$, we define $\alpha_\ell = 1$ if $\flat_\ell = *$, $\alpha_\ell =    -1$ if $\flat_\ell = \cdot$, $\beta_\ell = 1$ if $\sharp_\ell = \cdot$ and $\beta_\ell = -1$ if $\sharp_\ell = *$. In (\ref{eq:Pi2}), we require that, for every $j=1,\dots, n-1$, we have either $\sharp_j = \cdot$ and $\flat_j = *$ or $\sharp_j = *$ and $\flat_j = \cdot$ (so that the product $a_{\beta_\ell p_\ell}^{\sharp_\ell} a_{\alpha_\ell p_{\ell+1}}^{\flat_\ell}$ always preserves {} the number of particles, for all $\ell =1, \dots , n-1$). With this assumption, we find that the operator $\Pi^{(2)}_{\sharp,\flat} (f_1, \dots , f_n)$ maps $\cF^{\leq N}_+$ into itself. If, for some $\ell=1, \dots , n$, $\flat_{\ell-1} = \cdot$ and $\sharp_\ell = *$ (i.e. if the product $a_{\alpha_{\ell-1} p_\ell}^{\flat_{\ell-1}} a_{\beta_\ell p_\ell}^{\sharp_\ell}$ for $\ell=2,\dots , n$, or the product $b_{\alpha_0 p_1}^{\flat_0} a_{\beta_1 p_1}^{\sharp_1}$ for $\ell=1$, is not normally ordered) we require additionally that $f_\ell  \in \ell^1 (\Lambda^*_+)$. In position space, the same operator can be written as 
\begin{equation}\label{eq:Pi2-pos} \Pi^{(2)}_{\sharp, \flat} (f_1, \dots , f_n) = \int   \check{b}^{\flat_0}_{x_1} \check{a}_{y_1}^{\sharp_1} \check{a}_{x_2}^{\flat_1} \check{a}_{y_2}^{\sharp_2} \check{a}_{x_3}^{\flat_2} \dots  \check{a}_{y_{n-1}}^{\sharp_{n-1}} \check{a}_{x_n}^{\flat_{n-1}} \check{b}^{\sharp_n}_{y_n} \, \prod_{\ell=1}^n \check{f}_\ell (x_\ell - y_\ell) \, dx_\ell dy_\ell \end{equation}
An operator of the form (\ref{eq:Pi2}), (\ref{eq:Pi2-pos}) with all the properties listed above, will be called a $\Pi^{(2)}$-operator of order $N\in\NN$.

For $g, f_1, \dots , f_n \in \ell_2 (\Lambda^*_+)$, $\sharp = (\sharp_1, \dots , \sharp_n)\in \{ \cdot, * \}^n$, $\flat = (\flat_0, \dots , \flat_{n}) \in \{ \cdot, * \}^{n+1}$, we also define the operator 
\begin{equation}\label{eq:Pi1}
\begin{split} \Pi^{(1)}_{\sharp,\flat} &(f_1, \dots , f_n;g) \\ &= \sum_{p_1, \dots , p_n \in \Lambda^*}  b^{\flat_0}_{\alpha_0, p_1} a_{\beta_1 p_1}^{\sharp_1} a_{\alpha_1 p_2}^{\flat_1} a_{\beta_2 p_2}^{\sharp_2} a_{\alpha_2 p_3}^{\flat_2} \dots a_{\beta_{n-1} p_{n-1}}^{\sharp_{n-1}} a_{\alpha_{n-1} p_n}^{\flat_{n-1}} a^{\sharp_n}_{\beta_n p_n} a^{\flat n} (g) \, \prod_{\ell=1}^n f_\ell (p_\ell) \end{split} \end{equation}
where $\alpha_\ell$ and $\beta_\ell$ are defined as above. Also here, we impose the condition that, for all $\ell = 1, \dots , n$, either $\sharp_\ell = \cdot$ and $\flat_\ell = *$ or $\sharp_\ell = *$ and $\flat_\ell = \cdot$. This implies that $\Pi^{(1)}_{\sharp,\flat} (f_1, \dots , f_n;g)$ maps $\cF^{\leq N}_+$ back into $\cF_+^{\leq N}$. Additionally, we assume that $f_\ell \in \ell^1 (\Lambda^*_+)$ if $\flat_{\ell-1} = \cdot$ and $\sharp_\ell = *$ for some $\ell = 1,\dots , n$ (i.e. if the pair $a_{\alpha_{\ell-1} p_\ell}^{\flat_{\ell-1}} a^{\sharp_\ell}_{\beta_\ell p_\ell}$ is not normally ordered). In position space, the same operator can be written as
\begin{equation}\label{eq:Pi1-pos} \Pi^{(1)}_{\sharp,\flat} (f_1, \dots ,f_n;g) = \int \check{b}^{\flat_0}_{x_1} \check{a}_{y_1}^{\sharp_1} \check{a}_{x_2}^{\flat_1} \check{a}_{y_2}^{\sharp_2} \check{a}_{x_3}^{\flat_2} \dots  \check{a}_{y_{n-1}}^{\sharp_{n-1}} \check{a}_{x_n}^{\flat_{n-1}} \check{a}^{\sharp_n}_{y_n} \check{a}^{\flat n} (g) \, \prod_{\ell=1}^n \check{f}_\ell (x_\ell - y_\ell) \, dx_\ell dy_\ell \end{equation}
An operator of the form (\ref{eq:Pi1}), (\ref{eq:Pi1-pos}) will be called a $\Pi^{(1)}$-operator of order $N\in\NN$. Operators of the form $b(f)$, $b^* (f)$, for a $f \in \ell^2 (\Lambda^*_+)$, will be called $\Pi^{(1)}$-operators of order zero. 

The next lemma gives a detailed analysis of the nested commutators $\text{ad}^{(n)}_{B(\eta)} (b_p)$ and $\text{ad}^{(n)}_{B(\eta)} (b^*_p)$ for $n \in \bN$; the proof can be found in \cite[Lemma 2.5]{BBCS1}(it is a translation to momentum space of \cite[Lemma 3.2]{BS}). 
\begin{lemma}\label{lm:indu}
Let $B$ be defined as in (\ref{eq:genBog}), with coefficients $\eta_p$ as in (\ref{eq:defetaH}) and with $\alpha > 2\kappa$ (so that $\| \eta \| \to 0$, as $N \to \infty$). Let $n \in \bN$ and $p \in \Lambda^*$. Then the nested commutator $\text{ad}^{(n)}_{B} (b_p)$ can be written as the sum of exactly $2^n n!$ terms, with the following properties. 
\begin{itemize}
\item[i)] Possibly up to a sign, each term has the form
\begin{equation}\label{eq:Lambdas} \Lambda_1 \Lambda_2 \dots \Lambda_i \, N^{-k} \Pi^{(1)}_{\sharp,\flat} (\eta^{j_1}, \dots , \eta^{j_k} ; \eta^{s}_p \ph_{\alpha p}) 
\end{equation}
for some $i,k,s \in \bN$, $j_1, \dots ,j_k \in \bN \backslash \{ 0 \}$, $\sharp \in \{ \cdot, * \}^k$, $ \flat \in \{ \cdot, * \}^{k+1}$ and $\alpha \in \{ \pm 1 \}$ chosen so that $\alpha = 1$ if $\flat_k = \cdot$ and $\alpha = -1$ if $\flat_k = *$ (recall here that $\ph_p (x) = e^{-ip \cdot x}$). In (\ref{eq:Lambdas}), each operator $\Lambda_w : \cF^{\leq N} \to \cF^{\leq N}$, $w=1, \dots , i$, is either a factor $(N-\cN_+ )/N$, a factor $(N-(\cN_+ -1))/N$ or an operator of the form
\begin{equation}\label{eq:Pi2-ind} N^{-h} \Pi^{(2)}_{\sharp',\flat'} (\eta^{z_1}, \eta^{z_2},\dots , \eta^{z_h}) \end{equation}
for some $h, z_1, \dots , z_h \in \bN \backslash \{ 0 \}$, $\sharp,\flat  \in \{ \cdot , *\}^h$. 
\item[ii)] If a term of the form (\ref{eq:Lambdas}) contains $m \in \bN$ factors $(N-\cN_+ )/N$ or $(N-(\cN_+ -1))/N$ and $j \in \bN$ factors of the form (\ref{eq:Pi2-ind}) with $\Pi^{(2)}$-operators of order $h_1, \dots , h_j \in \bN \backslash \{ 0 \}$, then 
we have
\begin{equation*}%\label{eq:totalb}
 m + (h_1 + 1)+ \dots + (h_j+1) + (k+1) = n+1 \end{equation*}
\item[iii)] If a term of the form (\ref{eq:Lambdas}) contains (considering all $\Lambda$-operators and the $\Pi^{(1)}$-operator) the arguments $\eta^{i_1}, \dots , \eta^{i_m}$ and the factor $\eta^{s}_p$ for some $m, s \in \bN$, and $i_1, \dots , i_m \in \bN \backslash \{ 0 \}$, then \[ i_1 + \dots + i_m + s = n .\]
\item[iv)] There is exactly one term having of the form (\ref{eq:Lambdas}) with $k=0$ and such that all $\Lambda$-operators are factors of $(N-\cN_+ )/N$ or of $(N+1-\cN_+ )/N$. It is given by 
\begin{equation*}\label{eq:iv1} 
\left(\frac{N-\cN_+ }{N} \right)^{n/2} \left(\frac{N+1-\cN_+ }{N} \right)^{n/2} \eta^{n}_p b_p 
\end{equation*}
if $N\in\NN$ is even, and by 
\begin{equation*} \label{eq:iv2} 
- \left(\frac{N-\cN_+ }{N} \right)^{(n+1)/2} \left(\frac{N+1-\cN_+ }{N} \right)^{(n-1)/2} \eta^{n}_p b^*_{-p}  \end{equation*}
if $N\in\NN$ is odd.
\item[v)] If the $\Pi^{(1)}$-operator in (\ref{eq:Lambdas}) is of order $k \in \bN \backslash \{ 0 \}$, it has either the form  
\[ \sum_{p_1, \dots , p_k}  b^{\flat_0}_{\alpha_0 p_1} \prod_{i=1}^{k-1} a^{\sharp_i}_{\beta_i p_{i}} a^{\flat_i}_{\alpha_i p_{i+1}}  a^*_{-p_k} \eta^{2r}_p  a_p \prod_{i=1}^k \eta^{j_i}_{p_i}  \]
or the form 
\[\sum_{p_1, \dots , p_k} b^{\flat_0}_{\alpha_0 p_1} \prod_{i=1}^{k-1} a^{\sharp_i}_{\beta_i p_{i}} a^{\flat_i}_{\alpha_i p_{i+1}}  a_{p_k} \eta^{2r+1}_p a^*_p \prod_{i=1}^k \eta^{j_i}_{p_i}  \]
for some $r \in \bN$, $j_1, \dots , j_k \in \bN \backslash \{ 0 \}$. If it is of order $k=0$, then it is either given by $\eta^{2r}_p b_p$ or by $\eta^{2r+1}_p b_{-p}^*$, for some $r \in \bN$. 
\item[vi)] For every non-normally ordered term of the form 
\[ \begin{split} &\sum_{q \in \Lambda^*} \eta^{i}_q a_q a_q^* , \quad \sum_{q \in \Lambda^*} \, \eta^{i}_q b_q a_q^* \\  &\sum_{q \in \Lambda^*} \, \eta^{i}_q a_q b_q^*, \quad \text{or } \quad \sum_{q \in \Lambda^*} \, \eta^{i}_q b_q b_q^*  \end{split} \]
appearing either in the $\Lambda$-operators or in the $\Pi^{(1)}$-operator in (\ref{eq:Lambdas}), we have $i \geq 2$.
\end{itemize}
\end{lemma}
With Lemma \ref{lm:indu}, it follows from (\ref{eq:BCH}) that, if $\| \eta \|$ is sufficiently small, the series
\begin{equation}\label{eq:conv-serie}
\begin{split} e^{-B(\eta)} b_p e^{B (\eta)} &= \sum_{n=0}^\infty \frac{(-1)^n}{n!} \text{ad}_{B(\eta)}^{(n)} (b_p) \\
e^{-B(\eta)} b^*_p e^{B (\eta)} &= \sum_{n=0}^\infty \frac{(-1)^n}{n!} \text{ad}_{B(\eta)}^{(n)} (b^*_p) \end{split} \end{equation}
converge absolutely (the proof is a translation to momentum space of \cite[Lemma 3.3]{BS}).

As explained after (\ref{eq:bp-de}), the generalized creation and annihilation operators $b^*_p, b_p$ are close to the standard creation and annihilation operators on states with few excitations, in the sense that with $\cN_+ \ll N$. In particular, on such states one can expect the action of the generalized Bogoliubov transformation $e^B$ to be close to the action of a standard Bogoliubov transformation. This can be made more precise through the remainder operators 
\begin{equation} \label{eq:defdq}
d_q =\sum_{m\geq 0}\frac{1}{m!} \Big [\text{ad}_{-B(\eta)}^{(m)}(b_q) - \eta_q^m b_{\alpha_m q}^{\sharp_m }  \Big],\hspace{0.4cm} d^*_q =\sum_{m\geq 0}\frac{1}{m!} \Big [\text{ad}_{-B(\eta)}^{(m)}(b^*_q) - \eta_q^m b_{\alpha_m q}^{\sharp_{m+1}}  \Big]\end{equation}
where $q \in \L^*_+$, $ (\sharp_m, \alpha_m) = (\cdot, +1)$ if $m$ is even and $(\sharp_m, \alpha_m) = (*, -1)$ if $m$ is odd. It follows from (\ref{eq:conv-serie}) that 
\begin{equation}\label{eq:ebe} 
e^{-B(\eta)} b_q e^{B(\eta)} = \gamma_q  b_q +\s_q b^*_{-q} + d_q, \hspace{1cm} e^{-B(\eta)} b^*_q e^{B(\eta)} = \g_q b^*_q +\s_q b_{-q} + d^*_q  \end{equation} 
where we set $\g_q = \cosh (\eta_q)$ and $\s_q = \sinh (\eta_q)$. Given $x \in \Lambda$, it is also useful to define the operator valued distributions $\check{d}_x, \check{d}^*_x$ through
\begin{equation}\label{eq:ebex} e^{-B(\eta)} \check{b}_x e^{B(\eta)} = b ( \check{\g}_x)  +  b^* (\check{\s}_x) + \check{d}_x, \qquad 
e^{-B(\eta)} \check{b}^*_x e^{B(\eta)} = b^* ( \check{\gamma}_x)  +  b (\check{\s}_x) + \check{d}^*_x
\end{equation}
where $\check{\gamma}_x (y) = \sum_{q \in \Lambda^*} \cosh (\eta_q) e^{iq \cdot (x-y)}$ and $\check{\s}_x (y) = \sum_{q \in \Lambda^*} \sinh (\eta_q) e^{iq \cdot (x-y)}$.  

The next lemma confirms the intuition that remainder operators are small on states with $\cN_+ \ll N$, and provides estimates that will be crucial for our analysis. Its proof can be found in \cite{BBCS}.
\begin{lemma} \label{lm:dp} 
Let $B$ be defined as in (\ref{eq:genBog}), with coefficients $\eta_p$ as in (\ref{eq:defetaH}) and with $\alpha > 2\kappa$. Let $n \in \bN$, $p \in \Lambda^*$ and let $d_p$ be defined as in (\ref{eq:defdq}). There exists $C > 0$ such that  
\begin{equation}\label{eq:d-bds}
\begin{split} 
\| (\cN_+ + 1)^{n/2} d_p \xi \| &\leq \frac{C}{N} \left[ |\eta_p| \| (\cN_+ + 1)^{(n+3)/2} \xi \| + \| \eta \| \| b_p (\cN_+ + 1)^{(n+2)/2} \xi \| \right], \\ 
\| (\cN_+ + 1)^{n/2} d_p^* \xi \| &\leq \frac{C}{N} \, \| \eta \| \,\| (\cN_+ +1)^{(n+3)/2} \xi \| \end{split}  \end{equation}
for all $\xi \in \cF_+^{\leq N}$ and $N$ large enough. With $\bar{\bar{d}}_p = d_p + N^{-1} \sum_{q \in \L_+^*} \eta_q b_q^* a_{-q}^* a_p$, we also have, for $p \not \in \text{supp } \eta$, the improved bound 
\begin{equation}\label{eq:off} \| (\cN_+ + 1)^{n/2} \bar{\bar{d}}_p \xi \| \leq \frac{C}{N} \| \eta \|^2 \| a_p (\cN_+ + 1)^{(n+2)/2} \xi \| \end{equation}
In position space, with $\check{d}_x$ defined as in (\ref{eq:ebex}), we find   
 \begin{equation}\label{eq:dxy-bds} 
 \| (\cN_+ + 1)^{n/2} \check{d}_x \xi \| \leq  \frac{C }{N}\, \| \eta \| \Big[ \,\| (\cN_+ + 1)^{(n+3)/2} \xi \| +  \| b_x (\cN_+ + 1) ^{(n+2)/2}\xi \| \Big] 
\end{equation}
Furthermore, letting $\check{\bar{d}}_x = \check{d}_x  + (\cN_+ / N) b^*(\check{\eta}_x)$, we find 
\be \begin{split} \label{eq:splitdbd}
\| (\cN_+ &+ 1)^{n/2} \check{a}_y \check{\bar{d}}_x \xi \| \\ &\leq \frac{C}{N} \, \Big[ \, \|\eta \|^2  \| (\cN_+ + 1)^{(n+2)/2} \xi \|  + \| \eta \| |\check{\eta} (x-y)|  \| (\cN +1)^{(n+2)/2}  \xi \| \\
& \hspace{1cm} + \| \eta \| \| \check{a}_x (\cN_++1)^{(n+1)/2} \xi \| +  \|\eta \|^2 \|\check{a}_y (\cN_+ + 1)^{(n+3)/2} \xi \|\\
& \hspace{1cm}  + \| \eta \| \| \check{a}_x \check{a}_y (\cN +1)^{(n+2)/2}  \xi \|   \, \Big]
\end{split}\ee
and, finally, 
\begin{equation}\label{eq:ddxy}
\begin{split} 
\| (\cN_+ &+ 1)^{n/2} \check{d}_x \check{d}_y \xi \|  \\ &\leq \frac C {N^2} \Big[ \; \|\eta\|^2  \| (\cN_++ 1)^{(n+6)/2} \xi \| + \| \eta \| |\check{\eta} (x-y)|  \| (\cN_+ + 1)^{(n+4)/2}  \xi \| \\ 
&\hspace{1cm} + \|\eta \|^2 \| {a}_x (\cN_+ + 1)^{(n+5)/2} \xi \|   + \| \eta \|^2 
\|{a}_y (\cN_+ + 1)^{(n+5)/2} \xi \| \\ &\hspace{1cm}  
+ \| \eta \|^2\, \|{a}_x {a}_y (\cN_+ +  1)^{(n+4)/2} \xi \| \; \Big] 
\end{split} \end{equation}
for all $\xi \in \cF^{\leq n}_+$. 
\end{lemma}
  
We will also need to control commutators of the remainder operators $d_p, d_p^*$ with restricted number of particles operators $ \cN_{\leq cN^\gamma}$, where $c\geq0$ and $ \gamma\geq 0$ (recall here the definitions \eqref{eq:defcN>}). 

\begin{lemma}\label{lm:commdpNres}
Let $B$ be defined as in (\ref{eq:genBog}), with coefficients $\eta_p$ as in (\ref{eq:defetaH}) and with $\alpha > 2\kappa$. Let $n \in \bN$, $p \in \Lambda^*$ and let $d_p$ be defined as in (\ref{eq:defdq})
Moreover, given $ c\geq 0$ and $\gamma\geq 0$, denote by $ \chi\in \ell^2(\Lambda^*_+)$ the characteristic function of the set $ \{p\in\Lambda^*_+: |p|\leq cN^\gamma\}$. Then there exists $C>0$ s.t. 
\begin{equation} \label{eq:ngdp}
		\begin{split}
		&\big\| (\mathcal{N}_+ +1)^{n/2}[\mathcal{N}_{\leq cN^{\gamma}},d_p] \xi\big\| \\
		& \hspace{1cm}\leq \frac{C}N\left[ \, |\eta_p|   \| (\cN_+ + 1)^{(n+3)/2} \xi \|  + \big[|\chi_p|\|\eta\|  + \|\chi\eta\| \big]\| a_p (\cN_+ + 1)^{(n+2)/2} \xi \| \right],\\
		&\big\| (\mathcal{N}_+ +1)^{n/2}[\mathcal{N}_{\leq cN^{\gamma}},d^*_p] \xi\big\| \\
		&\hspace{1cm} \leq\frac{C}N\big[ \, |\eta_p|  + |\chi_p|\|\eta\|  + \|\chi\eta\| \big] \| (\cN_+ + 1)^{(n+3)/2} \xi \| 
		\end{split}
		\end{equation}
for all $p \in \Lambda_+^*, \xi \in \mathcal{F}_+^{\leq N}$. With $\bar{\bar{d}}_p = d_p + N^{-1} \sum_{q \in \Lambda_+^*} \eta_q b_q^* a_{-q}^* a_p$, we also have, for $p \notin \text{supp } \eta$, the improved bound
\begin{equation}\label{eq:ngdpbar}
\big\| (\mathcal{N}_+ +1)^{n/2}[\mathcal{N}_{\leq cN^{\gamma}}, \bar{\bar{d}}_p] \xi\big\| \leq\frac{C}{N}\big [  |\chi_p| \|\eta\|^2+\|\eta \chi\|  \|\eta\|  \big]  \big\| a_p(\mathcal{N}_++1)^{(n+2)/2} \xi \big\|.
\end{equation}
In position space, with the operators $\check{d}_x$ defined as in \eqref{eq:ebex}, let $ \check{\chi}_x\in L^2(\Lambda)$ be defined by $ \check{\chi}_x (y) = \chi(y-x)$ (s.t. $ \check{\chi}_x$ has Fourier coefficients $ \chi_p e^{-ipx}$). Then
		\begin{equation}
		\begin{split}\label{eq:ngdxy-bds} 
 		&\| (\cN_+ + 1)^{n/2} [\mathcal{N}_{\leq cN^{\gamma}},\check{d}_x ]\xi \|\\
		&\hspace{1cm}\leq   \frac{C }{N}\, \| \chi\eta \|\Big[ \,  \| (\cN_+ + 1)^{(n+3)/2} \xi \| +  \| \check{b}_x (\cN_+ + 1) ^{(n+2)/2}\xi \|\Big] \\
		&\hspace{1.5cm}+\frac CN\| \eta \|  \| b(\check{\chi}_x) (\cN_+ + 1) ^{(n+2)/2}\xi \|
		\end{split}
		\end{equation}
for all $\xi \in \cF^{\leq n}_+$. Furthermore, setting $\check{\bar{d}}_x = \check{d}_x  + (\cN_+ / N) b^*(\check{\eta}_x)$, we obtain 
		\be \begin{split} \label{eq:ngsplitdbd}
		\| (\cN_+ &+ 1)^{n/2}[ \cN_{\leq c N^\gamma},  \check{b}_y \check{\bar{d}}_x] \xi \| \\ 
		&\leq \frac{C}{N} \, \Big[ \, \|\chi \eta \|\|\eta\|  \| (\cN_+ + 1)^{(n+2)/2} \xi \| + \| \chi \eta \| \| \check{a}_x (\cN_++1)^{(n+1)/2} \xi \|   \\
		&\hspace{1cm} + \| \eta \| \| a(\check{\chi}_x) (\cN_++1)^{(n+1)/2} \xi \| +  \|\eta \|^2 \| a(\check{\chi}_y) (\cN_+ + 1)^{(n+3)/2} \xi \| \\
		& \hspace{1cm}  + \|\chi\eta\| \|\eta \| \|\check{a}_y (\cN_+ + 1)^{(n+3)/2} \xi \|\, \Big]\\
		 &\hspace{0.5cm} +\frac CN\Big[\,  \| \chi\eta \| |\check{\eta} (x-y)| + \| \eta \| |(\check{\chi}\ast \check{\eta}) (x-y)| \,\Big] \| (\cN +1)^{(n+2)/2}  \xi \|\\
		&\hspace{0.5cm} +\frac CN\Big[\, \| \eta \| \| a(\check{\chi}_x)\check{a}_y  (\cN +1)^{(n+2)/2}  \xi \| + \| \eta \| \|a(\check{\chi}_y) \check{a}_x  (\cN +1)^{(n+2)/2}  \xi \|  \\
		& \hspace{1.5cm}  + \| \chi\eta \| \| \check{a}_x \check{a}_y (\cN +1)^{(n+2)/2}  \xi \|  \, \Big] 
		\end{split}\ee
as well as
		\begin{equation}\label{eq:ngddxy}
		\begin{split} 
		\| (\cN_+ &+ 1)^{n/2} [\cN_{\leq c N^\gamma}, \check{d}_x \check{d}_y] \xi \|  \\ 
		&\leq \frac C {N^2} \|\chi \eta \| \|\eta\|\Big[ \;   \| (\cN_++ 1)^{(n+6)/2} \xi \| +  \| \check{a}_x (\cN_+ + 1)^{(n+5)/2} \xi \|   \\ 
		&\hspace{2.5cm}  + \|\check{a}_y (\cN_+ + 1)^{(n+5)/2} \xi \|+\|\check{a}_x \check{a}_y (\cN_+ +  1)^{(n+4)/2} \xi \|\,\Big] \\ 
		%&\hspace{0.5cm}+ \frac C {N^2} \|\chi \eta \|\Big[ \; \|\eta\|  \| (\cN_++ 1)^{(n+6)/2} \xi \| +  \|\eta\|  \| \check{a}_y(\cN_+ + 1)^{(n+5)/2}  \xi \|  \\ 
		%&\hspace{2.5cm} + \|\eta\|^2 \| (\cN_++ 1)^{(n+4)/2} \xi \| +  \|\eta\|  \| \check{a}_y(\cN_+ + 1)^{(n+3)/2}  \xi \|  \Big]\\
		&\hspace{0.5cm} +\frac C{N^2} \|\eta\|^2\Big[\,   \| a(\check{\chi}_x)  (\cN_++ 1)^{(n+5)/2} \xi \|  + \| a(\check{\chi}_y) (\cN_+ + 1) ^{(n+5)/2}\xi \|\,\Big] \\
		%&\hspace{0.5cm} +  \frac{C }{N^2}\, \|\eta\|^2 \| \chi\eta \|\Big[ \,  \| (\cN_+ + 1)^{(n+6)/2} \xi \| +  \| \check{a}_y (\cN_+ + 1) ^{(n+5)/2}\xi \|\Big] \\
		%&\hspace{2.5cm}+ \| a(\check{\chi}_y) (\cN_+ + 1) ^{(n+5)/2}\xi \|\,\Big]\\%
		%&\hspace{0.5cm} +\frac{C}{N^2} \, \Big[ \,  \| \eta \|^2 \| a(\check{\chi}_y) (\cN_++1)^{(n+3)/2} \xi \| \\
		%& \hspace{1.5cm}  + \|\chi\eta\| \|\eta \| \|\check{a}_x (\cN_+ + 1)^{(n+3)/2} \xi \|\, \Big]\\
		 &\hspace{0.5cm} +\frac C{N^2} \Big[\,  \|\chi\eta\| |\check{\eta} (x-y)|+ \| \eta \| |(\check{\chi}\ast \check{\eta}) (x-y)|  \,\Big] \| (\cN +1)^{(n+4)/2}  \xi \|\\
		&\hspace{0.5cm} +\frac C{N^2}\| \eta \|^2\Big[\,  \| a(\check{\chi}_y)\check{a}_x  (\cN +1)^{(n+4)/2}  \xi \| +\| a(\check{\chi}_x) \check{a}_y (\cN_++ 1)^{(n+4)/2} \xi \|\,\Big].
		\end{split} \end{equation} 
for all $\xi \in \cF^{\leq n}_+$.

\end{lemma}
\begin{proof}
For simplicity, we focus on the case $n=0$; the cases where $0\neq n\in \mathbb{Z}$ can be treated similarly, using that powers of $\cN_+ $ can be commuted easily with $ d_p ,d^*_p $ and $\check{d}_x, \check{d}_y $. 

Let us start with the first bound in \eqref{eq:ngdp}. By \eqref{eq:defdq}, linearity of the commutator with $\mathcal{N}_{\leq cN^{\gamma}}$ and by the triangle inequality, it is enough to estimate the r.h.s. of 
		\begin{equation}\label{eq:commNlow0} \| [\mathcal{N}_{\leq cN^{\gamma}},d_p] \xi\big\| \leq \sum_{m\geq 0}\frac1{m!} \Big\| \big[ \mathcal{N}_{\leq cN^{\gamma}}, \operatorname{ad}_{-B(\eta)}^{(m)}(b_p) - \eta_p^mb_{\alpha_m p }^{\sharp m} \big] \xi\Big\|.\end{equation}
Using Lemma \ref{lm:indu} and the fact that $ \cN_{\leq cN^\gamma}$ trivially commutes with the number of particles operator $\cN_+$, we can bound $\| [ \mathcal{N}_{\leq cN^{\gamma}}, \operatorname{ad}_{-B(\eta)}^{(m)}(b_q) - \eta_q^mb_{\alpha_m q }^{\sharp m} ] \xi\|$ by the sum of 
		\begin{equation}\label{eq:N-term} 
\left\| \left[ \left( \frac{N- \cN_+}{N} \right)^{\frac{m+ (1-\alpha_m)/2}{2}} \left( \frac{N+1-\cN_+}{N} \right)^{\frac{m-(1-\alpha_m)/2}{2}} - 1 \right] \eta_p^m \big[\cN_{\leq cN^\gamma}, b^{\sharp_m}_{\alpha_m p}\big] \xi \right\| \end{equation}
and exactly $2^m m! - 1$ terms of the form
\begin{equation}\label{eq:L-term} \left\| \big[ \cN_{\geq cN^\gamma} , \Lambda_1 \dots \Lambda_{i_1} N^{-k_1} \Pi^{(1)}_{\sharp,\flat} (\eta^{j_1} , \dots , \eta^{j_{k_1}} ; \eta^{\ell_1}_p \ph_{\alpha_{\ell_1} p})\big]\xi \right\| \end{equation}
where $i_1, k_1, \ell_1 \in \bN$, $j_1, \dots , j_{k_1} \in \bN \backslash \{ 0 \}$ and where each $\Lambda_r$-operator is either a factor $(N-\cN_+ )/N$, a factor $(N+1-\cN_+ )/N$ or a $\Pi^{(2)}$-operator of the form 
\begin{equation}\label{eq:Pi2-ex}
N^{-h} \Pi^{(2)}_{\underline{\sharp}, \underline{\flat}} (\eta^{z_1} , \dots, \eta^{z_h}) 
\end{equation}
with $h, z_1, \dots , z_h \in \bN \backslash \{ 0 \}$. Since we are considering the term (\ref{eq:N-term}) separately, each term of the form (\ref{eq:L-term}) must have either $k_1 > 0$ or it must contain at least one $\Lambda$-operator having the form (\ref{eq:Pi2-ex}). For $m=0$, (\ref{eq:N-term}) vanishes and for $m > 0$, it follows from 
		\[ \big[\cN_{\leq cN^\gamma}, b^{\sharp_m}_{\alpha_m p}\big]  = F(\sharp_m) \chi_p b^{\sharp_m}_{\alpha_m p}, \]
where set $ F(\sharp) = 1 $ if $ \sharp = *$ and $ F(\sharp)=-1$ if $ \sharp = \cdot$, that
\[ \begin{split} &\left\| \left[ \left( \frac{N- \cN_+}{N} \right)^{\frac{m+ (1-\alpha_m)/2}{2}} \left( \frac{N+1-\cN_+}{N} \right)^{\frac{m-(1-\alpha_m)/2}{2}} - 1 \right] \eta_p^m \big[\cN_{\leq cN^\gamma}, b^{\sharp_m}_{\alpha_m p}\big] \xi \right\|\\ & \hspace{8cm} \leq C^m |\eta_p\chi_p|^{m} N^{-1} \| (\cN_+ + 1)^{3/2} \xi \|. \end{split} \]
Hence, let's focus on terms of the form \eqref{eq:L-term} and let's write
		\begin{equation}\label{eq:commNlow1}\begin{split}
		&[\mathcal{N}_{\leq cN^{\gamma}},\Lambda_1\ldots \Lambda_{i_1} N^{-k_1} \Pi^{(1)}_{\sharp,\flat} (\eta^{j_1} , \dots , \eta^{j_{k_1}} ; \eta^{\ell_1}_p \ph_{\alpha_{\ell_1} p})] \\
		&\hspace{0.5cm}= \sum_{u=1}^{i_1} \Lambda_1\ldots\Lambda_{u-1} [\mathcal{N}_{\leq cN^{\gamma}},\Lambda_u] \Lambda_{u+1}\ldots \Lambda_{i_1}N^{-k_1} \Pi^{(1)}_{\sharp,\flat} (\eta^{j_1} , \dots , \eta^{j_{k_1}} ; \eta^{\ell_1}_p \ph_{\alpha_{\ell_1} p}) \\
		&\hspace{1cm}+ \Lambda_1\ldots \Lambda_{i_1} [\mathcal{N}_{\leq cN^{\gamma}}, N^{-k_1} \Pi^{(1)}_{\sharp,\flat} (\eta^{j_1} , \dots , \eta^{j_{k_1}} ; \eta^{\ell_1}_p \ph_{\alpha_{\ell_1} p})].
		\end{split}\end{equation}
It is clear that $ [\mathcal{N}_{\leq cN^{\gamma}},\Lambda_u] = 0 $ if $ \Lambda_u$ is of the form $ (N-\cN_+)/N$ or $(N-\cN_+-1)/N$. On the other hand, if $ \Lambda_u =N^{-h} \Pi^{(2)}_{\sharp',\flat'}(\eta^{z_1},\ldots,\eta^{z_h}) $ is of the form \eqref{eq:Pi2-ind} with
		\[\begin{split} &N^{-h} \Pi^{(2)}_{\sharp',\flat'}(\eta^{z_1},\ldots,\eta^{z_h}) \\
		&\hspace{0.5cm}= N^{-h}\!\!\!\!\!\sum_{p_1, \dots , p_h \in \Lambda^*}  b^{\flat'_0}_{\alpha_0 p_1} a_{\beta_1 p_1}^{\sharp'_1} a_{\alpha_1 p_2}^{\flat'_1} a_{\beta_2 p_2}^{\sharp'_2} a_{\alpha_2 p_3}^{\flat'_2} \dots  a_{\beta_{h-1} p_{h-1}}^{\sharp'_{h-1}} a_{\alpha_{h-1} p_h}^{\flat'_{h-1}} b^{\sharp'_h}_{\beta_h p_h} \, \prod_{\ell=1}^h \eta^{z_l}_{p_l},
		\end{split}\]
we use the identity
		\[%\label{eq:N2com} 
		[\mathcal{N}_{\leq cN^{\gamma}}, a_{\alpha p}^{ \flat} a_{\beta p}^{\sharp}] = (F(\flat) + F(\sharp)) \chi_p a_{\alpha p}^{ \flat} a_{\beta p}^{\sharp}\]
with which we obtain that
		\begin{equation}\label{eq:commNlow2}[\mathcal{N}_{\leq cN^{\gamma}}, N^{-h} \Pi^{(2)}_{\sharp',\flat'}(\eta^{z_1},\ldots,\eta^{z_h})] = \sum_{t=1}^h N^{-h} \Pi^{(2)}_{\sharp',\flat'}(\eta^{z_1},\ldots, (F(\flat'_{t-1})+ F(\sharp'_t)) \chi \eta^{z_t},\ldots,  \eta^{z_h}).\end{equation}
Similarly, if $N^{-k_1} \Pi^{(1)}_{\sharp,\flat} (\eta^{j_1} , \dots , \eta^{j_{k_1}} ; \eta^{\ell_1}_p \ph_{\alpha_{\ell_1} p}) $ is of the form
		\[ \begin{split} 
		&N^{-k_1} \Pi^{(1)}_{\sharp,\flat} (\eta^{j_1} , \dots , \eta^{j_{k_1}} ; \eta^{\ell_1}_p \ph_{\alpha_{\ell_1} p})\\
		 & \hspace{0.5cm}= N^{-k_1}\!\!\!\!\!\sum_{p_1, \dots , p_{k_1} \in \Lambda^*}  b^{\flat_0}_{\alpha_0, p_1} a_{\beta_1 p_1}^{\sharp_1} a_{\alpha_1 p_2}^{\flat_1}   \dots a_{\beta_{k_1-1} p_{k_1-1}}^{\sharp_{k_1-1}} a_{\alpha_{k_1-1} p_{k_1}}^{\flat_{k_1-1}} a^{\sharp_{k_1}}_{\beta_{k_1} p_{k_1}} \eta^{\ell_1}_pa^{\flat_{k_1}}_{\alpha_{\ell_1}p} \, \prod_{\ell=1}^{k_1}\eta_{p_\ell}^{j_\ell},
		\end{split}\]
we have that 
		\begin{equation}\label{eq:commNlow3}\begin{split}
		&[\mathcal{N}_{\leq cN^{\gamma}}, N^{-k_1} \Pi^{(1)}_{\sharp,\flat} (\eta^{j_1} , \dots , \eta^{j_{k_1}} ; \eta^{\ell_1}_p \ph_{\alpha_{\ell_1} p})] \\
		&\,=  \sum_{t=1}^{k_1} N^{-k_1} \Pi^{(1)}_{\sharp,\flat} (\eta^{j_1} , \ldots, (F(\flat_{t-1})+ F(\sharp_t)) \chi \eta^{j_t},\ldots , \eta^{j_{k_1}} ; \eta^{\ell_1}_p \ph_{\alpha_{\ell_1} p})\\
		&\,\hspace{0.5cm}  + N^{-k_1} \Pi^{(1)}_{\sharp,\flat} (\eta^{j_1} , \dots , \eta^{j_{k_1}} ; F(\flat_{k_1})\chi_p\eta^{\ell_1}_p \ph_{\alpha_{\ell_1} p}).
		\end{split}\end{equation}
Recalling that each term of the form (\ref{eq:L-term}) must have either $k_1 > 0$ or it must contain at least one $\Lambda$-operator having the form (\ref{eq:Pi2-ex}), the identities \eqref{eq:commNlow1}, \eqref{eq:commNlow2} and \eqref{eq:commNlow3} imply together with Lemma \ref{lm:indu} $ii), iii)$ that
		\begin{equation}\label{eq:ell1spl}\begin{split}
		& \left\| \big[ \cN_{\geq cN^\gamma} , \Lambda_1 \dots \Lambda_{i_1} N^{-k_1} \Pi^{(1)}_{\sharp,\flat} (\eta^{j_1} , \dots , \eta^{j_{k_1}} ; \eta^{\ell_1}_p \ph_{\alpha_{\ell_1} p})\big]\xi \right\| \\
		&\hspace{0.5cm} \leq C^m N^{-1} \left[ \|\chi\eta\| \|\eta\|^{m-\ell_1-1} \, |\eta_p|^{\ell_1}  \d_{\ell_1>0} \| (\cN_+ + 1)^{3/2} \xi \|  +\|\chi\eta\| \|\eta\|^{m-1} \| b_p (\cN_+ + 1) \xi \| \right] \\
		&\hspace{1cm}  + C^m N^{-1} \left[ \|\eta\|^{m-\ell_1} \, |\chi_p\eta_p|^{\ell_1}  \d_{\ell_1>0} \| (\cN_+ + 1)^{3/2} \xi \|  + \|\eta\|^m |\chi_p| \| b_p (\cN_+ + 1) \xi \| \right] \\
		& \hspace{0.5cm} \leq C^m  \|\eta\|^{m-1} N^{-1} \left[ \, |\eta_p|  \d_{m>0} \| (\cN_+ + 1)^{3/2} \xi \|  + \big[|\chi_p|\|\eta\|  + \|\chi\eta\|\big]\| b_p (\cN_+ + 1) \xi \| \right].
		\end{split}\end{equation}
Notice that we distinguished the cases $ \ell_1>0$ and $ \ell_1=0$ in the previous bound. Substituting the last bound into \eqref{eq:commNlow0} and summing over $m \geq 1$, we conclude the first bound in \eqref{eq:ngdp}. The second bound in \eqref{eq:ngdp} follows in the same way with the only difference that we bound $ \| b^*_p(\cN_++1)\xi\| \leq \| (\cN_++1)^{3/2}\xi\|$ in the cases where $ \ell_1=0$. The bound \eqref{eq:ngdpbar} follows from the fact that for $p\not \in \operatorname{supp}\eta$, the operator $ \bar{\bar{d}}_p$ is defined in such a way that all terms in \eqref{eq:commNlow0} for $m=1$ vanish. Moreover, all terms in \eqref{eq:ell1spl} for which $\ell_1>0$ vanish as well, since $\eta_p=0$.
		
Let us continue with the estimate \eqref{eq:ngdxy-bds}. %Compared to the momentum space bounds from above, it follows from $ \check{d}_x = \sum_{p\in\Lambda_+^*} e^{-ipx }d_p $ 
In position space, it suffices to bound 
		\begin{equation}\label{eq:commNlow0x} \| [\mathcal{N}_{\leq cN^{\gamma}},\check{d}_x] \xi\big\| \leq \sum_{m\geq 0}\frac1{m!} \Big\| \big[ \mathcal{N}_{\leq cN^{\gamma}}, \operatorname{ad}_{-B(\eta)}^{(m)}(\check{b}_x) - b^{\sharp m}(\check{\eta}_x^{(m)}) \big] \xi\Big\|,\end{equation}
where $ \check{\eta}^{(m)}_x \in L^2(\Lambda)$ is defined by its Fourier coefficients $ \eta_p^m e^{-ipx}$, for $p\in \Lambda_+^*$. Proceeding similarly as in momentum space, we first observe that
		\[\big[\cN_{\leq cN^\gamma}, b^{\sharp_m}(\check{\eta}_x^{(m)})\big] = F(\sharp_m)b^{\sharp_m}\big(\check{\chi}\ast \check{\eta}_x^{(m)}\big),  \]
where we set $ \check{\chi} =(\check{\chi}_x)_{|x=0}\in L^2(\Lambda)$ s.t. $ \check{\chi}\ast \eta_x^{(m)}\in L^2(\Lambda)$ has Fourier coefficients $ \chi_p\eta_p^m e^{-ipx}$. In particular, $ \| \check{\chi}\ast \check{\eta}_x^{(m)}\|\leq \|\chi\eta\|^m$, uniformly in $x\in\Lambda$. We then bound 
		\[%\label{eq:N-termx}
		\begin{split} 
		&\left\| \left[ \left( \frac{N- \cN_+}{N} \right)^{\frac{m+ (1-\alpha_m)/2}{2}} \left( \frac{N+1-\cN_+}{N} \right)^{\frac{m-(1-\alpha_m)/2}{2}} - 1 \right]  \big[\cN_{\leq cN^\gamma}, 	
		b^{\sharp_m}(\check{\eta}_x^{(m)})\big] \xi \right\|\\
		&\hspace{9cm} \leq C^m \| \chi \eta\|^mN^{-1}\| (\cN_++1)^{3/2}\xi\|
		\end{split}
		\]
for $m>0$ (recalling that this term vanishes for $m=0$). By Lemma \ref{lm:indu}, it then only remains to bound $2^m m! - 1$ terms of the form
\begin{equation}\label{eq:L-termx} \left\| \big[ \cN_{\geq cN^\gamma} , \Lambda_1 \dots \Lambda_{i_1} N^{-k_1} \Pi^{(1)}_{\sharp,\flat} (\eta^{j_1} , \dots , \eta^{j_{k_1}} ; \check{\eta}_x^{\ell_1})\big]\xi \right\| \end{equation}
where $i_1, k_1, \ell_1 \in \bN$, $j_1, \dots , j_{k_1} \in \bN \backslash \{ 0 \}$, where each $\Lambda_r$-operator is either a factor $(N-\cN_+ )/N$, a factor $(N+1-\cN_+ )/N$ or a $\Pi^{(2)}$-operator of the form \eqref{eq:Pi2-ex}. If we use the fact that $\big[\cN_{\leq cN^\gamma}, a^{\flat}(\check{\eta}_x^{(m)})\big] = F(\flat)a^{\flat}\big(\check{\chi}\ast \check{\eta}_x^{(m)}\big) $ and proceed then as in \eqref{eq:commNlow2}, \eqref{eq:commNlow3}, distinguishing the cases $ \ell_1>0$ and $ \ell_1=0$, we find that
		\begin{equation}\label{eq:ell1splx}\begin{split}
		& \left\| \big[ \cN_{\geq cN^\gamma} , \Lambda_1 \dots \Lambda_{i_1} N^{-k_1} \Pi^{(1)}_{\sharp,\flat} (\eta^{j_1} , \dots , \eta^{j_{k_1}} ; \check{\eta}_x^{\ell_1})\big]\xi \right\| \\
		&\hspace{0.5cm} \leq C^m N^{-1}  \|\chi\eta\| \|\eta\|^{m-1}\left[  \d_{\ell_1>0} \| (\cN_+ + 1)^{3/2} \xi \|  + \| \check{b}_x (\cN_+ + 1) \xi \| \right] \\
		&\hspace{1cm}  + C^m N^{-1} \left[ \|\eta\|^{m-\ell_1} \, \|\chi\eta\|^{\ell_1}  \d_{\ell_1>0} \| (\cN_+ + 1)^{3/2} \xi \|  + \|\eta\|^m \| b(\check{\chi}_x) (\cN_+ + 1) \xi \| \right] \\
		& \hspace{0.5cm} \leq C^mN^{-1}  \Big[\|\chi\eta\| \|\eta\|^{m-1} \| (\cN_+ + 1)^{3/2} \xi \|  +\|\chi\eta\| \|\eta\|^{m-1} \| \check{b}_x (\cN_+ + 1) \xi \| \Big]\\
		&\hspace{1cm} + C^mN^{-1} \|\eta\|^m \| b(\check{\chi}_x) (\cN_+ + 1) \xi \|.
		\end{split}\end{equation}
Summing over $m\geq 1$ in \eqref{eq:commNlow0x} proves \eqref{eq:ngdxy-bds}. Finally, the bounds \eqref{eq:ngsplitdbd} and \eqref{eq:ngddxy} can be proved with similar arguments. 		
\end{proof}

\subsection{Action of Quadratic Renormalization on Excitation Hamiltonian}

From (\ref{eq:cLN}) and (\ref{eq:GN}), we can decompose
\[ \cG_{N} = e^{-B(\eta_H)} \cL_N e^{B(\eta_H)} = \cG^{(0)}_{N} + \cG_{N}^{(2)} + \cG_{N}^{(3)} + \cG_{N}^{(4)} \]
with $ \cG_{N}^{(j)} = e^{-B(\eta_H)} \cL_N^{(j)} e^{B(\eta_H)}.$
In the following sections, we analyse the operators $\cG_{N}^{(j)}$, $j=0,2,3,4$, separately. Most of the analysis follows closely that of \cite[Section 7]{BBCS}, \cite[Section 7]{BBCS2} and we therefore focus on explaining the main steps only. Apart from the different scaling of the interaction, the only important difference consists in deriving additional commutator bounds of $\cG_{N}^{(j)}$ with restricted number of particle operators of the form $\cN_{\leq cN^\gamma}$. These bounds are based on Lemma \ref{lm:commdpNres} and will be explained in more detail in the following Subsections \ref{sub:G0} to \ref{sub:G4}. The usual commutator bounds in \eqref{eq:adjGN}, on the other hand, can be proved with the same arguments as in \cite[Section 7]{BBCS2} and we do not comment on them further. Finally, in Subsection \ref{sub:proofGN}, we prove Prop. \ref{prop:GN}. We assume throughout this section that $V \in L^3 (\bR^3)$ is compactly supported, pointwise non-negative and radial. 

%%%%%%%%%%%%%%%%%%%%%%%%%%%%%%%%%%%%%%%%%%%%%%%%%%%%%%%%%%%%%%%%%%%%
\subsubsection{Analysis of $ \cG_{N}^{(0)}$}  \label{sub:G0}

We recall from (\ref{eq:cLNj}) that
\begin{equation}\label{eq:cLN02} \cL_{N}^{(0)} =\; \frac{(N-1)}{2N} N^{\kappa} \widehat{V} (0) (N-\cN_+ ) + \frac{N^{\kappa}\widehat{V} (0)}{2N} \cN_+  (N-\cN_+ ) \end{equation}
and we define the error operator $\cE_{N}^{(0)}$ through the identity
\begin{equation}\label{eq:G0C}
 \begin{split}
\cG^{(0)}_{N} &= \frac{(N-1)}{2N} N^{\kappa}\widehat{V} (0) (N-\cN_+)+\frac{N^{\kappa}\widehat{V} (0)}{2N} \cN_+(N-\cN_+) + \cE_{N}^{(0)}.
 \end{split}
\end{equation}
% and
% \begin{equation}\label{eq:G0S}
%  \begin{split}
% \cG^{(0)}_{N,\ell} &= e^{-B_\ell(\eta)} \cL^{(0)}_N e^{B_\ell(\eta)}= \frac{(N-1)}{2} \widehat{V} (0)+ \D^{(0)}_{N,\ell}
%  \end{split}
% \end{equation}
\begin{prop}\label{prop:G0}
There exists a constant $C > 0$ such that
		\begin{align}\label{eq:E0C}
		\pm \cE_{N}^{(0)} &\leq C N^{2\kappa - \alpha/2} (\cN_+ +1), \\
		\label{eq:E0Ccomm} \pm \left[\mathcal{N}_{\leq cN^{\gamma}},\cE_{N}^{(0)}\right] &\leq C N^{2\kappa - \alpha/2} (\cN_+ +1),
%\\
%\label{eq:E0Cff}\pm\left [f (\cN_+/M) , \left [f (\cN_+ /M) ,\cE_{N}^{(0)}\right]\right] &\leq C N^{2\kappa %-\a/2} M^{-2} \|f'\|^2_{\infty} (\cN_+ +1),
		\end{align}
for all $\a > 2 \kappa$, $ \gamma\geq 0$, $c\geq 0$, $f$ smooth and bounded, $M \in \bN$ and $N \in \bN$ large enough.
\end{prop}

\begin{proof}
%From (\ref{eq:cLN02}) we have
%\begin{equation}\label{eq:cLN0-rew} \cL_N^{(0)} = \frac{(N-1)}{2} N^{\kappa}\widehat{V} (0) + \frac{1}{2N} N^{\kappa}\widehat{V} (0) \cN_+ - \frac{1}{2N} N^{\kappa}\widehat{V} (0) \cN_+^2 \end{equation}
As shown in \cite[Section 7.1]{BBCS}, $\cL_N^{(0)}$ can be written as
%\[ -\frac{\cN_+^2}{N} = \cN_+ \frac{N-\cN_+}{N} - \cN_+ =  \sum_{q \in \L^*_+} b_q^* b_q - \frac{\cN_+}{N} - \cN_+ \]
%Inserting in (\ref{eq:cLN0-rew}), we obtain
\[ \cL_N^{(0)} = \frac{(N-1)}{2} N^{\kappa}\widehat{V} (0) + \frac{N^{\kappa}\widehat{V} (0)}{2} \left[ \sum_{q \in \L^*_+} b_q^* b_q - \cN_+ \right] \]
and it follows from \eqref{eq:G0C} that
\begin{equation}\label{eq:e01}
 \begin{split}
\cE_{N}^{(0)} &= \frac{N^{\kappa} \widehat{V} (0)}{2} \sum_{q \in \L_+^*} \left[ e^{-B(\eta_H)} b_q^* b_q e^{B(\eta_H)} - b_q^* b_q \right] -  \frac{N^{\kappa} \widehat{V} (0)}{2}  \left[ e^{-B(\eta_H)} \cN_+ e^{B(\eta_H)} - \cN_+ \right]. \end{split} \end{equation}
Setting $\g_q = \cosh \eta_H (q)$, $\s_q = \sinh \eta_H (q)$ and recalling the definition of $d_q, d^*_q$ in (\ref{eq:defdq}), with $\eta$ replaced by $\eta_H (q) = \eta_q \chi (q \in P_H)$, we obtain that
\[ \sum_{q \in \L_+^*}  e^{-B(\eta_H)} b_q^* b_q e^{B(\eta_H)} = \sum_{q \in \L^*_+} \left[ \g_q b_q^* + \s_q b_{-q} + d^*_q \right] \left[ \g_q b_q + \s_q b_{-q}^* + d_q \right]  \]
Since $|\g_q^2 - 1| \leq C \eta_H (q)^2$, $|\s_q| \leq C |\eta_H (q)|$, we can use the first bound in (\ref{eq:d-bds}), Cauchy-Schwarz and the estimate $\| \eta_H \| \leq C N^{\kappa -\a/2}$ from (\ref{eq:etaHL2}) to deduce that
 \[ \frac{N^{\kappa}\widehat{V}(0)}{2} \Big| \sum_{q \in \L_+^*}  \langle \xi , \big[ e^{-B(\eta_H)} b_q^* b_q e^{B(\eta_H)} - b_q^* b_q \big] \xi \rangle \Big| \leq C N^{2\kappa -\a/2} \| (\cN_+ + 1)^{1/2} \xi \|^2. \]
Similarly, setting $\g_p^{(s)} = \cosh (s \eta_H (p))$ and $\s_p^{(s)} = \sinh (s \eta_H (p))$, and defining $d_p^{(s)}$ as in (\ref{eq:defD}) with $\eta$ replaced by $s \eta_H$, we expand the second term on the r.h.s. of (\ref{eq:e01}) as
\begin{equation*}
\begin{split}
e^{-B(\eta_H)} &\cN_+  e^{B(\eta_H)} - \cN_+ \\ 
%= \; &\int_0^1 e^{-sB(\eta_H)} [\cN_+ ,B(\eta_H)] e^{s B(\eta_H)}ds \\
%=\; &\int_0^1 \sum_{p \in P_{H}}  \eta_p \, e^{-s B(\eta_H)} ( b_p b_{-p}  + b^*_p b^*_{-p} ) e^{s B(\eta_H)} \, ds \\
= \; &\int_0^1 ds \sum_{p \in P_{H}}  \eta_p \, \left[ (\g_p^{(s)} b_p + \s_p^{(s)} b_{-p}^* + d_p^{(s)}) ( \g_p^{(s)} b_{-p} + \s_p^{(s)} b_{-p}^* + d_{-p}^{(s)}) + \hc \right].
\end{split}
\end{equation*}
We use $|\g^{(s)}_p| \leq C$ and $|\s_p^{(s)}| \leq C |\eta_p|$ as well as (\ref{eq:d-bds}) and (\ref{eq:etaHL2}) to deduce that
\[ \begin{split}\frac{N^{\kappa} \widehat{V}(0)}{2} \Big| \langle \xi, \big[& e^{-B(\eta_H)} \cN_+  e^{B(\eta_H)} - \cN_+ \big] \xi \rangle \Big| \\ &\leq C N^{\kappa} \| (\cN_+ + 1)^{1/2} \xi \|  \sum_{p \in P_H} |\eta_p| \left[ |\eta_p| \| (\cN_+ + 1)^{1/2} \xi \| + \| b_{p} \xi \| \right] \\ &\leq C N^{2\kappa -\a/2} \| (\cN_+ + 1)^{1/2} \xi \|^2. \end{split} \]
This proves the first bound (\ref{eq:E0C}). 

Let us continue with the commutator bound \eqref{eq:E0Ccomm}. Again, we consider the two contributions on the r.h.s. of (\ref{eq:e01}) separately. Let us notice first that $ [\cN_{\leq cN^\gamma}, b^*_qb_q] = 0$ and $ [\cN_{\leq cN^\gamma}, b^\sharp_q] = F(\sharp) \chi_q b^\sharp_q $ for every $q\in \Lambda_+^*$, where $ \chi\in \ell^2(\Lambda^*_+)$ denotes the characteristic function of $ \{q\in\Lambda^*_+: |q|\leq cN^\gamma \}$ and where  $ F(\sharp)=-1$ if $ \sharp=\cdot$ and $ F(\sharp) = 1$ if $ \sharp = *$. With this observation, we find with Cauchy-Schwarz that 
		\[\begin{split}
		&\frac{N^{\kappa}\widehat{V}(0)}{2} \Big| \sum_{q \in \L_+^*}  \langle \xi , \big[\cN_{\leq cN^\gamma},   e^{-B(\eta_H)} b_q^* b_q e^{B(\eta_H)} - b_q^* b_q \big] \xi \rangle \Big| \\
		&\hspace{0.5cm}\leq CN^{\kappa}\!\! \sum_{q \in \L_+^*}\!\! \big(|\chi_q \g_q| \|  b_q\xi\| +|\chi_q \s_q| \|  b_{-q}^*\xi\| + \| [\cN_{\leq cN^\gamma},d_q] \xi\| \big)\big(   |\s_q|\| b_{-q}^*\xi\|  + \| d_q\xi\|\big)\\
		&\hspace{1cm}+ CN^{\kappa}\!\! \sum_{q \in \L_+^*}\!\! \big(|  \g_q| \|  b_q\xi\| +|\s_q| \|  b_{-q}^*\xi\| + \| d_q \xi\| \big)\big(   |\chi_q  \s_q|\| b_{-q}^*\xi\|  + \| [\cN_{\leq cN^\gamma},d_q]\xi\|\big).
		\end{split}\]
Using once more that $ |\gamma_q|\leq C, |\sigma_q| \leq C |\eta_H(q)|$,  $\|\eta_H\|\leq CN^{\kappa-\alpha/2}$ as well as the bounds \eqref{eq:d-bds} and \eqref{eq:ngdp}, we obtain that
		\[\frac{N^{\kappa}\widehat{V}(0)}{2} \Big| \sum_{q \in \L_+^*}  \langle \xi , \big[\cN_{\leq cN^\gamma},  e^{-B(\eta_H)} b_q^* b_q e^{B(\eta_H)} - b_q^* b_q \big] \xi \rangle \Big| \leq C N^{2\kappa-\alpha/2}\| (\cN_++1)^{1/2}\xi\|^2.\]
Proceeding similarly for the second term on the r.h.s. of \eqref{eq:e01}, we find that
		\[ \begin{split} &\frac{N^{\kappa} \widehat{V}(0)}{2}  \Big| \langle \xi, [ \cN_{\leq cN^\gamma},  e^{-B(\eta_H)} \cN_+  e^{B(\eta_H)}  \big] \xi \rangle \Big| \\ 
		&\hspace{0.5cm} \leq CN^\kappa \int_0^1 ds \sum_{p \in P_{H}}  |\eta_p|  \big( \| (\cN_++1)^{1/2}\xi\| + |\eta_p|\| b_{-p}\xi\|  + \|[\cN_{\leq cN^\gamma}, (d_p^{(s)})^* ]\xi\| \big)\\
		&\hspace{4cm} \times \big( \|b_{-p}\xi\| + | \eta_p|\| (\cN_++1)^{1/2}\xi\| \big)  \\
		&\hspace{1cm} + CN^\kappa \int_0^1 ds \sum_{p \in P_{H}}  |\eta_p|  \big( \| (\cN_++1)^{1/2}\xi\| + |\eta_p|\| b_{-p}\xi\|  + \| (d_p^{(s)})^* \xi\| \big)\\
		&\hspace{4.5cm} \times \big( \|b_{-p}\xi\| + | \eta_p|\| (\cN_++1)^{1/2}\xi\| + \|[\cN_{\leq cN^\gamma}, (d_p^{(s)}) ]\xi\| \big)  \\
		&\hspace{0.5cm}\leq C N^{2\kappa-\alpha/2}\| (\cN_++1)^{1/2}\xi\|.
		\end{split}\]
Here, we used in the last step once more \eqref{eq:ngdp} and this proves the bound \eqref{eq:E0Ccomm}.
\end{proof}

%%%%%%%%%%%%%%%%%%%%%%%%%%%%%%%%%%%%%%%%%%%%%%%%%%%%%%%%%%%%%%%%%
%%%%%%%%%%%%%%%%%%%%%%%%%%%%%%%%%%%%%%%%%%%%%%%%%%%%%%%%%%%%%%%%%
%%%%%%%%%%%%%%%%%%%%%%%%%%%%%%%%%%%%%%%%%%%%%%%%%%%%%%%%%%%%%%%%%
%%%%%%%%%%%%%%%%%%%%%%%%%%%%%%%%%%%%%%%%%%%%%%%%%%%%%%%%%%%%%%%%%
%%%%%%%%%%%%%%%%%%%%%%%%%%%%%%%%%%%%%%%%%%%%%%%%%%%%%%%%%%%%%%%%%
\subsubsection{Analysis of $\cG_{N}^{(2)}=e^{-B(\eta_H)}\cL^{(2)}_N e^{B(\eta_H)}$}
\label{sub:G2}

We decompose $\cL_N^{(2)} = \cK + \cL_N^{(2,V)}$, setting $\cK = \sum_{p \in \Lambda_+^*} p^2 a_p^* a_p$ and 
\begin{equation}\label{eq:L2VN}
\cL^{(2,V)}_N =  \sum_{p \in \L^*_+} N^{\kappa}\widehat{V} (p/N^{1-\kappa}) a^*_pa_p \frac{N-\cN_+}{N} + \frac{1}{2} \sum_{p \in \L^*_+} N^{\kappa}\widehat{V} (p/N^{1-\kappa}) \left[ b_p^* b_{-p}^* + b_p b_{-p} \right].
\end{equation}
Hence, we can split $\cG_{N}^{(2)} $ into
\begin{equation}\label{eq:dec-G2} \cG_{N}^{(2)} = e^{-B(\eta_H)} \cK   e^{B(\eta_H)} + e^{-B(\eta_H)} \cL_N^{(2,V)} e^{B(\eta_H)} \end{equation}
and analyse the two contributions on the r.h.s. of the last equation separately.
%%%%%%%%%%%%%%%%%%%%%%%%%%%%%%%%%%%%%%%%%%%%%%%%%%%%%%%%%%%%%%%%%
%%%%%%%%%%%%%%%%%%%%%%%%%%%%%%%%%%%%%%%%%%%%%%%%%%%%%%%%%%%%%%%%%
\begin{prop}\label{prop:K}
There exists a constant $C > 0$ such that
\begin{equation} \label{eq:K-dec} \begin{split} e^{-B(\eta_H)}\cK e^{B(\eta_H)} = \; &\cK + \sum_{p \in P_{H}} p^2 \eta_p ( b_p b_{-p} + b^*_p b^*_{-p} ) \\&+ \sum_{p \in  P_{H}} p^2 \eta_p^2 \Big(\frac{N-\cN_+}{N}\Big) \Big(\frac{N-\cN_+ -1}{N}\Big) +\cE^{(K)}_{N}
 \end{split}
\end{equation}
where the self-adjoint operator $\cE^{(K)}_{N}$ satisfies
		\begin{align}\label{eq:errorKc}
		\pm \cE^{(K)}_{N} &\leq  C N^{3\kappa - \a/2} (\cH_N +1),\\
 		\label{eq:errorKccomm} \pm i\left [\mathcal{N}_{\leq cN^{\gamma}}, \cE^{(K)}_N \right] &\leq C \big(  N^{3\kappa -\alpha/2 } + N^{2\kappa - \a/2 + \gamma/2}\big)(\cH_N+1),
	%\\ 
	%	\label{eq:errCommK} \pm \left[ f (\cN_+/M), \left[ f (\cN_+ /M) ,\cE^{(K)}_{N} \right] \right] & 
	%\leq C M^{-2} \|f'\|^2_{\infty}\, N^{3\kappa-\a/2} \big( \cH_N + 1 \big)
		\end{align}
for all $\a \geq 2 \kappa$ with $ \alpha+\kappa\leq 1$, and for all $ \gamma \in [0;\alpha]$, $c\geq 0$, $f$ smooth and bounded, $M \in \bN$ and $N \in \bN$ large enough.
\end{prop}

\begin{proof}
As shown in \cite[Section 7.2]{BBCS}, we can use the relations \eqref{eq:ebe} and a first order Taylor expansion yields 
\begin{equation} \label{eq:cKterms}
\begin{split}
&e^{-B(\eta_H)} \cK e^{B(\eta_H)}-\cK \\ &= \int_0^1ds \sum_{p \in P_{H}} p^2 \eta_p \Big [\big(\gamma_p^{(s)} b_p+ \s^{(s)}_p b^*_{-p}\big) \big(\gamma_p^{(s)} b_{-p} + \s^{(s)}_p b^*_{p}\big)\, +\hc\Big ]\\
&\hspace{.3cm} + \int_0^1ds \sum_{p \in P_{H}} p^2 \eta_p \big[\big(\gamma_p^{(s)} b_p+ \s^{(s)}_p b^*_{-p}\big)
d_{-p}^{(s)}+ d_p^{(s)} \big( \gamma_p^{(s)} b_{-p}+ \s^{(s)}_p b^*_{p}\big)+\hc\big]\\
&\hspace{.3cm}+ \int_0^1ds \sum_{p \in P_{H}} p^2 \eta_p\big[ d_p^{(s)} d_{-p}^{(s)} + \hc \big]\\
&=: \text{G}_1+\text{G}_2+\text{G}_3.
\end{split}
\end{equation}
We recall that $\gamma_p^{(s)} = \cosh(s \eta_H (p))$, $\s^{(s)}_p =\sinh(s \eta_H (p))$ and that $d^{(s)}_p$ is defined as in (\ref{eq:defD}), with $\eta_p$ replaced by $s \eta_H (p)$. We consider the different contributions $ \text{G}_1$, $\text{G}_2$ and $\text{G}_3$, defined on the r.h.s. of the last equation \eqref{eq:cKterms}, separately.

Let us start with $ \text{G}_1$. By expanding the product, it was proved in \cite[Eq. (7.14)]{BBCS} that 
		\begin{equation} \label{eq:first2cK}
		\begin{split}
		\text{G}_1%&=\,\int_0^1ds \sum_{p \in P_{H}} p^2 \eta_p \Big [\big( (\gamma_p^{(s)})^2+(\s^{(s)}_p)^2\big)\big(b_pb_{-p}+b^*_{-p}b^*_{p}\big)+ \gamma_p^{(s)} \s^{(s)}_p (4b_p^*b_{p}-2N^{-1}a^*_pa_p)\big)\Big ]\\
		%&\quad+2\int_0^1ds \sum_{p \in P_{H}} p^2 \eta_p \gamma_p^{(s)} \s^{(s)}_p \left(1-\frac{\cN_+}{N}\right)\\
		%
		&= \sum_{p \in P_{H}} p^2 \eta_p \big(b_p b_{-p} + b^*_{-p} b^*_{p} \big)+ \sum_{p \in P_{H}} p^2 \eta_p^2 \left(1-\frac{\cN_+}{N}\right)+\cE^K_{1},
		\end{split}
		\end{equation}
where the error operator $ \cE^K_{1}$ is given by
\begin{equation} \nonumber
\begin{split}
\cE^K_{1}=&\,\int_0^1ds \sum_{p \in P_{H}} p^2 \eta_p \big[\big(( \gamma_p^{(s)})^2-1\big)+(\s^{(s)}_p)^2\big]\big(b_p b_{-p} + b^*_{-p} b^*_{p} \big)\\
&+\int_0^1ds \sum_{p \in P_{H}} p^2 \eta_p \gamma_p^{(s)} \s^{(s)}_p (4b_p^*b_{p}-2N^{-1}a^*_pa_p)\big)\\
&+2\int_0^1ds \sum_{p \in P_{H}} p^2 \eta_p  \left[ (\gamma_p^{(s)}-1) \s^{(s)}_p + (\s^{(s)}_p-s \eta_p) \right] \Big(1-\frac{\cN_+}{N}\Big).
\end{split}
\end{equation}
Using that $|\big((\gamma_p^{(s)})^2-1\big)|\leq C \eta_p^2$, $(\s^{(s)}_p)^2\leq C \eta_p^2$ and $p^2 \eta_p^2 \leq C N^{2\kappa-2\alpha}$, we obtain
		\begin{equation} \label{eq:cE0K}
		\begin{split}
		|\langle\xi, \cE^K_{1} \xi\rangle|  \leq \;  &C \sum_{p \in P_{H}} \Big[p^2 |\eta_p|^3 \|b_p\xi\|\|(\cN_++1)^{1/2} \xi\|+ p^2 \eta_p^2 \|a_p\xi\|^2+ p^2 \eta_p^4\Big]\\
		\leq \; &C N^{2\kappa - 2\alpha}\|(\cN_++1)^{1/2}\xi\|^2.
		\end{split}
		\end{equation}
Similarly, if $\chi\in \ell^2(\Lambda^*_+)$ denotes the characteristic function of $ \{q\in\Lambda^*_+: |q|\leq cN^\gamma \}$, we observe that 
		\[ \left[ \cN_{\leq cN^\gamma}, \cE^K_{1}\right]  = 2\int_0^1ds \sum_{p \in P_{H}} p^2 \eta_p \chi_p\big[\big(( \gamma_p^{(s)})^2-1\big)+(\s^{(s)}_p)^2\big]\big(b^*_{-p} b^*_{p}- b_p b_{-p} \big). \]
This term can be bounded as in \eqref{eq:cE0K} and we find that
		\begin{equation} \label{eq:cE0Kcomm}
		\begin{split}
		|\langle\xi, &[\cN_{\leq cN^\gamma}, \cE^K_{1} ]\xi\rangle| %\\ \leq \; & C \sum_{p \in P_{H}} p^2 |\eta_p|^3 \|b_p\xi\|\|(\cN_++1)^{1/2} \xi\|+C \sum_{p \in P_{H}} p^2 \eta_p^2 \|a_p\xi\|^2+C\sum_{p \in P_{H}} p^2 \eta_p^4\\
		\leq \; C N^{2\kappa - 2\alpha}\|(\cN_++1)^{1/2}\xi\|^2.
		\end{split}
		\end{equation}

Let us switch to the analysis of the contribution $\text{G}_2$, defined in \eqref{eq:cKterms}. As in \cite[Eq. (7.16)]{BBCS}, it is useful to further expand this into $\text{G}_2 = \text{G}_{21} + \text{G}_{22} + \text{G}_{23} + \text{G}_{24}$, where
\begin{equation} \label{eq:cKtermsG2}
\begin{split}
\text{G}_{21} &= \int_0^1ds \sum_{p \in P_{H}} p^2 \eta_p \left( \gamma_p^{(s)} b_p d_{-p}^{(s)} + \hc \right), \hspace{.23cm} \text{G}_{22} =  \int_0^1ds \sum_{p \in P_{H}} p^2 \eta_p \left( \s^{(s)}_p b^*_{-p} d_{-p}^{(s)} + \hc \right) \\
\text{G}_{23} &= \int_0^1ds \sum_{p \in P_{H}} p^2 \eta_p \left( \gamma_p^{(s)} d_p^{(s)} b_{-p}+ \hc \right), \hspace{.05cm} \text{G}_{24} =  \int_0^1ds \sum_{p \in P_{H}} p^2 \eta_p \left( \s^{(s)}_p d_p^{(s)} b^*_{p}+ \hc \right),
\end{split} \end{equation}
and to analyse the operators $\text{G}_{21}$, $\text{G}_{22} $, $\text{G}_{23} $ and $\text{G}_{24} $ separately. 

We start with $ \text{G}_{21}$. By, \cite[Eq. (7.17) \& Eq. (7.18)]{BBCS}, we have that
		\begin{equation}
		\label{eq:G21f} \text{G}_{21} = - \sum_{p \in P_H} p^2 \eta_p^2 \frac{\cN_+ + 1}{N} \frac{N-\cN_+}{N} + \left[  \cE_{2}^K + \hc \right], \end{equation}
where $\cE_{2}^K = \sum_{j=1}^5 \cE_{2j}^K$, with
		\begin{equation}\label{eq:cE2} \begin{split}
		\cE_{21}^K = \; &\frac{1}{2N} \sum_{p \in P_H} p^2 \eta_p^2 (\cN_++1) \big( b^*_p b_p - \frac 1 N a_p^* a_p\big) , \hspace{.3cm}
		\cE_{22}^K = \; \int_0^1 ds \sum_{p \in P_H} p^2 \eta_p (\gamma_p^{(s)} - 1)  b_p d_{-p}^{(s)} \\
		\cE_{23}^K = \; & \int_0^1 ds \sum_{p \in \Lambda_+^*} p^2 \eta_p b_p \bar{d}_{-p}^{(s)} , \hspace{3.2cm} \cE_{24}^K = \; - \int_0^1 ds \sum_{p \in P_H^c} p^2 \eta_p b_p \bar{\bar{d}}_{-p}^{(s)}, \\
		\cE_{25}^K =\; & \frac{1}{2N} \sum_{p \in P_H^c , q \in P_H} p^2 \eta_p \eta_q  a_q^* a_{-q}^* a_p a_{-p}(1-\cN_+/N).  \end{split} \end{equation}
Here, we set
		\begin{equation}\label{eq:barbarnot} \bar{d}^{(s)}_{-p} = d_{-p}^{(s)} +s \eta_{H} (p) \frac{\cN_+}{N} b_p^* \quad \text{ and } \quad \bar{\bar{d}}^{(s)}_{-p} = d_{-p}^{(s)} + \frac{1}{N} \sum_{q \in P_H} s \eta_q b_q^* a_{-q}^* a_{-p}.\end{equation} 
We easily find that
		\begin{equation}\label{eq:cEK21} |\langle \xi , \cE_{21}^K \xi \rangle | \leq C \sum_{p \in P_H} p^2 \eta_p^2 \| a_p \xi \|^2 \leq C N^{2\kappa - 2\a}  \|  \cN_+^{1/2} \xi \|^2
 		\end{equation}
and by applying (\ref{eq:d-bds}) in Lemma \ref{lm:dp}, we also find that
\begin{equation}\label{eq:cEK22} \begin{split} |\langle \xi , \cE_{22}^K \xi \rangle |  %&\leq \sum_{p \in P_H} p^2 |\eta_p|^3 \| \cN_+^{1/2} \xi \| \| d^{(s)}_{-p} \xi \| \\ 
&\leq \sum_{p \in P_H} p^2 |\eta_p|^3 \| \cN_+^{1/2} \xi \| \Big[ |\eta_p| \| \cN^{1/2}_+ \xi \|  + \| \eta_H \|  \| a_p \xi \| \Big] \\
&\leq  C N^{4 \kappa - 3\a} \| (\cN_+ + 1)^{1/2} \xi \|^2. \end{split} \end{equation}
Applying the bound \eqref{eq:off} for the operator $ \cE_{24}^K$, we obtain moreover that
		\begin{equation}\label{eq:cEK24}
		\begin{split} |\langle \xi , \cE_{24}^K \xi \rangle | %&\leq \int_0^1 ds \sum_{p \in P_H^c} p^2 |\eta_p| \| (\cN_+ + 1)^{1/2} \xi \| \| (\cN_+ + 1)^{-1/2} \bar{\bar{d}}_{-p}^{(s)} \xi \|\\ 
		&\leq C \| \eta_H \|^2 \| (\cN_+ + 1)^{1/2} \xi \| \sum_{p \in P_H^c} p^2 |\eta_p | \| a_p \xi \|\\ 
		%&\leq C N^{2\kappa -\a}  \| (\cN_+ + 1)^{1/2} \xi \| \| \cK^{1/2} \xi \| \Big[ \sum_{|p| \leq  N^{-\a} } p^2 \eta_p^2 \Big]^{1/2}\\ 
		&\leq C N^{3\kappa - \a/2} \| (\cN_+ + 1)^{1/2} \xi \| \| \cK^{1/2} \xi \|
 		\end{split}
  		\end{equation}
by Cauchy-Schwarz and, similarly, that
		\begin{equation}\label{eq:cEK25}
		\begin{split} |\langle  \xi , \cE_{25}^K \xi \rangle | &\leq C N^{-1}  \sum_{p \in P_H^c, q \in P_H} p^2 |\eta_p|  |\eta_q| \| a_q a_{-q} \xi \| \| a_p a_{-p} \xi \| 	
		%&\leq C \sum_{p \in P_H^c, q \in P_H} p^2 |\eta_p| |\eta_q| \| a_q \xi \| \| a_p \xi \|\\ 
		%&\leq C \Big[ \sum_{p \in P_H^c, q \in P_H} p^2 \eta^2_p q^2 \| a_q \xi \|^2 \Big]^{1/2}  \Big[  \sum_{p \in P_H^c, q \in P_H} q^{-2} \eta_q^2 p^2  \| a_p \xi \|^2 \Big]^{1/2}\\ 
		\leq C N^{2\kappa-\a} \| \cK^{1/2} \xi \|^2.
   		\end{split}
   		\end{equation}
Before we switch to the analysis of $ \cE_{23}$, let's quickly comment on the commutator of $ \cN_{\leq cN^\gamma}$ with $ \cE_{21}^K$, $\cE_{22}^K$, $\cE_{24}^K$ and $\cE_{25}^K$. We have that $ [\cN_{\leq cN^\gamma},\cE_{21}] = 0$ and referring to Lemma \ref{lm:commdpNres}, in particular to the bounds \eqref{eq:ngdp} and \eqref{eq:ngdpbar}, we obtain as above that 
		\[\begin{split}
		&|\langle  \xi , [\cN_{\leq cN^\gamma} ,  \cE_{22}^K] \xi \rangle | + |\langle  \xi , [\cN_{\leq cN^\gamma} ,  \cE_{24}^K] \xi \rangle | \\
		&\hspace{1.5cm} \leq C N^{4 \kappa - 3\a} \| (\cN_+ + 1)^{1/2} \xi \|^2+ C N^{3\kappa - \a/2} \| (\cN_+ + 1)^{1/2} \xi \| \| \cK^{1/2} \xi \|.
		\end{split}\]
Finally, if we denote by $ \chi\in\ell^2(\Lambda_+^*)$ the characteristic function of $ \{p\in\Lambda_+^*: |p|\leq cN^\gamma\}$, we observe that $[\mathcal{N}_{\leq cN^{\gamma}}, a_q^*a_{-q}^* a_pa_{-p}] =(-2 \chi(p) + 2 \chi(q))a_q^*a_{-q}^* a_pa_{-p}$ and obtain
		\[\begin{split} |\langle  \xi , [\cN_{\leq cN^\gamma} ,  \cE_{25}^K] \xi \rangle | &\leq C N^{-1} \!\!\!\! \sum_{p \in P_H^c, q \in P_H} p^2 |\eta_p|  |\eta_q| \| a_q a_{-q} \xi \| \| a_p a_{-p} \xi \| 	
		%&\leq C \sum_{p \in P_H^c, q \in P_H} p^2 |\eta_p| |\eta_q| \| a_q \xi \| \| a_p \xi \|\\ 
		%&\leq C \Big[ \sum_{p \in P_H^c, q \in P_H} p^2 \eta^2_p q^2 \| a_q \xi \|^2 \Big]^{1/2}  \Big[  \sum_{p \in P_H^c, q \in P_H} q^{-2} \eta_q^2 p^2  \| a_p \xi \|^2 \Big]^{1/2}\\ 
		\leq C N^{2\kappa-\a} \| \cK^{1/2} \xi \|^2.
   		\end{split}\]

Next, consider the remaining term $ \cE_{23}^K$. As in \cite{BBCS}, we use the scattering equation (\ref{eq:eta-scat0}) and rewrite $ \cE_{23}^K$ in position space as
		\begin{equation}\label{eq:bdbarpspaceterm1} \begin{split} \cE_{23}^K = \;& -\frac{1}2 N^\kappa\int_0^1 ds \int_{\Lambda^2}  dx dy\; N^{3-3\kappa} V(N^{1-\kappa}(x-y)) f_{N} (x-y) \check{b}_x \check{\bar{d}}^{(s)}_y \\ 
		&+ N^\kappa \lambda_\ell \int_0^1 ds   \int_{\Lambda^2} dx dy\; \chi_\ell (x-y) N^{3-3\kappa}f_{N} (x-y) \check{b}_x \check{\bar{d}}^{(s)}_y. \end{split} \end{equation}
With Lemma \ref{sceqlemma}, the bound (\ref{eq:splitdbd}) in Lemma \ref{lm:dp}, the upper bound \eqref{eq:etax} as well as the assumption $ \alpha+\kappa\leq 1$, we obtain that
		\begin{equation}\label{eq:bdbarpspaceterm2} \begin{split} 
		|\langle \xi , \cE_{23}^K \xi \rangle | \leq \; & N^\kappa \int_0^1 ds \int_{\Lambda^2} dx dy \left[ N^{3-3\kappa} V (N^{1-\kappa} (x-y)) + C \right]  \\ &\hspace{2.8cm}  \times \| (\cN_+ + 1)^{1/2} \xi \| \| (\cN_+ + 1)^{-1/2}  \check{a}_x \check{\bar{d}}^{(s)}_y \xi \| \\
		 \leq \; &C N^{\kappa-1} \| \eta_H \| 
 		\int_{\Lambda^2} dx dy \left[ N^{3-3\kappa} V (N^{1-\kappa} (x-y)) + C\right] \\ &\times \| (\cN_+ + 1)^{1/2} \xi \|  \Big[ N \| (\cN_++1)^{1/2} \xi \|  +  \| \check{a}_y \cN_+ \xi \| + \| \check{a}_x \check{a}_y \cN_+^{1/2} \xi \| \Big] \\
		\leq \; &C N^{2\kappa-\alpha/2} \| (\cN_+ + 1)^{1/2} \xi \|^2 + C N^{3\kappa/2-\alpha/2} \|(\cN_+ + 1)^{1/2} \xi \| \| \cV_N^{1/2} \xi \|. \end{split} \end{equation}
Similarly, if we use the commutator bound (\ref{eq:ngsplitdbd}) in Lemma \ref{lm:commdpNres}, we find that
		\begin{equation}\label{eq:bdbarpspaceterm3} \begin{split} 
		|\langle \xi , [\cN_{\leq cN^\gamma}, \cE_{23}^K] \xi \rangle | 
		&\leq \; C N^{2\kappa-\alpha/2-1}  
 \int_{\Lambda^2} \!\!dx dy \left[ N^{3-3\kappa} V (N^{1-\kappa} (x-y)) +C \right]\| (\cN_+ + 1)^{1/2} \xi \|  \\ 
 &\hspace{1cm}\times \Big[ (N +| (\check{\chi}\ast \check{\eta}_H)(x-y)| ) \| (\cN_++1)^{1/2} \xi \|   +  \| a(\check{\chi}_x )  (\cN_++1) \xi \| \\
 &\hspace{1.5cm} +  \| \check{a}_y \cN_+ \xi \| +  \| a(\check{\chi}_x ) \check{a}_y (\cN_++1)^{1/2} \xi \| + \| \check{a}_x \check{a}_y \cN_+^{1/2} \xi \| \Big]. 
 		\end{split} \end{equation}
Here, $ \check\chi_x\in L^2(\Lambda)$ has values $ \check{\chi}_x(y) = \check{\chi}(y-x)$, where $ \chi \in\ell^2(\Lambda_+^*)$ denotes the characteristic function of $ \{p\in\Lambda_+^*: |p|\leq cN^\gamma\}$. Notice that $  \check{\chi}\ast \check{\eta}_H$ has Fourier transform $ \chi \eta_H \in \ell^2(\Lambda_+^*)$ s.t. the assumptions $ \alpha+\kappa\leq 1$, $\gamma\leq \alpha$ and the bound \eqref{eq:etax} imply
		\begin{equation} \label{eq:xbndchigastetaH} | \check{\chi}\ast\eta_H (x)| \leq |\eta_H(x)|  + \sum_{\substack{ p\in\Lambda_+^*:|p| \leq cN^\gamma}} |\eta_p |\leq C N. 
		\end{equation}
Furthermore, switching to momentum space and observing that
		\[\begin{split}
		 &\int_{\Lambda^2} dx dy\; N^{3-3\kappa} V (N^{1-\kappa} (x-y))  a^*(\check{\chi}_x)\check{a}^*_y \check{a}_y a(\check{\chi}_x) \\
		 &\hspace{5cm}= \sum_{\substack{r\in\Lambda^*; p,q\in\Lambda_+^*: \\ p+r,q+r\neq 0}} \widehat{V}(r/N^{1-\kappa})\chi_{p+r}\chi_p a^*_{p+r}a^*_q a_pa_{q+r}, 
		\end{split}\]
we find as a consequence of Cauchy-Schwarz that
		\begin{equation}\label{eq:VNchixaybnd}\begin{split}
		&\int_{\Lambda^2} dx dy\; N^{3-3\kappa} V (N^{1-\kappa} (x-y))  \| \check{a}_y a(\check{\chi}_x)\xi\|^2\\
		&\leq  C  \bigg(\sum_{\substack{r\in\Lambda^*; p,q\in\Lambda_+^*: \\ |p+r|, |p|\leq cN^{\gamma}  }} |p|^{-2} |p+r|^2\| a_{p+r}a_q\xi \|^2\bigg)^{1/2} \bigg( \sum_{\substack{r\in\Lambda^*; p,q\in\Lambda_+^*: \\ |p+r|, |p|\leq cN^{\gamma} }} |p|^{2} |p+r|^{-2}\| a_{q+r}a_p\xi \|^2\bigg)^{1/2} \\
		&\leq C N^{\gamma} \| \cK^{1/2}(\cN_++1)^{1/2}\xi\|^2.
		\end{split}\end{equation}
Observing also that $ \int_{\Lambda} dx \; a^*(\check{\chi}_x ) a(\check{\chi}_x ) = \sum_{\substack{p\in\Lambda_+^*:  |p|\leq cN^\gamma}} a^*_pa_p $, 
we conclude that
		\[\begin{split} |\langle \xi , [\cN_{\leq cN^\gamma}, \cE_{23}^K] \xi \rangle |  \leq &\;C N^{2\kappa-\alpha/2} \| (\cN_+ + 1)^{1/2} \xi \|^2 + C N^{2\kappa-\alpha/2 +\gamma/2} \| (\cK + 1)^{1/2} \xi \|^2\\
		& + C N^{3\kappa/2-\alpha/2} \|(\cN_+ + 1)^{1/2} \xi \| \| \cV_N^{1/2} \xi \|.  \end{split}\]

Collecting all the previous bounds from (\ref{eq:cEK21}) to \eqref{eq:cEK25} as well as their associated commutator bounds, we deduce that we have for all $\xi\in \cF_+^{\leq N}$ that
		\begin{equation}\label{eq:cE2f} 
		\begin{split}
		| \langle\xi, \cE_2^K  \xi\rangle|  &\leq C N^{3\kappa-\alpha/2} \langle\xi, (\cH_N + 1)\xi\rangle, \\
		|\langle\xi, \left[ \cN_{\leq cN^\gamma}, \cE_2^K \right] \xi\rangle | &\leq C (N^{3\kappa-\alpha/2} + N^{2\kappa+\gamma/2-\alpha/2} ) \langle\xi, (\cH_N + 1)\xi\rangle.
		\end{split} 
		\end{equation}
		
Let us now switch to the contribution $\text{G}_{22}$, defined in \eqref{eq:cKtermsG2}. Lemma \ref{lm:dp} implies
		\begin{equation}\label{eq:G22}
		\begin{split}
		|\langle\xi,\text{G}_{22}\xi\rangle|&\leq C \sum_{p \in P_{H}} p^2 \eta_p^2\|b_{-p}\xi\|\| d_{-p} \xi\|
		%&\leq C \sum_{p \in P_{H}} p^2 \eta_{p}^2\|b_{-p}\xi\|\left[ |\eta_{-p}| \| (\cN_+ + 1)^{1/2} \xi \| + \| \eta \| \| b_{-p} \xi \| \right]\\
		\leq CN^{3\kappa - 5\a/2}\|(\cN_++1)^{1/2}\xi\|^2
		\end{split}
		\end{equation}
and, using Lemma \ref{lm:commdpNres}, we get similarly
		\begin{equation}\label{eq:G222}
		\begin{split}
		|\langle\xi,[\mathcal{N}_{\leq cN^{\gamma}},\text{G}_{22}]\xi\rangle|&\leq C \sum_{p \in P_{H}} p^2 \eta_p^2\|b_{-p}\xi\| \big (\chi(-p)\| d_{-p} \xi\|+ \| [\mathcal{N}_{\leq c N^{\gamma}},d_{-p}]\xi\|\big)\\
%&\leq C \sum_{p \in P_{H}} p^2 \eta_p^2\|b_{-p}\xi\|[\chi(p) ]\left[ |\eta_{-p}| \| (\cN_+ + 1)^{1/2} \xi \| + \| \eta \| \| b_{-p} \xi \| \right]\\
		&\leq CN^{3\kappa - 5\a/2}\|(\cN_++1)^{1/2}\xi\|^2.
		\end{split}
		\end{equation}

Consider next the term $\text{G}_{23}$, defined in \eqref{eq:cKtermsG2}. Recalling the notation  $\bar{\bar{d}}^{(s)}_p$, introduced in (\ref{eq:barbarnot}), it was shown in \cite{BBCS} that $\text{G}_{23}$ can be written as $\text{G}_{23} = \sum_{j=1}^4  \cE_{3j}^K + \hc$, where
		\begin{equation*} %\label{eq:G23}
		\begin{split}
		\cE_{31}^K = \; &\int_0^1ds \sum_{p \in P_{H}} p^2 \eta_p \big(\gamma_p^{(s)}-1\big) d_p^{(s)}b_{-p} \, , \hspace{1cm} \cE_{32}^K =  \int_0^1ds \sum_{p \in \L_+^*} p^2 \eta_p d_p^{(s)} b_{-p} \\  \cE_{33}^K =  \; & \frac{1}{2N} \sum_{p \in P_H^c, q \in P_H} p^2 \eta_p \eta_q b_q^* a_{-q}^* a_p b_{-p} \, , \hspace{1.1cm}
		\cE_{34}^K = -\int_0^1 ds \sum_{p\in P_H^c} p^2 \eta_p \bar{\bar{d}}^{(s)}_p b_{-p}
		\end{split}
		\end{equation*}
The contribution $\cE_{33}^K $ can be controlled exactly as $ \cE_{25}^K$, defined in \eqref{eq:cE2}. The errors $\cE_{31}^K $ and $\cE_{34}^K$ as well as their commutators with $\cN_{\leq cN^\gamma}$ can be controlled as above, using Lemma \ref{lm:dp} and Lemma \ref{lm:commdpNres}, respectively. We find that 
		\[ \begin{split}
		|\langle \xi, \cE_{31}^K \xi \rangle | + |\langle \xi, \cE_{34}^K \xi \rangle | & \leq C N^{4\kappa -3\a} \| (\cN_+ + 1)^{1/2} \xi \|^2 + C N^{3\kappa -\alpha/2}   \|  (\cN_+ + 1)^{1/2} \xi \| \| \cK^{1/2} \xi \|
		\end{split} \]
as well as
		\[\begin{split}
		&|\langle \xi, [\cN_{\leq cN^\gamma}, \cE_{31}^K] \xi \rangle | + |\langle \xi, [\cN_{\leq cN^\gamma}, \cE_{34}^K] \xi \rangle | \\
		&\hspace{1.5cm}\leq  C N^{4\kappa -3\a} \| (\cN_+ + 1)^{1/2} \xi \|^2 + CN^{3\kappa- \alpha/2}   \|  (\cN_+ + 1)^{1/2} \xi \| \| \cK^{1/2} \xi \|.
		\end{split}\]
Finally, the term $ \cE_{32}^K$ can be controlled similarly as the term $ \cE_{23}^K$, defined in \eqref{eq:cE2}. We switch to position space and apply Lemma \ref{lm:dp} so that
		\begin{equation}\label{eq:dbarbpspaceterm1}\begin{split} |\langle \xi , \cE_{32}^K \xi \rangle | &\leq CN^\kappa \int_0^1 ds \int_{\Lambda^2} dx dy \left[ N^{3-3\kappa} V(N^{1-\kappa}(x-y)) +  C \right]\| (\cN_+ + 1)^{1/2} \xi \| \\ 
		&\hspace{3cm} \times  \| (\cN_+ + 1)^{-1/2} \check{d}_x^{(s)} \check{a}_y \xi \| \\
		&\leq C N^{\kappa-1}  \| \eta_H \| \int_{\Lambda^2} dx dy \left[ N^{3-3\kappa} V(N^{1-\kappa}(x-y)) + C \right] \| (\cN_+ + 1)^{1/2} \xi \| \\ &\hspace{4cm} \times \left[ \| \check{a}_y (\cN_+ + 1) \xi \| + \| \check{a}_x \check{a}_y (\cN_+ + 1)^{1/2} \xi \| \right]  \\
		&\leq C N^{2\kappa - \alpha/2} \| (\cN_+ + 1)^{1/2} \xi \|^2 + C N^{3\kappa/2 -\a/2} \| (\cN_+ + 1)^{1/2} \xi \| \| \cV_N^{1/2} \xi \|.
 		\end{split} \end{equation}
For the commutator with $ \cN_{\leq cN^\gamma}$, we use Lemma \ref{lm:commdpNres} and obtain
		\begin{equation}\label{eq:dbarbpspaceterm2}\begin{split} |\langle \xi , [\cN_{\leq cN^\gamma}, \cE_{32}^K] \xi \rangle | %&\leq CN^\kappa \int_0^1 ds \int_{\Lambda^2} dx dy \left[ N^{3-3\kappa} V(N^{1-\kappa}(x-y)) +  C \right]\| (\cN_+ + 1)^{1/2} \xi \| \\ 
		%&\hspace{0.5cm} \times \big[  \| (\cN_+ + 1)^{-1/2} [\cN_{\leq cN^\gamma}, \check{d}_x^{(s)}] \check{a}_y \xi \| + \| (\cN_+ + 1)^{-1/2} \check{d}_x^{(s)} a(\check{\chi}_y) \xi \|\big] \\
		&\leq C N^{2\kappa-\alpha/2-1} \!\!\int_{\Lambda^2} dx dy \left[ N^{3-3\kappa} V(N^{1-\kappa}(x-y)) + C \right] \| (\cN_+ + 1)^{1/2} \xi \| \\ &\hspace{3cm} \times \Big[ \| \check{a}_y (\cN_+ + 1) \xi \| + \| a(\check{\chi}_y) \check{b}_x(\cN_+ + 1)^{1/2} \xi \|\\
		&\hspace{3.5cm} +\| a(\check{\chi}_y) (\cN_+ + 1) \xi \|  + \| \check{a}_x \check{a}_y (\cN_+ + 1)^{1/2} \xi \|  \Big]. 
 		\end{split} \end{equation}
Proceeding as in \eqref{eq:VNchixaybnd} and thereafter, we deduce that
		\[\begin{split} |\langle \xi , [\cN_{\leq cN^\gamma}, \cE_{32}^K] \xi \rangle |  \leq &\;C N^{2\kappa-\alpha/2} \| (\cN_+ + 1)^{1/2} \xi \|^2 + C N^{2\kappa-\alpha/2 +\gamma/2} \| (\cK + 1)^{1/2} \xi \|^2\\
		& + C N^{3\kappa/2-\alpha/2} \|(\cN_+ + 1)^{1/2} \xi \| \| \cV_N^{1/2} \xi \|.  \end{split}\]
Summarizing the last bounds, we have shown that for all $\xi\in \cF_+^{\leq N}$, we have that
		\begin{equation}\label{eq:G23f}
		\begin{split}
		| \langle\xi, \text{G}_{23}  \xi\rangle|  &\leq C N^{3\kappa-\alpha/2} \langle\xi, (\cH_N + 1)\xi\rangle, \\
		|\langle\xi, \left[ \cN_{\leq cN^\gamma}, \text{G}_{23} \right] \xi\rangle | &\leq C (N^{3\kappa-\alpha/2} + N^{2\kappa+\gamma/2-\alpha/2} ) \langle\xi, (\cH_N + 1)\xi\rangle.
		\end{split} 
		\end{equation}

It remains to control the term $\text{G}_{24}$ in \eqref{eq:cKtermsG2}. This follows as before, using (\ref{eq:d-bds}) in Lemma \ref{lm:dp} and the bound (\ref{eq:H1eta}). We find that
		\begin{equation*}
		\begin{split}
		|\langle &\xi,\text{G}_{24}\xi\rangle| \\ &\leq C \int_0^1 ds  \sum_{p \in P_{H}} p^2 \eta_p^2 \|(\cN_++1)^{1/2}\xi\|\|(\cN_++1)^{-1/2} d_p^{(s)}b^*_{p}\xi\| \\
		&\leq C N^{-1} \|(\cN_++1)^{1/2}\xi\|  \sum_{p \in P_{H}} p^2 \eta_p^2  \left[ |\eta_p| \| (\cN_+ + 1)^{3/2} \xi \| + \| \eta_H \| \| b_p b_p^* (\cN_+ +1)^{1/2} \xi \| \right] \\
		%&\leq C  N^{-1} \|(\cN_++1)^{1/2}\xi\| \\ &\hspace{.5cm} \times \sum_{p \in P_{H}} p^2 \eta_p^2  \left[ |\eta_p| \| (\cN_+ + 1)^{3/2} \xi \| + \| \eta_H \| \| (\cN_+ +1)^{1/2} \xi \| + \| \eta_H \| \| a_p (\cN_+ + 1) \xi \| \right] \\
		&\leq C N^{3\kappa -\a} \| (\cN_+ + 1)^{1/2} \xi \|^2 
		\end{split} 
		\end{equation*}
and, referring to (\ref{eq:ngdp}) in Lemma \ref{lm:commdpNres}, that
		\begin{equation*}
		\begin{split}
		|\langle &\xi,[\cN_{\leq cN^\gamma}, \text{G}_{24}]\xi\rangle| %\\ &\leq C \int_0^1 ds  \sum_{p \in P_{H}} p^2 \eta_p^2 \|(\cN_++1)^{1/2}\xi\|\|(\cN_++1)^{-1/2} d_p^{(s)}b^*_{p}\xi\| 
		%&\leq C N^{-1} \|(\cN_++1)^{1/2}\xi\|  \sum_{p \in P_{H}} p^2 \eta_p^2  \left[ |\eta_p| \| (\cN_+ + 1)^{3/2} \xi \| + \| \eta_H \| \| b_p b_p^* (\cN_+ +1)^{1/2} \xi \| \right] \\
		%&\leq C  N^{-1} \|(\cN_++1)^{1/2}\xi\| \\ &\hspace{.5cm} \times \sum_{p \in P_{H}} p^2 \eta_p^2  \left[ |\eta_p| \| (\cN_+ + 1)^{3/2} \xi \| + \| \eta_H \| \| (\cN_+ +1)^{1/2} \xi \| + \| \eta_H \| \| a_p (\cN_+ + 1) \xi \| \right] \\
		\leq C N^{3\kappa -\a} \| (\cN_+ + 1)^{1/2} \xi \|^2.
		\end{split} 
		\end{equation*}

Altogether, (\ref{eq:G21f}), (\ref{eq:cE2f}), (\ref{eq:G22}), (\ref{eq:G23f}) and the last two bounds prove that
		\[ \text{G}_2 =  - \sum_{p \in P_H} p^2 \eta_p \frac{\cN_+ + 1}{N} \frac{N-\cN_+}{N} + \cE_4^K \]
where for all $\xi\in \cF_+^{\leq N}$, it holds true that
		\begin{equation}\label{eq:G2fi} 
		\begin{split}
		| \langle\xi, \cE_4^K  \xi\rangle|  &\leq C N^{3\kappa-\alpha/2} \langle\xi, (\cH_N + 1)\xi\rangle, \\
		|\langle\xi, \left[ \cN_{\leq cN^\gamma}, \cE_4^K \right] \xi\rangle | &\leq C (N^{3\kappa-\alpha/2} + N^{2\kappa+\gamma/2-\alpha/2} ) \langle\xi, (\cH_N + 1)\xi\rangle.
		\end{split} 
		\end{equation}

Finally, we analyse $\text{G}_3$, defined in \eqref{eq:cKterms}. We follow \cite{BBCS} and split it into two terms, $\text{G}_3 = \cE_{51}^K + \cE_{52}^K + \hc$, where
\[ \cE_{51}^K = \int_0^1 ds \sum_{p \in \L_+^*} p^2 \eta_p d^{(s)}_p d_{-p}^{(s)}, \qquad \cE_{52}^K= - \int_0^1 ds \sum_{p \in P_H^c} p^2 \eta_p d^{(s)}_p d^{(s)}_{-p}.  \]	
The error $ \cE_{52}^K$ can be controlled as above, using the bounds \eqref{eq:d-bds} from Lemma \ref{lm:dp} and the bounds \eqref{eq:ngdp} from Lemma \ref{lm:commdpNres}. As a result, one obtains that
		\[ |\langle \xi, \cE_{52}^K \xi \rangle | + |\langle \xi, [\mathcal{N}_{\leq cN^{\gamma}},\cE_{52}^K] \xi \rangle |  \leq C N^{3\kappa -\a/2} \| (\cN_+ + 1)^{1/2} \xi \| \| \cK^{1/2} \xi \|. \] 
To deal with the remaining contribution $\cE_{51}^K $, we switch as usual to position space. Proceeding similarly as for the terms $\cE_{23}^K $ and $\cE_{32}^K$ above, but now using the bound \eqref{eq:ddxy} in Lemma \ref{lm:dp}, we find that 
		\begin{equation}\label{eq:ddxspacebnd1} \begin{split}  | \langle \xi , \cE_{51}^K \xi \rangle | \leq \; &CN^\kappa \int_0^1 ds \int_{\L^2}  dx dy \left[ N^{3-3\kappa} V(N^{1-\kappa}(x-y)) + C \right] \\ 
		&\hspace{2.5cm} \times  \| (\cN_+ + 1)^{1/2} \xi \| \| (\cN_+ + 1)^{-1/2} \check{d}^{(s)}_x \check{d}^{(s)}_y \xi \| \\
		\leq\;& C  N^{2\kappa-\a /2}  \| (\cN_+ + 1)^{1/2} \xi \|^2 + C N^{5\kappa/2 - \alpha} \| (\cN_+ + 1)^{1/2} \xi \| \| \cV_N^{1/2} \xi \|.
		\end{split} \end{equation}
In the same way, referring now to \eqref{eq:ngddxy} in Lemma \ref{lm:commdpNres}, we see that
		\begin{equation}\label{eq:ddxspacebnd2} \begin{split}  | \langle \xi , [\cN_{\leq cN^\gamma}, \cE_{51}^K] \xi \rangle | &\leq CN^\kappa \int_0^1 ds \int_{\L^2}  dx dy \left[ N^{3-3\kappa} V(N^{1-\kappa}(x-y)) + C \right] \\ 
		&\hspace{2cm} \times  \| (\cN_+ + 1)^{1/2} \xi \| \| (\cN_+ + 1)^{-1/2} [\cN_{\leq cN^\gamma}, \check{d}^{(s)}_x \check{d}^{(s)}_y] \xi \| \\
		&\leq C  N^{2\kappa-\a /2+\gamma/2}  \| (\cK + 1)^{1/2} \xi \|^2 + C N^{5\kappa/2 - \alpha} \| (\cN_+ + 1)^{1/2} \xi \| \| \cV_N^{1/2} \xi \|.
		\end{split} \end{equation}
Hence, collecting the last two bounds together with (\ref{eq:first2cK}), (\ref{eq:cE0K}) and (\ref{eq:G2fi}), we have proved the identity (\ref{eq:K-dec} with the error bounds (\ref{eq:errorKc}) and \eqref{eq:errorKccomm}. 
\end{proof}
%%%%%%%%%%%%%%%%%%%%%%%%%%%%%%%%%%%%%%%%%%%%%%%%%%%%%%%%%%%%%%%%%
%%%%%%%%%%%%%%%%%%%%%%%%%%%%%%%%%%%%%%%%%%%%%%%%%%%%%%%%%%%%%%%%%
Having analysed the conjugated kinetic energy $e^{-B(\eta_H)}\cK e^{B(\eta_H)}$, let's switch to the analysis of the second term on the r.h.s. of (\ref{eq:dec-G2}).
\begin{prop}\label{prop:G2V}
There is a constant $C > 0$ such that
		\begin{equation}\label{eq:cEV-defc}
		\begin{split}
		 e^{-B (\eta_H )} \cL^{(2,V)}_N  e^{B(\eta_H)}  = &\; \sum_{p \in P_{H}} N^{\kappa}\widehat{V} (p/N^{1-\kappa}) \eta_p\Big(\frac{N-\cN_+}{N}\Big)\Big(\frac{N-\cN_+-1}{N}\Big) \\
 		&\; +\sum_{p \in \Lambda^*_+} N^{\kappa} \widehat{V} (p/N^{1-\kappa})  a^*_pa_p \frac{N-\cN_+}{N} \\ 
		&\;+ \frac{1}{2}\sum_{p \in \Lambda^*_+} N^{\kappa}\widehat{V} (p/N^{1-\kappa}) \big( b_p b_{-p}+ b_{-p}^* b_p^*\big) +\cE_{N}^{(V)}
		\end{split}
		\end{equation}
where
		\begin{align}\label{eq:errorVc} \pm \cE_{N}^{(V)} &\leq  C N^{2\kappa- \a/2}(\cH_N +1),\\
		\label{eq:errorVccomm} \pm i[\mathcal{N}_{\leq cN^{\gamma}},\cE_{N}^{(V)}] & \leq  C N^{2\kappa- \a/2+\gamma/2} (\cH_N +1),
		%\\
		%\label{eq:errCommV} \pm \left[ f (\cN_+/M), \left[ f (\cN_+/M) ,\cE^{(V)}_{N} \right] \right] & %\leq C N^{2\kappa -\a/2} M^{-2} \|f'\|^2_{\infty}\,  \big( \cH_N + 1 \big)
		\end{align}
for all $\alpha \geq  2\kappa$ with $ \alpha+\kappa\leq 1$, and for all $ 0\leq \gamma\leq \alpha$, $c\geq 0$, $f$ smooth and bounded, $M \in \bN$ and $N \in \bN$ large enough.
\end{prop}
\begin{proof} We follow closely the analysis in \cite[Prop. 7.3]{BBCS} and briefly sketch the main steps to prove the bounds \eqref{eq:errorVc} and \eqref{eq:errorVccomm}. 

As in \cite{BBCS}, it is useful to decompose $e^{-B (\eta_H )} \cL^{(2,V)}_N  e^{B(\eta_H)} $ into the sum 
		\begin{equation}\label{eq:G2-deco} \begin{split}
		e^{-B(\eta_H)} \cL^{(2,V)}_{N}  e^{B(\eta_H)}  =\; & \sum_{p \in \Lambda^*_+} N^{\kappa} \widehat{V} (p/N^{1-\kappa}) e^{-B(\eta_H)} b_p^* b_p e^{B(\eta_H)} \\ &- \frac{1}{N} \sum_{p \in 		\Lambda^*_+} N^{\kappa}\widehat{V} (p/N^{1-\kappa}) e^{B(\eta_H)} a_p^* a_p e^{-B(\eta_H)} \\
		&+ \frac{1}{2} \sum_{p \in \Lambda^*_+} N^{\kappa}\widehat{V} (p/N^{1-\kappa}) e^{-B(\eta_H)} \big[ b_p b_{-p} + b_p^* b_{-p}^* \big] e^{B(\eta_H)} \\ =: \; &\text{F}_1 + \text{F}_2 +\text{F}_3
		\end{split} \end{equation}
and to analyse the contributions $ \text{F}_1$, $\text{F}_2$ and $\text{F}_3$ separately. Following \cite[Eq. (7.32)]{BBCS} and thereafter, we split $ \text{F}_1$ into
		\begin{equation*} %\label{eq:G21-deco}
		\begin{split}
		\text{F}_1 =\; & \sum_{ p\in  \L_+^*} N^{\kappa}\widehat{V} (p/N^{1-\kappa})  a_p^*  a_p \frac{N-\cN_+}{N}
		%&\frac{1}{N} \sum_{p \in \L_+^*} N^{\kappa}\widehat{V} (p/N^{1-\kappa}) a_p^* a_p  + \sum_{ p\in P_H} N^{\kappa}\widehat{V} (p/N^{1-\kappa}) \Big[ (\g_p^2-1) b_p^*b_p + \gamma_p \sigma_p (b_{-p}b_p+b_p^*b_{-p}^*) \\
		%&\hspace{5cm}+  \sigma_p^2 ( b_{p}^*b_p-N^{-1}a_{p}^*a_p)+ \sigma_p^2 \Big(\frac{N-\cN_+}{N}\Big)\Big]\\
		%&+\sum_{ p\in \L_+^*} N^{\kappa}\widehat{V} (p/N^{1-\kappa}) \big[ (\gamma_p b_p^* + \sigma_p b_{-p})  d_p+  d_p^*(\gamma_p b_p + \sigma_p b_{-p}^*)+ d_p^* d_p\big]\\
		 + \cE_{1}^V,
		\end{split} 
		\end{equation*}
where $ \cE_1^V = \sum_{i=j}^4 \text{F}_{1j}$ is defined through
		\[\begin{split}
		&\text{F}_{11} = \frac{1}{N} \sum_{p \in \L_+^*} N^{\kappa}\widehat{V} (p/N^{1-\kappa}) a_p^* a_p, \\
		& \text{F}_{12} = \sum_{ p\in P_H} N^{\kappa}\widehat{V} (p/N^{1-\kappa}) \Big[ (\g_p^2-1) b_p^*b_p + \gamma_p \sigma_p (b_{-p}b_p+b_p^*b_{-p}^*)\Big] , \\
		& \text{F}_{13} =   \sum_{ p\in P_H} N^{\kappa}\widehat{V} (p/N^{1-\kappa})  \Big[\sigma_p^2 ( b_{p}^*b_p-N^{-1}a_{p}^*a_p)+ \sigma_p^2 \Big(\frac{N-\cN_+}{N}\Big) \Big], \\
		&\text{F}_{14} = \sum_{ p\in \L_+^*} N^{\kappa}\widehat{V} (p/N^{1-\kappa}) \big[ (\gamma_p b_p^* + \sigma_p b_{-p})  d_p+  d_p^*(\gamma_p b_p + \sigma_p b_{-p}^*)+ d_p^* d_p\big].
		\end{split}\]
Using Cauchy-Schwarz, the pointwise bounds $|\g_p^2-1|\leq C |\eta_p|^2$, $|\s_p|\leq C|\eta_p|$ for all $p \in P_H$, the bound $\| \eta_H \| \leq C N^{\kappa - \a/2}$ and the fact that $ [\cN_{\leq cN^\gamma}, a^{\sharp}_p] = F(\sharp)\chi_p a^{\sharp}_p $, where $ F(*)= 1$, $ F(\cdot)=-1$ and where $ \chi$ denotes the characteristic function of the set \linebreak $ \{ p\in\Lambda_+^*: |p|\leq cN^\gamma\}$, it is straight-forward to verify that 
		\[\begin{split}
		 \pm ( \text{F}_{11} + \text{F}_{12}+\text{F}_{13}) \leq C N^{2\kappa-\a/2} (\cN_++1) ,\\
		 \pm i[\cN_{\leq cN^\gamma}, \text{F}_{11} + \text{F}_{12}+\text{F}_{13} ] \leq C N^{2\kappa-\a/2} (\cN_++1). 
		 \end{split}\]
To control the remaining error term $\text{F}_{14}$, we refer to the bounds \eqref{eq:d-bds} in Lemma \ref{lm:dp} and \eqref{eq:ngdp} in Lemma \ref{lm:commdpNres}. With Cauchy-Schwarz, they imply that
		\[\begin{split}
		 \pm ( \text{F}_{14} ) \leq C N^{2\kappa-\a/2} (\cN_++1) ,\hspace{0.5cm} \pm i[\cN_{\leq cN^\gamma}, \text{F}_{14}  ] \leq C N^{2\kappa-\a/2} (\cN_++1), 
		 \end{split}\]
so that altogether
		\[\begin{split}
		 \pm \cE_{1}^V &\leq  C N^{2\kappa-\a/2} (\cN_++1),\hspace{0.5cm}\pm i[\cN_{\leq cN^\gamma}, \cE_{1}^V ] \leq C N^{2\kappa-\a/2} (\cN_++1).
		 \end{split}\]

The analysis of the term $\text{F}_2$, defined in \eqref{eq:G2-deco}, is quite similar once we notice that 
		\[\begin{split} &\frac{1}{N} \sum_{p \in P_H} N^{\kappa}\widehat{V} (p/N^{1-\kappa})  e^{B(\eta_H)} a_p^* a_p e^{-B(\eta_H)} - \frac{1}{N} \sum_{p \in \Lambda^*_+} N^{\kappa}\widehat{V} (p/N^{1-\kappa})   a_p^* a_p   \\
		&=\frac{N^{\kappa}}{N}  \int_{0}^1ds   \sum_{p \in P_H} \widehat{V} (p/N^{1-\kappa})\eta_p \Big[ \big( \gamma_p^{(s)} b_p + \sigma_p^{(s)} b^*_{-p}\big) \big( \gamma_p^{(s)} b_{-p} + \sigma_p^{(s)} b^*_{p} \big)+\text{h.c.} \Big] \\
		&\hspace{0.5cm} + \frac{N^{\kappa}}{N}  \int_{0}^1ds  \!\!\sum_{p \in P_H} \widehat{V} (p/N^{1-\kappa})\eta_p \Big[ \big( \gamma_p^{(s)} b_p + \sigma_p^{(s)} b^*_{-p}\big) d_{-p}^{(s)}+ d_p^{(s)}\big( \gamma_p^{(s)} b_{-p} + \sigma_p^{(s)} b^*_{p} \big)+\text{h.c.} \Big] \\ 
		&\hspace{0.5cm} + \frac{N^{\kappa}}{N}  \int_{0}^1ds \sum_{p \in \Lambda^*_+} \widehat{V} (p/N^{1-\kappa})\eta_p \big[ d_p^{(s)}d_{-p}^{(s)} +\text{h.c.} \big].
		\end{split}\]
Here, we use the same notation as in \eqref{eq:cKterms}. Proceeding as above then results in
		\[\begin{split}
		 \pm \text{F}_2 &\leq  C N^{2\kappa-\a/2-1} (\cN_++1),\hspace{0.5cm}\pm i[\cN_{\leq cN^\gamma}, \text{F}_2 ] \leq C N^{2\kappa-\a/2-1} (\cN_++1).
		 \end{split}\]
We omit the details. 

Finally, let's consider the contribution $\text{F}_3$, defined in \eqref{eq:G2-deco}. As in \cite{BBCS}, we split it into
		\begin{equation}\label{eq:G23split}
		\begin{split}
		\text{F}_3 =\; &\frac{1}{2} \sum_{p \in \Lambda^*_+} N^{\kappa} \widehat{V} (p/N^{1-\kappa}) \Big[ \big( \gamma_p b_p + \sigma_p b_{-p}^* \big) \big( \gamma_p b_{-p} + \sigma_p b_p^* \big) +\text{h.c.}\Big]  \\
		&+ \frac{1}{2}  \sum_{p \in \Lambda^*_+} N^{\kappa}\widehat{V} (p/N^{1-\kappa}) \, \Big[ ( \g_p b_p+ \s_p b^*_{-p}) \, d_{-p} +
		d_p\, (\g_p b_{-p} + \s_p b^*_{p}) +\text{h.c.}\Big] \\
		&+\frac{1}{2}  \sum_{p \in \Lambda^*_+}N^{\kappa} \widehat{V} (p/N^{1-\kappa}) \big[ d_p d_{-p} + \text{h.c.}\big] \\
		=:  &\,\text{F}_{31}  + \text{F}_{32}+\text{F}_{33} 
\end{split} \end{equation}
and we start with the analysis of $ \text{F}_{31}$. The latter can be written as
		\begin{equation}\label{eq:fin-G23}
		\begin{split}
		\text{F}_{31}
		= \; & \frac{1}{2}\sum_{p \in \Lambda^*_+} N^{\kappa}\widehat{V} (p/N^{1-\kappa}) (b_p b_{-p}+b^*_p b^*_{-p}) +\sum_{p \in P_H} N^{\kappa}\widehat{V} (p/N^{1-\kappa}) \eta_p \frac{N-\cN_+}{N} + \cE^V_2\\
		\end{split} \end{equation}
with
		\begin{equation*}%\label{eq:G23-deco1}
		\begin{split}
		\cE_2^V= \; &\frac{1}{2}\sum_{p \in P_H} N^{\kappa}\widehat{V} (p/N^{1-\kappa}) \Big[ (\gamma_p^2-1) b_p b_{-p}+ \sigma_p^2b_{-p}^* b_p^*+2 \sigma_p \gamma_p b_{p}^* b_{p}  \\
 		&\hspace{4cm} -N^{-1} \gamma_p \sigma_p a^*_p a_p+ (\gamma_p \sigma_p - \eta_p) \frac{N-\cN_+}{N} \Big] +\text{h.c.}
		\end{split} \end{equation*}
Let us recall here that $\gamma_p =1$ and $\sigma_p = 0$ for $p \in P_H^c$. To control $ \cE_{2}^V$ and its commutator with $ \cN_{\leq cN^\gamma}$, we use once more the estimates $|\g_p^2-1|\leq C\eta_p^2$ and $|\s_p|\leq C |\eta_p|$ for all $p \in P_H$, and apply Lemmas \ref{lm:dp} and \ref{lm:commdpNres} to deduce with Cauchy-Schwarz as above that
		\[\begin{split}
		 \pm \cE_{2}^V &\leq  C N^{2\kappa-\a/2} (\cN_++1),\hspace{0.5cm}\pm i[\cN_{\leq cN^\gamma}, \cE_{2}^V ] \leq C N^{2\kappa-\a/2} (\cN_++1).
		 \end{split}\]

The analysis of the contributions $ \text{F}_{32}$ and $\text{F}_{33}$ in \eqref{eq:G23split} is slightly more tedious. We start with the term $ \text{F}_{32}$ and rewrite it similarly as in \cite[Eq. (7.39)]{BBCS} as
		\begin{equation}\label{eq:splitF32}
		\begin{split}
		\text{F}_{32} = \; & \frac{1}{2}  \sum_{p \in \Lambda^*_+} N^{\kappa}\widehat{V} (p/N^{1-\kappa}) \, \left[  \g_p b_pd_{-p}+\g_pd_p b_{-p}+ \s_p b^*_{-p} d_{-p}  +\s_pd_p  b^*_{p} +\text{h.c.} \right]  \\
		=: \; & \sum_{j=1}^4 \big(\text{F}_{32j} +\text{h.c.}\big). 
		\end{split} \end{equation}
As explained in \cite[Eq. (7.39) \& (7.40)]{BBCS}, massaging a bit the first term $\text{F}_{321}$ yields
		\begin{equation*} %\label{eq:F321}
		\text{F}_{321} = - \sum_{p \in P_H} N^{\kappa}\widehat{V} (p/N^{1-\kappa}) \eta_p \, \left(\frac{N-\cN_+}{N}\right)\left(\frac{\cN_+ +1}{N}\right) +\cE_4^V \end{equation*}
where $\cE_4^V = \cE_{41}^V + \cE_{42}^V + \cE_{43}^V +\text{h.c.} $ is defined through
		\begin{equation} \label{eq:F321error}
		\begin{split}
		\cE^V_{41} = \; &\frac{1}{2}  \sum_{p \in \Lambda^*_+} N^{\kappa}\widehat{V} (p/N^{1-\kappa}) \, (\g_p - 1) b_p d_{-p} \, , \qquad
		\cE_{42}^V = \frac{1}{2} \sum_{p \in \L_+^*} N^{\kappa}\widehat{V} (p/N^{1-\kappa}) b_p \bar{d}_{-p} \\
		\cE_{43}^V = \; &- \frac{1}{2} \sum_{p \in P_H} N^{\kappa}\widehat{V} (p/N^{1-\kappa}) \eta_p \frac{\cN_+ + 1}{N} (b_p^* b_p - N^{-1} a_p^* a_p ).
		\end{split} \end{equation}
Here, we set $\bar{d}_{-p} = d_{-p} +   \eta_H (p)  (\cN_+/N) b_p^*$. The error terms $ \cE^V_{41}$ and $ \cE^V_{43}$ are easily controlled with the same arguments as above, using the pointwise bounds $|\g_p^2-1|\leq C\eta_p^2$ and $|\s_p|\leq C |\eta_p|$ for all $p \in P_H$, the bound $ \|\eta_H\|\leq CN^{\kappa-\alpha/2}$ and by applying \eqref{eq:d-bds} in Lemma \ref{lm:dp} as well as \eqref{eq:ngdp} in Lemma \ref{lm:commdpNres}. This results in 
		\[\begin{split}
		\pm  ( \cE^V_{41} +  \cE^V_{43}) \leq CN^{2\kappa-\alpha/2} (\cN_++1), \hspace{0.5cm} \pm  i [\cN_{\leq cN^\gamma},  \cE^V_{41} +  \cE^V_{43}] \leq CN^{2\kappa-\alpha/2} (\cN_++1).
		\end{split}\]
The remaining error $\cE^V_{42}$ reads in position space 
		\[\cE^V_{42} = \frac{1}{2} N^\kappa \int_{\Lambda^2}  dxdy\; N^{3-3\kappa}V(N^{1-\kappa}(x-y)) \check{b}_x \check{\bar{d}}_y. \]
We can compare this to the position space representation of the error term $ \cE_{23}^K$ in \eqref{eq:bdbarpspaceterm1}. Notice that $ \cE_{42}^V$ is, up to the uniformly bounded factor $ f_N(x-y)$, equal to the first term on the r.h.s. of \eqref{eq:bdbarpspaceterm1}. Thus, if we proceed exactly as in \eqref{eq:bdbarpspaceterm2} and \eqref{eq:bdbarpspaceterm3}, we find that
		\[\begin{split} 
		\pm \cE_{42}^V  \leq &\;C N^{2\kappa-\alpha/2} (\cH_N + 1), \hspace{0.5cm}\pm i [\cN_{\leq cN^\gamma}, \cE_{42}^V]  \leq \; C N^{2\kappa-\alpha/2 +\gamma/2} (\cH_N+1). 
		\end{split}\]
		
Going back to \eqref{eq:splitF32}, consider now the second term $ \text{F}_{322}$. We decompose $ \gamma_p = 1+ (\gamma_p-1)$ and control the resulting term that contains $ (\gamma_p-1)$ by Cauchy-Schwarz, using that $ |\gamma_p-1|\leq C\eta_p^2$. The other term can be estimated by switching to position space. Indeed, in position space this term can be bounded exactly as the error term $ \cE_{32}^K$ in \eqref{eq:dbarbpspaceterm1} and \eqref{eq:dbarbpspaceterm2} in the proof of the last proposition. Altogether, one finds that
		\[\begin{split} 
		\pm  \text{F}_{322} \leq &\;C N^{2\kappa-\alpha/2} (\cH_N + 1),\hspace{0.5cm} \pm i [\cN_{\leq cN^\gamma}, \text{F}_{322}]  \leq \; C N^{2\kappa-\alpha/2 +\gamma/2} (\cH_N+1). 
		\end{split}\]

Finally, the two remaining contributions $ \text{F}_{323}$ and $ \text{F}_{324}$, defined in \eqref{eq:splitF32}, can be controlled using Lemma \ref{lm:dp} and Lemma \ref{lm:commdpNres}. By \eqref{eq:d-bds}, we have for instance that
		\begin{equation*} %\label{eq:F322}
		\begin{split}
		|\langle\xi,\text{F}_{323}\xi\rangle| %\leq \; &C N^{\kappa}\sum_{p \in P_H} |\eta_p| \| b_{-p} \xi\|\| d_{-p} \xi\| \\ 
		&\leq C  N^{\kappa} \sum_{p \in P_H} |\eta_p| \| b_{-p} \xi\| \left[ |\eta_p| \| (\cN_+ + 1)^{1/2} \xi \| + \| \eta_H \| \| b_{-p} \xi \| \right] \\ 
		&\leq C N^{3\kappa - 5\a/2} \| (\cN_+ + 1)^{1/2} \xi \|^2  
		\end{split} \end{equation*}
and, similarly, by \eqref{eq:ngdp} that
		\begin{equation*} %\label{eq:F322}
		\begin{split}
		|\langle\xi,[\mathcal{N}_{\leq cN^{\gamma}},\text{F}_{323}]\xi\rangle| %\leq \; &  C N^{\kappa}\sum_{p \in P_H} |\eta_p| [\| [\mathcal{N}_{\leq cN^{\gamma}},b_{-p}] \xi\|\| d_{-p} \xi\| + \norm{b_{-p}\xi}\norm{[\mathcal{N}_{\leq cN^{\gamma}},d_{-p}]\xi}] \\ 
		&\leq C  N^{\kappa} \sum_{p \in P_H} |\eta_p|  \| b_{-p} \xi\| \left[ |\eta_p| \| (\cN_+ + 1)^{1/2} \xi \| + \| \eta_H \| \| b_{-p} 		\xi \| \right] \\ 
		&\leq C N^{3\kappa - 5\a/2} \| (\cN_+ + 1)^{1/2} \xi \|^2.  \end{split} \end{equation*}
To control the remaining term $ \text{F}_{324}$, on the other hand, a simple analysis (using the same arguments as above) shows that it is enough to bound 
		\[ \cE_{5}^V :=  \frac{1}{2}  \sum_{p \in \Lambda^*_+} N^{\kappa}\widehat{V} (p/N^{1-\kappa})\eta_p \big[  d_p  b^*_{p} +\text{h.c.}\big]. \]	
To apply Lemma \ref{lm:dp} and Lemma \ref{lm:commdpNres} in the usual way, we observe first of all
		\[\sum_{p \in \Lambda^*_+} \big|\widehat{V} (p/N^{1-\kappa})  \eta_p \big| \leq \bigg(\sum_{p \in \Lambda^*_+} |p|^{-2}\big|\widehat{V} (p/N^{1-\kappa}) \big|^2\bigg)^{1/2} \bigg(\sum_{p \in \Lambda^*_+} |p|^2 \eta_p^2\bigg)^{1/2} \]
and, by Plancherel, that 
		\[\sum_{p \in \Lambda^*_+} p^{-2} \big|\widehat{V}(p/N^{1-\kappa})\big|^2 \leq C \sum_{\substack{ p\in \Lambda_+^*:\\|p| \leq N^{1-\kappa}}}| p|^{-2} + N^{2\kappa -2} \sum_{\substack{ p\in \Lambda_+^*:\\|p| \geq N^{1-\kappa}}} \big|\widehat{V}(p/N^{1-\kappa})\big|^2 \leq C N^{1-\kappa}.\]	
Together with \eqref{eq:H1eta}, this shows that 
		\begin{equation} \label{eq:Vetapspacebnd}
		\sum_{p \in \Lambda^*_+} \big|\widehat{V} (p/N^{1-\kappa})  \eta_p \big| \leq CN.
		\end{equation}	
Now, applying \eqref{eq:d-bds} in Lemma \ref{lm:dp} proves that
		\[\begin{split}
		|\langle\xi, \cE_{5}^V \xi \rangle|  &\leq \;  CN^{\kappa-1}   \sum_{p \in P_H} \big|\widehat{V} (p/N^{1-\kappa}) \big||\eta_p|\|(\cN_++1)^{1/2}\xi \|  \\ 
		&\hspace{0.5cm} \times\left[ |\eta_p| \| (\cN_+ + 1)^{3/2} \xi \| + \| \eta_H \| \| (\cN_+ + 1)^{1/2} \xi \| + \| \eta_H \| \| a_p (\cN_+ + 1) \xi \| \right] \\ &\leq \; C N^{2\kappa-\a/2} \|(\cN_++1)^{1/2}\xi\|^{2}
		\end{split}\]
and, referring to \eqref{eq:ngdp} in Lemma \ref{lm:commdpNres}, also that 
		\[\begin{split}
		|\langle\xi, [\cN_{\leq cN^\gamma}, \cE_{5}^V] \xi \rangle|  &\leq \; C N^{2\kappa-\a/2} \|(\cN_++1)^{1/2}\xi\|^{2}.
		\end{split}\]
Altogether, we arrive at 
		\[\begin{split} 
		\pm  \text{F}_{32} \leq &\;C N^{2\kappa-\alpha/2} (\cH_N + 1), \hspace{0.5cm}\pm i [\cN_{\leq cN^\gamma}, \text{F}_{32}]  \leq \; C N^{2\kappa-\alpha/2 +\gamma/2} (\cH_N+1). 
		\end{split}\]
		
To complete the proof of the proposition, we still need to analyse the last term $\text{F}_{33}$ in \eqref{eq:G23split}. This term, however, can be analysed exactly as the error $\cE_{51}^K$ in \eqref{eq:ddxspacebnd1} and \eqref{eq:ddxspacebnd2}, after switching to position space. One finds 
		\[\begin{split} 
		\pm  \text{F}_{33} \leq &\;C N^{2\kappa-\alpha/2} (\cH_N + 1), \hspace{0.5cm}\pm i [\cN_{\leq cN^\gamma}, \text{F}_{33}]  \leq \; C N^{2\kappa-\alpha/2 +\gamma/2} (\cH_N+1), 
		\end{split}\]
and collecting all the bounds on $ \text{F}_1, \text{F}_2 $ and $\text{F}_3$, defined in \eqref{eq:G2-deco}, proves the bounds \eqref{eq:errorVc} and \eqref{eq:errorVccomm}. 	
\end{proof}

Let us finish this section and summarize the results of Prop. \ref{prop:K} and Prop. \ref{prop:G2V}.
 \begin{prop}\label{prop:G2}
 There exists a constant $C > 0$ such that
 		\begin{equation*} %\label{eq:def-EK}
		\begin{split}
		\cG_{N}^{(2)} =&\; \cK + \sum_{p \in P_{H}} \Big[p^2 \eta_p^2 + N^{\kappa}\widehat{V} (p/N^{1-\kappa}) \eta_p\Big]\Big(\frac{N-\cN_+}{N}\Big) \Big(\frac{N-\cN_+ -1}{N}\Big)\\
		&+ \sum_{p \in P_{H}} p^2 \eta_p \big( b^*_p b^*_{-p} + b_p b_{-p} \big) +\sum_{p \in \Lambda^*_+}  N^{\kappa}\widehat{V} (p/N^{1-\kappa})  a^*_pa_p \frac{N-\cN_+}{N} \\ &+ \frac{1}		{2}\sum_{p \in \Lambda^*_+}N^{\kappa} \widehat{V} (p/N^{1-\kappa}) \big( b_p b_{-p}+ b_{-p}^* b_p^*\big) + \cE_{N}^{(2)}
		\end{split}
		\end{equation*}
where the self-adjoint operator $ \cE_{N}^{(2)} $ satisfies
		\begin{align}
		\nonumber\pm \cE^{(2)}_{N} & \leq C N^{3\kappa -\alpha/2} (\cH_N + 1),\\
		\nonumber\pm i[\mathcal{N}_{\leq cN^{\gamma}},\cE^{(2)}_{N}] &\leq C \big(N^{3\kappa -\alpha/2} + N^{2\kappa - \a/2+\gamma/2}\big) (\cH_N + 1), 
		%\\
		%\nonumber\pm \left[ f (\cN_+/M), \left[ f (\cN_+ /M) ,\cE^{(2)}_{N} \right] \right] & \leq C  %M^{-2} \|f'\|^2_{\infty} N^{3\kappa-\alpha/2} \big( \cH_N + 1 \big)
		\end{align}
for all $\alpha \geq  2\kappa$ with $ \alpha+\kappa\leq 1$, and for all $ 0\leq \gamma\leq \alpha$, $c\geq 0$, $f$ smooth and bounded, $M \in \bN$ and $N \in \bN$ large enough.
\end{prop}

%%%%%%%%%%%%%%%%%%%%%%%%%%%%%%%%%%%%%%%%%%%%%%%%%%%%%%%%%%%%%%%%%
%%%%%%%%%%%%%%%%%%%%%%%%%%%%%%%%%%%%%%%%%%%%%%%%%%%%%%%%%%%%%%%%%
%%%%%%%%%%%%%%%%%%%%%%%%%%%%%%%%%%%%%%%%%%%%%%%%%%%%%%%%%%%%%%%%%
%%%%%%%%%%%%%%%%%%%%%%%%%%%%%%%%%%%%%%%%%%%%%%%%%%%%%%%%%%%%%%%%%
%%%%%%%%%%%%%%%%%%%%%%%%%%%%%%%%%%%%%%%%%%%%%%%%%%%%%%%%%%%%%%%%%
\subsubsection{Analysis of $ \cG_{N}^{(3)}$} \label{sub:G3}

In this section, we analyse the operator $\cG_{N}^{(3)} $ which, by (\ref{eq:cLNj}), is equal to
\be \label{eq:G3N2}
 \cG_{N}^{(3)} = \frac{1}{\sqrt{N}} \sum_{p,q \in \L^*_+ : p + q \not = 0} N^{\kappa}\widehat{V} (p/N^{1-\kappa}) e^{-B(\eta_H)} b^*_{p+q} a^*_{-p} a_q e^{B(\eta_H)} + \hc  \ee

\begin{prop}\label{prop:GN-3}
There exists a constant $C > 0$ such that
\begin{equation}\label{eq:def-E3}
\cG_{N}^{(3)} = \frac{1}{\sqrt{N}} \sum_{p,q \in \L^*_+ : p + q \not = 0} N^{\kappa}\widehat{V} (p/N^{1-\kappa}) \left[ b_{p+q}^* a_{-p}^* a_q  + \emph{h.c.} \right] + \cE^{(3)}_{N} \end{equation}
where the self-adjoint operator $ \cE_N^{(3)}$ satisfies
		\begin{align}
		\label{eq:lm-GN31} \pm \cE_{N}^{(3)} &\leq C N^{2\kappa-\a/2} \big(\cH_N + 1\big), \\
		\label{eq:lm-GN31comm} \pm i[\mathcal{N}_{\leq cN^{\gamma}},\cE_{N}^{(3)}] &\leq C N^{2\kappa-\a/2+\gamma/2} \big(\cH_N + 1\big),
		%\\
		%\label{eq:E3Cff} \pm [f (\cN_+/M) , [ f (\cN_+ / M) , \cE_{N}^{(3)}]] &\leq C M^{-2} \|f'\|%^2_{\infty}N^{2\kappa-\a/2} \big(\cH_N + 1\big)
		\end{align}
for all $\alpha \geq  2\kappa$ with $ \alpha+\kappa\leq 1$, and for all $ 0\leq \gamma\leq \alpha$, $c\geq 0$, $f$ smooth and bounded, $M \in \bN$ and $N \in \bN$ large enough.
\end{prop}
\begin{proof} Let us indicate the main steps to prove \eqref{eq:lm-GN31} and \eqref{eq:lm-GN31comm}. %, the last bound \eqref{eq:E3Cff} being proved as \eqref{eq:lm-GN31}. 
To this end, we follow the proof of \cite[Prop. 7.5]{BBCS} which shows that $\cE_{N}^{(3)} $ in \eqref{eq:def-E3} takes the form
		\begin{equation} \label{eq:deco-cE3} \begin{split}
		\cE^{(3)}_{N} &=  \frac{N^{\kappa}}{\sqrt{N}} \sum_{p,q \in \Lambda^*_+ : p+q \not = 0} \widehat{V} (p/N^{1-\kappa}) \big((\g_{p+q}-1) b^*_{p+q} + \s_{p+q} b_{-p-q} + d_{p+q}^* \big) \, a_{-p}^* a_q  \\
		&\hspace{.3cm}+ \frac{N^{\kappa}}{\sqrt{N}} \sum_{\substack{ p,q \in \Lambda_+^* ,\\ p+q \not = 0}} \widehat{V} (p/N^{1-\kappa}) \eta_H (p) \, e^{-B(\eta_H)}b^*_{p+q}e^{B(\eta_H)} \int_0^1 ds\, e^{-sB(\eta_H)} b_{p} b_{q} e^{sB(\eta_H)}\\
		&\hspace{.3cm}+ \frac{N^{\kappa}}{\sqrt{N}} \sum_{\substack{p,q \in \Lambda_+^* , \\p+q \not = 0}} \widehat{V} (p/N^{1-\kappa}) \eta_H (q) \, e^{-B(\eta_H)} b^*_{p+q}e^{B(\eta_H)} \int_0^1 ds\,  e^{-sB(\eta_H)}b_{-p}^*b^*_{-q} e^{sB(\eta_H)}  \\
		&\hspace{.3cm}+ \hc \\ &=:  \; \cE^{(3)}_1 + \cE_2^{(3)} + \cE_3^{(3)} + \hc
		\end{split} \end{equation}
Let us consider the three terms $\cE_1^{(3)}, \cE_2^{(3)}, \cE_3^{(3)}$ separately and explain why they all satisfy  (\ref{eq:lm-GN31}) and \eqref{eq:lm-GN31comm}. Starting with $\cE_1^{(3)}$, it is useful to split it into
		\[ \begin{split}
		\cE^{(3)}_1 =\; &  \frac{N^{\kappa} }{\sqrt{N}} \sum_{p,q \in \Lambda^*_+ : p+q \not = 0}\widehat{V} (p/N^{1-\kappa}) \Big[(\g_{p+q}-1) b^*_{p+q} + \s_{p+q} b_{-p-q} + d_{p+q}^* \Big]  a_{-p}^* a_q \\
		=: & \; \cE^{(3)}_{11} + \cE^{(3)}_{12} +\cE^{(3)}_{13}.
		\end{split}\]
The first two terms $ \cE_{11}^{(3)}$ and $ \cE_{12}^{(3)}$ can be bounded using Cauchy-Schwarz, the fact that $|\g_{p+q}-1|\leq C\eta_{p+q}^2$ and $\| \eta_H \| \leq C N^{\kappa -\a/2}$. One proceeds similarly as in the proof of Proposition \ref{prop:G2} and obtains for instance that
		\[\begin{split}
		|\langle \xi, \cE^{(3)}_{11} \xi \rangle|
		& \leq  \frac{CN^{\kappa}}{\sqrt{N}} \sum_{p,q \in \Lambda^*_+ : p+q \not = 0} |\widehat{V} (p/N^{1-\kappa})| | \eta_H (p+q)|^2 \, \| b_{p+q} a_{-p} \xi \| \| a_q \xi \| \\
		& \leq   \frac{C N^{\kappa}}{\sqrt{N}} \bigg(\sum_{p,q \in \Lambda^*_+ : p+q \not = 0}   |\eta_H (p+q)|^2 \, \| a_{-p}(\cN_++1)^{1/2} \xi \|^2 \bigg)^{1/2}\\ &\hspace{4cm} \times  \bigg(\sum_{p,q \in \Lambda^*_+ : p+q \not = 0} |\eta_H (p+q)|^2   \| a_q \xi \|^2 \bigg)^{1/2}\\
		& \leq C N^{\kappa}\| \eta_H \|^2 \| (\cN_++1)^{1/2} \xi \|^2 \leq C N^{3\kappa-\a} \| (\cN_+ + 1)^{1/2} \xi \|^2
		\end{split} \]
Arguing similarly for $ \cE_{12}^{(3)}$ as well as the commutator of $\cN_{\leq cN^\gamma}$ with $ \cE_{11}^{(3)}$ and $ \cE_{12}^{(3)}$ (recall that $ [\cN_{\leq cN^\gamma}, b^{\sharp}_p] = F(\sharp) \chi_p b^\sharp_p$ with $ F(*)=1, F(\cdot)=-1$ and where $\chi\in\ell^{2}(\Lambda_+^*)$ denotes the characteristic function of $\{p\in\Lambda_+^*: |p|\leq cN^\gamma\}$), we find that 
		\[\pm ( \cE_{11}^{(3)} +  \cE_{12}^{(3)}) \leq CN^{2\kappa-\alpha/2} (\cN_++1), \hspace{0.5cm} \pm i[\cN_{\leq cN^\gamma},  \cE_{11}^{(3)} +  \cE_{12}^{(3)}] \leq CN^{2\kappa-\alpha/2}(\cN_++1).\]
As for $ \cE_{13}^{(3)}$, we follow \cite{BBCS} and rewrite $d^*_{p+q}= \bar{d}^*_{p+q} - \frac{(\cN_++1)}{N} \eta_H (p+q) b_{-p-q}$. A simple bound as above then shows that it is enough to control the term involving $\bar{d}^*_{p+q} $, i.e.
		\[ \begin{split} %\label{eq:G3-P13}
		\widetilde{\cE}^{(3)}_{13}  & = \frac{N^\kappa}{\sqrt N}\sum_{p,q \in \Lambda^*_+ : p+q \not = 0} \widehat{V} (p/N^{1-\kappa}) \, \bar d^*_{p+q}   a^*_{-p} a_q \\
		& = \frac{N^\kappa}{\sqrt N}\int_{\L^2}  dx dy\; N^{3-3\kappa} V(N^{1-\kappa}(x-y)) \check{\bar d}^*_x \check{a}^*_y \check{a}_x. 
		\end{split}\]
We bound this term using \eqref{eq:splitdbd} in Lemma \ref{lm:dp} and find that
		\begin{equation*}\begin{split} 
		|\langle \xi, \widetilde{\cE}^{(3)}_{13} \xi \rangle|  & \leq \frac{N^\kappa}{\sqrt N} \int_{\L^2}  dx dy N^{3-3\kappa} V(N^{1-\kappa}(x-y))  \| \check{a}_x \xi \| \|\check{a}_y \check{\bar d}_x  \xi \|  \\
 		& \leq  C\frac{N^\kappa}{\sqrt N} \| \eta_H \| \int_{\L^2}  dx dy N^{3-3\kappa} V(N^{1-\kappa}(x-y))\| \check{a}_x \xi \|\\ 		
		& \hspace{2.5cm} \times  \Big[N^{-1}|\check{\eta}_H(x-y)|\|(\cN_++1)\xi \| + \| \check{a}_x (\cN_+ + 1)^{1/2} \xi \|  \\
		&\hspace{3cm}+ \| \check{a}_y (\cN_+ + 1)^{1/2} \xi \| + \| \check{a}_x \check{a}_y \xi \| \Big]  \\
 		%& \leq  C N^{\kappa/2}\| \eta_H \| \| \cN_+^{1/2} \xi \| \Big[  \int_{\L^2}  dx dy N^{2-2\kappa} V(N^{1-\kappa}(x-y)) \\ & \hspace{3.5cm} \times  \big[|\|(\cN_++1)\xi \|^2 + \| \check{a}_x (\cN_+ + 1)^{1/2} \xi \|^2 + \| \check{a}_x \check{a}_y \xi \|^2 \big] \Big]^{1/2}   \\
 		%& \leq C N^{\kappa/2}\| \eta_H \| \| \cN_+^{1/2} \xi \| \big[ N^{\kappa/2}\| (\cN_+ + 1)^{1/2} \xi \| + \| \cV_N^{1/2} \xi \| \big] \\
 		& \leq C N^{2\kappa-\a/2} \| (\cN_+ + 1)^{1/2} \xi \|^2 +N^{3\kappa/2-\a/2} \| (\cN_++1)^{1/2}\xi\| \|\cV_N^{1/2} \xi \|.
 		\end{split}\end{equation*}
The commutator with $ \cN_{\leq cN^\gamma}$ is controlled similarly. With $ \check{\chi}_x \in L^2(\Lambda)$ taking values $\check{\chi}_x(y) =\check{\chi}(y-x) $, and $ \chi\in \ell^2(\Lambda_+^*)$ denoting the characteristic function of $ \{ p\in\Lambda_+^*: |p|\leq cN^\gamma \}$, we recall in particular the bounds \eqref{eq:xbndchigastetaH}, \eqref{eq:VNchixaybnd} as well as the identity $$\int_\Lambda dx\, a^*(\check{\chi}_x)a(\check{\chi}_x) = \sum_{p\in\Lambda_+^*: |p|\leq cN^\gamma} a^*_pa_p. $$
Together with Cauchy-Schwarz and the bound \eqref{eq:ngsplitdbd} from Lemma \ref{lm:commdpNres}, we find that
		\begin{equation}\begin{split} \label{eq:GN3-P13tl}
		 |\langle \xi, [ \cN_{\leq cN^\gamma} ,\widetilde{\cE}^{(3)}_{13} ] \xi \rangle| %& \leq \frac{N^\kappa}{\sqrt N} \int_{\L^2}  dx dy N^{3-3\kappa} V(N^{1-\kappa}(x-y))  \| \check{a}_x \xi \| \|\check{a}_y \check{\bar d}_x  \xi \|  \\
 		& \leq  C N^{2\kappa-\alpha/2} \int_{\L^2}  dx dy N^{3-3\kappa} V(N^{1-\kappa}(x-y))\Big[ \| \check{a}_x \xi \| + \| a(\check{\chi}_x)\xi\|\Big]\\ 		
		& \hspace{0.8cm} \times  \Big[ C \|(\cN_++1)^{1/2}\xi \| + \| \check{a}_x  \xi \|  + \| \check{a}_y  \xi \| +\| a(\check{\chi}_x)  \xi \| +\| a(\check{\chi}_y) \xi \|\\
		&\hspace{1.3cm} +N^{-1/2}\| a(\check{\chi}_y) \check{a}_x \xi \| +N^{-1/2} \| a(\check{\chi}_x) \check{a}_y \xi \|+ N^{-1/2}\| \check{a}_x \check{a}_y \xi \| \Big]  \\
 		%& \leq  C N^{\kappa/2}\| \eta_H \| \| \cN_+^{1/2} \xi \| \Big[  \int_{\L^2}  dx dy N^{2-2\kappa} V(N^{1-\kappa}(x-y)) \\ & \hspace{3.5cm} \times  \big[|\|(\cN_++1)\xi \|^2 + \| \check{a}_x (\cN_+ + 1)^{1/2} \xi \|^2 + \| \check{a}_x \check{a}_y \xi \|^2 \big] \Big]^{1/2}   \\
 		%& \leq C N^{\kappa/2}\| \eta_H \| \| \cN_+^{1/2} \xi \| \big[ N^{\kappa/2}\| (\cN_+ + 1)^{1/2} \xi \| + \| \cV_N^{1/2} \xi \| \big] \\
 		& \leq C N^{2\kappa-\a/2} \| (\cN_+ + 1)^{1/2} \xi \|^2 + CN^{2\kappa-\a/2+\gamma/2} \| (\cK+1)^{1/2}\xi\| ^2 \\
		&\hspace{0.5cm} +CN^{3\kappa/2-\a/2} \| (\cN_++1)^{1/2}\xi\| \| \cV_N^{1/2} \xi \|.
 		\end{split}\end{equation}
Collecting the previous bounds on $ \cE_{13}^{(3)}, \cE_{23}^{(3)} $ and $\cE_{33}^{(3)}$, we summarize that 
		\begin{equation} \label{eq:G3N-P1end}
		\begin{split}
		\pm \cE^{(3)}_1  &\leq C N^{2\kappa -\a/2} (\cH_N + 1), \hspace{0.5cm}\pm i[\mathcal{N}_{\leq cN^{\gamma}},\cE^{(3)}_1]  \leq C N^{2\kappa - \a/2 + \gamma/2} (\cH_N+1).
		\end{split}
		\end{equation}

We continue with the analysis of the second error term $\cE^{(3)}_2$, defined in (\ref{eq:deco-cE3}). Following \cite[Eq. (7.50)]{BBCS}, we rewrite this term as 
\begin{equation}\label{eq:cE32-deco} \begin{split}
\cE^{(3)}_2 =\; & \frac{1}{\sqrt{N}} \sum_{p,q \in \Lambda_+^* , p+q \not = 0} N^{\kappa}\widehat{V} (p/N^{1-\kappa}) \eta_H (p) \, e^{-B( \eta_H )}b^*_{p+q}e^{B( \eta_H)}\\
& \hskip 1cm \times \int_0^1 ds\, \big( \g^{(s)}_p  \g^{(s)}_{q} b_{p} b_q + \s^{(s)}_p  \s^{(s)}_{q}b^*_{-p} b^*_{-q} +\g^{(s)}_p   \s^{(s)}_{q} b^*_{-q} b_{p} +  \s^{(s)}_p  \g^{(s)}_{q}  b^*_{-p} b_{q}  \big)  \\
& +\frac{1}{\sqrt{N}} \sum_{p,q \in \Lambda_+^* , p+q \not = 0}N^{\kappa} \widehat{V} (p/N^{1-\kappa}) \eta_H (p) \, e^{-B( \eta_H )}b^*_{p+q}e^{B( \eta_H)} \int_0^1 ds\,  \g^{(s)}_p  \s^{(s)}_{q} [b_{p},  b^*_{-q}]  \\
& + \frac{1}{\sqrt{N}} \sum_{p,q \in \Lambda_+^* , p+q \not = 0} N^{\kappa}\widehat{V} (p/N^{1-\kappa})  \eta_H (p) \, e^{-B( \eta_H )}b^*_{p+q}e^{B( \eta_H )}\\
& \hskip 1cm\times \int_0^1 ds\, \Big[ d^{(s)}_p \big( \g^{(s)}_{q} b_{q} + \s^{(s)}_{q} b^*_{-q} \big) + \big( \g^{(s)}_p b_{p} + \s^{(s)}_p b^*_{-p} \big)   d^{(s)}_{q} +   d^{(s)}_{p}  d^{(s)}_{q}   \Big]\\
=: &\; \cE^{(3)}_{21} + \cE^{(3)}_{22} + \cE^{(3)}_{23}
\end{split}\end{equation}
Let us recall here that for any $s \in [0;1]$ and $p \in \L_+^*$, we write $\gamma^{(s)}_p = \cosh (s \eta_H (p))$, $\s^{(s)}_p = \sinh (s \eta_H (p))$ and $d^{(s)}_p$ is defined as in (\ref{eq:defD}) (with $\eta$ replaced by $s \eta_H$).

The operators $\cE^{(3)}_{21} $ and $\cE^{(3)}_{22}$ as well as their commutators with $ \cN_{\leq cN^\gamma}$ can be controlled by applying Cauchy-Schwarz and using the bounds \eqref{eq:modetap}, \eqref{eq:etaHL2} on $\eta_H$ together with Lemmas \ref{lm:dp} and \ref{lm:commdpNres}. We omit the details and summarize that this results in
		\[\begin{split}
		\pm \big(\cE^{(3)}_{21} + \cE^{(3)}_{22}\big)  &\leq C N^{2\kappa -\a/2} (\cN_+ + 1), \\
		\pm i\big[\mathcal{N}_{\leq cN^{\gamma}},\cE^{(3)}_{21} + \cE^{(3)}_{22}\big] & \leq C N^{2\kappa - \a/2 } (\cN_++1).
		\end{split}\]

Hence, let's switch to the last term on the r.h.s. of (\ref{eq:cE32-deco}). Since $ |\gamma_p^{(s)}-1|\leq C\eta_p^2$ and $ |\sigma_p^{(s)}|\leq C\eta_p$, uniformly in $s\in [0;1]$, the usual Cauchy-Schwarz bounds, together with the Lemmas \ref{lm:dp} and \ref{lm:commdpNres}, imply that it indeed suffices to consider $ \widetilde{\cE}_{23}^{(3)}$, defined by
		\[\widetilde{\cE}_{23}^{(3)} := \frac{N^{\kappa}}{\sqrt{N}}\int_0^1 ds\, \sum_{\substack{ p,q \in \Lambda_+^* ,\\ p+q \not = 0}} \widehat{V} (p/N^{1-\kappa})  \eta_H (p) \, e^{-B(\eta_H)}b^*_{p+q}e^{B(\eta_H)} \big[   b_{p}   \bar{d}^{(s)}_{q} +   d^{(s)}_{p}  d^{(s)}_{q}   \big], \] 
while 
		\[\begin{split}
		\pm \big(\cE^{(3)}_{23} -\widetilde{\cE}_{23}^{(3)} \big)  &\leq C N^{2\kappa -\a/2} (\cN_+ + 1), \\
		\pm i\big[\mathcal{N}_{\leq cN^{\gamma}},\cE^{(3)}_{23} -\widetilde{\cE}_{23}^{(3)}\big] & \leq C N^{2\kappa - \a/2 } (\cN_++1).
		\end{split}\]
Recall here the notation that $ \bar{d}^{(s)}_{q} = d^{(s)}_{q} + (\cN_+/N)s \eta_H(q)b^*_{-q}$. To control the term $ \widetilde{\cE}_{23}^{(3)}$, on the other hand, we switch to position space where $ \widetilde{\cE}_{23}^{(3)}$ takes the form
		\begin{equation}\label{eq:wtE233xspace} \begin{split}
		\widetilde{\cE}^{(3)}_{23} =\;& \frac{N^{\kappa}}{\sqrt N}\int_0^1 ds\, \int_{\L^3} dx dy dz\,  N^{3-3\kappa} V(N^{1-\kappa}(x-z))  \check{\eta}_H (z-y)    \\
		& \hskip 2.5cm\times  e^{-B(\eta_H)} \check{b}^*_xe^{B(\eta_H)} \Big[   \check{b}_x \check{ \bar{d}}^{(s)}_x    +  \check{ d}^{(s)}_y \check{ d}^{(s)}_x  \Big].
		\end{split}\end{equation}
By Cauchy-Schwarz and Lemma \ref{lm:dp}, we find that
		\[\begin{split} 
		|\langle \xi , \widetilde{\cE}^{(3)}_{23} \xi \rangle | \leq \; & C N^\kappa\| \eta_H \| \int_{\Lambda^3} dx dy dz \, N^{3-3\kappa} V(N^{1-\kappa}(x-z))   | \check{\eta}_H (y-z)|\, \| \check{b}_x e^{B( \eta_H)}  \xi \| \\ &\hspace{2cm} \times \, \Big[ N^{-1/2}\| \check{a}_{x} \check{a}_{y} \xi \| + \| (\cN_++1)^{1/2} \xi \| + \| \check{a}_x  \xi \| + \| \check{a}_y  \xi \| \Big] \\
		%\leq \; & \frac{C N^{\kappa}\| \eta_H \|^2}{\sqrt{N}} \|  \cN_+^{1/2}  \xi \| \| (\cN_++ 1) \xi \| \\
		\leq \; &C N^{3\kappa -\alpha} \| (\cN_+ + 1)^{1/2} \xi \|^2.
		\end{split} \]
To control $[ \cN_{\leq cN^\gamma},\widetilde{\cE}_{23}^{(3)}]$, we expand $ e^{-B(\eta_H)} \check{b}^*_x e^{B(\eta_H)} = b^*(\check{\gamma}_x) + b(\check{\sigma}_x) + \check{d}^*_x $ s.t.		\begin{equation}\label{eq:commNresbxconj}
		 \big[ \cN_{\leq cN^\gamma}, b^*(\check{\gamma}_x) + b(\check{\sigma}_x) + \check{d}^*_x\big] = b^*(\check{\chi}_x)  + b^*(\text{p}_{x}) + b(\text{r}_{x}) + [\cN_{\leq cN^\gamma},\check{d}^*_x].  
		 \end{equation}
Here, $ \check{\chi}_x\in L^2(\Lambda)$ takes values $ \check{\chi}_x(y) = \check{\chi}(y-x)$, with $\chi\in\ell^2(\Lambda_+^*)$ denoting the characteristic function of the set $ \{ p\in\Lambda_+^*: |p|\leq cN^\gamma \}$. Moreover, $ \text{p}_{x} \in L^2(\Lambda)$ denotes the inverse Fourier transform of $ ( (\gamma_p-1)\chi_pe^{-ipx})_{p\in \Lambda_+^*}\in\ell^2(\Lambda_+^*)$ and $ \text{r}_{x} \in L^2(\Lambda)$ denotes the inverse Fourier transform of $ (\sigma_p\chi_pe^{-ipx})_{p\in \Lambda_+^*}\in\ell^2(\Lambda_+^*)$. In particular, we have that
		\[ \sup_{x\in \Lambda}\| \text{p}_x\| \leq C\|\eta_H\|^2\leq CN^{2\kappa-\alpha},\hspace{0.5cm}  \sup_{x\in \Lambda}\| \text{r}_x\|\leq C\|\eta_H\| \leq CN^{\kappa-\alpha/2} \]
by Plancherel's theorem. If we then use the estimates \eqref{eq:ngdxy-bds}, \eqref{eq:ngsplitdbd} and \eqref{eq:ngddxy} from Lemma \ref{lm:commdpNres}, we obtain similarly to the previous bound that
		\begin{equation}\begin{split} \label{eq:E233-commNres}
		 |\langle \xi, [ \cN_{\leq cN^\gamma} , \widetilde{\cE}_{23}^{(3)} ] \xi \rangle| & \leq  C N^{2\kappa-\alpha/2} \int_{\L^2}  dx dydz \; N^{3-3\kappa} V(N^{1-\kappa}(x-z))| \check{\eta}_H (y-z)|\\
		 &\hspace{4.5cm}\times \Big[ \| \check{a}_x \xi \| + \| a(\check{\chi_x})\xi\| + \| (\cN_++1)^{1/2}\xi\| \Big]\\	
		& \hspace{1cm} \times  \Big[ C \|(\cN_++1)^{1/2}\xi \| + \| \check{a}_x  \xi \|  + \| \check{a}_y  \xi \| +\| a(\check{\chi}_x)  \xi \| +\| a(\check{\chi}_y) \xi \|\\
		&\hspace{1.6cm} +N^{-1/2}\| a(\check{\chi}_y) \check{a}_x \xi \| +N^{-1/2} \| a(\check{\chi}_x) \check{a}_y \xi \|+ N^{-1/2}\| \check{a}_x \check{a}_y \xi \| \Big]  \\
 		%& \leq  C N^{\kappa/2}\| \eta_H \| \| \cN_+^{1/2} \xi \| \Big[  \int_{\L^2}  dx dy N^{2-2\kappa} V(N^{1-\kappa}(x-y)) \\ & \hspace{3.5cm} \times  \big[|\|(\cN_++1)\xi \|^2 + \| \check{a}_x (\cN_+ + 1)^{1/2} \xi \|^2 + \| \check{a}_x \check{a}_y \xi \|^2 \big] \Big]^{1/2}   \\
 		%& \leq C N^{\kappa/2}\| \eta_H \| \| \cN_+^{1/2} \xi \| \big[ N^{\kappa/2}\| (\cN_+ + 1)^{1/2} \xi \| + \| \cV_N^{1/2} \xi \| \big] \\
 		& \leq CN^{3\kappa-\alpha} \| (\cN_++1)^{1/2}\xi\| ^2.
 		\end{split}\end{equation}
Now, let's collect the bounds on $ \cE_{21}^{(3)}$, $ \cE_{22}^{(3)}$ and $ \cE_{23}^{(3)}$, defining $ \cE_2^{(3)}$ in Eq. \eqref{eq:cE32-deco}, so that 
		\begin{equation}\label{eq:cE32f}\begin{split}
		\pm \cE^{(3)}_{2}   &\leq C N^{2\kappa -\a/2} (\cN_+ + 1), \hspace{0.5cm} \pm i\big[\mathcal{N}_{\leq cN^{\gamma}},\cE^{(3)}_{2} \big]  \leq C N^{2\kappa - \a/2 } (\cN_++1).
		\end{split}\end{equation}
		
Finally, going back to (\ref{eq:deco-cE3}), it remains to consider the error term $ \cE_3^{(3)}$ or, equivalently, its adjoint. Similarly as in \cite{BBCS}, we write the adjoint as 
		\begin{equation*} \begin{split} %\label{eq:G3N-P2}
\cE^{(3)*}_3 %=\; & \frac{1}{\sqrt{N}} \sum_{p,q \in \Lambda_+^* , p+q \not = 0}N^{\kappa} \widehat{V} (p/N^{1-\kappa})  \eta_H (q) \int_0^1 ds\, e^{-sB(\eta_H)} b_{-q}e^{sB( \eta_H )}\\
%& \hskip 1cm \times  \big( \g^{(s)}_p b_{-p} + \s^{(s)}_p b^*_{p} +  d^{(s)}_{-p}\big) \big( \g_{p+q} b_{p+q} + \s_{p+q} b^*_{-p-q} +  d_{p+q}\big) \\
%
 =\; &\frac{N^{\kappa}}{\sqrt{N}} \sum_{\substack{ p,q \in \Lambda_+^*,\\ p+q \not = 0 }} \widehat{V} (p/N^{1-\kappa})  \eta_H (q) \, \int_0^1 ds\,e^{-sB(\eta_H)} b_{-q}e^{sB( \eta_H)}\\
&  \times \Big[ \, \g^{(s)}_p  \g_{p+q} b_{-p} b_{p+q} + \s^{(s)}_p  \s_{p+q} b^*_{p} b^*_{-p-q}  + \g^{(s)}_p  \s_{p+q} b^*_{-p-q} b_{-p}+  \g_{p+q}  \s^{(s)}_p b^*_{p} b_{p+q} \\
& \hskip 1.5cm +  d^{(s)}_{-p} \big( \g_{p+q} b_{p+q} + \s_{p+q} b^*_{-p-q}\big) +  \big( \g^{(s)}_p b_{-p} + \s^{(s)}_p b^*_{p}\big)  \bar{d}_{p+q} +   d^{(s)}_{-p}  d_{p+q}\Big]  \\
& +\frac{N^{\kappa}}{\sqrt{N}} \sum_{\substack{ p,q \in \Lambda_+^* ,\\ p+q \not = 0}} \widehat{V} (p/N^{1-\kappa})  \eta_H (q) \, \int_0^1 ds\,e^{-sB( \eta_H )} b_{-q}e^{sB( \eta_H)}  \\
&\hspace{1.3cm}\times \Big[  \g^{(s)}_p  \s_{p+q} [b_{-p},b^*_{-p-q}] - \big( \g^{(s)}_p b_{-p} + \s^{(s)}_p b^*_{p}\big) (\cN_+/N)\eta_H(p+q)b^*_{-p-q} \Big]
\\ =: & \, \cE_{31}^{(3)} + \cE_{32}^{(3)} 
\end{split} \end{equation*}
The operator $\cE^{(3)}_{32} $ and its commutator with $ \cN_{\leq cN^\gamma}$ can be controlled by Cauchy-Schwarz in momentum space, using the bounds \eqref{eq:modetap}, \eqref{eq:etaHL2} on $\eta_H$ together with Lemmas \ref{lm:dp} and \ref{lm:commdpNres}. The bounds are analogous to, for instance, \cite[Eq. (7.54)]{BBCS} and we obtain
		\[\begin{split}
		\pm \cE^{(3)}_{32}  &\leq C N^{3\kappa -\alpha } (\cN_+ + 1), \hspace{0.5cm}\pm i\big[\mathcal{N}_{\leq cN^{\gamma}},\cE^{(3)}_{32} \big]  \leq C N^{3\kappa - \alpha} (\cN_++1).
		\end{split}\]
The error term $\cE_{31}^{(3)}$, on the other hand, reads in position space
		\[\begin{split}
 		\cE_{31}^{(3)} &= \frac{N^\kappa}{\sqrt N}\int_0^1 ds \int_{\Lambda^2} dx dy \, N^{3-3\kappa} V(N^{1-\kappa}(x-y)) \, e^{-sB( \eta_H)} b(\check{\eta}_{H,x}) e^{sB( \eta_H)} \\
		& \hskip 2cm \times \, \Big[  b(\check{ \g}^{(s)}_x) b(\check{ \g}_y) + b^*(\check{ \s}^{(s)}_x) b^*(\check{ \s}_y) +b^*(\check{ \s}_y) b(\check{ \g}^{(s)}_x)+ b^*(\check{ \s}^{(s)}_x) b(\check{ \g}_y)  \\
		& \hskip 3.3 cm + \check{ d}_x^{(s)} \big(  b(\check{ \g}_y) + b^*(\check{ \s}_y) \big) + \big(  b(\check{ \g}^{(s)}_x) + b^*(\check{ \s}^{(s)}_x) \big) \check{ \bar{d}}_y +\check{ d}_x^{(s)}\check{ d}_y\Big]
		\end{split} \]
and its analysis is quite similar to that of the error term $ \widetilde{\cE}_{23}^{(3)}$, defined in position space in \eqref{eq:wtE233xspace}. Together with Lemma \ref{lm:dp}, Cauchy-Schwarz implies
		\begin{equation*} \begin{split}  \label{eq:GN3-P2} 
		|\langle \xi , \cE^{(3)}_{31} \xi \rangle | \leq \; & N^{2\kappa-\alpha/2} \int_0^1  ds\,  \int_{\Lambda^2} dx dy \, N^{3-3\kappa} V(N^{1-\kappa}(x-y)) \|(\cN_++1)^{1/2}  \xi \| \\ 
		&\hspace{3cm} \times  \Big[  N^{-1/2}\| \check{b}_{x}\check{b}_{y} \xi \| +  \| \check{b}_x  \xi \| + \| \check{b}_y  \xi \| + \| (\cN_++1)^{1/2} \xi \| \Big]\\
		\leq \;& N^{2\kappa-\alpha/2}\| (\cN_+ +1)^{1/2} \xi \|^2  + N^{3\kappa/2-\a/2} \| (\cN_+ +1)^{1/2} \xi \|\| \cV_N^{1/2}\xi \|.
		\end{split} \end{equation*}
To control $ [\cN_{\leq cN^\gamma},\cE^{(3)}_{31}]$, we use the identity \eqref{eq:commNresbxconj}, the decomposition 
	\[ e^{-sB( \eta_H)} b(\check{\eta}_{H,x}) e^{sB( \eta_H)} = \sum_{p\in\Lambda_+^*} \big[ \eta_H(p)\gamma_p e^{-ipx} b_p + \eta_H(p)\sigma_p e^{-ipx} b_{-p}^* +\eta_H(p) e^{-ipx}d_p\big]  \] 
and, as a consequence of \eqref{eq:ngdp} in Lemma \ref{lm:commdpNres}, the upper bound
		\[ \sup_{x\in \Lambda} \| [ \cN_{\leq cN^\gamma}, e^{-sB( \eta_H)} b(\check{\eta}_{H,x}) e^{sB( \eta_H)}] \xi\| \leq C \|\eta_H\| \|(\cN_++1)^{1/2}  \xi \|\leq CN^{2\kappa-\alpha/2}.\]
Proceeding then similarly to \eqref{eq:GN3-P13tl}, we omit further details and summarize that
		\[|\langle \xi , [\cN_{\leq cN^\gamma},\cE^{(3)}_{31}] \xi \rangle | \leq C N^{2\kappa -\a/2 + \gamma/2}\|(\cH_N+1)^{1/2}\xi \|^2. \]
In conclusion, the bounds on $ \cE_{31}^{(3)} $ and $ \cE_{32}^{(3)}$ show that
		\begin{equation}\label{eq:G3N-e3f} \pm \cE^{(3)}_3 \leq C N^{2\kappa - \a/2} (\cH_N + 1), \hspace{0.5cm} \pm i[\mathcal{N}_{\leq cN^{\gamma}}, \cE^{(3)}_3] \leq C N^{2\kappa -\a/2 + \gamma/2}(\cH_N+1)
		\end{equation}
Combining \eqref{eq:G3N-P1end}, \eqref{eq:cE32f}, \eqref{eq:G3N-e3f} with \eqref{eq:deco-cE3} concludes the proof.
\end{proof}

%%%%%%%%%%%%%%%%%%%%%%%%%%%%%%%%%%%%%%%%%%%%%%%%%%%%%%%%%%%%%%%%%
%%%%%%%%%%%%%%%%%%%%%%%%%%%%%%%%%%%%%%%%%%%%%%%%%%%%%%%%%%%%%%%%%
%%%%%%%%%%%%%%%%%%%%%%%%%%%%%%%%%%%%%%%%%%%%%%%%%%%%%%%%%%%%%%%%%
%%%%%%%%%%%%%%%%%%%%%%%%%%%%%%%%%%%%%%%%%%%%%%%%%%%%%%%%%%%%%%%%%
%%%%%%%%%%%%%%%%%%%%%%%%%%%%%%%%%%%%%%%%%%%%%%%%%%%%%%%%%%%%%%%%%
\subsubsection{Analysis of $ \cG_{N}^{(4)}$} \label{sub:G4}
In this section, we analyse $ \cG_{N}^{(4)} = e^{-B(\eta_H)} \cL^{(4)}_{N} e^{B(\eta_H)} $, with $\cL^{(4)}_N$ as defined in (\ref{eq:cLNj}).
\begin{prop}\label{prop:GN-4}
There exists a constant $C > 0$ such that
		\begin{equation*} \begin{split} %\label{eq:def-E4}
		\cG_{N}^{(4)} %= \; &e^{-B(\eta_H)} \cL^{(4)}_{N} e^{B(\eta_H)} \\ 
		= \; &\cV_N + \frac{1}{2N} \sum_{\substack{q \in \Lambda^*_+, r\in \Lambda^* \\ q,\, q+ r \in P_H}} N^{\kappa}\widehat{V} (r/N^{1-\kappa}) \eta_{q+r} \eta_q \left( 1-\frac{\cN_+ }{N} \right) \left( 1 - \frac{\cN_+ +1}{N} \right)  \\ 
		&+ \frac{1}{2N} \sum_{\substack{q \in \Lambda^*_+, r \in \L^*: \\ q +r \in P_H}} N^{\kappa}\widehat{V} (r/N^{1-\kappa}) \, \eta_{q+r} \left(  b_q b_{-q} + b^*_q b^*_{-q} \right)  + \cE^{(4)}_{N},		\end{split} \end{equation*}
where the self-adjoint operator $ \cE_N^{(4)}$ satisfies
		\begin{align}\label{eq:E4bound2}
  		\pm \cE_{N}^{(4)} &\leq C N^{2\kappa- \a/2} \big( \cH_N +1 \big),\\
		\label{eq:E4boundcomm}
		 \pm i[\mathcal{N}_{\leq cN^{\gamma}},\cE_{N}^{(4)}] &\leq C N^{2\kappa- \a/2+\gamma/2} \big( \cH_N +1 \big),
		 %\\
%		\label{eq:E4Cff}
%		\pm [f (\cN_+/M), [f (\cN_+/M),\cE_{N}^{(4)}]] &\leq C M^{-2} \|f'\|^2_{\infty}N^{2\kappa- \a/2} %\big(\cH_N + 1\big)
		\end{align}
for all $\alpha \geq  2\kappa$ with $ \alpha+\kappa\leq 1$, and for all $ 0\leq \gamma\leq \alpha$, $c\geq 0$, $f$ smooth and bounded, $M \in \bN$ and $N \in \bN$ large enough.
\end{prop}
For the proof of Prop. \ref{prop:GN-4} we need a slight extension of \cite[Lemma 7.7]{BBCS} to our setting.

\begin{lemma}\label{lm:prel4}
Let $\eta_H \in \ell^2 (\L^*_+)$ be defined as in (\ref{eq:defetaH}) and assume that $\alpha \geq  2\kappa$ with $ \alpha+\kappa\leq 1$, $ 0\leq \gamma\leq \alpha$ as well as $c\geq 0$. Moreover, let $ \chi \in \ell^{2}(\Lambda_+^*)$ denote the characteristic function of the set $ \{p\in\Lambda_+^*: |p|\leq cN^\gamma\}$ and define $ \check{\chi}_x\in L^2(\Lambda)$ s.t. $ \check{\chi}_x(y) = \check{\chi}(y-x)$, for all $x,y\in\Lambda$. Then, there exists a constant $C > 0$ such that
		\be \begin{split} \label{eq:prel4-2}
		&\| (\cN_++1)^{n/2} e^{-B(\eta_H)} \check{b}_x \check{b}_y  e^{B( \eta_H)}\xi \| \\
		&\hspace{1cm}\leq  C  \Big [\;  N \| (\cN_+ +1)^{n/2} \xi \| + \| \check{a}_y (\cN_+ +1)^{(n+1)/2} \xi \| \\
		&\hspace{2cm}  + \| \check{a}_x (\cN_+ +1)^{(n+1)/2} \xi \|   +  \| \check{a}_x \check{a}_y (\cN_++1)^{n/2} \xi \| \Big ]
		\end{split}\ee
and such that 
		\be \begin{split} \label{eq:prel4-3}
		&\| (\cN_++1)^{n/2} [\cN_{\leq cN^\gamma}, e^{-B(\eta_H)} \check{b}_x \check{b}_y  e^{B( \eta_H)}]\xi \| \\
		&\hskip0.5cm\leq C  \Big [\;  N \| (\cN_+ +1)^{n/2} \xi \| + \| \check{a}_y (\cN_+ +1)^{(n+1)/2} \xi \| + \| \check{a}_x (\cN_+ +1)^{(n+1)/2} \xi \| \\
		&\hskip 1.5cm  +  \| a(\check{\chi}_y)  (\cN_++1)^{(n+1)/2} \xi \| +  \| a(\check{\chi}_x)  (\cN_++1)^{(n+1)/2} \xi \|\\
		&\hskip1.5cm + \| a(\check{\chi}_y) \check{a}_x (\cN_+ +1)^{n/2} \xi \| + \| a(\check{\chi}_x) \check{a}_y (\cN_+ +1)^{n/2} \xi \| \\
		&\hskip1.5cm+  \| \check{a}_x \check{a}_y (\cN_++1)^{n/2} \xi \| \Big ]
		\end{split}\ee
for all $\xi \in \cF_+^{\leq N}$ and $n \in \bZ$.  

More generally, given any $f\in L^2(\Lambda)$ and $x\in\Lambda$, denote by $f_x \in L^2(\Lambda)$ the function with values $f_x(y) = f(y-x)$, for all $y\in\Lambda$. Then, for $f,g\in L^2(\Lambda)$, we have that
		\begin{equation}\label{eq:prel4-4}
		\begin{split}
		&\| (\cN_++1)^{n/2} e^{-B(\eta_H)} b^\sharp(f_x) b^\flat(g_y)  e^{B( \eta_H)}\xi \|  \\
		&\hspace{2cm}\leq  C  \|f\|\|g\|  \| (\cN_++1)^{(n+2)/2}\xi\|, \\
		&\| (\cN_++1)^{n/2} e^{-B(\eta_H)} b(f_x) \check{b}_y e^{B( \eta_H)} \xi \| \\
		&\hspace{2cm}\leq C  \|f\|   \| (\cN_++1)^{(n+2)/2}\xi\|+ C  \|f\|   \| \check{a}_y(\cN_++1)^{(n+1)/2}\xi\|,
		\end{split}
		\end{equation}
where $ (\sharp, \flat)\in \{*,\cdot\}^2$. Similarly, for the commutator with $\cN_{\leq cN^\gamma}$, we have that 
		\begin{equation}\label{eq:prel4-5}
		\begin{split}
		&\| (\cN_++1)^{n/2} [\cN_{\leq cN^\gamma }, e^{-B(\eta_H)} b^*(f_x) b^*(g_y)  e^{B( \eta_H)}]\xi \| \\
		&\hspace{1.5cm} \leq  C  \|f\|\|g\|  \| (\cN_++1)^{(n+2)/2}\xi\|, \\
		&\| (\cN_++1)^{n/2} [\cN_{\leq cN^\gamma }, e^{-B(\eta_H)} b (f_x) \check{b}_y e^{B( \eta_H)}]\xi \| \\
		&\hspace{1.5cm}\leq  C  \|f\|  \| (\cN_++1)^{(n+2)/2}\xi\| + C  \|f\|   \| a(\check{\chi}_y)(\cN_++1)^{(n+1)/2}\xi\| \\
		&\hspace{2cm} +C  \|f\|   \| \check{a}_y(\cN_++1)^{(n+1)/2}\xi\|.
		\end{split}
		\end{equation}
\end{lemma}
\begin{proof}
For simplicity, consider the case $n=0$; the general case follows along the same lines. The proof of \eqref{eq:prel4-2} and \eqref{eq:prel4-3} follows as in \cite[Lemma 7.7]{BBCS}. We simply expand
		\[ e^{-B(\eta)} \check{b}_x \check{b}_y  e^{B( \eta)} = \big( \check{b}_x + b (\text{p} _x) + b^*(\check{\s}_x) + \check{d}_x \big) \big( \check{b}_y + b(\text{p}_y) + b^*(\check{\s}_y) + \check{d}_y\big)  \]
and consider different cases. Here, $ \text{p}_x \in L^2(\Lambda)$ denotes the inverse Fourier transform of $((\gamma_p-1)\chi_p e^{-ipx})_{p\in \Lambda_+^*}\in\ell^2(\Lambda_+^*)$ whose norm satisfies $\sup_{x\in\Lambda}\| \text{p}_x\|\leq C\|\eta_H\|^2\leq C $. Using this and the results of Lemmas \ref{lm:dp} and \ref{lm:commdpNres} proves the bounds \eqref{eq:prel4-2} and \eqref{eq:prel4-3}. 

The first bound in \eqref{eq:prel4-4} is a direct consequence of Lemma \ref{lm:Ngrow} and the second bound in \eqref{eq:prel4-4} follows from Lemma \ref{lm:Ngrow} and \eqref{eq:dxy-bds}, after expanding $e^{-B(\eta_H)}\check{b}_y  e^{B( \eta)}  $ as above. 

Finally, let's consider the two commutator bounds in \eqref{eq:prel4-5} and let's start with the second bound. Here, it is useful to expand 
		\[ e^{-B(\eta_H)} b (f_x)e^{B( \eta_H)} = \sum_{p\in\Lambda_+^*} \big( \widehat{f}_p \gamma_p e^{ipx} b_p + \widehat{f}_p \sigma_p e^{ipx} b^*_{-p}+\widehat{f}_p  e^{ipx} d_p\big)  \] 
so that 
		\[ \big[ \cN_{\leq cN^\gamma}, e^{-B(\eta_H)} b (f_x)e^{B( \eta_H)} \big] = \sum_{p\in\Lambda_+^*} \big( - \widehat{f}_p \chi_p \gamma_p e^{ipx} b_p + \widehat{f}_p \chi_p \sigma_p e^{ipx} b^*_{-p}+\widehat{f}_p  e^{ipx} [ \cN_{\leq cN^\gamma}, d_p]\big). \]
In particular, using that $ f\in L^2(\Lambda)$ and the bounds \eqref{eq:d-bds}, \eqref{eq:ngdp}, we have that
		\[ \big\|  \big[ \cN_{\leq cN^\gamma}, e^{-B(\eta_H)} b (f_x)e^{B( \eta_H)} \big] \xi\big\| \leq C\|f\| \| (\cN_++1)^{1/2}\xi \| \]
for any $\xi\in \cF_+^{\leq N}$. Using this bound, expanding the factor $ e^{-B(\eta_H)}\check{b}_y e^{B( \eta_H)}$ in position space as in the first step and using \eqref{eq:ngdxy-bds} then proves the second bound in \eqref{eq:prel4-5}. For the first bound in \eqref{eq:prel4-5}, we expand $e^{-B(\eta_H)} b^*(f_x) b^*(g_y)  e^{B( \eta_H)}$ into
		\[\begin{split}
		 &e^{-B(\eta_H)} b^*(f_x) b^*(g_y)  e^{B( \eta_H)}\\
		 &=  \Big( b^*( \text{ch}(f)_x) +b^*( \text{sh}(f)_x)+ \sum_{p \in \Lambda^*_+}  \widehat{f}_p e^{ipx} d^*_p\Big) \Big( b^*( \text{ch}(g)_y)  +b^*( \text{sh}(g)_y)+ \sum_{q \in \Lambda^*_+}  \widehat{g}_q e^{ipy} d^*_q\Big).
		  \end{split}\] 
Here, we define $\text{ch}(f) \in L^2(\Lambda)$ and $ \text{sh}(f) \in L^2(\Lambda)$ through their Fourier coefficients $ \widehat{\text{ch}(f)}(p) = \widehat{f}_p\gamma_p$ and $ \widehat{\text{ch}(f)}(p) = \widehat{f}_p\gamma_p$, for all $p\in\Lambda_+^*$. In particular, for any $f\in L^2(\Lambda)$,
		\[\sup_{x\in \Lambda} \| \text{ch}(f)_x \| \leq C\|f\|, \hspace{0.5cm}\sup_{x\in \Lambda} \| \text{sh}(f)_x \| \leq C\|f\|. \]
To derive the first bound in \eqref{eq:prel4-5}, we then proceed as in the first step with the only difference that, if the commutator $[\cN_{\leq cN^\gamma}, \cdot] $ hits one of the $d^*_p $ or $d^*_q$ operators, we need to use the commutator expansion from Lemma \ref{lm:indu}, similarly as in the proof of Lemma \ref{lm:commdpNres}, and control each term of the expansion. Since by assumption $ f,g\in L^2(\Lambda)$, this can be done as above and we omit further details.	 
\end{proof}

\begin{proof}[Proof of Prop. \ref{prop:GN-4}] We proceed as in \cite[Eq. (7.58) \& (7.59)]{BBCS} and decompose the operator $ \cG_N^{(4)}$ into $  \cG^{(4)}_{N}  =\cV_N+  \text{W}_1 + \text{W}_2 + \text{W}_3 + \text{W}_4, $ where 
		\begin{equation}\label{eq:defW}
		\begin{split}
		\text{W}_1 = \; & \frac{N^{\kappa}}{2N} \sum_{q \in \Lambda_+^* , r \in \Lambda^* : r \not = -q} \widehat{V} (r/N^{1-\kappa}) \eta_H ({q+r}) \int_0^1 ds \Big [e^{-s B(\eta_H)} b_{q} b_{-q}\,e^{sB(\eta_H)}   + \text{h.c.} \Big] \\
  		\text{W}_2 = \; & \frac{N^{\kappa}}{N} \sum_{\substack{ p,q \in \Lambda_+^* ,\\ r \in \Lambda^* : r \not = p,-q}} \widehat{V} (r/N^{1-\kappa})\,  \eta_H ({q+r}) \int_0^1 ds\,  \Big[  e^{-s B(\eta_H)} b^*_{p+r} b^*_{q} e^{s B( \eta_H )}  a^*_{-q-r} a_p + \text{h.c.} \Big]\\
		\text{W}_3 = \; &   \frac{N^{\kappa}}{N} \sum_{p,q\in \Lambda^*_+, r \in \Lambda^* : r \not = -p -q} \widehat{V} (r/N^{1-\kappa}) \eta_H ({q+r}) \eta_H (p) \, \\ &\hspace{1.5cm} \times  \int_0^1 ds\,\int_0^s d\t \, \Big[e^{-sB( \eta_H)} b^*_{p+r} b^*_q e^{sB( \eta_H)} e^{-\t B( \eta_H)} b^*_{-p} b^*_{-q-r} e^{\t B( \eta_H )}+ \hc \Big] \\
		\text{W}_4 = \; & \frac{N^{\kappa}}{N} \sum_{p,q\in \Lambda^*_+, r \in \Lambda^* : r \not = -p -q} \widehat{V} (r/N^{1-\kappa}) \, \eta^2_H ({q+r}) \\ &\hspace{1.5cm} \times \int_0^1 ds\,\int_0^s d\t \,  \Big[ e^{-sB( \eta_H)} b^*_{p+r} b^*_q e^{sB( \eta_H)} e^{-\t B( \eta_H)} b_{p} b_{q+r} e^{\t B( \eta_H)} + \hc  \Big].
		\end{split} \end{equation}
We analyse $ \text{W}_1$ to $\text{W}_4$ separately and start with $\text{W}_1$. Setting $\gamma^{(s)}_q = \cosh (s \eta_H (q))$, $\s_q^{(s)} = \sinh (s \eta_H (q))$ and recalling that $d_q^{(s)}$ is defined as in (\ref{eq:defD}), with $\eta$ replaced by $s \eta_H$, we may proceed as in \cite[(7.59) \& (7.61)]{BBCS} and find that
		\begin{equation}\label{eq:W1-dec}
		\begin{split}
		\text{W}_1 = \; & \frac{N^{\kappa}}{2N} \sum_{q \in \Lambda_+^* , r \in \Lambda^* : r \not = -q}  \widehat{V} (r/N^{1-\kappa}) \eta_H (q+r) \int_0^1 ds (\g^{(s)}_q)^2  (b_q b_{-q} + \hc ) \\ 
		&+ \frac{1}{2N} \sum_{q \in \Lambda_+^* , r \in \Lambda^* : r \not = -q} N^{\kappa} \widehat{V} (r/N^{1-\kappa}) \eta_H (q+r) \int_0^1 ds \, \g_q^{(s)} \s^{(s)}_q \big( [b_q , b_q^*] + \hc \big) \\
		&+ \frac{N^{\kappa}}{2N} \sum_{q \in \Lambda_+^* , r \in \Lambda^* : r \not = -q} \widehat{V} (r/N^{1-\kappa}) \eta_H (q+r) \int_0^1 ds  \, \g_q^{(s)} \big( b_q d_{-q}^{(s)}+\hc \big) + \cE_{10}^{(4)} \\
		=: & \; \text{W}_{11} + \text{W}_{12} + \text{W}_{13} + \cE^{(4)}_{10}.
		\end{split}
		\end{equation}
Here, the operator $ \cE_{10}^{(4)} = \sum_{j=1}^5\cE_{10j}^{(4)} $ is defined through
		\begin{equation}\label{eq:cE410} \begin{split} \cE^{(4)}_{101} &= \frac{N^{\kappa}}{2N} \sum_{\substack{ q \in \Lambda_+^* ,\\ r \in \Lambda^* : r \not = -q }}\!\!\!\! \widehat{V} (r/N^{1-\kappa}) \eta_H (q+r) \int_0^1 ds \Big[ 2 \gamma_q^{(s)} \s_q^{(s)} b_q^* b_q  + (\s_q^{(s)})^2 b_{-q}^* b_q^* +\hc \Big] \\
		 \cE^{(4)}_{102} &= \frac{N^{\kappa}}{2N} \sum_{q \in \Lambda_+^* , r \in \Lambda^* : r \not = -q} \widehat{V} (r/N^{1-\kappa}) \eta_H (q+r) \int_0^1 ds \, \sigma_q^{(s)} \big[ b_{-q}^* d_{-q}^{(s)} + \hc \big]  \\
  		\cE^{(4)}_{103} &= \frac{N^{\kappa}}{2N} \sum_{q \in \Lambda_+^* , r \in \Lambda^* : r \not = -q} \widehat{V} (r/N^{1-\kappa}) \eta_H (q+r) \int_0^1 ds \, \sigma_q^{(s)} \big[ d^{(s)}_q b_q^* + \hc \big] \\
 		\cE^{(4)}_{104} &= \frac{N^{\kappa}}{2N} \sum_{q \in \Lambda_+^* , r \in \Lambda^* : r \not = -q} \widehat{V} (r/N^{1-\kappa}) \eta_H (q+r) \int_0^1 ds \, \g_q^{(s)} \big[ d^{(s)}_q b_{-q} + \hc \big] \\
   		\cE^{(4)}_{105} &= \frac{N^{\kappa}}{2N} \sum_{q \in \Lambda_+^* , r \in \Lambda^* : r \not = -q}  \widehat{V} (r/N^{1-\kappa}) \eta_H (q+r) \int_0^1 ds \big[ d^{(s)}_q d^{(s)}_{-q} + \hc \big].
		\end{split} \end{equation}
Let us start with the analysis of the operators in \eqref{eq:cE410}. To control them and to control their commutators with $ \cN_{\leq cN^\gamma}$, we will use the two pointwise bounds
		\begin{equation}\label{eq:VetaN}
		\sup_{q \in \Lambda_+^*}  \sum_{r \in \L_+^*} |\widehat{V} (r/N^{1-\kappa}) \eta_{q+r}|  \leq CN, \sum_{\substack{ q \in \L_+^*,\\ r \in \L^*, r \not = -q}} \hspace{-0.2cm}|\widehat{V} (r/N^{1-\kappa}) \eta_H (q+r) \eta_H (q)| \leq CN^2.
		\end{equation}
Here, $C>0$ denotes a constant which is independent of $N\in\mathbb{N}$. The pointwise estimates in \eqref{eq:VetaN} can be proved, with minor modifications, like the pointwise bound \eqref{eq:Vetapspacebnd}. 

Applying \eqref{eq:d-bds} from Lemma \ref{lm:dp}, the terms $\cE^{(4)}_{101}, \cE^{(4)}_{102} $ and $\cE^{(4)}_{103}$ can all be bounded in the usual way by Cauchy-Schwarz. By \eqref{eq:VetaN}, we have for instance that 
		\[ \begin{split} 
		|\langle \xi, \cE_{103}^{(4)} \xi \rangle | %& \leq \frac{C N^{\kappa}\| (\cN_+ + 1)^{1/2} \xi \|}{N} \sum_{q \in \Lambda_+^* , r \in \Lambda^* : r \not = -q} |\widehat{V} (r/N^{1-\kappa})| |\eta_H (q+r)| |\eta_H (q)|\\ 
		%&\hspace{5.5cm} \times \left[ |\eta_q | \| b_q^* \xi \| + N^{-1} \| \eta_H \| \| b_q b_q^* \cN^{1/2}_+ \xi \| \right] \\ 
		& \leq C N^{\kappa-1} \| (\cN_+ + 1)^{1/2} \xi \|  \sum_{q \in \Lambda_+^* , r \in \Lambda^* : r \not = -q} |\widehat{V} (r/N^{1-\kappa})| |\eta_H (q+r)| |\eta_H (q)| \\ 
		&\hspace{4cm} \times \left[ \big(|\eta_q | + N^{-1} \| \eta_H \|\big)\| (\cN_++1)^{1/2} \xi \| + \| \eta_H \| \| b_q  \xi \| \right] \\ 
		&\leq C N^{2\kappa - \a/2} \| (\cN_++1)^{1/2} \xi \|^2.
		\end{split} \]
Proceeding similarly for $\cE^{(4)}_{101}$ and $ \cE^{(4)}_{102} $, we find that 
		\[ \pm \big( \cE^{(4)}_{101} + \cE^{(4)}_{102}+\cE^{(4)}_{103}\big) \leq C N^{2\kappa - \a/2} (\cN_++1).  \]
Similarly, if we use that $ [\cN_{\leq cN^\gamma}, b^\sharp(f)] = F(\sharp) b^\sharp(f)$ for any $f\in L^2(\Lambda)$ and with $ F(*)=1, F(\cdot)=-1$, and if we use the bound \eqref{eq:ngdp} to commute $\cN_{\leq cN^\gamma} $ with $d^{(s)}_q$, we find that 
		\[ \pm  i\big[\cN_{\leq cN^\gamma}, \cE^{(4)}_{101} + \cE^{(4)}_{102}+\cE^{(4)}_{103}\big] \leq C N^{2\kappa - \a/2} (\cN_++1).  \]

As for $ \cE^{(4)}_{104} $ and $ \cE^{(4)}_{105}$, it is useful to switch to position space. Following \cite{BBCS}, we first split $ \cE^{(4)}_{104}$ into $\cE_{104}^{(4)} = \cE_{1041}^{(4)} + \cE_{1042}^{(4)} +\text{h.c.}$, where
		\[ \begin{split} \cE_{1041}^{(4)} &= \frac{N^{\kappa}}{2N} \sum_{q \in \Lambda_+^* , r \in \Lambda^* : r \not = -q} \widehat{V} (r/N^{1-\kappa}) \eta_H (q+r) \int_0^1 ds \, (\g_q^{(s)} -1) d^{(s)}_q b_{-q} \\
		\cE_{1042}^{(4)} &= \frac{N^{\kappa}}{2N} \sum_{q \in \Lambda_+^* , r \in \Lambda^* : r \not = -q} \widehat{V} (r/N^{1-\kappa}) \eta_H (q+r) \int_0^1 ds  \, d^{(s)}_q b_{-q}.
		\end{split} \]
Using that $|(\g_q^{(s)} -1)|\leq C\eta_H(q)^2 $ for all $q\in \Lambda_+^*$ and arguing as for the error terms $\cE^{(4)}_{101}, \cE^{(4)}_{102} $ and $\cE^{(4)}_{103}$, a straight forward computation shows that
		\[ \pm   \cE^{(4)}_{1041} \leq C N^{4\kappa - 3\a} (\cN_++1),\hspace{0.5cm} \pm  i\big[\cN_{\leq cN^\gamma}, \cE^{(4)}_{1041}\big] \leq C N^{4\kappa -3\a } (\cN_++1). \] 
To deal with $\cE^{(4)}_{1042} $, on the other hand, we go to position space and apply \eqref{eq:dxy-bds} s.t.
		\[\begin{split}
		| \langle \xi , \cE_{1042}^{(4)} \xi \rangle | &= \Big|\frac{1}{2} \int_0^1 ds \int_{\Lambda^2} dx dy N^{2-2\kappa} V(N^{1-\kappa}(x-y)) \check{\eta}_H (x-y) \langle \xi, \check{d}^{(s)}_x \check{b}_y \xi \rangle \Big| \\ 
		&\leq C N^\kappa \| (\cN_+ + 1)^{1/2} \xi \| \int_0^1ds\, \int_{\Lambda^2} dx dy\, N^{3-3\kappa}  V(N^{1-\kappa}(x-y)) \\
		&\hspace{6cm}\times  \| (\cN_+ + 1)^{-1/2} \check{d}^{(s)}_x \check{b}_y \xi \|
 \\ 
		&\leq  C N^{2\kappa-\alpha/2} \| (\cN_+ + 1)^{1/2} \xi \| \int_{\Lambda^2} dx dy\,  N^{3-3\kappa} V(N^{1-\kappa}(x-y))  \\ 
		&\hspace{6cm} \times \left[ \| \check{a}_y \xi \| +N^{-1/2} \| \check{a}_x \check{a}_y \xi \| \right] \\
		&\leq C N^{2\kappa-\alpha/2}\| (\cN_+ + 1)^{1/2} \xi \| ^2 + CN^{3\kappa/2-\alpha/2}\| (\cN_+ + 1)^{1/2} \xi \| \| \cV_N^{1/2}\xi\|.
		\end{split}\]
If we use \eqref{eq:ngdxy-bds} from Lemma \ref{lm:commdpNres} to control the commutator with $\cN_{\leq cN^\gamma}$ and if we recall the estimate \eqref{eq:VNchixaybnd}, we conclude altogether that $\cE^{(4)}_{1042}$ satisfies
		\[ \pm   \cE^{(4)}_{1042} \leq C N^{2\kappa - \a/2} (\cN_++1),\hspace{0.5cm} \pm  i\big[\cN_{\leq cN^\gamma}, \cE^{(4)}_{1042}\big] \leq C N^{2\kappa -\a/2+\gamma/2 } (\cH_N+ 1). \] 
Finally, for $\cE^{(4)}_{1045} $ we proceed very similarly. We switch to position space and find that
		\[ \begin{split} |\langle \xi ,\big[\cN_{\leq cN^\gamma}, \cE_{105}^{(4)}] \xi \rangle | &\leq CN^\kappa \int_{\L^2} dx dy \, N^{3-3\kappa} V(N^{1-\kappa}(x-y))  |\langle \xi, \check{d}_x \check{d}_y \xi \rangle|\\ 
		&\leq  C N^{2\kappa-\alpha/2} \| (\cN_+ + 1)^{1/2} \xi \| \int_{\Lambda^2} dx dy\,  N^{3-3\kappa} V(N^{1-\kappa}(x-y))  \\ 
		&\hspace{5cm} \times \left[ \| \check{a}_x \xi \| + \| \check{a}_y \xi \| +N^{-1/2} \| \check{a}_x \check{a}_y \xi \| \right] \\
		&\leq C N^{2\kappa-\alpha/2}\| (\cN_+ + 1)^{1/2} \xi \| ^2 + CN^{3\kappa/2-\alpha/2}\| (\cN_+ + 1)^{1/2} \xi \| \| \cV_N^{1/2}\xi\|
		\end{split}\]
as well as 
		$ |\langle \xi ,\big[\cN_{\leq cN^\gamma}, \cE_{105}^{(4)}] \xi \rangle | \leq  C N^{2\kappa -\a/2+\gamma/2 } (\cH_N+ 1). $
In fact, to prove this latter commutator bound, we use the identity $ \int_{\Lambda}dx\;a^*(\check{\chi}_x)a(\check{\chi}_x) = \sum_{p\in\Lambda_+^*:|p|\leq cN^\gamma}a^*_pa_p$, the bound \eqref{eq:ngddxy} in Lemma \ref{lm:commdpNres} and the estimate \eqref{eq:VNchixaybnd}. 

Collecting all the previous bounds on $\cE^{(4)}_{10k}$, $k\in \{ 1,\dots,5\}$, we arrive at
		\begin{equation}\label{eq:E104} \pm \cE_{10}^{(4)}\leq C N^{2\kappa -\a/2} (\cH_N + 1), \hspace{0.5cm} \pm i[\mathcal{N}_{\leq cN^{\gamma}}, \cE_{10}^{(4)}] \leq C N^{2\kappa -\a/2 + \gamma/2} (\cH_N+1). \end{equation}

Next, let's go back to \eqref{eq:W1-dec} and analyse the operators $\text{W}_{11}, \text{W}_{12} $ and $\text{W}_{13}$. We follow \cite[Eq. (7.64) \& (7.65)]{BBCS} and write $ \text{W}_{11}$ and $\text{W}_{12}$ as 
		\begin{equation}\label{eq:W11W12}
		\begin{split}
		 \text{W}_{11} & =  \frac{N^{\kappa}}{2N} \sum_{q \in \Lambda_+^* , r \in \Lambda^* : r \not = -q} \widehat{V} (r/N^{1-\kappa}) \eta_H (q+r) (b_q b_{-q} + \hc ) + \cE_{11}^{(4)},\\
		 \text{W}_{12} &= \frac{N^{\kappa}}{2N} \sum_{q \in \Lambda_+^* , r \in \Lambda^* : r \not = -q} \widehat{V} (r/N^{1-\kappa}) \eta_H (q+r) \eta_H (q)  \Big(1- \frac{\cN_+}{N} \Big) + \cE_{12}^{(4)},
		 \end{split}
		 \end{equation}
with the error terms  $ \cE_{11}^{(4)}$ and $\cE_{12}^{(4)}$ defined by
		\[ \begin{split} 
		\cE_{11}^{(4)} = \;& \frac{N^{\kappa}}{2N} \sum_{q \in \Lambda_+^* , r \in \Lambda^* : r \not = -q}  \widehat{V} (r/N^{1-\kappa}) \eta_H (q+r) \int_0^1 ds \big[ (\g^{(s)}_q) ^2- 1\big]  (b_q b_{-q} + \hc )  ,\\
		\cE_{12}^{(4)}  = \; & -\frac{1}{2N^2} \sum_{q \in \Lambda_+^* , r \in \Lambda^* : r \not = -q} N^{\kappa}\widehat{V} (r/N^{1-\kappa}) \eta_H (q+r) \int_0^1 ds \g_q^{(s)} \s^{(s)}_q a_q^* a_q \\ &+ \frac{1}{2N} \sum_{q \in \Lambda_+^* , r \in \Lambda^* : r \not = -q}N^{\kappa} \widehat{V} (r/N^{1-\kappa}) \eta_H (q+r) \int_0^1 ds (\g_q^{(s)} \s^{(s)}_q -s \eta_H (q)) \Big( 1- \frac{\cN_+}{N} \Big). 
		\end{split} \]
Since $ |  (\g^{(s)}_q) ^2- 1|\leq C\eta_H(q)^2$ and $| \g_q^{(s)} \s^{(s)}_q -s \eta_H (q))| | \leq C |\eta_H (q)|^3$, we use the same arguments with which we controlled the error $ \cE_{10}^{(4)}$ to deduce that 
		\[ \pm ( \cE_{11}^{(4)} +\cE_{12}^{(4)})\leq C N^{3\kappa -5\a/2} (\cN_+ + 1), \hspace{0.5cm} \pm i[\mathcal{N}_{\leq cN^{\gamma}}, \cE_{11}^{(4)} +\cE_{12}^{(4)} ] \leq C N^{3\kappa -5\a/2} (\cN_++1). \]	
We omit the details. Similar arguments apply to the operator $ \text{W}_{13}$, defined in \eqref{eq:W1-dec}, but here we partly need to switch to position space again. We split $\text{W}_3$ into
		\begin{equation}\label{eq:W13} 
		\text{W}_{13} = - \frac{N^{\kappa}}{2N} \sum_{\substack{ q \in \Lambda_+^* ,\\ r \in \Lambda^* : r \not = -q}} \widehat{V} (r/N^{1-\kappa}) \eta_H (q+r) \eta_H (q) \left(1- \frac{\cN_+}{N} \right) \frac{\cN_+ +1}{N}  + \cE^{(4)}_{13},
		\end{equation}
where the error $\cE^{(4)}_{13} = \cE_{131}^{(4)} +  \cE_{132}^{(4)} + \cE_{133}^{(4)}$ is defined through
\begin{equation*}%\label{eq:cE413}
\begin{split}
\cE^{(4)}_{131} = \; &\frac{N^{\kappa}}{2N} \sum_{q \in \Lambda_+^* , r \in \Lambda^* : r \not = -q} \widehat{V} (r/N^{1-\kappa}) \eta_H (q+r) \int_0^1 ds (\g_q^{(s)} -1) b_q d_{-q}^{(s)} +\hc  \\
\cE^{(4)}_{132} = \; &\frac{N^{\kappa}}{2N} \sum_{q \in \Lambda_+^* , r \in \Lambda^* : r \not = -q} \widehat{V} (r/N^{1-\kappa}) \eta_H (q+r) \int_0^1 ds \, b_q \left[ d_{-q}^{(s)} + s \eta_H (q) \frac{\cN_+}{N} b_{q}^* \right] + \hc \\
\cE^{(4)}_{133} = \;&- \frac{N^{\kappa}}{2N} \sum_{q \in \Lambda_+^* , r \in \Lambda^* : r \not = -q} \widehat{V} (r/N^{1-\kappa}) \eta_H (q+r) \eta_H (q) a^*_q  a_{q}\frac{(N-\cN_+)}{N} \frac{\cN_+ +1}{N}.
%\cE^{(4)}_{134} = \; & \frac{1}{2N^2} \sum_{q \in \Lambda_+^* , r \in \Lambda^* : r \not = -q}N^{\kappa} \widehat{V} (r/N^{1-\kappa}) \eta_H (q+r) \eta_H (q) a_q^* a_q \frac{\cN_+ +1}{N}
\end{split} \end{equation*}
The last term $\cE^{(4)}_{133}$ is easily seen to be bounded by $  \pm \cE_{133}^{(4)} \leq C N^{2\kappa-2\a} (\cN_+ + 1)  $ and we also notice that $[\mathcal{N}_{\leq cN^{\gamma}}, \cE^{(4)}_{133}] =0$. Hence, let's focus on the first two errors $\cE^{(4)}_{131}$ and $\cE^{(4)}_{132}$. Since $|\g_q^{(s)} -1| \leq C \eta_H (q)^2$, Lemma \ref{lm:dp}, Lemma \ref{lm:commdpNres} and (\ref{eq:VetaN}) imply that 
		\[  \begin{split} |\langle \xi , \cE_{131}^{(4)} \xi \rangle | &\leq C N^{\kappa-1} \sum_{q \in \Lambda_+^* , r \in \Lambda^* : r \not = -q} |\widehat{V} (r/N^{1-\kappa})| |\eta_H (q+r)| |\eta_H (q)|^2  \| (\cN_+ +1)^{1/2} \xi \| \\ 
		&\hspace{5cm} \times \left[ |\eta_H (q)| \| (\cN_+ +1)^{1/2} \xi \| + \| \eta_H \| \| b_q \xi \| \right] \\ 
		&\leq C N^{4\kappa - 3\a} \| (\cN_+ + 1)^{1/2} \xi \|^2 \end{split} \]
and that
		\[  \begin{split} |\langle \xi ,[\cN_{\leq cN^\gamma},  \cE_{131}^{(4)}] \xi \rangle | %&\leq \frac{C N^{\kappa}}{N} \sum_{q \in \Lambda_+^* , r \in \Lambda^* : r \not = -q} |\widehat{V} (r/N^{1-\kappa})| |\eta_H (q+r)| |\eta_H (q)|^2  \| (\cN_+ +1)^{1/2} \xi \| \\ 
		%&\hspace{5cm} \times \left[ |\eta_H (q)| \| (\cN_+ +1)^{1/2} \xi \| + \| \eta_H \| \| b_q \xi \| \right] \\ 
		&\leq C N^{4\kappa - 3\a} \| (\cN_+ + 1)^{1/2} \xi \|^2. \end{split} \]
As for the term $\cE_{132}^{(4)} $, we switch to position space where it reads
		\[\cE_{132}^{(4)}  = \int_0^1 ds \int_{\L^2} dx dy N^{2-2\kappa} V(N^{1-\kappa}(x-y)) \check{\eta}_H (x-y)  \check{b}_x \check{\bar{d}}_y. \]
Recall here the notation $\check{\bar{d}}^{(s)}_y = d^{(s)}_y + s (\cN_+ / N) b^* (\check{\eta}_{H,y})$. Due to the pointwise estimate $ \|\check{\eta}_H \|_\infty\leq CN$, we may proceed as in \eqref{eq:bdbarpspaceterm1}, \eqref{eq:bdbarpspaceterm2} and thereafter to conclude that 
		\[\begin{split}
		\pm \cE_{132}^{(4)} &\leq CN^{2\kappa-\alpha/2}(\cN_++1) + CN^{3\kappa/2-\alpha/2}(\cV_N+1),\\
		\pm \cE_{132}^{(4)} &\leq CN^{2\kappa-\alpha/2}(\cN_++1) + CN^{2\kappa+\gamma/2-\alpha/2}(\cH_N+1).
		\end{split}\]
		
Now, if we collect the bounds \eqref{eq:E104}, \eqref{eq:W11W12} and \eqref{eq:W13}, we see altogether that  
		\begin{equation}\label{eq:W1f} 
		\begin{split}  \text{W}_1 = \; &\frac{N^{\kappa}}{2N} \sum_{\substack{ q \in \Lambda_+^* , r \in \Lambda^* :\\ r \not = -q }} \hspace{-0.3cm}\widehat{V} (r/N^{1-\kappa}) \eta_H (q+r) \eta_H (q)  \Big(1- \frac{\cN_+}{N} \Big)\Big(1- \frac{\cN_+ + 1}{ N} \Big) \\ 
		&+\frac{N^{\kappa}}{2N} \sum_{\substack{q \in \L_+^*, r \in \L^* :\\ r \not = -q}} \widehat{V} (r/N^{1-\kappa}) \eta_H (q+r)  \Big(b_q b_{-q} + \hc \Big) + \cE^{(4)}_1, 
		\end{split} \end{equation}
where the error operator $\cE_{1}^{(4)}$ satisfies the estimates
		\[ \pm \cE_{1}^{(4)} \leq C N^{2\kappa-\a/2} (\cH_N + 1), \hspace{0.5cm} \pm i[\mathcal{N}_{\leq cN^{\gamma}},\cE_{1}^{(4)}] \leq C N^{2\kappa-\a/2+\gamma/2} (\cH_N + 1). \]
This concludes the analysis of the first contribution $\text{W}_1$ in Eq. \eqref{eq:defW}. It remains to analyse the contributions $ \text{W}_2, \text{W}_3$ and $ \text{W}_4$. To this end, it is useful to switch to position space and to use the results of Lemma \ref{lm:prel4}. Considering for instance $\text{W}_2$, we have that 
		\begin{equation*} \begin{split}
		\text{W}_2 = \int_{\Lambda^2} dx dy N^{2-2\kappa} V(N^{1-\kappa}(x-y))  \int_0^1  ds \Big[ e^{-sB( \eta_H)} \check{b}^*_x \check{b}^*_y  e^{s B( \eta_H)}
		a^* (\check{\eta}_{H,x}) \check{a}_y  + \text{h.c.}   \Big].
		\end{split}
		\end{equation*}
Using Cauchy-Schwarz, Lemma \ref{lm:prel4} and the bound
		\[ \| (\cN_+ +1)^{-1/2} a^* (\check{\eta}_{H,x}) \check{a}_y \xi \| \leq C \|  \eta_H \| \| \check{a}_y \xi \| \leq C N^{\kappa-\a/2} \| \check{a}_y \xi \|,  \]
we obtain that
		\begin{equation*} \begin{split} 
		|\langle \xi, \text{W}_2 \xi \rangle |&\leq  C N^{2\kappa-\a/2} \int_{\Lambda^2}  dx dy \, N^{3-3\kappa} V(N^{1-\kappa}(x-y))  \| \check{a}_y \xi \| \\ 
		&\hspace{2cm} \times \Big[   \| (\cN_+ +1)^{1/2} \xi \| +   \| \check{a}_x \xi \| +  \| \check{a}_y \xi \| + N^{-1/2}  \| \check{a}_x \check{a}_y \xi \| \Big] \\
		& \leq C  N^{2\kappa-\a/2} \|(\cN_++1)^{1/2}\xi\|^2 + CN^{3\kappa/2 - \a/2} \|(\cN_++1)^{1/2}\xi\| \cV_N^{1/2} \xi \|.
		\end{split} \end{equation*}
Similarly, to control the commutator with $ \cN_{\leq cN^\gamma}$, we use the estimate
		\[ \| (\cN_+ +1)^{-1/2} [\cN_{\leq cN^\gamma},a^* (\check{\eta}_{H,x}) \check{a}_y] \xi \|\leq C N^{\kappa-\a/2}\Big( \| \check{a}_y \xi \| +  \| a(\check{\chi}_y) \xi \|\Big) \]
and find with the help of Lemma \ref{lm:prel4} that
		\begin{equation*} \begin{split}
		|\langle \xi, [\cN_{\leq cN^\gamma},\text{W}_2] \xi \rangle |&\leq  C N^{2\kappa-\a/2} \int_{\Lambda^2}  dx dy \, N^{3-3\kappa} V(N^{1-\kappa}(x-y)) \Big[ \| \check{a}_y \xi \| +  \| a(\check{\chi}_y) \xi \|\Big] \\ 
		&\hspace{1.3cm} \times \Big[   \| (\cN_+ +1)^{1/2} \xi \| +   \| \check{a}_x \xi \| +  \| \check{a}_y \xi \| +  \| a(\check{\chi}_y) \xi \| + \| a(\check{\chi}_x) \xi \| \\
		&\hspace{1.8cm}  + N^{-1/2}  \| a(\check{\chi}_y) \check{a}_x \xi \| + N^{-1/2}  \| a(\check{\chi}_x) \check{a}_y \xi \| + N^{-1/2}  \| \check{a}_x \check{a}_y \xi \| \Big] \\
		& \leq C  N^{2\kappa-\a/2} \|(\cN_++1)^{1/2}\xi\|^2  + CN^{2\kappa-\a/2+\gamma/2} \|(\cH_N+1)^{1/2}\xi\|^2.
		\end{split} \end{equation*}
Here, the last inequality follows as in \eqref{eq:bdbarpspaceterm2} and thereafter.
 
Finally, controlling the remaining two contributions $\text{W}_3$ and $\text{W}_4$, defined in \eqref{eq:defW}, follows along the same lines. We skip the details and summarize that 
		\[\begin{split}
		\pm (\text{W}_3 +\text{W}_3) &\leq  CN^{3\kappa-\alpha}(\cN_++1) + CN^{5\kappa/2 - \a}(\cV_N+1) ,\\
		 \pm i[\cN_{\leq cN^\gamma}, \text{W}_3+\text{W}_4] &\leq CN^{3\kappa-\alpha}(\cN_++1) + CN^{3\kappa - \a +\gamma/2}(\cH_N+1).
		 \end{split}\] 
Altogether, we have thus shown that
	\[\begin{split} 
	\cG^{(4)}_{N} = \; & \frac{N^{\kappa}}{2N} \sum_{q \in \Lambda_+^* , r \in \Lambda^* : r \not = -q} \widehat{V} (r/N^{1-\kappa}) \eta_H (q+r) \eta_H (q)  \Big(1- \frac{\cN_+}{N} \Big) \Big(1- \frac{\cN_+ + 1}{ N} \Big) \\ 
	& + \frac{N^{\kappa}}{2N} \sum_{q \in \L_+^*, r \in \L^* : r \not = -q} \widehat{V} (r/N^{1-\kappa}) \eta_H (q+r)  \big(b_q b_{-q} + \hc \big) +\cV_N + \cE^{(4)}_{N} 
	\end{split} \]
with the error operator $\cE_{N}^{(4)}$ satisfying the bounds \eqref{eq:E4bound2} and \eqref{eq:E4boundcomm}. The proof of \eqref{eq:E4Cff} is similar to that of \eqref{eq:E4bound2} and we omit the details.
\end{proof}

%%%%%%%%%%%%%%%%%%%%%%%%%%%%%%%%%%%%%%%%%%%%%%%%%%%%%%%%%%%%%%%%%
%%%%%%%%%%%%%%%%%%%%%%%%%%%%%%%%%%%%%%%%%%%%%%%%%%%%%%%%%%%%%%%%%
%%%%%%%%%%%%%%%%%%%%%%%%%%%%%%%%%%%%%%%%%%%%%%%%%%%%%%%%%%%%%%%%%
%%%%%%%%%%%%%%%%%%%%%%%%%%%%%%%%%%%%%%%%%%%%%%%%%%%%%%%%%%%%%%%%%
%%%%%%%%%%%%%%%%%%%%%%%%%%%%%%%%%%%%%%%%%%%%%%%%%%%%%%%%%%%%%%%%%
\subsection{Proof of Prop. \ref{prop:GN}} \label{sub:proofGN}
The goal of this section is to prove Proposition \ref{prop:GN}. With the results of the previous sections \ref{sub:G0} - \ref{sub:G4}, the proof follows as in \cite[Section 7.5]{BBCS}, suitably adjusted to our setting. In addition to the arguments of \cite[Section 7.5]{BBCS}, however, we need to provide some further bounds to control commutators with $\cN_{\leq cN^\gamma}$. Let us sketch the main steps and let us focus for simplicity on proving \eqref{eq:GN-prel} to \eqref{eq:GeffE} as well as \eqref{eq:errComm} (the last bound \eqref{eq:adjGN} follows then as explained at the beginning of \cite[Section 7]{BBCS2}). First of all, Propositions \ref{prop:G0}, \ref{prop:G2}, \ref{prop:GN-3} and \ref{prop:GN-4} imply that the excitation Hamiltonian $\cG_{N}$, defined in \eqref{eq:GN}, has the form
\begin{equation} \begin{split} \label{eq:proofGNell-1}
\cG_{N} = \; &    \frac{N^{\kappa}\widehat{V} (0)}{2}\, (N +\cN_+ -1) \, (1-\cN_+/N) +\sum_{p \in \Lambda^*_+} N^{\kappa} \widehat{V} (p/N^{1-\kappa})  a^*_pa_p (1-\cN_+/N)\\
%%%
& + \sum_{p \in P_{H}} \eta_p \Big[p^2 \eta_p + N^{\kappa}\widehat{V} (p/N^{1-\kappa}) \\
&\hspace{2cm}+ \frac {N^{\kappa}} {2N} \sum_{\substack{r \in \L^*,\\ p+r \in P_H}} \widehat{V} (r/N^{1-\kappa})  \eta_{p+r}\Big]\frac{(N-\cN_+)}{N}\frac{(N-\cN_+-1)}{N} \\
&+ \sum_{p \in P_{H}} \Big[\; p^2 \eta_p + \frac 12 N^{\kappa}\widehat{V} (p/N^{1-\kappa}) + \frac {N^{\kappa}} {2N} \sum_{r \in \L^*:\; p+r \in P_H}\hskip -0.5cm \widehat{V} (r/N^{1-\kappa})  \eta_{p+r} \; \Big]  \big( b^*_p b^*_{-p} + b_p b_{-p} \big) \\
 & + \frac{1}{2}\sum_{p \in P_H^c} \Big[ N^{\kappa}\widehat{V} (p/N^{1-\kappa}) + \frac {N^{\kappa}} {N} \sum_{r \in \L^*:\; p+r \in P_H}\hskip -0.5cm \widehat{V} (r/N^{1-\kappa})  \eta_{p+r} \Big]\big( b_p b_{-p}+ b_{-p}^* b_p^*\big) \\
& + \frac{N^{\kappa}}{\sqrt{N}} \sum_{p,q \in \L^*_+ : p + q \not = 0} \widehat{V} (p/N^{1-\kappa}) \left[ b_{p+q}^* a_{-p}^* a_q  + \hc \right]  +\cK+\cV_N   + \cE_{1}
\end{split} \end{equation}
where the self-adjoint operator $ \cE_1$ satisfies
		\begin{equation*}%\label{eq:cE10}
		\begin{split}
		\pm \cE_1 &\leq C  N^{3\kappa-\alpha/2} \big(\cH_N + 1 \big),\\
		 \pm i[\cN_{\leq cN^\gamma}, \cE_1] &\leq C(N^{3\kappa-\alpha/2} + N^{2\kappa+\gamma/2-\alpha/2}) \big(\cH_N + 1 \big), \\
		\pm i[\cN_{> cN^\gamma}, \cE_1] &\leq C(N^{3\kappa-\alpha/2} + N^{2\kappa+\gamma/2-\alpha/2}) \big(\cH_N + 1 \big).
		\end{split}
		\end{equation*}
Notice for the last line that $ \cN_{>cN^\gamma} = \cN_+ - \cN_{\leq cN^\gamma}$. Using the scattering equation \eqref{eq:eta-scat}, the bound \eqref{eq:etaHL2} and Lemma \ref{sceqlemma}, we deduce that (see also \cite[Eq. (7.71)]{BBCS})
		\[\begin{split}
		 &\sum_{p \in P_{H}} \eta_p \Big[p^2 \eta_p + N^{\kappa}\widehat{V} (p/N^{1-\kappa}) + \frac {N^{\kappa}} {2N} \sum_{\substack{r \in \L^*,\\ p+r \in P_H}} \widehat{V} (r/N^{1-\kappa})  \eta_{p+r}\Big]\frac{(N-\cN_+)}{N}\frac{(N-\cN_+-1)}{N}\\
		 %&= \frac{1}{2} \sum_{p \in P_H} N^{\kappa}\widehat{V} (p/N^{1-\kappa}) \eta_p (1-\cN_+/N)^2 + \cE_2 \\
		 & =  \left[ 4\pi \frak{a}_0 N^{1+\kappa} - \frac12 N^{1+\kappa} \widehat{V} (0) \right](1-\cN_+/N)^2   +\cE_2,
		\end{split}\]
where the error $ \cE_2$ satisfies $ \pm \cE_2\leq CN^{2\kappa + \alpha }$ and $ [\cN_{\leq cN^\gamma}, \cE_2] = [\cN_+, \cE_2]=0$. Similarly, 
		\[\begin{split} %\label{eq:proofGNellQ}
		\sum_{p \in P_{H}} \Big[\, p^2 \eta_p + \frac 12 &N^{\kappa}\widehat{V} (p/N^{1-\kappa}) + \frac 1 {2N} \sum_{r \in \L^*:\; p+r \in P_H}\hskip -0.5cm N^{\kappa}\widehat{V} (r/N^{1-\kappa})  \eta_{p+r} \, \Big]  \big( b^*_p b^*_{-p} + b_p b_{-p} \big) \\
		& = N^\kappa (N^{3-3\kappa} \lambda_\ell )\sum_{p \in P_H}  (\widehat{\chi}_\ell\ast \widehat{f}_N)(p) \big( b^*_p b^*_{-p} + b_p b_{-p} \big) \\
		& \hskip .3cm -  \frac {N^{\kappa}} {2N} \sum_{\substack{p,q \in \L^*: \\ p\in P_H,\, q \in P_H^c}}\hskip -0.5cm \widehat{V} ((p-q)/N^{1-\kappa}) \eta_{q} \big( b^*_p b^*_{-p} + b_p b_{-p} \big)
		\end{split} \]
so that, once more by Lemma \ref{sceqlemma}, the bounds $| N^{3-3\kappa}\lambda_\ell |\leq C$ and $\|(\widehat{\chi}_\ell\ast \widehat{f}_N)\|\leq C$ imply
		\[\begin{split} 
		& \pm N^\kappa (N^{3-3\kappa} \lambda_\ell )\sum_{p \in P_H}  (\widehat{\chi}_\ell\ast \widehat{f}_N)(p) \big( b^*_p b^*_{-p} + b_p b_{-p} \big)\leq CN^{\kappa-\alpha} (\cK+1), \\
		& \pm N^\kappa (N^{3-3\kappa} \lambda_\ell )\sum_{p \in P_H}  (\widehat{\chi}_\ell\ast \widehat{f}_N)(p) [\cN_{\leq cN^\gamma},  b^*_p b^*_{-p} + b_p b_{-p}  ]\leq CN^{\kappa-\alpha} (\cK+1).
		\end{split}\]
Analogously, we can write
		\[\begin{split}
		&\frac {N^{\kappa}} {2N} \sum_{\substack{p,q \in \L^*: \\ p\in P_H,\, q \in P_H^c}}\hskip -0.5cm \widehat{V} ((p-q)/N^{1-\kappa}) \eta_{q} \big( b_p b_{-p} +\text{h.c.} \big)\\
		& \hspace{0.5cm}= \frac12\sum_{q\in\Lambda^*:q \in P_H^c} \int_{\Lambda^2}dxdy\; N^{2-2\kappa}V(N^{1-\kappa}(x-y)) e^{iq(x-y)} \eta_{q}\big( \check{b}_x \check{b}_{y} +\text{h.c.} \big)\\
		&\hspace{1cm}-\frac {N^{\kappa}} {2N} \sum_{\substack{p,q \in \L^*: \\ p\in P_H^c,\, q \in P_H^c}}\hskip -0.5cm \widehat{V} ((p-q)/N^{1-\kappa}) \eta_{q} \big( b_p b_{-p} +\text{h.c.} \big).
		\end{split}\]
These terms can be controlled (with $\alpha+\kappa\leq 1$) by
		%\[\begin{split}
		%& \int_{\Lambda^2}dxdy\; N^{2-2\kappa}V(N^{1-\kappa}(x-y)) \| \check{b}_x \check{b}_{y}  \xi\| 	\bigg \| \sum_{q\in\Lambda^*:q \in P_H^c} e^{iq(x-y)} \eta_{q} \xi\bigg\|\\
		%&\leq  \| \cV_{N}^{1/2}\xi\| \|\xi\|\bigg( N^{\kappa-1}\hspace{-0.5cm} \sum_{p,q \in \Lambda^*:p,q \in P_H^c} \widehat{V}((p-q)/N^{1-\kappa}) \eta_p\eta_q  \bigg)^{1/2} 
		%\end{split}\]			
%Proceeding similarly for the commutator with $ \cN_{\leq cN^\gamma}$, we conclude that
		\[\begin{split} 
		 \pm \frac {N^{\kappa}} {2N} \sum_{\substack{p,q \in \L^*: \\ p\in P_H,\, q \in P_H^c}}\hskip -0.5cm \widehat{V} ((p-q)/N^{1-\kappa}) \eta_{q} \big( b^*_p b^*_{-p} + b_p b_{-p} \big)&\leq CN^{\alpha+3\kappa/2-1/2} (\cH_N +1) , \\
		 \pm \frac {N^{\kappa}} {2N}\hspace{-0.2cm} \sum_{\substack{p,q \in \L^*: \\ p\in P_H,\, q \in P_H^c}}\hskip -0.5cm \widehat{V} ((p-q)/N^{1-\kappa}) \eta_{q} \Big(i [\cN_{\leq cN^\gamma}, b^*_p b^*_{-p} + b_p b_{-p}]\Big)&\leq CN^{\alpha+3\kappa/2-1/2} (\cH_N +1), \\
		 \pm \frac {N^{\kappa}} {2N}\hspace{-0.2cm} \sum_{\substack{p,q \in \L^*: \\ p\in P_H,\, q \in P_H^c}}\hskip -0.5cm \widehat{V} ((p-q)/N^{1-\kappa}) \eta_{q} \Big(i [\cN_{> cN^\gamma}, b^*_p b^*_{-p} + b_p b_{-p}]\Big)&\leq CN^{\alpha+3\kappa/2-1/2} (\cH_N +1). \\
		\end{split}\]
Arguing similarly for the term in the fourth line of \eqref{eq:proofGNell-1}, we find  
		\[ \begin{split} %\label{eq:quadr4}
 		&\frac{N^{\kappa}}{2} \sum_{p \in P_H^c} \Big[ \widehat{V} (p/N^{1-\kappa}) + \frac{1}{N} \sum_{r \in \L^*:\; p+r \in P_H}\hskip -0.5cm \widehat{V} (r/N^{1-\kappa})  \eta_{p+r} \Big]\big( b_p b_{-p}+ b_{-p}^* b_p^*\big) \\
 		& = \frac{N^{\kappa}}{2}\sum_{p \in P_H^c}  (\widehat{V}(\cdot/N^{1-\kappa}) * \widehat f_{N})_p \big( b_p b_{-p}+ b_{-p}^* b_p^*\big) + \cE_3,
		 \end{split}\]
where the error $\cE_3$ satisfies 
		\begin{equation*}%\label{eq:cE10}
		\begin{split}
		\pm \cE_3 &\leq CN^{\alpha+3\kappa/2-1/2} (\cH_N+1),\; \pm i[\cN_{\leq cN^\gamma}, \cE_3]\leq CN^{\alpha+3\kappa/2-1/2} (\cH_N+1), \\
		\pm i[\cN_{> cN^\gamma}, \cE_3] &\leq CN^{\alpha+3\kappa/2-1/2} (\cH_N+1).
		\end{split}
		\end{equation*}
Collecting the error bounds from above, we summarize that 
		\begin{equation*} \begin{split} % \label{eq:proofGNell-2}
		\cG_{N} = \; & 4 \pi \frak{a}_0 N^{\kappa} (N-\cN_+) +N^{\kappa} \big[ \widehat{V}(0) - 4\pi \frak{a}_0 \big] \cN_+ (1-\cN_+/N) \\
		&  +N^{\kappa}\sum_{p \in \Lambda^*_+}  \widehat{V} (p/N^{1-\kappa})  a^*_pa_p (1-\cN_+/N)\\
		&+ \frac{N^{\kappa}}{2}\sum_{p \in P_H^c}  (\widehat{V}(\cdot/N^{1-\kappa}) * \widehat f_{N})_p \big( b_p b_{-p}+ b_{-p}^* b_p^*\big) \\
		& + \frac{N^{\kappa}}{\sqrt{N}} \sum_{p,q \in \L^*_+ : p + q \not = 0} \widehat{V} (p/N^{1-\kappa}) \left[ b_{p+q}^* a_{-p}^* a_q  + \hc \right]  +\cH_N   + \cE_4,
		\end{split} \end{equation*}
where $ \cE_4$ satisfies
		\[ \pm \cE_4 \leq C\big(N^{3\kappa-\alpha/2}+ N^{\alpha+3\kappa/2-1/2}\big)  (\cH_N+1) + CN^{\alpha+2\kappa} \]
as well as the commutator bounds
		\begin{equation}\label{eq:commcNE4GN}\pm i[\cN_{\leq cN^\gamma}, \cE_4], \,\pm i[\cN_{> cN^\gamma}, \cE_4]  \leq C\big(N^{3\kappa-\alpha/2}+N^{\gamma/2+ 2\kappa-\alpha/2}+ N^{\alpha+3\kappa/2-1/2}\big)  (\cH_N+1). \end{equation}
With a few more simplifications, using $|N^{\kappa}\widehat{V} (p/N^{1-\kappa}) -  N^{\kappa}\widehat{V} (0)| \leq C |p| N^{2\kappa-1}$ and
		\[ \begin{split} \big|(\widehat{V} (\cdot/N^{1-\kappa}) * \widehat f_{N})_p  - 8 \pi \frak{a}_0 \big| &\leq \int_\Lambda dx \, N^{3-3\kappa} V(N^{1-\kappa}x) f_\ell (N^{1-\kappa}x) \big| e^{ip \cdot x} - 1 \big| \\
		&\hspace{0.5cm}+ \big|(\widehat{V} (\cdot/N^{1-\kappa}) * \widehat f_{N})(0)  - 8 \pi \frak{a}_0 \big|  \leq C \big(|p| + 1\big) N^{\kappa-1} \end{split} \]
as well as the estimate (similar to \cite[eq. (8.38)]{BBCS})
		\[\begin{split}
		&\bigg| \frac{N^{\kappa}}{\sqrt{N}} \sum_{p,q \in \L^*_+ : |q| > N^\beta; p + q \not = 0} \widehat{V} (p/N^{1-\kappa}) \langle\xi, \left[ b_{p+q}^* a_{-p}^* a_q  + \hc \right] \xi\rangle \bigg| \\
		&\hspace{0.5cm} \leq C \sqrt{N} \int_{\Lambda^2}dxdy\; N^{2-2\kappa} V(N^{1-\kappa}(x-y)) \| \check{a}_x\check{a}_y\xi\| \Big\| \sum_{q\in\Lambda_+^*: |q| > N^{\beta}} e^{iqx} a_q  \xi\Big\|\\
		&\hspace{0.5cm}\leq C N^{\kappa/2-\beta } \langle \xi, \cH_N \xi\rangle,
		\end{split}\]
we arrive at the decomposition $ \cG_N = \cG_N^{\text{eff}}+\cE_{\cG_N}$, where the error $ \cE_{\cG_N}$ satisfies 
		\[ \pm \cE_{\cG_N} \leq C\big(N^{3\kappa-\alpha/2}+ N^{\alpha+3\kappa/2-1/2}+N^{\kappa/2-\beta}\big)  (\cH_N+1) + CN^{\alpha+2\kappa}. \]
%and the commutator bounds
%		\[\pm i[\cN_{\leq cN^\gamma}, \cE_{\cG_N}], \,\pm i[\cN_{> cN^\gamma}, \cE_{\cG_N}]  \leq C\big(N^{3\kappa-\alpha/2}+N^{\gamma/2+ 2\kappa-\alpha/2}+ N^{\alpha+3\kappa/2-1/2}\big)  (\cH_N+1).\]
This proves in particular the bound \eqref{eq:GeffE}. To prove the remaining bounds, i.e. \eqref{eq:Gbd0}, \eqref{eq:theta-err} and \eqref{eq:errComm}, we need to analyse further the operator $ \cG_N^{\text{eff}}$. 

 To show \eqref{eq:Gbd0}, we use \eqref{eq:GeffE} and Cauchy-Schwarz to see that 
 		\begin{equation}\label{eq:pairingterm} N^\kappa\sum_{p\in P_H^c} | \langle\xi, b^*_{p} b^*_{-p}\xi\rangle| \leq \delta \|(\cK+1)^{1/2}\xi\|^2 + C\delta^{-1}N^{\alpha+2\kappa}\| (\cN_++1)^{1/2}\xi\|^2
		\end{equation}
as well as
		\[\begin{split}
		& \frac{N^{\kappa}}{\sqrt{N}} \sum_{p,q \in \L^*_+ : p + q \not = 0} \widehat{V} (p/N^{1-\kappa})\big| \langle\xi, \left[ b_{p+q}^* a_{-p}^* a_q  + \hc \right] \xi\rangle \big| \\
		&\leq  C \sqrt{N}\!\int_{\Lambda^2}\!\!dxdy\,N^{2-2\kappa} V(N^{1-\kappa}(x-y)) \| \check{a}_x \check{a}_y\xi\| \|\check{a}_x\xi\|  \leq \delta \|\cV_N^{1/2}\xi\|^2 + C\delta^{-1}N^{\kappa}\| \cN_+^{1/2}\xi\|^2.
		\end{split}\]
Controlling the remaining terms in $ \cG_N^{\text{eff}}$ is similar and we arrive at \eqref{eq:Gbd0}. To get the improved lower bound \eqref{eq:theta-err}, we can complete the square in \eqref{eq:pairingterm} (see the arguments before \cite[Eq. (7.81)]{BBCS}, the adaption to our setting is straight-forward). Finally, using the commutator bound on $ \cE_4$ in \eqref{eq:commcNE4GN} and the assumptions that $ \alpha>6\kappa$ and $ 2\alpha+3\kappa<1$, the bounds \eqref{eq:errComm} follow similarly. The only additional ingredient is to use the bounds \eqref{eq:Ntheta-bds} when controlling the commutator of $ \cN_{>cN^\gamma}$ with the quadratic term \eqref{eq:pairingterm}. It is straight-forward to prove that this produces an error bounded by $ CN^{\alpha/2+\kappa-\gamma}(\cH_N+1)$. This concludes the proof of Prop. \ref{prop:GN}.

%%%%%%%%%%%%%%%%%%%%%%%%%%%%%%%%%%%%%%%%%%%%%%%%%%%%%%%%%%%%%%%%%%%
%%%%%%%%%%%%%%%%%%%%%%%%%%%%%%%%%%%%%%%%%%%%%%%%%%%%%%%%%%%%%%%%%%%
%%%%%%%%%%%%%%%%%%%%%%%%%%%%%%%%%%%%%%%%%%%%%%%%%%%%%%%%%%%%%%%%%%%
%%%%%%%%%%%%%%%%%%%%%%%%%%%%%%%%%%%%%%%%%%%%%%%%%%%%%%%%%%%%%%%%%%%
%%%%%%%%%%%%%%%%%%%%%%%%%%%%%%%%%%%%%%%%%%%%%%%%%%%%%%%%%%%%%%%%%%%

\section{Analysis of $ \cJ_N $ }\label{sec:JN}

This section is devoted to the proof of Proposition \ref{prop:JN}. An important role in the proof of Prop. \ref{prop:JN} is played by the estimates obtained in Lemma \ref{lm:NresgrowA}, Lemma \ref{lm:KresgrowA} and Corollary \ref{cor:VNgrowA}, to control the growth of the restricted number of particles, of the restricted kinetic energy and, respectively, of the potential energy, with respect to conjugation through the unitary operator $e^A$. We will also need control on the action of $e^A$ on the full kinetic energy operator $\cK$. To this end, we first consider the commutator of $\cK$ with $A = A_1 - A_1^*$, with $A_1$ defined as in (\ref{eq:defA}). 
\begin{prop}\label{prop:commKA} Assume the exponents $\alpha,\beta$ satisfy (\ref{eq:alphabeta}). Let $m_0  \in\mathbb{R}$ be such that $ m_0 \beta =\alpha$. Then there exists $C> 0$ such that 
		\begin{equation} \label{eq:propcommKA}
		\begin{split}
		[\cK, A] & = -\frac{1}{\sqrt N} \sum_{\substack{ u\in \Lambda_+^*, p\in P_L:\\ p+u\neq 0 }} N^{\kappa} (\widehat{V}(./N^{1-\kappa})\ast  \widehat{f}_{N})(u) \big( b^*_{p+u} a^*_{-u}a_p +\emph{h.c.}\big) + \cE_{[\cK,A]}
		\end{split}
		\end{equation}
where the self-adjoint operator $\cE_{[\cK,A]} $ satisfies 
		\begin{equation}\label{eq:propcommKAerror}
		\begin{split}
		\pm \cE_{[\cK,A]} \leq&\;   C N^{-3\beta/2}\cK + CN^{3\beta/4+\kappa}\cK_{\leq N^{3\beta/2}} +\delta  \cK + C \delta^{-1}\sum_{j=3}^{2\lfloor m_0\rfloor -1}N^{ j\beta/2 + \beta/2+2\kappa}\cN_{\geq \frac12 N^{j\beta/2}} \\
		&\; + C \delta^{-1}N^{ \alpha +2\kappa}\cN_{\geq \frac12 N^{\lfloor m_0 \rfloor \beta}} + C \delta^{-1} m_0 N^{\alpha +2\kappa}		
		\end{split}
		\end{equation}	
for all $\delta>0$ and $N\in\NN$ sufficiently large.
\end{prop}
\begin{proof} We compute $[\cK, A]  = [\cK, A_1] +\text{h.c.}$ with $A_1$ defined in \eqref{eq:defA}, and we find
		\[\begin{split} 
		[\cK, A_1] + \text{h.c.}&= \frac{1}{\sqrt N } \sum_{r\in P_H, v\in P_L }2r^2\eta_r b^*_{r+v}a^*_{-r}a_v \\
		&\hspace{0.4cm} + \frac{2}{\sqrt N } \sum_{r\in P_H, v\in P_L } r\cdot v\,\eta_r b^*_{r+v}a^*_{-r}a_v +\text{h.c.}
		\end{split}\]
Using the scattering equation \eqref{eq:eta-scat0}, this implies that
		\begin{equation}\label{eq:commKA1}
		\begin{split}
		[\cK, A_1] + \text{h.c.}&= -\frac{1}{\sqrt N} \sum_{\substack{ r\in \Lambda_+^*, v\in P_L:\\ v+r\neq 0 }} N^{\kappa} (\widehat{V}(./N^{1-\kappa})\ast  \widehat{f}_{N})(r)  b^*_{v+r} a^*_{-r}a_v\\
		&\hspace{0.4cm} + \Pi_{1} +\Pi_{2} + \Pi_{3}  +\text{h.c.}, 
		\end{split}
		\end{equation}
where
		\begin{equation}\label{eq:commKA2}
		\begin{split}
		\Pi_{1} =&\; \frac{1}{\sqrt N} \sum_{\substack{ r\in P_H^c, v\in P_L:\\ v+r\neq 0 }} N^{\kappa} (\widehat{V}(./N^{1-\kappa})\ast  \widehat{f}_{N})(r)  b^*_{v+r} a^*_{-r}a_v  , \\ 
		\Pi_{2} =&\; \frac{1}{\sqrt N} \sum_{\substack{ r\in P_H, v\in P_L:\\ v+r\neq 0 }} N^{3-2\kappa}\lambda_\ell (\widehat{\chi}_\ell\ast \widehat{f}_N)(r)  b^*_{v+r} a^*_{-r}a_v  , \\ 
		\Pi_{3} =&\;  \frac{2}{\sqrt N } \sum_{r\in P_H, v\in P_L } r\cdot v\,\eta_r b^*_{r+v}a^*_{-r}a_v.
		\end{split}
		\end{equation}
Since the first term on the r.h.s. of \eqref{eq:commKA1} appears explicitly in \eqref{eq:propcommKA}, let us explain how to control the operators $ \Pi_{1}, \Pi_{2}$ and $\Pi_{3}$, defined in \eqref{eq:commKA2}.
		
To bound the operator $ \Pi_{1}$, we note first that Lemma \ref{sceqlemma} $ii)$ implies that 
		\[ \big|(\widehat{V}(./N^{1-\kappa})\ast  \widehat{f}_{N})(r)\big|\leq (\widehat{V}(./N^{1-\kappa})\ast  \widehat{f}_{N})(0) = \int dx \;V(x)f_\ell(x)\leq C. \]
Given $\xi\in\cF_+^{\leq N}$, we apply (\ref{eq:Ntheta-bds}) and estimate $ \Pi_1$ by
		\begin{equation}\begin{split}\label{eq:commKA3}
		&\bigg| \frac{1}{\sqrt N} \sum_{\substack{ r\in P_H^c, v\in P_L:\\ v+r\neq 0 }} N^{\kappa} (\widehat{V}(./N^{1-\kappa})\ast  \widehat{f}_{N})(r)  \langle\xi,  b^*_{v+r} a^*_{-r}a_v \xi\rangle\bigg|\\
		&\leq \frac{CN^\kappa}{\sqrt N}  \sum_{\substack{r\in\Lambda^*_+, v\in P_L: \\  |r|\leq N^{3\beta/2},  v+r\neq 0 }} \|a_{-r}  a_{v+r}\xi\| \|a_v \xi \|  \\
		&\hspace{0.4cm}+ \frac{CN^\kappa}{\sqrt N} \sum_{j=3}^{2\lfloor m_0\rfloor-1} \!\!\!\!\sum_{\substack{r\in P_H^c, v\in P_L: \\ N^{j\beta/2}\leq |r|\leq N^{(j+1)\beta/2}, \\ v+r\neq 0 }} \hspace{-0.75cm}\|a_{-r} (\cN_{\geq \frac12 N^{j\beta/2}}+1)^{-1/2} a_{v+r}\xi\| \|a_v(\cN_{\geq \frac12 N^{j\beta/2}}+1)^{1/2}\xi \| \\
		&\hspace{0.4cm}+ \frac{CN^\kappa}{\sqrt N} \sum_{\substack{r\in P_H^c, v\in P_L: \\ N^{\lfloor m_0 \rfloor \beta}\leq |r|\leq N^{\alpha}, \\ v+r\neq 0 }} \hspace{-0.75cm}\|a_{-r} (\cN_{\geq \frac12 N^{\lfloor m_0 \rfloor\beta}}+1)^{-1/2} a_{v+r}\xi\| \|a_v(\cN_{\geq \frac12 N^{\lfloor m_0 \rfloor\beta}}+1)^{1/2}\xi \| 
		%&\leq  CN^{\kappa+3\beta/2} \langle\xi, \cN_+\xi\rangle   +\sum_{j=3}^{m_0}CN^{\kappa+3j\beta/2} \langle\xi, (\cN_{\geq \frac12 N^{(j-1)\beta}}+1)\xi\rangle
		\end{split}\end{equation}	
Cauchy-Schwarz implies that the first term on the r.h.s. of \eqref{eq:commKA3} is bounded by
		\[\frac{CN^\kappa}{\sqrt N}  \sum_{\substack{r\in\Lambda^*_+, v\in P_L: \\ 0\leq  |r|\leq N^{3\beta/2},  v+r\neq 0 }} \|a_{-r}  a_{v+r}\xi\| \|a_v \xi \|\leq CN^{3\beta/4+\kappa} \langle\xi, \cK_{\leq N^{3\beta/2} } \xi\rangle. \]
To bound the second contribution on the r.h.s. of \eqref{eq:commKA3}, we estimate
		%\[\begin{split}
		%&\frac{CN^\kappa}{\sqrt N}\sum_{\substack{r\in P_H^c, v\in P_L: \\ N^{\beta}\leq |r|\leq N^{3\beta/2}, \\ v+r\neq 0 }} \hspace{-0.75cm}\|a_{-r} (\cN_{\geq N^{\beta}}+1)^{-1/2} a_{v+r}\xi\| \|a_v(\cN_{\geq N^{\beta}}+1)^{1/2}\xi\| \\
		%&\; \leq \delta \langle \xi, \cK \xi\rangle + \delta^{-1}CN^{\beta +2\kappa}\langle\xi, (\cN_{\geq N^{\beta}}+1)\xi\rangle.
		%\end{split}\]
		\begin{equation}\label{eq:commKA31}\begin{split}
		&\frac{CN^\kappa}{\sqrt N} \sum_{j=3}^{2\lfloor m_0\rfloor-1} \!\!\!\!\sum_{\substack{r\in P_H^c, v\in P_L: \\ N^{j\beta/2}\leq |r|\leq N^{(j+1)\beta/2}, \\ v+r\neq 0 }} \hspace{-0.75cm}\|a_{-r} (\cN_{\geq \frac12 N^{j\beta/2}}+1)^{-1/2} a_{v+r}\xi\| \|a_v(\cN_{\geq \frac12 N^{j\beta/2}}+1)^{1/2}\xi \| \\
		&= \frac{CN^\kappa}{\sqrt N} \sum_{j=3}^{2\lfloor m_0\rfloor-1} \!\!\!\!\sum_{\substack{r\in P_H^c, v\in P_L: \\ N^{j\beta/2}\leq |r|\leq N^{(j+1)\beta/2}, \\ v+r\neq 0 }} \Big(|v+r|\|a_{-r} (\cN_{\geq \frac12 N^{j\beta/2}}+1)^{-1/2} a_{v+r}\xi\| \Big)\\
		&\hspace{6cm}\times\Big(|v+r|^{-1} \|a_v(\cN_{\geq \frac12 N^{j\beta/2}}+1)^{1/2}\xi \|\Big) \\
		&\;\leq  C \sum_{j=3}^{2\lfloor m_0\rfloor-1}N^{ j\beta/4 + \beta/4+\kappa} \| \cK^{1/2}\xi\|  \| (\cN_{\geq \frac12 N^{j\beta/2}}+1)^{1/2}\xi\|.
		\end{split}\end{equation}
Similarly, we find that 
		\begin{equation}\label{eq:commKA32}\begin{split}
		&\frac{CN^\kappa}{\sqrt N} \sum_{\substack{r\in P_H^c, v\in P_L: \\ N^{\lfloor m_0 \rfloor \beta}\leq |r|\leq N^{\alpha}, \\ v+r\neq 0 }} \hspace{-0.75cm}\|a_{-r} (\cN_{\geq \frac12 N^{\lfloor m_0 \rfloor\beta}}+1)^{-1/2} a_{v+r}\xi\| \|a_v(\cN_{\geq \frac12 N^{\lfloor m_0 \rfloor\beta}}+1)^{1/2}\xi \|  \\
		&\;\leq  C N^{ \alpha/2 +\kappa} \| \cK^{1/2}\xi\|  \| (\cN_{\geq \frac12 N^{\lfloor m_0 \rfloor \beta}}+1)^{1/2}\xi\|.
		\end{split}\end{equation}

Collecting the previous three bounds and using $ |ab|\leq |a|^2+|b|^2$ shows that
		\begin{equation}\label{eq:commKA4}
		\begin{split}
		\pm \Pi_1\leq &\; \delta  \cK   + CN^{3\beta/4+\kappa}\cK_{\leq N^{3\beta/2}} + C \delta^{-1}\sum_{j=3}^{2\lfloor m_0\rfloor -1}N^{ j\beta/2 + \beta/2+2\kappa}\cN_{\geq \frac12 N^{j\beta/2}} \\
		&\; + C \delta^{-1}N^{ \alpha +2\kappa}\cN_{\geq \frac12 N^{\lfloor m_0 \rfloor \beta}} + C \delta^{-1} m_0 N^{\alpha +2\kappa}
		\end{split}
		\end{equation}
for some constant $C>0$ and all $\delta>0$.

Next, let us switch to $ \Pi_2$ and $\Pi_3$, defined in \eqref{eq:commKA2}. From Lemma \ref{sceqlemma} $i)$, we recall that $ |N^{3-2\kappa}\lambda_\ell | \leq C N^\kappa $. Moreover, with $ (\widehat{\chi}_\ell\ast \widehat{f}_N)(r) =  \widehat{\chi_\ell}(r) + N^{-1}\eta_r$ and the representation
		\[\widehat{\chi_\ell}(r) = \frac{4\pi}{|r|^2} \Big(\frac{\sin(\ell|r|)}{|r|}-\ell\cos(\ell |r|) \Big),\]
we find for all $r\in\Lambda_+^*$ that
		\begin{equation}\label{eq:chiastfNbnd}
		|(\widehat{\chi}_\ell\ast \widehat{f}_N)(r)| \leq C (1+N^{\kappa-1})|r|^{-2}\leq C|r|^{-2}.
		\end{equation}
Consequently, Cauchy-Schwarz implies
		\begin{equation}\label{eq:commKA5}
		\begin{split}
		|\langle\xi, \Pi_2\xi\rangle| \leq &\; \frac{CN^{\kappa} }{\sqrt N} \sum_{\substack{ r\in P_H, v\in P_L:\\ v+r\neq 0 }} |r|^{-2} \|  a_{v+r} (\cN_{\geq \frac12 N^{\alpha} }+1)^{-1/2}a_{-r} \xi\| \|a_v (\cN_{\geq \frac12 N^{\alpha}}+1)^{1/2}\xi\|\\
		\leq &\;CN^{\kappa-\alpha/2} \langle\xi, (\cN_{\geq \frac12 N^{\alpha}}+1)\xi\rangle.
		\end{split}
		\end{equation}
Similarly, we obtain
		\begin{equation}\label{eq:commKA6}
		\begin{split}
		|\langle\xi, \Pi_3\xi\rangle| \leq \;&  \frac{C}{\sqrt N } \sum_{r\in P_H, v\in P_L } |r| | v|  |\eta_r| \|  a_{v+r} (\cN_{\geq \frac12 N^{\alpha} }+1)^{-1/2}a_{-r} \xi\| \|a_v (\cN_{\geq \frac12 N^{\alpha}}+1)^{1/2}\xi\|\\
		\leq &\; CN^{\kappa-\alpha/2-1/2} \| \cK^{1/2}\xi\|  \| \cK_L^{1/2}(\cN_{\geq \frac12 N^{\alpha}}+1)^{1/2} \xi\|. 
		\end{split}
		\end{equation}
Combining \eqref{eq:commKA4}, \eqref{eq:commKA5} and \eqref{eq:commKA6} and defining $ \cE_{[\cK,A]} = \sum_{i=1}^3(\Pi_i+\text{h.c.})$ proves the claim.
\end{proof}

With the help of Prop. \ref{prop:commKA}, we obtain a rough bound for the action of $e^A$ on $\cK$. 
\begin{cor}\label{cor:KgrowA}  Assume the exponents $\alpha,\beta$ satisfy (\ref{eq:alphabeta}). Let $ m_0\in\mathbb{R}$ be such that $m_0 \beta =\alpha$ ($3 < m_0 < 5$ from (\ref{eq:alphabeta})). Then, there exists a constant $C>0$ such that the operator inequality
		\begin{equation} \label{eq:corKgrowA}
		\begin{split}
		e^{-sA} \cK e^{sA} &\leq C N^{3\beta/4+\kappa}\cK+  C \cV_N +CN^{\kappa } \cN_+ + C m_0 N^{\alpha +2\kappa}.
		\end{split}\end{equation}	
for all $ s\in [-1;1]$ and for all $N\in\NN$ sufficiently large.
\end{cor}
\begin{proof}
We apply Gronwall's lemma to $\varphi_{\xi}(s) = \langle \xi, e^{-sA} \cK e^{sA}\xi\rangle$ for some $\xi \in \cF_+^{\leq N}$ with $ \|\xi\|=1$, which has derivative
		\[\partial_s \varphi_\xi(s) = \langle \xi, e^{-sA} [\cK, A] e^{sA}\xi\rangle.\] 
We use the identity \eqref{eq:propcommKA} and bound
		\[\begin{split}
		\pm \frac{1}{\sqrt N} \sum_{\substack{ u\in \Lambda_+^*, p\in P_L:\\ p+u\neq 0 }} N^{\kappa} (\widehat{V}(./N^{1-\kappa})\ast  \widehat{f}_{N})(u) \big( b^*_{p+u} a^*_{-u}a_p +\text{h.c.}\big) \leq C \cV_N + CN^\kappa  (\cN_++1).
		\end{split}\]
This estimate can be proved in the same way as \eqref{eq:corVN1} by observing $ \sup_{x\in\Lambda} |f_N(x)|\leq 1$, by Lemma \ref{sceqlemma} $ii)$. Together with \eqref{eq:propcommKA}, \eqref{eq:propcommKAerror} (choosing $ \delta=1$), Corollary \ref{cor:VNgrowA}, Lemma \ref{lm:NresgrowA} and Lemma \ref{lm:KresgrowA}, we obtain for sufficiently large N
		\[\begin{split}
		\varphi_{\xi}(s) \leq &\; C \varphi_{\xi}(s)  + C\langle\xi, \cV_N\xi\rangle + C N^{\kappa}\langle\xi, (\cN_++1)\xi\rangle + CN^{3\beta/4+\kappa}\langle \xi, \cK_{\leq N^{3\beta/2}} \xi\rangle\\
		&\; + C \sum_{j=3}^{2\lfloor m_0\rfloor -1}N^{ j\beta/2 + \beta/2+2\kappa} \langle\xi, \cN_{\geq \frac12 N^{j\beta/2}} \xi\rangle + C N^{ \alpha +2\kappa} \langle\xi, \cN_{\geq \frac12 N^{\lfloor m_0 \rfloor \beta}}\xi\rangle + C N^{\alpha +2\kappa}\\
		 \leq &\; C \varphi_{\xi}(s)  + C\langle\xi, \cV_N\xi\rangle + C N^{\kappa}\langle\xi, \cN_+\xi\rangle + CN^{3\beta/4+\kappa}\langle \xi, \cK_{\leq N^{3\beta/2}} \xi\rangle\\
		&\; + C N^{ 2\kappa -\beta } \langle\xi, \cK \xi\rangle + C N^{ 2\kappa-\beta} \langle\xi, \cK \xi\rangle + C N^{\alpha +2\kappa}\\
		 \leq &\; C \varphi_{\xi}(s) +C N^{3\beta/4+\kappa}\langle \xi, \cK  \xi\rangle + C\langle\xi, \cV_N\xi\rangle + C N^{\kappa}\langle\xi, \cN_+\xi\rangle  + C N^{\alpha +2\kappa},
		\end{split}\]
where we used as usual the operator inequality $ \cN_{\geq \Theta}\leq \Theta^{-2} \cK $ and, moreover, that $ \alpha -2\lfloor m_0\rfloor \beta \leq ( 1-\lfloor m_0 \rfloor)\beta \leq -\beta$ for $ m_0\beta=\alpha$ and $ \alpha>3\beta+2\kappa\geq 3\beta$ (from (\ref{eq:alphabeta})). The claim follows now from Gronwall's lemma.
\end{proof}
%%%%%%%%%%%%%%%%%%%%%%%%%%%%%%%%%%%%%%%%%%%%%%%%%%%%%%%%%%%%%%%%%%%
%%%%%%%%%%%%%%%%%%%%%%%%%%%%%%%%%%%%%%%%%%%%%%%%%%%%%%%%%%%%%%%%%%%
%%%%%%%%%%%%%%%%%%%%%%%%%%%%%%%%%%%%%%%%%%%%%%%%%%%%%%%%%%%%%%%%%%%

\subsection{Action of Cubic Renormalization on Excitation Hamiltonian} 

In this subsection, we are going to determine the main contributions to the excitation Hamiltonian $\cJ_N= e^{-A} \cG_{N}^\text{eff} e^{A}$. From (\ref{eq:GNeff}), and recalling the definition of the sets $P_H = \{ p \in \L^*_+ : |p| \geq N^\alpha \}, P_L = \{ p \in \L^*_+ : |p| \leq N^\beta \}$, we can decompose 
\[ \cJ_N = \cJ_N^{(0)} +\cJ_N^{(2)} +\cJ_N^{(3)} +\cJ_N^{(4)}   \]
where the self-adjoint operators $ \cJ_N^{(i)}, i=0, 2, 3, 4,$ are defined by
\begin{equation}\begin{split}\label{eq:JN0to4} 
\cJ_N^{(0)} =&\;4\pi \frak{a}_0 N^{\kappa} e^{-A}(N-\cN_+) e^A+ \big[\widehat V(0)-4\pi \frak{a}_0\big]N^{\kappa} e^{-A}\cN_+ (1-\cN_+/N) e^A, \\
 \cJ_N^{(2)}=&\; N^{\kappa}\widehat V(0)\sum_{p\in P_H^c}   e^{-A}b^*_pb_p e^{A} + 4\pi \frak{a}_0N^{\kappa} \sum_{p\in P^c_H}  e^{-A}\big[ b^*_p b^*_{-p} + b_p b_{-p} \big] e^A\\
\cJ_N^{(3)} =&\;   \frac{1}{\sqrt N}\sum_{\substack{p\in\Lambda_+^*, q\in P_L:\\ p+q\neq 0}} N^{\kappa}\widehat V(p/N^{1-\kappa}) e^{-A}\big[ b^*_{p+q}a^*_{-p}a_q+ \text{h.c.}\big]e^A,\\
\cJ_N^{(4)}=&\; e^{-A}\cH_N e^A   = e^{-A}\cK e^A +  e^{-A}\cV_N e^A.  
\end{split}\end{equation} 

\subsubsection{Analysis of $ \cJ_N^{(0)} $}\label{sec:JN0}
The goal of this section is to determine the main contributions to $ \cJ_N^{(0)}$, which was defined in equation \eqref{eq:JN0to4}. We recall that 
		\[\cJ_N^{(0)} = 4\pi \frak{a}_0 N^{\kappa} e^{-A}(N-\cN_+) e^A+ \big[\widehat V(0)-4\pi \frak{a}_0\big]N^{\kappa} e^{-A}\cN_+ (1-\cN_+/N) e^A. \]
In order to determine the main contributions to $ \cJ_N^{(0)}$, we first prove a slight generalization of \cite[Lemma 8.6]{BBCS}. The lemma will also be useful in the following Section \ref{sec:JN2}. \begin{lemma}\label{lm:FinftycNopA} Assume $\alpha,\beta$ satisfy (\ref{eq:alphabeta}). Let $ k\in\NN_0$ and let $ F = (F_p)_{p\in\Lambda_+^*} \in \ell^\infty(\Lambda_+^*) $. Then, there exists $C>0$ s.t.
		\begin{equation}\label{eq:FinftycNopA}
		\pm \bigg(\sum_{p\in\Lambda_+^*} F_p e^{-A} a^*_pa_p\,\cN_+^k e^{A}   - \sum_{p\in\Lambda_+^*} F_p a^*_pa_p \,\cN_+^k \bigg)\leq C N^{-3\beta/2}\|F\|_\infty (\cN_{\geq \frac12 N^\alpha}+1)(\cN_++1)^k
		\end{equation}
for all $N\in\NN$ sufficiently large.
\end{lemma}
\begin{proof}
We compute that
		\[\begin{split}
		 &\sum_{p\in\Lambda_+^*} F_p e^{-A} a^*_pa_p\,\cN_+^k e^{A}   - \sum_{p\in\Lambda_+^*} F_p a^*_pa_p \,\cN_+^k= \int_0^1ds\;\sum_{p\in\Lambda_+^*}    F_p e^{-sA} [a^*_pa_p\cN_+^k, A_1] e^{sA} +\text{h.c.}\\
		  %&= \frac{1}{\sqrt N } \int_0^1ds\; \sum_{r\in P_H, v\in P_L }  (F_{r+v} + F_{-r}- F_v) \eta_r e^{-sA}  b^*_{r+v} a^*_{-r}a_v \cN_+^ke^{sA} \\
		  %&\hspace{0.4cm} + \frac{k}{\sqrt N }  \int_0^1ds\; \sum_{\substack{p\in\Lambda_+^*, r\in P_H,\\ v\in P_L }} F_p\eta_r e^{-sA} a^*_p a_p  b^*_{r+v} a^*_{-r}a_v (\cN_+ + \Theta(\cN_+))^{k-1} e^{sA}  +\text{h.c.}\\
		  &= \frac{1}{\sqrt N } \int_0^1ds\; \sum_{r\in P_H, v\in P_L }  (F_{r+v} + F_{-r}- F_v) \eta_r e^{-sA}  b^*_{r+v} a^*_{-r}a_v\, \cN_+^ke^{sA} \\
		  &\hspace{0.4cm} +  \frac{k}{\sqrt N } \int_0^1ds\; \sum_{r\in P_H, v\in P_L }  (F_{r+v} + F_{-r}) \eta_r e^{-sA}   b^*_{r+v} a^*_{-r}a_v (\cN_+ + \Theta(\cN_+))^{k-1} e^{sA} \\
		  &\hspace{0.4cm} + \frac{k}{\sqrt N }  \int_0^1ds\; \sum_{\substack{p\in\Lambda_+^*, r\in P_H,\\ v\in P_L }} F_p\eta_r e^{-sA}   b^*_{r+v} a^*_{-r}a^*_p a_pa_v (\cN_+ + \Theta(\cN_+))^{k-1} e^{sA}  +\text{h.c.}\\
		 \end{split}\]
for some function $ \Theta:\NN\to (0,1)$, by the mean value theorem. Applying (\ref{eq:Ntheta-bds}), Lemma \ref{lm:NresgrowA} and Cauchy-Schwarz, the first two contributions on the r.h.s. of the last equation can be controlled by
		\[\begin{split}
		&\bigg| \frac{1}{\sqrt N } \int_0^1ds\; \sum_{r\in P_H, v\in P_L }  (F_{r+v} + F_{-r}- F_v) \eta_r \langle \xi, e^{-sA}  b^*_{r+v} a^*_{-r}a_v \cN_+^ke^{sA} \xi\rangle\bigg| \\
		&\leq \frac{\|F\|_\infty}{\sqrt N } \int_0^1ds\; \sum_{r\in P_H, v\in P_L }  \|(\cN_{\geq \frac12N^\alpha}+1)^{-1/2} a_{r+v} a_{-r} (\cN_++1)^{k/2}e^{sA}\xi\| \\
		&\hspace{6.5cm}\times| \eta_r| \| (\cN_{\geq \frac12N^\alpha}+1)^{1/2}a_v \cN_+^{k/2}e^{sA} \xi\|\\
		&\leq C N^{\kappa-\alpha/2} \|F\|_\infty \int_0^1ds\; \langle \xi, e^{-sA} (\cN_{\geq \frac12N^\alpha}+1)(\cN_++1)^{k}e^{sA}\xi\rangle \\
		&\leq C N^{-3\beta/2} \|F\|_\infty \langle \xi,  (\cN_{\geq \frac12N^\alpha}+1)(\cN_++1)^{k}\xi\rangle
		\end{split}\] 
and, analogously, 
		\[\begin{split}
		&\bigg| \frac{k}{\sqrt N } \int_0^1ds\; \sum_{r\in P_H, v\in P_L }  (F_{r+v} + F_{-r}) \eta_r \langle \xi, e^{-sA}  b^*_{r+v} a^*_{-r}a_v (\cN_+ + \Theta(\cN_+))^{k-1} e^{sA} \xi\rangle\bigg| \\
		&\leq \frac{C\|F\|_\infty}{\sqrt N } \int_0^1ds\; \sum_{r\in P_H, v\in P_L }  \|(\cN_{\geq \frac12N^\alpha}+1)^{-1/2} a_{r+v} a_{-r} (\cN_++1)^{(k-1)/2}e^{sA}\xi\| \\
		&\hspace{6.5cm}\times| \eta_r| \| (\cN_{\geq \frac12N^\alpha}+1)^{1/2}a_v (\cN_++1)^{(k-1)/2}e^{sA} \xi\|\\
		& \leq C N^{-3\beta/2} \|F\|_\infty \langle \xi,  (\cN_{\geq \frac12N^\alpha}+1) (\cN_++1)^{k-1}\xi\rangle.
		\end{split}\] 
Finally, we also bound 
		\[\begin{split}
		&\bigg| \frac{k}{\sqrt N } \int_0^1ds\; \sum_{p\in\Lambda_+^* r\in P_H, v\in P_L }  F_p \eta_r \langle \xi, e^{-sA}  b^*_{r+v} a^*_{-r}a^*_pa_pa_v (\cN_+ + \Theta(\cN_+))^{k-1} e^{sA} \xi\rangle\bigg| \\
		&\leq \frac{\|F\|_\infty}{\sqrt N } \int_0^1ds\; \sum_{p\in\Lambda_+^*, r\in P_H, v\in P_L }  \|(\cN_{\geq \frac12N^\alpha}+1)^{-1/2} a_p a_{r+v} a_{-r} (\cN_++1)^{(k-1)/2}e^{sA}\xi\| \\
		&\hspace{6cm}\times| \eta_r| \| (\cN_{\geq \frac12N^\alpha}+1)^{1/2} a_pa_v (\cN_++1)^{(k-1)/2}e^{sA} \xi\|\\
		& \leq C N^{-3\beta/2} \|F\|_\infty \langle \xi,  (\cN_{\geq \frac12N^\alpha}+1) (\cN_++1)^{k}\xi\rangle.
		\end{split}\] 
Combining the last three estimates concludes the proof of \eqref{eq:FinftycNopA}.
\end{proof}

\begin{cor}\label{cor:JN0}
Assume $\alpha,\beta$ satisfy (\ref{eq:alphabeta}). Let $ \cJ_N^{(0)}$ be defined as in \eqref{eq:JN0to4}. Then 
		\[\cJ_N^{(0)} =4\pi \frak{a}_0 N^{\kappa} (N-\cN_+) + \big[\widehat V(0)-4\pi \frak{a}_0\big]N^{\kappa} \cN_+ (1-\cN_+/N)  + \cE_{\cJ_N}^{(0)} \]
where the self-adjoint operator $ \cE_{\cJ_N}^{(0)}$ satisfies
		\[ \pm e^A  \cE_{\cJ_N}^{(0)} e^{-A}\leq C N^{\kappa-3\beta/2} (\cN_{\geq \frac12 N^\alpha}+1)\]
for all $N\in\NN$ sufficiently large.
\end{cor}
\begin{proof} The claim follows immediately from Lemma \ref{lm:FinftycNopA} and Lemma \ref{lm:NresgrowA}. 
\end{proof}

%%%%%%%%%%%%%%%%%%%%%%%%%%%%%%%%%%%%%%%%%%%%%%%%%%%%%%%%%%%%%%%%%%%
%%%%%%%%%%%%%%%%%%%%%%%%%%%%%%%%%%%%%%%%%%%%%%%%%%%%%%%%%%%%%%%%%%%
%%%%%%%%%%%%%%%%%%%%%%%%%%%%%%%%%%%%%%%%%%%%%%%%%%%%%%%%%%%%%%%%%%%

\subsubsection{Analysis of $ \cJ_N^{(2)} $}\label{sec:JN2}

In this section, we determine the main contributions to $ \cJ_N^{(2)}$, defined in \eqref{eq:JN0to4}. We recall 
		\begin{equation}\label{eq:JN2} \cJ_N^{(2)}= N^{\kappa}\widehat V(0)\sum_{p\in P_H^c}   e^{-A}b^*_pb_p e^{A} + 4\pi \frak{a}_0N^{\kappa} \sum_{p\in P^c_H}  e^{-A}\big[ b^*_p b^*_{-p} + b_p b_{-p} \big] e^A. \end{equation}
\begin{prop}\label{prop:JN2}
Assume $\alpha,\beta$ satisfy (\ref{eq:alphabeta}). Let $ \cJ_N^{(2)}$ be defined as in \eqref{eq:JN0to4}. Then 
\[\cJ_N^{(2)} = N^{\kappa}\widehat V(0)\sum_{p\in P_H^c}  b^*_pb_p + 4\pi \frak{a}_0N^{\kappa} \sum_{p\in P^c_H}  \big[ b^*_p b^*_{-p} + b_p b_{-p} \big]  +\cE_{\cJ_N}^{(2)} \]
and there exists a constant $C>0$ such that  
\[ \pm e^A \cE_{\cJ_N}^{(2)} e^{-A}\leq CN^{-3\beta}\cK + CN^{\alpha+2\kappa}\]
for all $N\in\NN$ sufficiently large.
\end{prop}
\begin{proof} From $ b^*_pb_p = a^*_pa_p(1-\cN_+/N+1/N)$ and Corollary \ref{cor:JN0}, we conclude that 
		\[N^{\kappa}\widehat V(0)\sum_{p\in P_H^c}  e^{-A} b^*_pb_p e^{A} = N^{\kappa}\widehat V(0)\sum_{p\in P_H^c} b^*_pb_p+  \cE_{\cJ_N}^{(21)},\]
where the the self-adjoint operator $ \cE_{\cJ_N}^{(21)}$ is such that
		\[ \pm e^A  \cE_{\cJ_N}^{(21)} e^{-A}\leq C N^{\kappa-3\beta/2} (\cN_{\geq \frac12 N^\alpha}+1)\leq C N^{\kappa-2\alpha-3\beta/2} \cK + C N^{\kappa-3\beta/2}.\]
Notice that we used here the operator inequality $ \cN_{\geq \frac12 N^\alpha}\leq 4N^{-2\alpha}\cK$. Since 
		\[C N^{\kappa-2\alpha-3\beta/2} \cK + C N^{\kappa-3\beta/2}\leq CN^{-3\beta}\cK + CN^{\alpha+2\kappa}\]
for all $\alpha>3\beta+2\kappa\geq 0$, it only remains to analyse the second contribution in \eqref{eq:JN2}. To this end, we compute
		\[\begin{split}
		 [b^*_p b^*_{-p}+b_p b_{-p}, b^*_{v+r}a^*_{-r}a_v ] &= - b^*_{v+r} b^*_{p }b^*_{-r}\delta_{-p,v}- b^*_{v+r} b^*_{-r}b^*_{-p }\delta_{p,v}\\
		 &\hspace{0.4cm}+a^*_{-r}a_v b_p (1-\cN_+/N)  \delta_{-p,r+v} - \frac1N  a^*_{r+v}a^*_{-r}a_va_{-p}b_p \\
		 &\hspace{0.4cm} + a^*_{-r}a_v b_{-p} (1-\cN_+/N) \delta_{p,r+v}- \frac1N  a^*_{r+v}a^*_{-r}a_va_{p}b_{-p}
		 \end{split}\]
for all $ p\in P_H^c$, $ r\in P_H$ and $v\in P_L$. As a consequence, we find that
		\begin{equation}\label{eq:JN21}
		\begin{split}
		&\sum_{p\in P^c_H}  e^{-A}\big[ b^*_p b^*_{-p} + b_p b_{-p} \big] e^A - \sum_{p\in P^c_H}  \big[ b^*_p b^*_{-p} + b_p b_{-p} \big]\\
		&=\frac2{\sqrt N}\int_0^1 ds \!\!\!\sum_{\substack{ r\in P_H, v\in P_L}} \!\!\!\eta_r e^{-sA} \big[  a^*_{-r}a_v b_{-r-v}\chi_{\{|r+v|\leq N^{\alpha}\}} (1-\cN_+/N) - b^*_{v+r} b^*_{-v }b^*_{-r}   \big]e^{sA}\\
		&\hspace{0.4cm} + \frac2{ N^{3/2}}\int_0^1 ds\; \sum_{\substack{p\in P_H^c, r\in P_H, \\ v\in P_L }} \eta_r  e^{-sA}a^*_{r+v}a^*_{-r}a_va_{p}b_{-p}e^{sA} +\text{h.c.},
		\end{split}\end{equation}
where $ \chi_{\{|\cdot|\leq N^{\alpha}\}}$ denotes the characteristic function for the set $\{p\in\Lambda_+^*:|p|\leq N^{\alpha}\}$. 

Let us now estimate the size of the different contributions on the r.h.s. in \eqref{eq:JN21}. Applying (\ref{eq:Ntheta-bds}), Lemma \ref{lm:NresgrowA} and Cauchy-Schwarz, we obtain on the one hand that
		\[\begin{split}
		&\bigg| \frac2{\sqrt N}\int_0^1 ds\!\!\!\sum_{\substack{ r\in P_H, v\in P_L}}\!\!\! \eta_r \langle e^{sA} \xi,  \big[  a^*_{-r}a_v b_{-r-v}\chi_{\{|r+v|\leq N^{\alpha}\}} (1-\cN_+/N) - b^*_{v+r} b^*_{-v }b^*_{-r}   \big]e^{sA}\xi\rangle\bigg|\\
		&\leq C\int_0^1 ds\!\!\!\sum_{\substack{ r\in P_H, v\in P_L}} \Big( |\eta_r| \|  a_{-r-v} e^{sA}\xi\| \| a_{-r} e^{sA}\xi\| \\
		&\hspace{3.5cm}+ |\eta_r| \| (\cN_{\geq \frac12 N^\alpha}+1)^{-1/2} a_{-r} a_{r+v} e^{sA}\xi\| \| (\cN_{\geq \frac12 N^\alpha}+1)^{1/2} e^{sA}\xi\|  \Big)\\
		&\leq CN^{3\beta/2+\kappa-\alpha/2}\int_0^1ds\; \langle \xi, e^{-sA} (\cN_{\geq \frac12N^\alpha}+1)e^{sA}\xi\rangle\leq C\langle \xi,  (\cN_{\geq \frac12N^\alpha}+1)\xi\rangle
		\end{split}\]
for all $ \alpha>3\beta +2\kappa\geq 0$ and $N\in\NN$ sufficiently large. On the other hand, proceeding in the same way, the last contribution on the r.h.s. in \eqref{eq:JN21} is bounded by
		\[\begin{split}
		&\bigg| \frac2{ N^{3/2}}\int_0^1 ds\; \sum_{\substack{p\in P_H^c, r\in P_H, \\ v\in P_L }} \eta_r  \langle e^{sA} \xi,a^*_{r+v}a^*_{-r}a_va_{p}b_{-p}e^{sA}\xi\rangle \bigg| \\
		&\leq \frac{C}{N^{3/2}} \int_0^1 ds\; \sum_{\substack{p\in P_H^c, r\in P_H, \\ v\in P_L }}  \| (\cN_{\geq \frac12 N^\alpha}+1)^{-1/2}a_{r+v}a_{-r} e^{sA}\xi\|\\
		&\hspace{5cm} \times |\eta_r|\| (\cN_{\geq \frac12 N^\alpha}+1)^{1/2}a_va_{p}b_{-p}e^{sA}\xi\| \\
		&\leq CN^{\alpha+\kappa}\langle \xi,  (\cN_{\geq \frac12N^\alpha}+1)\xi\rangle.
		\end{split}\]
Combining the last two estimates with the identity \eqref{eq:JN21}, we conclude that 
		\[4\pi \frak{a}_0N^{\kappa} \sum_{p\in P^c_H}  e^{-A}\big[ b^*_p b^*_{-p} + b_p b_{-p} \big] e^A = 4\pi \frak{a}_0N^{\kappa} \sum_{p\in P^c_H}  \big[ b^*_p b^*_{-p} + b_p b_{-p} \big] +  \cE_{\cJ_N}^{(22)},\]
where the the self-adjoint operator $ \cE_{\cJ_N}^{(22)}$ is such that
		\[ \pm e^A  \cE_{\cJ_N}^{(22)} e^{-A}\leq CN^{\alpha+2\kappa} (\cN_{\geq \frac12N^\alpha}+1) \leq CN^{-3\beta}\cK + CN^{\alpha+2\kappa}.\]
Setting $ \cE_{\cJ_N}^{(2)} = \cE_{\cJ_N}^{(21)}+ \cE_{\cJ_N}^{(22)}$, this concludes the proposition.
\end{proof}

%%%%%%%%%%%%%%%%%%%%%%%%%%%%%%%%%%%%%%%%%%%%%%%%%%%%%%%%%%%%%%%%%%%
%%%%%%%%%%%%%%%%%%%%%%%%%%%%%%%%%%%%%%%%%%%%%%%%%%%%%%%%%%%%%%%%%%%
%%%%%%%%%%%%%%%%%%%%%%%%%%%%%%%%%%%%%%%%%%%%%%%%%%%%%%%%%%%%%%%%%%%

\subsubsection{Analysis of $ \cJ_N^{(3)} $} \label{sec:JN3}

In this section, we determine the main contributions to $ e^{-A}\cC_N e^A$, where we recall that
		\begin{equation}\label{eq:defCNsecCN} \cC_{N} =   \frac{1}{\sqrt N}\sum_{p\in\Lambda_+^*, q\in P_L: p+q\neq 0} N^{\kappa}\widehat V(p/N^{1-\kappa})\big(b^*_{p+q}a^*_{-p}a_q+ \text{h.c.}\big). 
		\end{equation}
The following proposition summarizes important properties of $ [\cC_N, A]$.
\begin{prop}\label{prop:commCNA} Assume $\alpha,\beta$ satisfy (\ref{eq:alphabeta}). Then 
there exists $C>0$ such that 
		\begin{equation} \label{eq:propcommCNA}
		\begin{split}
		[\cC_N, A] & = \frac {2N^\kappa}N \sum_{r\in P_H, v\in P_L} \!\!\! \big[ \widehat{V}(r/N^{1-\kappa})\eta_r + \widehat{V}((v+r)/N^{1-\kappa}))\eta_r\big] a^*_va_v (1-\cN_+/N) \\ &+ \cE_{[\cC_N,A]}
\end{split}
\end{equation}
where 
\begin{equation}\label{eq:propcommCNAerror}\begin{split}
\pm \cE_{[\cC_N,A]} \leq C N^{\kappa/2-3\beta/2}\cV_N + CN^\kappa (\cN_{\geq \frac12 N^\alpha}+1)+ C  N^{\kappa/2-3\beta/2-1} (\cN_{\geq \frac12 N^\alpha}+1)^2 
\end{split}\end{equation}
for all $N\in\NN$ sufficiently large.
\end{prop}
\begin{proof} With \eqref{eq:defA} we have that
		\[[\cC_N, A]  = [\cC_N, A_1] +\text{h.c.}  \]
Now, let $ p\in\Lambda_+^*, q\in P_L, r\in P_H$ and $v\in P_L$. For $N\in\NN$ sufficiently large, we have $ |v+r|\geq N^\alpha-N^\beta>\frac12 N^\alpha>N^\beta$ so that 
		\[ [a^*_{-p}a_q, b^*_{v+r}] = [a^*_{-p}a_q, a^*_{-r}] = 0. \]
As a consequence, we obtain
		\[ [b^*_{p+q}a^*_{-p}a_q, b^*_{v+r}a^*_{-r}a_v] = -b^*_{v+r }b^*_{-r}a^*_{-p}a_q\delta_{p+q,v}  - b^*_{v+r}b^*_{p+q}a^*_{-r}a_q \delta_{-p,v} \]
as well as
		\[\begin{split}
		&[a^*_qa_{-p}b_{p+q}, b^*_{v+r}a^*_{-r}a_v]\\ %&= [a^*_q a_{-p}, b^*_{v+r}]b_{p+q}a^*_{-r}a_v + a^*_q a_{-p}[b_{p+q},b^*_{v+r}]a^*_{-r}a_v\\
		%& \hspace{0.4cm} + b^*_{v+r}[a^*_{q}a_{-p}, a^*_{-r}a_v]b_{p+q} + b^*_{v+r}a^*_{q}a_{-p}[b_{p+q}, a^*_{-r}a_v]\\
		&= a^*_q a_v (1-\cN_+/N) ( \delta_{v+r,-p}\delta_{p+q,-r} + \delta_{p+q,v+r} \delta_{p,r})+ a^*_qa^*_{-r}a_{p+q}a_v (1-\cN_+/N)\delta_{v+r,-p}\\
		&\hspace{0.4cm}  + a^*_{v+r}a^*_{q}a_{p+q}a_v (1-\cN_+/N)\delta_{p,r}+ a^*_{v+r}a^*_{q}a_{-p}a_v (1-\cN_+/N)\delta_{p+q,-r} \\
		&\hspace{0.4cm} + a^*_qa^*_{-r}a_{-p}a_v (1-\cN_+/N)\delta_{p+q,v+r} - b^*_{v+r}a^*_{-r}a_{-p}b_{p+q}\delta_{v,q} -\frac1N  a^*_q a^*_{v+r}a^*_{-r}a_va_{-p}a_{p+q}. 
		\end{split}\]
Hence, we conclude that
		\begin{equation}\label{eq:commCNA1}
		\begin{split}
		[\cC_N, A_1]+\text{h.c.}& = \frac {2N^\kappa}N \sum_{r\in P_H, v\in P_L} \!\!\! \big[ \widehat{V}(r/N^{1-\kappa})\eta_r + \widehat{V}((v+r)/N^{1-\kappa}))\eta_r\big] a^*_va_v (1-\cN_+/N)\\
		&\hspace{0.4cm}+ (\Xi_{1}+ \Xi_{2}+\Xi_{3}+\Xi_{4}+\Xi_{5}+\text{h.c.}),
		\end{split}
		\end{equation}
where
		\begin{equation}\label{eq:defXi1to5}
		\begin{split}
		\Xi_{1} &= - \frac{1}{N} \sum_{r\in P_H, q, v\in P_L}^* N^\kappa\big[  \widehat{V}((v-q)/N^{1-\kappa})\eta_r + \widehat{V}(v/N^{1-\kappa})\eta_r\big] b^*_{v+r }b^*_{-r}a^*_{q-v}a_q , \\
		\Xi_{2} &=  \frac{1}{N} \sum_{r\in P_H, q, v\in P_L}^* N^\kappa \big[ \widehat{V}((v+r)/N^{1-\kappa})\eta_r+\widehat{V}((v+r+q)/N^{1-\kappa})\eta_r\big] \\
		&\hspace{7cm}\times a^*_qa^*_{r}a_{q+v+r}a_{-v} (1-\cN_+/N) , \\
		\Xi_{3} &=  \frac{1}{N} \sum_{r\in P_H, q, v\in P_L}^* N^\kappa \big[ \widehat{V}(r/N^{1-\kappa})\eta_r +\widehat{V}((r+q)/N^{1-\kappa})\eta_r \big]\\
		&\hspace{7cm}\times a^*_{v+r}a^*_{q}a_{r+q}a_v (1-\cN_+/N), \\
		\Xi_{4} &= - \frac{1}{N} \sum_{p\in\Lambda_+^*, r\in P_H, v\in P_L}^* N^\kappa \widehat{V}(p/N^{1-\kappa})\eta_r b^*_{v+r}a^*_{-r}a_{-p}b_{p+v} , \\
		\Xi_{5} &= - \frac{1}{N^2} \sum_{p\in\Lambda_+^*, r\in P_H, q, v\in P_L}^* N^\kappa \widehat{V}(p/N^{1-\kappa})\eta_r  a^*_q a^*_{v+r}a^*_{-r}a_va_{-p}a_{p+q}. \\		
		\end{split}
		\end{equation}
Let us next estimate the size of the operators $ \Xi_{1}$ to $\Xi_{5}$, defined in \eqref{eq:defXi1to5}. We find
		\[\begin{split}
		|\langle \xi, \Xi_{1}\xi \rangle| &\leq \frac{CN^\kappa}{N} \!\!\!\sum_{r\in P_H, q, v\in P_L} \!\!\! |\eta_r| \|  a_{-r }(\cN_{\geq \frac12 N^\alpha}+1)^{-1/2}a_{v+r}a_{q-v}\xi\| \| a_q (\cN_{\geq \frac12 N^\alpha}+1)^{1/2} \xi\|\\
		&\leq CN^{3\beta/2 + 2\kappa-\alpha/2} \langle\xi, (\cN_{\geq \frac12 N^\alpha}+1)\xi\rangle
 		\end{split} \]
for all $ \xi\in\cF_+^{\leq N}$ and $N\in\NN$ sufficiently large. Similarly, the operators $ \Xi_{2}$ and $\Xi_{3}$ can be controlled by
		\[\begin{split}
		|\langle \xi, \Xi_{2}\xi \rangle| + |\langle \xi, \Xi_3\xi \rangle| & \leq \frac{CN^\kappa}{N} \!\!\!\sum_{r\in P_H, q, v\in P_L} \!\!\! |\eta_r| \Big( \|  a_{q}a_{r}\xi\| \| a_{q+v+r} a_{-v} \xi\| +\|  a_{v+r }a_{q}\xi\| \| a_{q+r} a_{v} \xi\| \Big)\\
		&\leq CN^{3\beta/2 + 2\kappa-\alpha/2} \langle\xi, (\cN_{\geq \frac12 N^\alpha}+1)\xi\rangle.
		\end{split}\]
Switching to position space, the operator $\Xi_{4}$ can be bounded by
		\[\begin{split}
		|\langle \xi, \Xi_{4}\xi \rangle| &= \bigg|  \sum_{\substack{r\in P_H, v\in P_L: \\v+r\neq 0}}\int_{\Lambda^2}dxdy\;  N^{2-2\kappa} V(N^{1-\kappa}(x-y)) \eta_r e^{-ivy} \langle\xi,  b^*_{v+r}a^*_{-r}\check{a}_{x}\check{b}_{y}\xi\rangle \bigg|\\
		&\leq  C N^{\kappa-\alpha/2} \| \cV_N^{1/2}\xi\| \bigg(N^{\kappa-1}\sum_{\substack{r\in P_H}}\int_{\Lambda}dy\;  \Big\| \sum_{v\in P_L}e^{-ivy}a_{v+r}a_{-r}\xi\Big\|^2 \bigg)^{1/2}\\
		& = C N^{\kappa-\alpha/2} \| \cV_N^{1/2}\xi\| \bigg(N^{\kappa-1}\sum_{\substack{r\in P_H, v\in P_L}} \langle\xi, a^*_{v+r}a^*_{-r}a_{-r}a_{v+r}\xi\rangle \bigg)^{1/2}\\
		&\leq C N^{3\kappa/2-\alpha/2-1/2}\| \cV_N^{1/2}\xi\| \| (\cN_{\geq \frac12 N^\alpha}+1)\xi\|
		\end{split}\]
and, similarly, we control the operator $\Xi_{5}$ by 
		\[\begin{split}
		|\langle \xi, \Xi_{5}\xi\rangle| &=   \bigg| \frac{1}{N}\!\! \sum_{r\in P_H, q, v\in P_L} \!\int_{\Lambda^2}\!\!dxdy\, N^{2-2\kappa}V(N^{1-\kappa}(x-y)) e^{-iqy}\eta_r \langle\xi, a^*_q a^*_{v+r}a^*_{-r}a_v\check{a}_{x}\check{a}_{y} \xi\rangle\bigg| \\
		&\leq C N^{\kappa-\alpha/2-1/2}\|\cV_N^{1/2}\xi\| \bigg(N^{\kappa-1}\!\! \sum_{r\in P_H, v\in P_L} \int_{\Lambda}dy\;  \Big\|\sum_{q\in P_L}e^{-iqy} a_q a_{v+r}a_{-r}\xi\Big\|^2  \bigg)^{1/2}\\
		&\leq C N^{3\kappa/2-\alpha/2-1/2}\| \cV_N^{1/2}\xi\| \| (\cN_{\geq \frac12 N^\alpha}+1)\xi\|.
		\end{split}\]
Summarizing the previous estimates and using that $\alpha>3\beta+2\kappa $, we conclude that
		\begin{equation*}\begin{split}
		\pm  \sum_{i=1}^5(\Xi_{i} +\text{h.c.})& \leq C N^{\kappa/2-3\beta/2}\cV_N + CN^\kappa (\cN_{\geq \frac12 N^\alpha}+1)+ C  N^{\kappa/2-3\beta/2-1} (\cN_{\geq \frac12 N^\alpha}+1)^2.
		\end{split}\end{equation*}
Defining $ \cE_{[\cC_N,A]} = \sum_{i=1}^5(\Xi_{i}+ \text{h.c.})$, this concludes the proof.
\end{proof}
The following corollary describes the main contributions to $ e^{-sA} \cC_N e^{sA}$, for any fixed $s\in [0;1]$. In particular, for $s=1$, it determines $ \cJ_N^{(3)}$, up to small errors. The slightly more general result about $ e^{-sA} \cC_N e^{sA}$ for any $ s\in [0;1]$ will be useful in Section \ref{sec:JNHN}. 
\begin{cor} \label{cor:eACN} Assume $\alpha,\beta$ satisfy (\ref{eq:alphabeta}). Then there exists $C>0$ such that  
		\begin{equation*} 
		\begin{split}
		 &e^{-sA}\cC_N e^{sA} \\
		 &= \cC_N + 2sN^\kappa \!\!\!\sum_{r\in P_H, v\in P_L} \!\!\! \big[ \widehat{V}(r/N^{1-\kappa})\eta_r/N + \widehat{V}((v+r)/N^{1-\kappa}))\eta_r/N\big] a^*_va_v (1-\cN_+/N)\\
		 &\hspace{0.4cm}+  \cE_{\cJ_N}^{(3)}(s)
		\end{split}
		\end{equation*}
where the self-adjoint operator $ \cE_{\cJ_N}^{(3)}(s) $ is such that
		\begin{equation*}
		\begin{split}
		\pm e^{A} \cE_{\cJ_N}^{(3)}(s)e^{-A} \leq C N^{(\kappa-3\beta)/2}\cV_N + C N^{(3\kappa-7\beta)/2}\cK +C N^{(4\kappa-3\beta)/2}\cN_+ +CN^\kappa
		\end{split}\end{equation*}
for all $s\in [0;1]$ and for all $N\in\NN$ sufficiently large.
\end{cor}
\begin{proof}
With Prop. \ref{prop:commCNA}, we expand
		\[\begin{split} 
		&e^{-sA}\cC_N e^{sA}-\cC_N\\
		 &=  2N^\kappa \int_0^sdt \; \!\!\! \!\!\!\sum_{r\in P_H, v\in P_L} \!\!\!\!\! \big[ \widehat{V}(r/N^{1-\kappa})\eta_r/N + \widehat{V}((v+r)/N^{1-\kappa}))\eta_r/N\big] e^{-tA}a^*_va_v (1-\cN_+/N) e^{tA} \\
		 &\hspace{0.4cm} +\int_0^sdt \; e^{-tA}\cE_{[\cC_N, A]} e^{-tA}.
		\end{split}\]
Now, using Plancherel's theorem observe that 
		\[\begin{split}
		&\bigg|\sum_{r\in P_H}\big[ \widehat{V}(r/N^{1-\kappa})\eta_r/N + \widehat{V}((v+r)/N^{1-\kappa}))\eta_r/N\big]\bigg|\\
		&\hspace{2cm}\leq CN^{\alpha+2\kappa-1} + C\int_{\Lambda}N^{3-3\kappa}V(N^{1-\kappa}x)w_\ell(N^{1-\kappa}x)\leq CN^{\kappa}. 
		\end{split}\]
The claim is now an immediate consequence of Lemma \ref{lm:FinftycNopA}, the bound \eqref{eq:propcommCNAerror}, Corollary \ref{cor:VNgrowA}, Lemma \ref{lm:NresgrowA}, Lemma \ref{lm:KresgrowA} and using the operator inequality $\cN_{\geq \frac 12 N^\alpha} \leq 4 N^{-2\alpha}\cK $. 
\end{proof}

%%%%%%%%%%%%%%%%%%%%%%%%%%%%%%%%%%%%%%%%%%%%%%%%%%%%%%%%%%%%%%%%%%%
%%%%%%%%%%%%%%%%%%%%%%%%%%%%%%%%%%%%%%%%%%%%%%%%%%%%%%%%%%%%%%%%%%%
%%%%%%%%%%%%%%%%%%%%%%%%%%%%%%%%%%%%%%%%%%%%%%%%%%%%%%%%%%%%%%%%%%%

\subsubsection{Analysis of $ \cJ_N^{(4)}$} \label{sec:JNHN}

The goal of this section is to determine the main contributions to $ \cJ_N^{(4)} =  e^{-A}\cH_N e^A$. As in \cite{BBCS}, it turns out that conjugating $\cH_N$ with the cubic exponential $e^A$ leads to a renormalization of the cubic term $ \cC_N$ of the quadratically renormalized Hamiltonian $\cG_N^{\text{eff}}$, defined in \eqref{eq:GNeff}. To see this, let's recall \eqref{eq:propcommKA}, \eqref{eq:propcomm1VN} and compute
		\begin{equation}\label{eq:eAHN1}\begin{split}
		 \cJ_N^{(4)}&= \cH_N + \int_0^1ds\;e^{-sA}[\cK + \cV_N, A]e^{sA}\\
		&= \cH_N - \int_0^1ds\; \frac{1}{\sqrt N}\sum_{u\in\Lambda^*_+, p\in P_L}N^\kappa \big(\widehat{V}(\cdot/N^{1-\kappa})\ast (\widehat{f}_N-\eta/N)\big)(u) \\
		&\hspace{5cm}\times e^{-sA}\big(b^*_{p+u}a^*_{-u}a_p+\text{h.c.}\big)e^{sA}\\
		&\hspace{0.4cm} +  \int_0^1ds\; e^{-sA}\big(  \cE_{[\cK,A]} + \cE_{[\cV_N,A]} \big)e^{sA}.
		\end{split}\end{equation}
Here, let us recall the definitions \eqref{eq:commKA2}, \eqref{eq:commVNA12} and that the operators $ \cE_{[\cK,A]} $ and $\cE_{[\cV_N,A]}$ are explicitly given by
		\begin{equation}\label{eq:commKAVNAerrors}
		\begin{split}
		\cE_{[\cK,A]} &=  \sum_{j=1}^3 (\Pi_j+\text{h.c.}), \hspace{0.4cm}\cE_{[\cV_N,A]} = \sum_{j=1}^4 (\Theta_j+\text{h.c.}).
		\end{split}
		\end{equation}
With $ (\widehat{f}_N-\eta/N)(q) = \delta_{q,0}$ for all $q\in \Lambda_+^*$, we obtain from \eqref{eq:eAHN1} and Corollary \ref{cor:eACN} that
		\begin{equation}\label{eq:eAHN2}
		\begin{split}
		 \cJ_N^{(4)}&= \cH_N- \cC_N  -\int_0^1ds\; \cE_{\cJ_N}^{(3)}(s)+  \int_0^1ds\; e^{-sA}\big(  \cE_{[\cK,A]} + \cE_{[\cV_N,A]} \big)e^{sA}\\
		&\hspace{0.4cm} -  N^\kappa \!\!\!\sum_{r\in P_H, v\in P_L} \!\!\! \big[ \widehat{V}(r/N^{1-\kappa})\eta_r/N + \widehat{V}((v+r)/N^{1-\kappa}))\eta_r/N\big] a^*_va_v (1-\cN_+/N)
		\end{split}
		\end{equation}
and we observe that the contribution $ -\cC_N$ will cancel exactly the contribution $\cC_N$ in $  \cJ_N^{(3)}$, determined in Corollary \ref{cor:eACN} (for $s=1$). Moreover, the quadratic contribution in the last line of \eqref{eq:eAHN2} combines with the corresponding contribution to  $  \cJ_N^{(3)}$ as well. 

To finish this section, it remains to extract the remaining leading order contributions to  $  \cJ_N^{(4)}$ from the integral terms in \eqref{eq:eAHN2}. It turns out that all contributions, but the term $ (\Pi_1+\text{h.c.})$ contained in  $\cE_{[\cK,A]}$, are error terms which can be neglected. To make this more precise, Corollary \ref{cor:eACN} implies first of all that there exists a constant $C>0$ s.t.
		\begin{equation}\label{eq:eAHN3} \pm \int_0^1ds\;e^A  \cE_{\cJ_N}^{(3)}(s) e^{-A}\leq C N^{(\kappa-3\beta)/2}\cV_N + C N^{(3\kappa-7\beta)/2}\cK +C N^{(4\kappa-3\beta)/2}\cN_+ +CN^\kappa
		\end{equation}
for all $ \alpha>3\beta + 2 \kappa \geq 0$ with $ \alpha+\kappa\leq 1$, $ 2\kappa-3\beta\leq 0$ and for all $N\in\NN$ sufficiently large. 

Next, we use \eqref{eq:commKAVNAerrors} and recall the bounds \eqref{eq:commKA5} and \eqref{eq:commKA6}. They imply that
		\[\begin{split}
		& \pm \int_0^1ds\; e^{(1-s)A} (\Pi_2+\Pi_3+\text{h.c.}) e^{-(1-s)A} \\
		 &\hspace{0cm} \leq CN^{-3\beta/2}\int_0^1ds\;  e^{(1-s)A} (\cN_{\geq \frac12N^\alpha}+1)e^{(1-s)A} + C \delta N^{-3\beta/2}\int_0^1ds\; e^{(1-s)A} \cK e^{-(1-s)A}\\
		 &\hspace{0.4cm} +\delta^{-1} C N^{-3\beta/2} \int_0^1ds\;  e^{(1-s)A} \cK_L e^{-(1-s)A}
		 \end{split}\]
for all $ \delta>0$, $ \alpha> 3\beta+2\kappa\geq 0$ and $N\in\NN$ sufficiently large. With Lemma \ref{lm:NresgrowA}, Lemma \ref{lm:KresgrowA}, Corollary \ref{cor:KgrowA} and $ \cN_{\geq \frac12N^\alpha}\leq N$ in $\cF_+^{\leq N}$, we deduce from the previous bound that
		\[\begin{split}
		& \pm \int_0^1ds\; e^{(1-s)A} (\Pi_2+\Pi_3+\text{h.c.}) e^{-(1-s)A} \\
		 &\hspace{0.4cm} \leq CN^{-3\beta/2} (\cN_{\geq \frac12N^\alpha}+1) + \delta C  N^{\kappa-3\beta/4} \cK + \delta C N^{-3\beta/2}\cV_N+\delta C N^{\alpha+2\kappa-3\beta/2} \\
		 &\hspace{0.8cm}+\delta^{-1} C N^{-3\beta/2}\cK_L  + \delta^{-1} C N^{-5\beta/2}(\cN_{\geq \frac12 N^\alpha}+1)
		 \end{split}\]
for all $\delta>0$, $ \alpha> 3\beta+2\kappa\geq 0$ with $ \alpha+\kappa\leq 1$, $ \beta\geq \kappa$. Choosing $ \delta = N^{-\beta}$ in the last bound and using that $ \cN_{\geq \frac12 N^{\alpha}}\leq 4 N^{-2\alpha}\cK$, we find
		\begin{equation}\label{eq:eAHN4}\begin{split}
		& \pm \int_0^1ds\; e^{(1-s)A} (\Pi_2+\Pi_3+\text{h.c.}) e^{-(1-s)A} \\
		 &\hspace{3cm} \leq C N^{-\beta/2}\cK   +   C N^{-5\beta/2}\cV_N + C  N^{ \alpha+2\kappa-5\beta/2}
		 \end{split}\end{equation}
for all $ \alpha> 3\beta+2\kappa\geq 0$ with $ \alpha+\kappa\leq1$, $ \beta\geq \kappa$.

Going back to \eqref{eq:eAHN2} and using the estimate \eqref{eq:propcommVNAerror}, we obtain that
		\[\begin{split}
		&\pm \int_0^1ds\; e^{(1-s)A} \cE_{[\cV_N,A]} e^{-(1-s)A}\\
		&\hspace{0.4cm}\leq \int_0^1ds\; e^{(1-s)A} \Big[\wt\delta \cV_N   + \wt\delta^{-1}CN^{\kappa-2\beta -1} \cK_L (\cN_{\geq \frac12 N^{\alpha}}+1)+ \wt\delta^{-1}C N^{2\alpha  +3\kappa-2} \cN_+\\
		&\hspace{3.4cm} +  \wt\delta^{-1}CN^{\kappa -1} (\cN_{\geq \frac12 N^{\alpha}}+1)^2 \Big]e^{-(1-s)A} \\
		&\hspace{0.4cm}\leq \int_0^1ds\; e^{(1-s)A} \Big[\wt \delta \cV_N   + \wt\delta^{-1}CN^{\kappa-2\beta } \cK_L  \Big]e^{-(1-s)A} + \wt\delta^{-1}C N^{2\alpha  +3\kappa-2} \cN_+\\
		 &\hspace{0.8cm}+\wt\delta^{-1}C N^{2\alpha  +3\kappa-2}+ \wt\delta^{-1}CN^{\kappa -1}.
		\end{split}\]
Setting $ \widetilde \delta = N^{\mu} $ for $ \mu = \max (\alpha+3/2\kappa-1, \kappa/2-\beta)$, Corollary \ref{cor:VNgrowA} and Lemma \ref{lm:KresgrowA} imply 
		\begin{equation}\label{eq:eAHN5}\begin{split}
		&\pm \int_0^1ds\; e^{(1-s)A} \cE_{[\cV_N,A]} e^{-(1-s)A}\leq CN^{\mu+\kappa}\cH_N     + CN^{\mu } \cN_+  +CN^{\mu } +CN^{\mu+2\beta-1} 
		 \end{split}\end{equation}
for all $ \alpha> 3\beta+2\kappa\geq 0$ with $ \alpha+\kappa\leq 1$ and all $N\in\NN$ large enough.

Looking back at \eqref{eq:eAHN2} and collecting the bounds \eqref{eq:eAHN3} to \eqref{eq:eAHN5}, it only remains to analyse the operator $( \Pi_1+\text{h.c.})$ in the definition \eqref{eq:commKAVNAerrors} of the operator $\cE_{[\cK,A]} $. Using that $ \widehat{V}(./N^{1-\kappa})\ast  \widehat{f}_{N}= \widehat{Vf_\ell}(./N^{1-\kappa}) $, let us recall that $ \Pi_1$ is explicitly given by
		\[\Pi_1 =\frac{1}{\sqrt N} \sum_{\substack{ p\in P_H^c, q\in P_L:\\ p+q\neq 0 }} N^{\kappa} \widehat{Vf_\ell}(p/N^{1-\kappa}) b^*_{q+p} a^*_{-p}a_q.  \]
The following lemma analyses slightly more general operators than $ \Pi_1$, after conjugation with $e^{sA}$ for any $ s\in [-1;1]$. It will also be useful in the proof of Proposition \ref{prop:JN}.
\begin{lemma}\label{lm:cubicrenA}
Assume $\alpha,\beta$ satisfy (\ref{eq:alphabeta}). Let $ F= (F_p)_{p\in\Lambda_+^*)}\in \ell^\infty( \Lambda_+^* )$. Then, there exists a constant $ C>0$ s.t.
		\begin{equation*}
		\begin{split}
		\pm&\bigg[ \frac{1}{\sqrt N} \sum_{\substack{ p\in P_H^c, q\in P_L:\\ p+q\neq 0 }} F_p e^{-sA}\big(b^*_{q+p} a^*_{-p}a_q+\text{h.c.}\big)e^{sA} - \frac{1}{\sqrt N} \sum_{\substack{ p\in P_H^c, q\in P_L:\\ p+q\neq 0 }} F_p \big(b^*_{q+p} a^*_{-p}a_q+\emph{h.c.}\big) \bigg] \\
		&\hspace{0cm}\leq C\| F\|_\infty N^{  \kappa +\alpha   } \langle \xi, (\cN_{\geq \frac12 N^\alpha}+1)\xi \rangle
		\end{split}
		\end{equation*}
for all $ s\in [-1;1]$ and $N\in\NN$ sufficiently large.
\end{lemma}
%\frac{1}{\sqrt N} \sum_{\substack{ r\in P_H^c, v\in P_L:\\ v+r\neq 0 }} N^{\kappa} (\widehat{V}(./N^{1-\kappa})\ast  \widehat{f}_{N})(r)  (b^*_{v+r} a^*_{-r}a_v  +\emph{h.c.})
\begin{proof}
Let us set 
		\[\text{X}:= \frac{1}{\sqrt N} \sum_{\substack{ p\in P_H^c, q\in P_L:\\ p+q\neq 0 }} F_p b^*_{q+p} a^*_{-p}a_q.\]
From the identity
		\begin{equation}\label{eq:cubicrenA1}
		\begin{split}
		&e^{-sA}\big( \text{X} +\text{h.c.}) e^{sA} - (\text{X}+\text{h.c.}) 
		%& = \int_{0}^{s}dt\; \frac{1}{\sqrt N} \sum_{\substack{ p\in P_H^c, q\in P_L:\\ p+q\neq 0 }} F_p [e^{-tA}b^*_{q+p} a^*_{-p}a_q + a^*_qa_{-p}b_{q+p}, A_1]e^{tA} +\text{h.c.}\\
		 =  \int_{0}^{s}dt\; e^{-tA}[ ( \text{X} +\text{h.c.}), A_1] e^{tA} +\text{h.c.},
		\end{split}
		\end{equation}
we conclude that it suffices to control the commutator $[ (\text{X}+\text{h.c.}), A_1] $ after conjugation with $e^{tA}$, uniformly in $t\in [-1;1]$. From the proof of Proposition \ref{prop:commCNA}, we collect
		\[[b^*_{p+q}a^*_{-p}a_q, b^*_{v+r}a^*_{-r}a_v] = -b^*_{v+r }b^*_{-r}a^*_{-p}a_q\delta_{p+q,v}  - b^*_{v+r}b^*_{p+q}a^*_{-r}a_q \delta_{-p,v} \]
and 
		\[\begin{split}
		&[a^*_qa_{-p}b_{p+q}, b^*_{v+r}a^*_{-r}a_v]\\ %&= [a^*_q a_{-p}, b^*_{v+r}]b_{p+q}a^*_{-r}a_v + a^*_q a_{-p}[b_{p+q},b^*_{v+r}]a^*_{-r}a_v\\
		%& \hspace{0.4cm} + b^*_{v+r}[a^*_{q}a_{-p}, a^*_{-r}a_v]b_{p+q} + b^*_{v+r}a^*_{q}a_{-p}[b_{p+q}, a^*_{-r}a_v]\\
		&= a^*_q a_v (1-\cN_+/N)  \delta_{v+r,-p}\delta_{p+q,-r} + a^*_qa^*_{-r}a_{p+q}a_v (1-\cN_+/N)\delta_{v+r,-p}\\
		&\hspace{0.4cm}  + a^*_{v+r}a^*_{q}a_{-p}a_v (1-\cN_+/N)\delta_{p+q,-r} + a^*_qa^*_{-r}a_{-p}a_v (1-\cN_+/N)\delta_{p+q,v+r}\\
		&\hspace{0.4cm}  - b^*_{v+r}a^*_{-r}a_{-p}b_{p+q}\delta_{v,q} -\frac1N  a^*_q a^*_{v+r}a^*_{-r}a_va_{-p}a_{p+q} \\
		\end{split}\]
for all $ p\in P_H^c, q\in P_L$, $ r\in P_H, v\in P_L$ and $N\in\NN$ large enough. Consequently, we find
		\[[ (\text{X}+\text{h.c.}), A_1]  = \sum_{j=1}^7 (\Upsilon_j+\text{h.c.}\big),\]
where the operators $ \Upsilon_j, j=1,\dots,5$, are defined by
		\begin{equation}\label{eq:cubicrenA2}
		\begin{split}
		\Upsilon_{1} &= \frac{1}{N} \sum_{\substack{ r\in P_H, v\in P_L:\\ r+v\in P_H^C  }}^*  F_{r+v}  \eta_r a^*_{v  }a_v(1-\cN_+/N) , \\
		\Upsilon_{2} &= - \frac{1}{N} \sum_{\substack{ r\in P_H, q, v\in P_L }}^* \big[ F_{v-q}\eta_r +F_v\eta_r\big] b^*_{v+r }b^*_{-r}a^*_{q-v}a_q , \\
		\Upsilon_{3} &=  \frac{1}{N} \sum_{\substack{r\in P_H, q, v\in P_L:\\ r+v\in P_H^C }}^*  F_{v+r}\eta_ra^*_qa^*_{r}a_{q+v+r}a_{-v} (1-\cN_+/N), \\		\Upsilon_{4} &=  \frac{1}{N} \sum_{\substack{r\in P_H, q, v\in P_L:\\ q+r+v\in P_H^C }}^*F_{v+r+q}\eta_r  a^*_qa^*_{r}a_{q+v+r}a_{-v} (1-\cN_+/N) , \\
		\Upsilon_{5} &=  \frac{1}{N} \sum_{\substack{r\in P_H, q, v\in P_L:\\ r+q\in P_H^C}}^*  F_{r+q}\eta_r a^*_{v+r}a^*_{q}a_{r+q}a_v (1-\cN_+/N), \\
		\Upsilon_{6}&= - \frac{1}{N} \sum_{\substack{ p\in P_H^c, r\in P_H, v\in P_L}}^* F_p\eta_r b^*_{v+r}a^*_{-r}a_{-p}b_{p+v} , \\
		\Upsilon_{7} &= - \frac{1}{N^2} \sum_{p\in P_H^c, r\in P_H, q, v\in P_L}^*  F_p\eta_r  a^*_q a^*_{v+r}a^*_{-r}a_va_{-p}a_{p+q}. \\	
		\end{split}
		\end{equation}
To control the different contributions in \eqref{eq:cubicrenA2}, we apply as usual Cauchy-Schwarz. Given any $\xi\in\cF_+^{\leq N}$, we bound the first two operators in \eqref{eq:cubicrenA2} by
		\[\begin{split}
		|\langle \xi,  \Upsilon_1 \xi\rangle| &\leq  C \|F\|_\infty N^{\kappa-1}\sum_{N^\alpha\leq |r|\leq N^\alpha+N^\beta} |r|^{-2} \langle \xi, (\cN_+ +1)\xi \rangle \\
		&\leq C\|F\|_\infty N^{\beta+\kappa-1}\langle \xi, (\cN_+ +1)\xi \rangle\leq C \|F\|_\infty N^{\beta+\kappa}
		\end{split}\]
and
		\[\begin{split}
		|\langle \xi,  \Upsilon_2 \xi\rangle| &\leq \frac{C\|F\|_\infty}{N} \!\!\!\!\!\sum_{\substack{ r\in P_H, q, v\in P_L }}^*\!\!\!\!\!\! |\eta_r| \| (\cN_{\geq \frac12 N^\alpha}+1)^{-1/2}a_{q-v} a_{v+r }a_{-r}\xi\| \| (\cN_{\geq \frac12 N^\alpha}+1)^{1/2} a_q \xi\| \\
		&\leq C\|F\|_\infty N^{3\beta/2 + \kappa -\alpha/2  } \langle \xi, (\cN_{\geq \frac12 N^\alpha}+1)\xi \rangle\leq C \|F\|_\infty \langle \xi, (\cN_{\geq \frac12 N^\alpha}+1)\xi \rangle
		\end{split}\]
for all $ \alpha>3\beta+2\kappa\geq 0$ and $N\in\NN$ sufficiently large. In the same way, we find that
		\[\begin{split}
		|\langle \xi,  (\Upsilon_3 +\Upsilon_4+\Upsilon_5) \xi\rangle| &\leq C \|F\|_\infty\langle \xi, (\cN_{\geq \frac12 N^\alpha}+1)\xi \rangle
		\end{split}\]
as well as
		\[\begin{split}
		|\langle \xi,  \Upsilon_6 \xi\rangle| &\leq \frac{C \|F\|_\infty}{N} \!\!\!\!\!\!\sum_{\substack{ p\in P_H^c, r\in P_H, \\v\in P_L}}^*\!\!\!\!\!\! |\eta_r | \| (\cN_{\geq \frac12 N^\alpha}+1)^{-1/2}a_{v+r}a_{-r} \xi\| \| (\cN_{\geq \frac12 N^\alpha}+1)^{1/2}a_{-p}a_{p+v}\xi\|\\
		&\leq C\|F\|_\infty N^{  \kappa +\alpha   } \langle \xi, (\cN_{\geq \frac12 N^\alpha}+1)\xi \rangle.
		\end{split}\]
Finally, we have that 
		\[\begin{split}
		|\langle \xi,  \Upsilon_7 \xi\rangle| &\leq \frac{C\|F\|_\infty}{N^2} \sum_{p\in P_H^c, r\in P_H, q, v\in P_L}^*  \|   (\cN_{\geq \frac12 N^\alpha}+1)^{-1/2}a_q a_{v+r}a_{-r}\xi \| \\
		&\hspace{4.5cm}\times |\eta_r|\|(\cN_{\geq \frac12 N^\alpha}+1)^{1/2}a_va_{-p}a_{p+q} \xi\| \\
		&\leq C \|F\|_\infty N^{  \kappa +\alpha   } \langle \xi, (\cN_{\geq \frac12 N^\alpha}+1)\xi \rangle.
		\end{split}\]
Combining the last five bounds with equation \eqref{eq:cubicrenA1}, equation \eqref{eq:cubicrenA2} and Lemma \ref{lm:NresgrowA}, we conclude the proof of Lemma \ref{lm:cubicrenA}.
\end{proof}
Since $\sup_{p\in\Lambda^*_+}|\widehat{Vf_\ell}(./N^{1-\kappa}) |\leq C$, we obtain immediately the following corollary.
\begin{cor} Assume $\alpha,\beta$ satisfy (\ref{eq:alphabeta}). Let $ \Pi_1$ be defined as in equation \eqref{eq:commKA2}. Then, there exists a constant $ C>0$ such that 
		\begin{equation}\label{eq:cubicrenAPi1}
		\begin{split}
		\pm&\big[ e^{-sA}(\Pi_1+\emph{h.c.})e^{sA} - (\Pi_1+\emph{h.c.}) \big] \leq CN^{  2\kappa +\alpha   } \langle \xi, (\cN_{\geq \frac12 N^\alpha}+1)\xi \rangle
		\end{split}
		\end{equation}
for all $ s\in [-1;1]$ and $N\in\NN$ sufficiently large.
\end{cor}
Another consequence is the following result describing $ \cJ_N^{(4)}$, up to small errors. 
\begin{cor}\label{cor:JN4} Assume $\alpha,\beta$ satisfy (\ref{eq:alphabeta}) and let $ \mu = \max(\alpha+3\kappa/2-1,\kappa/2-\beta)$. Let $ \cJ_N^{(4)}$ be defined as in \eqref{eq:JN0to4}. Then, we have  
		\[\begin{split} 
		\cJ_N^{(4)}&= \cH_N- \cC_N + \frac{1}{\sqrt N} \sum_{\substack{ p\in P_H^c, q\in P_L:\\ p+q\neq 0 }} N^{\kappa} \widehat{Vf_\ell}(p/N^{1-\kappa}) (b^*_{q+p} a^*_{-p}a_q+\emph{h.c.}) \\
		&\hspace{0.4cm} -  N^\kappa \!\!\!\sum_{r\in P_H, v\in P_L} \!\!\! \big[ \widehat{V}(r/N^{1-\kappa})\eta_r/N + \widehat{V}((v+r)/N^{1-\kappa}))\eta_r/N\big] a^*_va_v (1-\cN_+/N) \\
		&\hspace{0.4cm} +  \cE_{\cJ_N}^{(4)},
		\end{split}\]
where there exists a constant $C>0$ such that 
		\[\begin{split}
		\pm e^A \cE_{\cJ_N}^{(4)}e^{-A} & \leq C m_0 ( N^{-\beta/2} + N^{\kappa +\mu } )\cK + CN^{\kappa+\mu }\cV_N + CN^\mu \cN_+ + CN^{\alpha+2\kappa}+CN^{\mu}
		\end{split}\]
for all $N\in\NN$ sufficiently large.
\end{cor}
\begin{proof}This is an immediate consequence of the identity \eqref{eq:eAHN2} and the bounds \eqref{eq:eAHN3}, \eqref{eq:eAHN4}, \eqref{eq:eAHN5} as well as \eqref{eq:cubicrenAPi1}.
\end{proof}

\subsection{Proof of Proposition \ref{prop:JN}}\label{sec:JNprop}

Applying Corollary \ref{cor:JN0}, Proposition \ref{prop:JN2}, Corollary \ref{cor:eACN} and Corollary \ref{cor:JN4}, we obtain 
		\begin{equation}\label{eq:propJN1}\begin{split}
		 \cJ_N &=  4\pi \frak{a}_0 N^{\kappa} (N-\cN_+) + \big[\widehat V(0)-4\pi \frak{a}_0\big]N^{\kappa} \cN_+ (1-\cN_+/N)  \\
		 &\hspace{0.4cm} + N^{\kappa}\widehat V(0)\sum_{p\in P_H^c}  b^*_pb_p + 4\pi \frak{a}_0N^{\kappa} \sum_{p\in P^c_H}  \big[ b^*_p b^*_{-p} + b_p b_{-p} \big]\\
		 &\hspace{0.4cm} + N^\kappa \!\!\!\sum_{r\in P_H, v\in P_L} \!\!\! \big[ \widehat{V}(r/N^{1-\kappa})\eta_r/N + \widehat{V}((v+r)/N^{1-\kappa}))\eta_r/N\big] a^*_va_v (1-\cN_+/N)\\
		 &\hspace{0.4cm} + \frac{1}{\sqrt N} \sum_{\substack{ p\in P_H^c, q\in P_L:\\ p+q\neq 0 }} N^{\kappa} \widehat{Vf_\ell}(p/N^{1-\kappa}) (b^*_{q+p} a^*_{-p}a_q+\text{h.c.})+\cH_N  \\
		 &\hspace{0.4cm}+ \cE_{\cJ_N}^{(0)}+ \cE_{\cJ_N}^{(2)} +\cE_{\cJ_N}^{(3)}+\cE_{\cJ_N}^{(4)},
		 \end{split}\end{equation}
where we have set $ \cE_{\cJ_N}^{(3)} = \cE_{\cJ_N}^{(3)}(1)$. For $\mu = \max(3\alpha/2+\kappa-1, \kappa/2-\beta) $, we know that 
		\begin{equation}\label{eq:propJN2}\begin{split}
		\pm &e^A\Big( \cE_{\cJ_N}^{(0)}+ \cE_{\cJ_N}^{(2)} +\cE_{\cJ_N}^{(3)}+\cE_{\cJ_N}^{(4)}\Big)e^{-A}\\
		&\hspace{1cm}\leq C( N^{-\beta/2} + N^{\kappa+\mu} ) \cK + CN^{\kappa+\mu}\cV_N + CN^{\mu}\cN_+ + C( N^{\alpha+2\kappa} + N^{\mu}).
		\end{split}\end{equation}
Hence, let's evaluate the remaining contributions to $\cJ_N$. With \eqref{eq:modetap}, we use the bound
		\[ \sum_{r\in P_H^c\cup\{0\} }\big| \widehat{V}(r/N^{1-\kappa})\eta_r/N + \widehat{V}((v+r)/N^{1-\kappa}))\eta_r/N\big|\leq C N^{\alpha+\kappa-1} \]
to conclude that 
		\begin{equation}\label{eq:propJN3}\begin{split}
		&\pm \bigg(N^\kappa \!\!\!\sum_{r\in P_H, v\in P_L} \!\!\! \big[ \widehat{V}(r/N^{1-\kappa})\eta_r/N + \widehat{V}((v+r)/N^{1-\kappa}))\eta_r/N\big] a^*_va_v (1-\cN_+/N)\\
		&\hspace{0.7cm} + N^\kappa \int_{\Lambda}dx\;V(x)w_\ell(x) \sum_{ v\in P_L}a^*_va_v(1-\cN_+/N) + N^\kappa \sum_{ v\in P_L} \widehat{Vw_\ell}(v/N^{1-\kappa}) b^*_vb_v\bigg)\\
		&\hspace{0.4cm}\leq CN^{\alpha+2\kappa-1} (\cN_++1). 
		\end{split}\end{equation}
Now, by Lemma \ref{sceqlemma}, we have that 
		\[ \int_{\Lambda}dx\;V(x)w_\ell(x) = \widehat{V}(0) - 8\pi \mathfrak{a}_0 + \mathcal{O}(N^{\kappa-1}) \]
and, for $ v\in P_L$, we find similarly that
		\[\widehat{Vw_\ell}(v/N^{1-\kappa}) = \widehat{Vw_\ell}(0) + \mathcal{O}(|v|/N^{1-\kappa}) =\widehat{V}(0) - 8\pi \mathfrak{a}_0 + \mathcal{O}(N^{\beta+\kappa-1}).  \] 
As a consequence, we deduce
		\begin{equation}\label{eq:propJN4}\begin{split}
		&\pm \bigg(N^\kappa \!\!\!\sum_{r\in P_H, v\in P_L} \!\!\! \big[ \widehat{V}(r/N^{1-\kappa})\eta_r/N + \widehat{V}((v+r)/N^{1-\kappa}))\eta_r/N\big] a^*_va_v (1-\cN_+/N)\\
		&\hspace{0.7cm} -(8\pi\mathfrak{a}_0-\widehat{V}(0) )  N^\kappa  \sum_{ v\in P_L}a^*_va_v(1-\cN_+/N) + (8\pi\mathfrak{a}_0-\widehat{V}(0) ) N^\kappa \sum_{ v\in P_L}  b^*_vb_v\bigg)\\
		&\hspace{0.4cm}\leq CN^{\alpha+2\kappa-1} (\cN_++1). 
		\end{split}\end{equation}
Finally, we use the operator bounds
		\[ \pm \sum_{\substack{v\in\Lambda_+^*: v\in P_L^c}} a^*_va_v(1-\cN_+/N)\leq C \cN_{\geq N^\beta} \hspace{0.5cm}\text{and}\hspace{0.5cm} \pm \sum_{\substack{v\in\Lambda_+^*: \\ v\in P_L^c\cap P_H^c}} b^*_vb_v(1-\cN_+/N) \leq C\cN_{\geq N^\beta} \]
to conclude that 
		\begin{equation}\label{eq:propJN5}\begin{split}
		&\pm \bigg(N^\kappa \!\!\!\sum_{r\in P_H, v\in P_L} \!\!\! \big[ \widehat{V}(r/N^{1-\kappa})\eta_r/N + \widehat{V}((v+r)/N^{1-\kappa}))\eta_r/N\big] a^*_va_v (1-\cN_+/N)\\
		&\hspace{0.7cm} -(8\pi\mathfrak{a}_0-\widehat{V}(0) )  N^\kappa  \cN_+(1-\cN_+/N) - (8\pi\mathfrak{a}_0-\widehat{V}(0) ) N^\kappa \sum_{ p\in P_H^c}  b^*_pb_p\bigg)\\
		&\hspace{0.4cm}\leq CN^{\alpha+2\kappa-1} (\cN_++1) + CN^{\kappa} \cN_{\geq N^\beta}. 
		\end{split}\end{equation}
Collecting the estimates \eqref{eq:propJN3} to \eqref{eq:propJN5}, we summarize that 
		\begin{equation}\label{eq:propJN6}\begin{split}
		 \cJ_N &=  4\pi \frak{a}_0 N^{1+\kappa}  - 4\pi \frak{a}_0 N^{\kappa} \cN_+^{\hspace{0.03cm}2}/N +\cH_N  \\
		 &\hspace{0.4cm} + 8 \pi \frak{a}_0N^{\kappa}\sum_{p\in P_H^c}  b^*_pb_p + 4\pi \frak{a}_0N^{\kappa} \sum_{p\in P^c_H}  \big[ b^*_p b^*_{-p} + b_p b_{-p} \big]\\
		 &\hspace{0.4cm} + \frac{1}{\sqrt N} \sum_{\substack{ p\in P_H^c, q\in P_L:\\ p+q\neq 0 }} N^{\kappa} \widehat{Vf_\ell}(p/N^{1-\kappa}) (b^*_{q+p} a^*_{-p}a_q+\text{h.c.})  \\
		 &\hspace{0.4cm}+ \cE_{\cJ_N}^{(0)}+ \cE_{\cJ_N}^{(2)} +\cE_{\cJ_N}^{(3)}+\cE_{\cJ_N}^{(4)} + \cE_{\cJ_N}^{(5)},
		 \end{split}\end{equation}
where the self-adjoint operator $ \cE_{\cJ_N}^{(5)}$ satisfies the operator inequalities
		\begin{equation}\label{eq:propJN7}
		\begin{split}
		\pm e^A \cE_{\cJ_N}^{(5)}e^{-A}& \leq CN^{\kappa} \cN_{\geq N^\beta} + CN^{\alpha+2\kappa-1} \cN_+ +CN^\kappa\\
		&\leq  CN^{\kappa-2\beta }\cK + CN^{\alpha+2\kappa-1} \cN_+ +CN^\kappa.
		\end{split}
		\end{equation}
Notice that we applied Lemma \ref{lm:NresgrowA} here. Now, let us consider the cubic term on the r.h.s. of equation \eqref{eq:propJN1}. To this end, let's define the self-adjoint operator $ \cE_{\cJ_N}^{(6)}$ by
		\[\cE_{\cJ_N}^{(6)}  = \frac{1}{\sqrt N} \sum_{\substack{ p\in P_H^c, q\in P_L:\\ p+q\neq 0 }} N^{\kappa}\big ( \widehat{Vf_\ell}(p/N^{1-\kappa}) -8\pi \mathfrak{a}_0\big)(b^*_{q+p} a^*_{-p}a_q+\text{h.c.}).	\]
Since $\sup_{p\in P_H^c}|\widehat{Vf_\ell}(p/N^{1-\kappa}) -8\pi \mathfrak{a}_0 | \leq CN^{\alpha+\kappa-1} $, we conclude with Lemma \ref{lm:cubicrenA} that
		\[\begin{split}
		\pm \Big( e^{A}\cE_{\cJ_N}^{(6)} e^{-A} - \cE_{\cJ_N}^{(6)}\Big)\leq C N^{2\alpha+3\kappa-1} \cN_{\geq \frac12 N^\alpha} + C N^{2\alpha+3\kappa-1} \leq C N^{3\kappa-1}\cK + C N^{2\alpha+3\kappa-1}.
		\end{split}\]
Then, we recall that if $m_0\in\mathbb{R}$ is such that $  m_0\beta  =\alpha$, we know that $ m_0\leq 5$. Hence, using once again that $\sup_{p\in P_H^c}|\widehat{Vf_\ell}(p/N^{1-\kappa}) -8\pi \mathfrak{a}_0 | \leq CN^{\alpha+\kappa-1} $ and controlling $ \cE_{\cJ_N}^{(6)}$ as in \eqref{eq:commKA31} to \eqref{eq:commKA32} (using $ m_0\in [3;5]$ and that $ |ab|\leq \delta |a|^2+\delta^{-1}|b|^2$ with $ \delta = N^{\kappa-\beta/2}$), the previous estimate implies that
		\begin{equation}\label{eq:propJN8}
		\begin{split}
		\pm e^{A}\cE_{\cJ_N}^{(6)} e^{-A} & \leq C (N^{\alpha+ 2\kappa-\beta/2-1}+  N^{3\kappa-1})\cK + CN^{ 2\alpha+\beta/2+2\kappa-1} \\
		&\leq C N^{\mu +\kappa/2- \beta/2}\cK + CN^{2\alpha+\beta/2+2\kappa-1}.
		\end{split}
		\end{equation}
Combining the previous estimates and collecting \eqref{eq:propJN2}, \eqref{eq:propJN7} and \eqref{eq:propJN8}, we find that
		\begin{equation}\label{eq:propJN9}\begin{split}
		 \cJ_N &=  4\pi \frak{a}_0 N^{1+\kappa}  - 4\pi \frak{a}_0 N^{\kappa} \cN_+^{\hspace{0.03cm}2}/N + 8 \pi \frak{a}_0N^{\kappa}\sum_{p\in P_H^c} \Big[  b^*_pb_p +\frac12 b^*_p b^*_{-p} + \frac12 b_p b_{-p} \Big]  \\
		 &\hspace{0.4cm} + \frac{8\pi \mathfrak{a}_0N^\kappa}{\sqrt N} \sum_{\substack{ p\in P_H^c, q\in P_L:\\ p+q\neq 0 }} \big[b^*_{q+p} a^*_{-p}a_q+\text{h.c.}\big]  +\cH_N  + \cE_{\cJ_N} = \cJ_N^{\text{eff}}+\cE_{\cJ_N},
		 \end{split}\end{equation}
where the self-adjoint operator $ \cE_{\cJ_N} = \cE_{\cJ_N}^{(0)}+ \cE_{\cJ_N}^{(2)} +\cE_{\cJ_N}^{(3)}+\cE_{\cJ_N}^{(4)} + \cE_{\cJ_N}^{(5)}+ \cE_{\cJ_N}^{(6)}$ satisfies 
		\[\begin{split}
		\pm e^A\cE_{\cJ_N}e^{-A} & \leq  C(N^{-\beta/2}+ N^{\kappa+\mu} )\cK + CN^{\kappa+\mu}\cV_N  + CN^{\mu}\cN_+ +CN^{\alpha+2\kappa}(1+ N^{\alpha+\beta/2-1}). 
		\end{split}\]
This is precisely the bound \eqref{eq:propJNerrorbnd} and thus finishes the proof of the proposition.
\qed

%%%%%%%%%%%%%%%%%%%%%%%%%%%%%%%%%%%%%%%%%%%%%%%%%%%%%%%%%%%%%%%%%%%
%%%%%%%%%%%%%%%%%%%%%%%%%%%%%%%%%%%%%%%%%%%%%%%%%%%%%%%%%%%%%%%%%%%
%%%%%%%%%%%%%%%%%%%%%%%%%%%%%%%%%%%%%%%%%%%%%%%%%%%%%%%%%%%%%%%%%%%
%%%%%%%%%%%%%%%%%%%%%%%%%%%%%%%%%%%%%%%%%%%%%%%%%%%%%%%%%%%%%%%%%%%
%%%%%%%%%%%%%%%%%%%%%%%%%%%%%%%%%%%%%%%%%%%%%%%%%%%%%%%%%%%%%%%%%%%


\begin{thebibliography}{55}

%\bibitem{BDS}
%N.~{Benedikter}, G.~{de Oliveira} and B.~{Schlein}. {Quantitative derivation of the Gross-Pitaevskii equation}. \emph{Comm. Pure Appl. Math.} (2014).

\bibitem{BFKT}
T. Balaban, J. Feldman, H. Kn\"orrer, E. Trubowitz. Complex Bosonic Many-Body Models:
Overview of the Small Field Parabolic Flow. {\it Ann. Henri Poincar\'e } {\bf 18} (2017), 2873--2903.

\bibitem{BBCS0}
C. Boccato, C. Brennecke, S. Cenatiempo, B. Schlein. Complete Bose-Einstein condensation in the Gross-Pitaevskii regime. {\it Comm. Math. Phys.} {\bf 359} (2018), no. 3, 975--1026.

\bibitem{BBCS1}
C. Boccato, C. Brennecke, S. Cenatiempo, B. Schlein. The excitation spectrum of Bose gases interacting through singular potentials.  Preprint arXiv:1704.04819. To appear on {\it J. Eur. Math. Soc.} 

\bibitem{BBCS2}
C. Boccato, C. Brennecke, S. Cenatiempo, B. Schlein. Bogoliubov Theory in the Gross-Pitaevskii limit. \emph{Acta Mathematica} \textbf{222} (2019), no. 2, 219–335.

\bibitem{BBCS}
C. Boccato, C. Brennecke, S. Cenatiempo, B. Schlein. Optimal Rate for Bose-Einstein Condensation in the Gross-Pitaevskii Regime. \emph{Comm. Math. Phys.} (2019), doi:10.1007/s00220-019-03555-9. Preprint arXiv:1812.03086.

\bibitem{bog}
N. N. Bogoliubov. On the theory of superfluidity.
{\it Izv. Akad. Nauk. USSR} {\bf 11} (1947), 77. Engl. Transl. {\it J. Phys. (USSR)} {\bf 11} (1947), 23. 

\bibitem{BCS}
C. Brennecke, M. Caporaletti, B. Schlein. Excitation spectrum for Bose gases beyond the Gross-Pitaevskii regime. In preparation. 

\bibitem{BS}
C. Brennecke, B. Schlein. Gross-Pitaevskii dynamics for Bose-Einstein condensates. \emph{Analysis \& PDE} \textbf{12} (2019), no. 6, 1513–1596.

\bibitem{BFS}
B. Brietzke, S. Fournais, J.P. Solovej. A Simple 2nd Order Lower Bound to the Energy of Dilute Bose Gases. \emph{Comm. Math. Phys.} \textbf{376} (2020), 323--351.

%\bibitem{BriS}
%B. Brietzke. On the Second Order Correction to the Ground State Energy of the Dilute Bose Gas. PhD %thesis (2017).
% B. Brietzke, J.P. Solovey. The Second Order Correction to the Ground State Energy of the Dilute Bose Gas. Preprint (available in Brietzke PhD thesis)



%\bibitem{DN}
%J.~ Derezi\'nski, M.~Napi\'orkowski. Excitation Spectrum of Interacting Bosons in the Mean-Field Infinite-%Volume Limit. {\it Annales Henri Poincar\'e} {\bf 
%15} (2014), 2409-2439. 

%\bibitem{Dy}
%F.J. Dyson. Ground-State Energy of a Hard-Sphere Gas. {\it Phys. Rev.} {\bf 106} (1957), 20--26.
%

%\bibitem{ESY0} 
%L.~{Erd\H{o}s}, B.~{Schlein} and H.-T.~{Yau}.
%Derivation of the Gross-Pitaevskii hierarchy for the dynamics of Bose-Einstein condensate. {\it Comm. Pure  Appl. Math.} {\bf 59} (2006), no. 12, 1659--1741.

%\bibitem{ESY}
%L. Erd\H os, B. Schlein, H.-T. Yau. Ground-state energy of a low-density Bose gas: a second order upper %bound. {\it Phys. Rev. A} {\bf 78} (2008), 053627.
%

%\bibitem{ESY2}
%L.~{Erd{\H{o}}s}, B.~{Schlein}, H.-T.~{Yau}.
%\newblock Derivation of the {G}ross-{P}itaevskii equation for the dynamics of
%  {B}ose-{E}instein condensate,
%\newblock \emph{Ann. of Math.} \textbf{172} (2010), no. 1, 291--370.


%\bibitem{EY}
%L.~{Erd{\H{o}}s} and H.-T.~{Yau}. {Derivation of the
%  nonlinear {S}chr\"{o}dinger equation from a many-body {C}oulomb system}. \emph{Adv.
%  Theor. Math. Phys.} \textbf{5} (2001),  no. 6, 1169--1205.

\bibitem{F}
S. Fournais. Length scales for BEC in the dilute Bose gas. Preprint arXiv:2011.00309. 

\bibitem{FS}
S. Fournais, J.P. Solovej. The energy of dilute Bose gases. \emph{Ann. of Math.} \textbf{192}, No. 3 (2020), 893--976.

%\bibitem{GiuS}
%A. Giuliani, R. Seiringer. The ground state energy of the weakly interacting Bose gas at high density. {\it J. %Stat. Phys.} {\bf 135} (2009), 915.
%
%\bibitem{GS}
%P. Grech, R. Seiringer. The excitation spectrum for weakly interacting bosons in a trap. {\it Comm. Math. %Phys.} {\bf 322} (2013), no. 2, 559-591.
%
%\bibitem{LNS} M.~{Lewin}, P.~T.~{Nam} and B.~{Schlein}. 
%{Fluctuations around Hartree states in the mean-field regime}. %Preprint arXiv:1307.0665.
%
%\bibitem{Lan}
%L.D. Landau. Theory of the superfluidity of Helium II.
%{\it Phys. Rev.} {\bf 60} (1941), 356–-358.
%
%\bibitem{LNR1} M. Lewin, P.~T.~Nam, N. Rougerie. Derivation of Hartree's theory for generic
%mean-field Bose gases. {\it Adv. Math.} {\bf 254} (2014), pp. 570-621. 

%\bibitem{LNR2} M. Lewin, P.~T.~Nam, N. Rougerie. The mean-field approximation and the 
%non-linear Schr\"odinger functional for trapped  {B}ose gases. {\it Trans. Amer. Math. Soc.} {\bf 368} 
%(2016), no. 9, 6131-6157. 

\bibitem{H} C. Hainzl. Another proof of BEC in the GP-limit. Preprint arXiv:2011.09450. 

\bibitem{HY} K. Huang, C. N. Yang. Quantum-Mechanical Many-Body Problem with Hard-Sphere Interaction. \emph{Phys. Rev.} \textbf{105}, 3, (1957), 767 --757.

\bibitem{LHY} T. D. Lee, K. Huang, C. N. Yang. Eigenvalues and Eigenfunctions of a Bose System of Hard Spheres and Its Low-Temperature Properties. \emph{Phys. Rev.} \textbf{106} (1957), 6, 1135--1145.

\bibitem{LY} T. D. Lee, C. N. Yang. Many-Body Problem in Quantum Mechanics and Quantum Statistical Mechanics. \emph{Phys. Rev.} \textbf{105} (1957), 1119 -- 1120.

\bibitem{LNSS} M.~Lewin, P.~T.~{Nam}, S.~{Serfaty}, J.P. {Solovej}. Bogoliubov spectrum of interacting Bose gases.  \newblock \emph{Comm. Pure Appl. Math.} \textbf{68} (2014), 3, 413 -- 471. 

\bibitem{LS1}
E.~H.~Lieb and R.~Seiringer.
\newblock Proof of {B}ose-{E}instein condensation for dilute trapped gases.
\newblock \emph{Phys. Rev. Lett.} \textbf{88} (2002), 170409.

\bibitem{LS2}
E.~H.~Lieb and R.~Seiringer.
\newblock Derivation of the Gross-Pitaevskii equation for rotating Bose gases.
\newblock \emph{Comm. Math. Phys. } \textbf{264}:2 (2006), 505-537.

\bibitem{LSSY}
E.~H.~Lieb, R.~Seiringer, J. P. Solovej and J.~Yngvason. \textit{
The Mathematics of the Bose Gas and its Condensation}. Series: Oberwolfach Seminars. 
Birkhäuser Verlag, 2005.


\bibitem{LSY}
E.~H.~Lieb, R.~Seiringer, and J.~Yngvason.
\newblock Bosons in a trap: A rigorous derivation of the {G}ross-{P}itaevskii
  energy functional. \newblock \emph{Phys. Rev. A} \textbf{61} (2000), 043602.

%\bibitem{LSo}
%E.~H.~Lieb, J. P. Solovej. Ground state energy of the one-component charged Bose gas. {\it Comm. Math. Phys.} {\bf 217} (2001), 127--163. Errata: {\it Comm. Math. Phys.} {\bf 225} (2002), 219-221.

%\bibitem{LSo2}
%E.~H.~Lieb, J. P. Solovej. Ground state energy of the two-component charged Bose gas. {\it Comm. Math. %Phys.} {\bf 252} (2004), 485--534.

%\bibitem{LY} 
%E.~H.~Lieb, J.~Yngvason. Ground State Energy of the low density Bose Gas. {\it Phys. Rev. Lett.} {\bf 80} %(1998), 2504–2507. 

\bibitem{NNRT} 
P.~T.~{Nam}, M. Napi{\'o}rkowski, J. Ricaud, A. Triay. Optimal rate of condensation for trapped bosons in the Gross--Pitaevskii regime. Preprint arXiv:2001.04364. 

\bibitem{NRS}
P.~T.~{Nam}, N. ~{Rougerie}, R.~Seiringer. Ground states of large bosonic systems: The Gross-Pitaevskii limit revisited. 
\emph{Analysis and PDE.} {\bf 9} (2016), no. 2, 459--485

%\bibitem{NRS1}
%M. Napi{\'o}rkowski, R. Reuvers, J. P. Solovej. The Bogoliubov free energy functional I. Existence of %minimizers and phase diagrams. Preprint arxiv:1511.05935.
%
%\bibitem{NRS2}
%M. Napi{\'o}rkowski, R. Reuvers, J. P. Solovej. The Bogoliubov free energy functional II. The dilute limit.  %Preprint arxiv:1511.05953.
%
%
%
%\bibitem{P1}
%A. Pizzo. Bose particles in a box I. A convergent expansion of the ground state of a three-modes %Bogoliubov Hamiltonian in the mean field limiting regime. Preprint arxiv:1511.07022.
%
%\bibitem{P2}
%A. Pizzo. Bose particles in a box II. A convergent expansion of the ground state of the Bogoliubov %Hamiltonian in the mean field limiting regime. Preprint arxiv:1511.07025.
%
%\bibitem{P3}
%A. Pizzo. Bose particles in a box III. A convergent expansion of the ground state of the Hamiltonian in 
%the mean field limiting regime. Preprint arxiv:1511.07026.
%

\bibitem{Sei}
R. Seiringer. The Excitation Spectrum for Weakly Interacting Bosons. {\it Comm. Math. Phys.} {\bf 
306} (2011), 565–-578. 

%\bibitem{So}
%J. P. Solovej. Upper bounds to the ground state energies of the one- and two-component charged Bose %gases. {\it Comm. Math. Phys.} {\bf 266} (2006), no. 3, 797-818.
%
%\bibitem{YY}
%H.-T. Yau, J. Yin. The second order upper bound for the ground state energy of a Bose gas. {\it J. Stat. %Phys.} {\bf 136} (2009), no. 3, 453--503.

\end{thebibliography}
\end{document}